\documentclass[a4paper]{article}

\usepackage{amsmath}
\usepackage{url}
\usepackage{graphicx}
\usepackage{a4wide}
\usepackage{calc}
\usepackage{color}
\usepackage{hyperref}
\newcommand{\ig}[2][scale=1.0]{\includegraphics[#1]{#2}}

\usepackage{tikz}
\usetikzlibrary{decorations.markings,fit,backgrounds}

\newcommand{\etal}{et al.~}

\newcommand{\ie}{i.e.~}
\newcommand{\eg}{e.g.~}
\newcommand{\wloge}{without loss of generality,}
\newcommand{\mc}{\mathcal}
\newcommand{\poly}{\mathrm{poly}}
\newcommand{\nca}{\mathrm{nca}}
\newcommand{\dist}{\mathrm{dist}}
\newcommand{\ds}[1]{\gamma(#1)}
\newcommand{\dsg}[2]{\ds{#1 \mid #2}}
\newcommand{\cds}[1]{\gamma_{c}(#1)}

\newcommand{\cdsa}[2]{\cds{#1\; \dag\; #2}}

\newcommand{\eds}[1]{\rho(#1)}
\newcommand{\mm}[1]{\nu(#1)}
\newcommand{\HT}{\ensuremath{\widehat{H}}}

\DeclareSymbolFont{AMSb}{U}{msb}{m}{n}
\DeclareSymbolFontAlphabet{\mathbb}{AMSb}

\newtheorem{theorem}{Theorem}[section]
\newtheorem{proposition}[theorem]{Proposition}

\newtheorem{lemma}[theorem]{Lemma}
\newtheorem{corollary}[theorem]{Corollary}

\newtheorem{definition}[theorem]{Definition}
\newtheorem{example2}[theorem]{Example}

\newenvironment{proof}{\begin{genproof}}{\end{genproof}}

\newenvironment{genproof}[1][]{\begin{trivlist}\item \textbf{Proof#1:} }{\nolinebreak\end{trivlist}}

\newcounter{cclaim}[theorem]
\renewcommand{\thecclaim}{\thetheorem.\arabic{cclaim}}
\newenvironment{cclaim}{\refstepcounter{cclaim}\par\medskip\noindent\textbf{Claim~\thecclaim}}{}
\newenvironment{cproof}{\par\smallskip\noindent\textbf{Proof:}\ }{\nolinebreak\qquad\nolinebreak$\triangle${}\par\medskip}

\newcommand{\ccase}[1]{\medskip\noindent\textbf{Case #1:} }
\newcommand{\csubcase}[1]{\smallskip\noindent\textbf{Case #1:} }

\newcommand{\joinTime}{O(nm^{3/2})}
\newcommand{\stripTime}{O(n^{2}m^{3/2})}

\newcommand{\problemDS}{\textsc{Dominating Set}}
\newcommand{\problemCDS}{\textsc{Connected Dominating Set}}
\newcommand{\problemCEDS}{\textsc{Connected Edge Dominating Set}}
\newcommand{\problemEDS}{\textsc{Edge Dominating Set}}
\newcommand{\problemCVC}{\textsc{Connected Vertex Cover}}
\newcommand{\problemVC}{\textsc{Vertex Cover}}
\newcommand{\problemWDS}{\textsc{Weighted Dominating Set}}
\newcommand{\problemWCDS}{\textsc{Weighted Connected Dominating Set}}
\newcommand{\problemSAT}{\textsc{3-SAT}}
\newcommand{\problemMCC}{\textsc{Multicolored Clique}}
\newcommand{\Wh}[1][1]{\textnormal{W[#1]}}
\newcommand{\NPh}{\textnormal{NP}}

\newcommand{\defparproblem}[4]{

\bigskip
\noindent\fbox{
	\begin{minipage}{.96\linewidth}
	\textsc{#1}
	
	\smallskip
	\noindent\begin{tabular}{@{}l@{ }l}
	\emph{Input:} & \begin{minipage}[t]{\linewidth-\widthof{Parameter:\ \ \ \ }} #3\end{minipage}\\[.3pt]
	\emph{Parameter:} & \begin{minipage}[t]{\linewidth-\widthof{Parameter:\ \ \ \ }} #2\end{minipage}\\[.3pt]
	\emph{Question:} & \begin{minipage}[t]{\linewidth-\widthof{Parameter:\ \ \ \ }} #4\end{minipage}
	\end{tabular}
	\end{minipage}
}

\medskip
}

\begin{document}
\title{Domination when the Stars Are Out\footnote{An extended abstract of this paper appeared in the proceedings of ICALP 2011, LNCS vol.~6755, Springer, Berlin, 2011, pp.~462--473. This work is supported by the People Programme (Marie Curie Actions) of the European Union's Seventh Framework Programme (FP7/2007-2013) under REA grant agreement number 631163.11, by the Israel Science Foundation (grant no.\ 551145/14), by the European Research Council under ERC Starting Grant 306465 (BeyondWorstCase), and by DFG grant MN 59/1-1.}}
\author{%
Danny Hermelin
\and Matthias Mnich
\and Erik Jan van Leeuwen
\and\\ Gerhard J. Woeginger
}
\date{\today}
\maketitle

\begin{abstract}
  We algorithmize the structural characterization for claw-free graphs by Chudnovsky and Seymour.
  Building on this result, we show that \textsc{Dominating Set} on claw-free graphs is (i) fixed-parameter tractable and (ii) even possesses a polynomial kernel.
  To complement these results, we establish that \textsc{Dominating Set} is unlikely to be fixed-parameter tractable on the slightly larger class of graphs that exclude $K_{1,4}$ as an induced subgraph ($K_{1,4}$-free graphs).
  We show that our algorithmization can also be used to show that the related \textsc{Connected Dominating Set} problem is fixed-parameter tractable on claw-free graphs.
  To complement that result, we show that \textsc{Connected Dominating Set} is unlikely to have a polynomial kernel on claw-free graphs and is unlikely to be fixed-parameter tractable on $K_{1,4}$-free graphs.
  Combined, our results provide a dichotomy for \textsc{Dominating Set} and \textsc{Connected Dominating Set} on $K_{1,\ell}$-free graphs and show that the problem is fixed-parameter tractable if and only if $\ell \leq 3$.
\end{abstract}

\section{Introduction}
\label{sec:intro}
\textsc{(Connected) Dominating Set} is the problem to determine whether a given graph $G$ has a (connected) dominating set of size at most $k$.
A subset $D \subseteq V(G)$ is a \emph{dominating set} if every vertex in $G$ is either contained in $D$ or adjacent to some vertex in $D$, and $D$ is a \emph{connected dominating set} if $D$ is a dominating set and $G[D]$ (the graph induced by $D$) is connected. 
Dominating sets play a prominent role in both algorithmics and combinatorics (see \eg\cite{HaynesHS1998a,HaynesHS1998b}).
Since the \textsc{(Connected) Dominating Set} problem is hard in its decision~\cite{Karp1972,GareyJ1979}, approximation~\cite{Feige1998,GuhaK1998}, and parameterized versions~\cite{DowneyF1992}, research has focused on finding graph classes for which the problem becomes tractable.

In this paper, we focus on the class of claw-free graphs.
A graph is \emph{claw-free} if no vertex has three pairwise nonadjacent neighbors, \ie if it does not contain $K_{1,3}$ as an induced subgraph.
The class of claw-free graphs contains several well-studied graph classes, including line graphs, unit interval graphs, complements of triangle-free graphs, and graphs of several polyhedra and polytopes.
Throughout the years, this graph class attracted much interest, and is by now the subject of hundreds of mathematical research papers; for an overview we refer to the survey by Faudree~\etal\cite{FaudreeFR1997}.

In the context of algorithms, most research on claw-free graphs has focussed on the \textsc{Independent Set} problem.
Building on Edmond's classical polynomial-time algorithm~\cite{Edmonds1965} for \textsc{Independent Set} on line graphs (better known as \textsc{Matching}), Sbihi~\cite{Sbihi1980} and Minty~\cite{Minty1980} (the latter corrected by Nakamura and Tamura~\cite{NakamuraT2001}) already gave polynomial-time algorithms for \textsc{Independent Set} on claw-free graphs over 30 years ago. Recently, significantly faster algorithms have been discovered that follow a similar approach~\cite{FaenzaOS2014,NobiliS2015}.

In contrast to the \textsc{Independent Set} problem, however, \textsc{Dominating Set} is $\mathsf{NP}$-complete on claw-free graphs.
In fact, \textsc{Dominating Set} is $\mathsf{NP}$-complete even on line graphs~\cite{YannakakisG1980}.
Nevertheless, Fernau~\cite{Fernau2006} recently showed that \textsc{Dominating Set} on line graphs (also known as \textsc{Edge Dominating Set}) has an $f(k) \cdot n^{O(1)}$ time algorithm, where $k$ is the size of the solution, meaning that this problem is \emph{fixed-parameter tractable} (see~\eg\cite{cygan-book,DowneyF1999,DowneyFellows2013,FlumGrohe2006}).
Moreover, Fernau~\cite{Fernau2006} showed that any instance of \textsc{Dominating Set} on line graphs can be reduced in polynomial time to have $O(k^{2})$ vertices, which implies that the problem admits a \emph{polynomial kernel} (see~\eg\cite{cygan-book,DowneyF1999,DowneyFellows2013,FlumGrohe2006}).
Both results were recently slightly improved~\cite{XiaoKP2013,IwaideN2015}.
It has been left an open question, however, whether such  algorithms also exist for \textsc{Dominating Set} on claw-free graphs.

In a wider picture, claw-free graphs are a member of the more general family of graphs that exclude $K_{1,\ell}$ as an induced subgraph for $\ell \in \mathbb{N}$, \ie \emph{$\ell$-claw-free graphs}.
These graphs generalize many important classes of geometric intersection graphs; for example, unit square graphs are $K_{1,5}$-free, and unit disk graphs are $K_{1,6}$-free.
Marx showed that \textsc{Dominating Set} is $\mathsf{W}[1]$-hard on unit square graphs when parameterized by the solution size $k$, implying it is $\mathsf{W}[1]$-hard on $K_{1,5}$-free graphs~\cite{Marx2006}, which makes it unlikely that the problem is fixed-parameter tractable (see~\eg\cite{cygan-book,DowneyF1999,DowneyFellows2013,FlumGrohe2006}).
However, the problem becomes trivial on $K_{1,2}$-free graphs, since these graphs are just disjoint unions of cliques.
Hence, the computational and parameterized complexity of \textsc{Dominating Set} has been left open on $K_{1,3}$-free (claw-free) and $K_{1,4}$-free graphs.

In this paper, we resolve the question whether \textsc{Dominating Set} is fixed-parameter tractable on $K_{1,3}$-free and/or $K_{1,4}$-free graphs.
We answer this same question also for the related \textsc{Connected Dominating Set} problem.
Additionally, we determine the $\ell$-claw-free graphs on which these problems admit a polynomial kernel.
These results combined completely settle the parameterized complexity of both problems on $\ell$-claw-free graphs. 

\paragraph{Our Results.}
In the first part of the paper, we present our main contribution: an algorithmic version of a recent, highly nontrivial structural decomposition theorem for claw-free graphs by Chudnovsky and Seymour.
This decomposition shows that every claw-free graph can be built by applying certain gluing operations to certain atomic structures.
The proof that such a decomposition exists is contained in a sequence of seven papers by Chudnovsky and Seymour~\cite{ChudnovskyS2007-1,ChudnovskyS2008-2,ChudnovskyS2008-3,ChudnovskyS2008-4,ChudnovskyS2008-5,ChudnovskyS2010-6,ChudnovskyS2012-7}; an accessible summary of the proof can be found in the announcement of their results~\cite{ChudnovskyS2005}.
The original proof of their decomposition theorem, however, does not directly imply an algorithm to find the decomposition. 

In order to obtain our algorithmic decomposition theorem for claw-free graphs, we first develop polynomial-time algorithms that undo the aforementioned gluing operations (Sect.~\ref{sec:joins}).
These algorithms could be of independent interest, as the structures that these algorithms find are not specific to claw-free graphs.
Then, we develop polynomial-time algorithms that recognize several of the atomic structures that make up claw-free graphs (Sect.~\ref{sec:recog}).
Finally, we make several changes to the proof of the original decomposition theorem by Chudnovsky and Seymour that simplify it and make it easier to algorithmize (Sect.~\ref{sec:structure}).
Finally, we put all these pieces together to give an algorithmic decomposition theorem for claw-free graphs that runs in $\stripTime$ time (Sect.~\ref{sec:decomp}).

The structural result for claw-free graphs that is implied by our algorithmic decomposition theorem is inspired by a claim of Chudnovsky and Seymour~\cite[Claim~3.1]{ChudnovskyS2005} in the announcement of their structural result for claw-free graphs. However, as far as we are aware, this claim is not explicitly proved in their final work~\cite{ChudnovskyS2007-1,ChudnovskyS2008-2,ChudnovskyS2008-3,ChudnovskyS2008-4,ChudnovskyS2008-5,ChudnovskyS2010-6,ChudnovskyS2012-7}.
Our algorithmic decomposition theorem can be seen as a variant or an interpretation of~\cite[Claim~3.1]{ChudnovskyS2005} with the hindsight of knowing the final work of Chudnovsky and Seymour, as well as an explicit proof and an algorithm to actually find the decomposition.

\medskip
In the second part of the paper, we apply our algorithmic decomposition theorem for claw-free graphs to establish the following:

\begin{itemize}
  \item \textsc{Dominating Set} on claw-free graphs is fixed-parameter tractable when parameterized by the solution size.
  To be precise, we show that we can decide the existence of a dominating set of size at most $k$ in $9^{k} \cdot O(n^{5})$ time when the graph is claw-free (Sect.~\ref{sec:ds}).
  \item \textsc{Connected Dominating Set} on claw-free graphs is fixed-parameter tractable when parameterized by the solution size. 
  To be precise, we show that we can decide the existence of a connected dominating set of size at most $k$ in $36^{k} \cdot n^{O(1)}$ time when the graph is claw-free (Sect.~\ref{sec:cds}).
  This resolves an open question by Misra~\etal\cite{MisraPRS2010}.
  \item \textsc{Dominating Set} on claw-free graphs has a polynomial kernel when parameterized by the solution size.
  To be precise, we show that given a claw-free graph $G$ and an integer $k$, we can output a graph $G'$ with $O(k^{3})$ vertices and an integer $k'$ such that $G$ has a dominating set of size $k$ if and only if $G'$ has a dominating set of size $k'$; the presented algorithm runs in $O(n^{5})$ time (Sect.~\ref{sec:kernel}).
\end{itemize}

In the third part of the paper (Sect.~\ref{sec:hard}), we complement the above results and show that:

\begin{itemize}
  \item \textsc{Dominating Set} and \textsc{Connected Dominating Set} are $\mathsf{W}[1]$-hard on $K_{1,4}$-free graphs. 
  \item \textsc{Connected Dominating Set} has no polynomial kernel on claw-free graphs (even on line graphs), unless the polynomial hierarchy collapses to the third level. 
  \item \textsc{Dominating Set} and \textsc{Connected Dominating Set} on claw-free graphs (even on line graphs) cannot have an algorithm that runs in $2^{o(k)} n^{O(1)}$ time, where $k$ is the size of the solution, unless the Exponential Time Hypothesis fails.
  \item The weighted variants of \textsc{Dominating Set} and \textsc{Connected Dominating Set} are $\mathsf{W}[1]$-hard on claw-free graphs (even on cobipartite graphs).
\end{itemize}
Combined, this sequence of results completely determines the computational and parameterized complexity of \textsc{Dominating Set} and \textsc{Connected Dominating Set} in $K_{1,\ell}$-free graphs for all $\ell$.

\paragraph{Further Applications and Outlook.}
Since the publication of the extended abstract of our paper, several results have appeared that apply our algorithmic decomposition for claw-free graphs to other problems.
Martin~\etal\cite{MartinPvL2018,MartinPvL2018-arxiv} showed that the \textsc{Disconnected Cut} problem (find a vertex cut that itself induces at least two connected components) is polynomial-time solvable on claw-free graphs.
Golovach~\etal\cite{GolovachPvL2015} showed that the \textsc{Induced Disjoint Paths} problem (find disjoint and `nonadjacent' paths between each of $k$ given pairs of terminal vertices) is fixed-parameter tractable on claw-free graphs.
Hermelin~\etal\cite{HermelinMvL2014} showed that a generalization of \textsc{Independent Set} called \textsc{Induced $H$-Matching} (find $k$ independent copies of a fixed graph $H$ in a given graph) is fixed-parameter tractable and has a polynomial kernel on claw-free graphs.
Both results rely heavily on the algorithmic decomposition theorem for claw-free graphs that we develop in this paper. 
In fact, the latter paper presents a stronger version of the algorithmic decomposition theorem than we present here; that version, however, would not be possible without the structural results and fundamental algorithms that we develop in Sect.~\ref{sec:joins} and~\ref{sec:structure} of this paper.

The basic idea behind all three application of our technique (\textsc{(Connected) Dominating Set} in this paper, and \textsc{Induced Disjoint Paths} and \textsc{Induced $H$-Matching} in follow-up work~\cite{GolovachPvL2015,HermelinMvL2014}) builds upon the constructs and notions that we develop in this paper.
However, significant technical effort is needed to tailor the decomposition theorem and the way we apply it to the problem at hand. 
We believe, however, that using the algorithmic techniques developed in this paper there is a strong potential to develop algorithms for many other problems on claw-free graphs.

\paragraph{Related Work.}
Since the announcement of the Chudnovsky-Seymour decomposition theorem for claw-free graphs, several results appeared that use it to attack problems on the class of claw-free graphs, and particularly on a subclass called quasi-line graphs (see \eg\cite{ChudnovskyO2007,ChudnovskyF2010,GalluccioGV2010}).
However, most of these results are structural and give no algorithms to find the decomposition.

Recently and independent of our work, two papers appeared that find an algorithmic decomposition theorem for claw-free graphs. First, the decomposition theorem obtained by King and Reed~\cite{King2009,KingR2015} is based on the work of Chudnovsky and Seymour and has subtle differences when compared to ours, but his algorithmic methods are completely different. Intuitively, King and Reed find the individual parts (strips) of the decomposition by considering the local structure around $5$-wheels, while we find the structures (joins) that hold the different parts together and only then recognize the parts. Since the joins are central to the decomposition by Chudnovsky and Seymour, we think that our approach is more flexible. Moreover, the approach by King and Reed is heavily geared towards an application to the coloring problem and therefore does not need to identify several parts in their full generality.

Second, the decomposition theorem obtained by Faenza~\etal\cite{FaenzaOS2014} is not based on the proof of the decomposition theorem by Chudnovsky and Seymour, and thus is substantially different compared to ours.
Faenza~\etal use their decomposition to obtain a faster polynomial-time algorithm for \textsc{Weighted Independent Set} on claw-free graphs (recently improved by Nobili and Sassano~\cite{NobiliS2015}).
Although their decomposition can potentially be used to show that \textsc{Dominating Set} on claw-free graphs is fixed-parameter tractable~\cite{Stauffer2011}, it is not clear whether a polynomial kernel would follow as well.
Conversely, the ideas behind our work can be adapted to give a polynomial-time algorithm for \textsc{Weighted Independent Set} on claw-free graphs, albeit with a worse run time than the algorithms by Faenza~\etal and Nobili and Sassano.

A recent paper by Cygan~\etal\cite{CyganPPPW2012} also proves that \textsc{Dominating Set} is fixed-parameter tractable on claw-free graphs.
Their algorithm does not use the decomposition theorem by Chudnovsky and Seymour, and runs in time $2^{O(k^{2})} \cdot n^{O(1)}$, compared to our $9^{k} \cdot n^{O(1)}$-time algorithm.
Moreover, their methods do not extend to a polynomial kernel.
Cygan~\etal\cite{CyganPPPW2012} also show that \textsc{Dominating Set} and \textsc{Connected Dominating Set} are $\mathsf{W}[2]$-complete on $K_{1,4}$-free graphs, compared to our (slightly weaker) $\mathsf{W}[1]$-hardness results.

From the perspective of approximation algorithms, it is known that \textsc{Dominating Set} has a polynomial-time $(\ell-1)$-approximation algorithm on $\ell$-claw-free graphs, and \textsc{Connected Dominating Set} has a polynomial-time $2(\ell-1)$-approximation algorithm~\cite{MaratheBHRR1995}.
These approximation factors, however, are not known to be tight~\cite{ChlebikC2008}.

Finally, it is important to note that in our paper we exclude $K_{1,\ell}$ as an \emph{induced} subgraph, and not as a subgraph. In the latter case, \textsc{Dominating Set} is already known to be fixed-parameter tractable~\cite{PhilipRS2009}.

\section*{Part I -- Algorithmic Decomposition Theorem for Claw-Free Graphs}

\section{Definitions}
The definitions given in this section are the same as those by Chudnovsky and Seymour~\cite{ChudnovskyS2008-5}.
Although the reader could find them there, in order to be self-contained, we repeat those definitions that we need for our algorithmic structure theorem.
Moreover, we sometimes need additional properties that can be somewhat hidden in Chudnovsky and Seymour~\cite{ChudnovskyS2008-5}, and it will be convenient to highlight them here explicitly.
We also highlight the most important definitions for our algorithmic structure theorem in explicit definitions.

\subsection{Trigraphs and Basic Definitions}
In this paper, we work with a generalization of the notion of a graph, called a trigraph.
Roughly speaking, a trigraph is a graph with a distinguished subset of edges that are called \emph{semi-edges} and that form a matching in the graph.
Intuitively, semi-edges capture the idea of `fuzzy' edges in the graph that are both there and not there at the same time.
This intuition will become more clear later, when we consider the notion of a thickening. 
\begin{definition}[Trigraph]
  A \emph{trigraph} $G$ has a finite set of vertices $V(G)$ and an adjacency function $\phi_{G}: V(G) \times V(G) \rightarrow \{-1, 0, 1\}$, such that
  \begin{itemize}
    \item $\phi_{G}(v,v) = 0$ for all $v \in V(G)$,
    \item $\phi_{G}(u,v) = \phi_{G}(v,u)$ for all $u,v \in V(G)$,
    \item at most one of $\phi_{G}(u,v), \phi_{G}(u,w) = 0$ for all distinct $u,v,w \in V(G)$.
  \end{itemize}
\end{definition}
The pairs $u,v \in V(G)$ with $u\not= v$ for which $\phi_{G}(u,v) \in \{1,-1\}$ are the regular edges and nonedges, respectively, whereas those for which $\phi_{G}(u,v) = 0$ constitute the semi-edges.
Observe that, by definition, the set of semi-edges is indeed a matching.
We also note that a trigraph without semi-edges is just a normal graph.
Conversely, given a graph $G'$, it can be \emph{regarded as a trigraph} $G$ by setting $V(G) = V(G')$ and $\phi_{G}(u,v) = 1$ if $u,v$ are adjacent in $G'$ and $\phi_{G}(u,v)=-1$ otherwise.

At first sight, it would appear that trigraphs only form a distraction and complicate our graph-theoretic framework.
However, it turns out that trigraphs make it much easier to describe claw-free graphs, and in particular, the decomposition theorems for them.

We now extend some classical notions of graph theory (such as adjacency) to trigraphs.
In a trigraph $G$ with adjacency function $\phi_{G}$, we say that distinct $u,v \in V(G)$ are \emph{strongly adjacent} if $\phi_{G}(u,v) = 1$, \emph{semiadjacent} if $\phi_{G}(u,v) = 0$, and \emph{strongly antiadjacent} if $\phi_{G}(u,v) = -1$.
We then say that $u,v$ are \emph{adjacent} if $u,v$ are either strongly adjacent or semiadjacent, and $u,v$ are \emph{antiadjacent} if $u,v$ are either strongly antiadjacent or semiadjacent.
Moreover, a vertex $u$ is a \emph{neighbor} of a vertex~$v$ if $u,v$ are adjacent and $u,v$ are \emph{strong neighbors} if $u$ and $v$ are strongly adjacent.
The notions of \emph{antineighbor} and \emph{strong antineighbor} are similarly defined.
We denote by $N(v)$ the set of neighbors of a vertex $v$ and define $N[v] = N(v) \cup \{v\}$ as the closed neighborhood of $v$. 
Similarly, we define $N(X) = \{v \mid v \in N(x) \setminus X, x \in X\}$ and $N[X] = N(X) \cup X$.

Given disjoint sets $A,B \subseteq V(G)$, we say that $A$ is \emph{complete to $B$} or \emph{$B$-complete} if every vertex of $A$ is adjacent to every vertex of $B$.
We say that $A$ is \emph{strongly complete to $B$} or \emph{strongly $B$-complete} if every vertex of $A$ is strongly adjacent to every vertex of $B$.
If we say that $a \in V(G)$ is (strongly) complete to $B$, we mean that $\{a\}$ is (strongly) complete to $B$.
The notions of \emph{anticomplete} and \emph{strongly anticomplete} are defined similarly.

A set $C \subseteq V(G)$ is a \emph{clique} if every pair of vertices of $C$ is adjacent, and a \emph{strong clique} if every pair of vertices of $C$ is strongly adjacent.
A vertex $v$ of a trigraph is \emph{simplicial} if $N[v]$ is a clique, and \emph{strongly simplicial} if $N[v]$ is a strong clique. 

\begin{definition}[Stable set, $\alpha(G)$]
  A set $I \subseteq V(G)$ is \emph{stable} or \emph{independent} if every pair of vertices of~$I$ is antiadjacent, and \emph{strongly stable} or \emph{strongly independent} if every pair of vertices of $I$ is strongly antiadjacent.
  Let $\alpha(G)$ denote the size of a largest subset of $V(G)$ that is stable.
  Sometimes, $\alpha(G)$ will be called the \emph{stability number} or \emph{independence number} of $G$.
\end{definition}

The following notion of a thickening is crucial to turn semi-edges into `normal edges'.
See also Fig.~\ref{fig:thick}.

\begin{definition}[Thickening]
  A trigraph $G$ is a \emph{thickening} of a trigraph $G'$ if there is a set $\mathcal{W} = \{W_{v} \subseteq V(G) \mid v \in V(G')\}$ such that each set $W_{v}$ is nonempty, such that
  \begin{itemize}
    \item $W_{u} \cap W_{v} = \emptyset$ for all distinct $u,v \in V(G')$ and $\bigcup_{v \in V(G')} W_{v} = V(G)$,
    \item $W_{v}$ is a strong clique in $G$ for each $v \in V(G')$,
    \item if $u,v$ are strongly adjacent in $G'$, then $W_{u}$ is strongly complete to $W_{v}$ in $G$,
    \item if $u,v$ are strongly antiadjacent in $G'$, then $W_{u}$ is strongly anticomplete to $W_{v}$ in $G$,
    \item if $u,v$ are semiadjacent in $G'$, then $W_{u}$ is neither strongly complete nor strongly anticomplete to $W_{v}$ in $G$.
  \end{itemize}
  We sometimes talk about the thickening $\mathcal{W}$ of $G'$ to $G$.
\end{definition}
Note that a trigraph is always a thickening of itself.
Also note that if $G$ is a thickening of $G'$ and $G'$ is a thickening of $G''$, then $G$ is also a thickening of $G''$.
Finally, note that if a graph $G$ is a thickening of trigraph $G'$ and $u,v$ are semiadjacent in $G'$, then $|W_{u}|+|W_{v}| \geq 3$.

\tikzset{tn/.style = {circle, fill=green},tt/.style = {circle, fill=black, inner sep=0.8mm}}

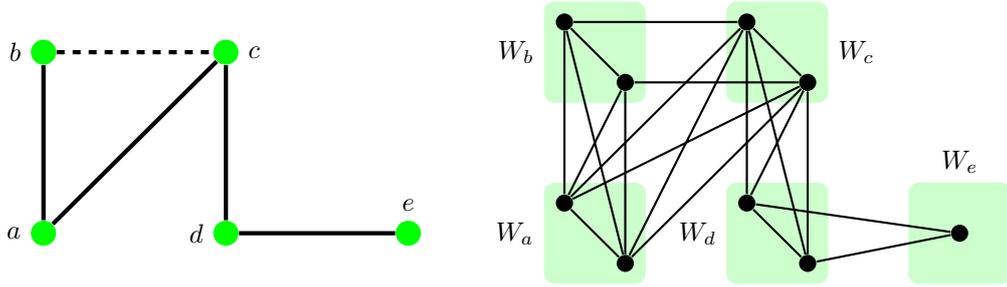
\begin{figure}[t]
  \centering
  \begin{tikzpicture}[scale=0.8]
    \node[tn,label=left:{$a$}] (v1) at (0,0) {};
    \node[tn,label=left:{$b$}] (v2) at (0,3) {};
    \node[tn,label=left:{$d$}] (v3) at (3,0) {};
    \node[tn,label=right:{$c$}] (v4) at (3,3) {};
    \node[tn,label=above:{$e$}] (v5) at (6,0) {};
    \node at (7,0) {};
    \node at (-1,-1) {};
    \draw [-, ultra thick] (v2) -- (v1) -- (v4) -- (v3) -- (v5);
    \draw [-, ultra thick, dashed] (v2) -- (v4);
  \end{tikzpicture}
  \begin{tikzpicture}[scale=0.8]
    \node at (-1,-1) {};
    \node[tt] (v1a) at (-0.5,0.5) {};
    \node[tt] (v1b) at (0.5,-0.5) {};
    \node[tt] (v2a) at (-0.5,3.5) {};
    \node[tt] (v2b) at (0.5,2.5) {};
    \node[tt] (v3a) at (2.5,0.5) {};
    \node[tt] (v3b) at (3.5,-0.5) {};
    \node[tt] (v4a) at (2.5,3.5) {};
    \node[tt] (v4b) at (3.5,2.5) {};
    \node[] (sv5a) at (5.5,0.5) {};
    \node[] (sv5b) at (6.5,-0.5) {};
    \node[tt] (v5) at (6,0) {};

    \draw [-, thick] (v1a) -- (v1b) -- (v2b) -- (v2a) -- (v1a) -- (v2b);
    \draw [-, thick] (v1b) -- (v2a);
    \draw [-, thick] (v4a) -- (v1b) -- (v4b) -- (v4a) -- (v1a) -- (v4b);
    \draw [-, thick] (v2a) -- (v4a);
    \draw [-, thick] (v2b) -- (v4b);
    \draw [-, thick] (v3a) -- (v3b) -- (v4b) -- (v3a) -- (v4a) -- (v3b);
    \draw [-, thick] (v3a) -- (v5) -- (v3b);

    \begin{pgfonlayer}{background}
      \node [rectangle, rounded corners, fill=green!20, fit=(v1a) (v1b), inner sep=1.5mm,label=left:{$W_a$}] {};
      \node [rectangle, rounded corners, fill=green!20, fit=(v2a) (v2b), inner sep=1.5mm,label=left:{$W_b$}] {};
      \node [rectangle, rounded corners, fill=green!20, fit=(v3a) (v3b), inner sep=1.5mm,label=left:{$W_d$}] {};
      \node [rectangle, rounded corners, fill=green!20, fit=(v4a) (v4b), inner sep=1.5mm,label=right:{$W_c$}] {};
      \node [rectangle, rounded corners, fill=green!20, fit=(sv5a) (sv5b), inner sep=1.5mm,label=above:{$W_e$}] {};
    \end{pgfonlayer}
  \end{tikzpicture}
  \caption{The left panel shows a trigraph $G'$.
    Here the dotted line represents a semi-edge, the thick lines represent edges, and no line between two vertices represents a nonedge.
    The right panel shows a graph $G$ that is a thickening $\mc{W}$ of $G'$.
    Observe that since $b$ is strongly adjacent to $a$, all vertices of $W_b$ are adjacent to all vertices of $W_a$.
    Moreover, since $b$ is strongly antiadjacent to $d$, all vertices of $W_b$ are antiadjacent to all vertices of $W_d$.
    Finally, since $b$ is semiadjacent to $c$, $W_b$ is neither complete nor anticomplete to $W_c$.
    Note that $(W_b,W_c)$ is a W-join in $G$; in fact, it is a proper W-join.
    The sets $W_a,W_d,W_e$ each form a homogeneous clique or twin set in~$G$.}
\label{fig:thick}
\end{figure}

For any $X \subseteq V(G)$, $G[X]$ is the trigraph \emph{induced by} $X$, which is the trigraph with vertex set $X$ and adjacency determined by the restriction of $\phi_{G}$ to $X \times X$.
We say that $H$ is an \emph{induced subtrigraph} of~$G$ if $H$ is isomorphic to $G[X]$ for some $X \subseteq V(G)$. We define $G \setminus X = G[V(G)\setminus X]$.
Isomorphism between trigraphs is defined as expected.

\begin{definition}[claw, claw-free]
  A \emph{claw} is a trigraph with four vertices $a,b,c,d$, such that $\{b,c,d\}$ is stable and complete to $a$.
  Then $a$ is the \emph{center} of the claw.
  If no induced subtrigraph of a trigraph~$G$ is isomorphic to a claw, then $G$ is \emph{claw-free}.
\end{definition}

\subsection{Twins and Joins}
\label{sec:def:joins}
In this subsection, we describe the basic graph structures used in the decomposition theorem of claw-free (tri)graphs.
Later, in Sect.~\ref{sec:joins}, we will describe algorithms to actually find these structures in a graph.

Throughout this section, let $G$ be a trigraph.
A strong clique $X$ of $G$ is \emph{homogeneous} if every vertex in $G\setminus X$ is either strongly complete or strongly anticomplete to $X$.
This is equivalent to requiring that for $x,x' \in X$, $x$ and $x'$ have the same closed neighborhoods and all their neighbors are strong neighbors. 

\begin{definition}[twins]
  A homogeneous strong clique is sometimes also called a \emph{twin set}.
  Then $G$ admits \emph{twins} if $G$ has a homogeneous strong clique of size $2$.
\end{definition}
See also Fig.~\ref{fig:thick}.

A pair of strong cliques $(A,B)$ is \emph{homogeneous} if every vertex $v \in V(G)\setminus(A\cup B)$ is either strongly complete or strongly anticomplete to $A$, and is either strongly complete or strongly anticomplete to~$B$.
In other words, the set $V(G) \setminus (A \cup B)$ can be partitioned into four sets: those vertices strongly adjacent to $A$ and strongly antiadjacent to $B$, those strongly adjacent to $B$ and strongly antiadjacent to $A$, those strongly adjacent to both $A$ and $B$, and those strongly antiadjacent to both $A$ and $B$.
Observe that if $G$ is a graph and $V(G) \setminus (A \cup B) = \emptyset$, then $G$ is a cobipartite graph (\ie the complement of a bipartite graph).

\begin{definition}[(proper) W-join]
  A homogeneous pair of cliques $(A,B)$ is a \emph{W-join} if $A$ is neither strongly complete nor strongly anticomplete to $B$, and $A$ or $B$ has size at least $2$.
  A W-join is \emph{proper} if no member of $A$ is strongly complete or strongly anticomplete to $B$ and no member of $B$ is strongly complete or strongly anticomplete to $A$.
\end{definition}

It is important to observe that if a trigraph $G$ is a thickening of a trigraph $G'$ with $\{W_{v} \mid v \in V(G')\}$ and $|W_{v}| > 1$ for some $v \in V(G')$, then $W_{v}$ is a twin set if $v$ is not semiadjacent to any vertex of $G'$, and $(W_{u},W_{v})$ is a W-join if $v$ is semiadjacent to a vertex $u \in V(G')$.
The latter observation follows from the fact that the semi-edges form a matching in $G'$.
See also Fig.~\ref{fig:thick}.

A partition $(V_{1}, V_{2})$ of $V(G)$ is a \emph{$0$-join} if $V_{1}$ is strongly anticomplete to $V_{2}$ and $V_{1},V_{2} \not= \emptyset$.
If $G$ does not admit a $0$-join, then $G$ is called \emph{connected}.
Observe that if $G$ is a graph rather than a trigraph, then this corresponds to the standard notion of connectedness.

A partition $(V_{1}, V_{2})$ of $V(G)$ is a \emph{$1$-join} if there are sets $A_{1} \subseteq V_{1}$, $A_{2} \subseteq V_{2}$ such that
\begin{itemize}
  \item $A_{1} \cup A_{2}$ is a strong clique,
  \item $V_{1}\setminus A_{1}$ is strongly anticomplete to $V_{2}$, and $V_{2}\setminus A_{2}$ is strongly anticomplete to $V_{1}$,
  \item $A_{i}, V_{i}\setminus A_{i} \not= \emptyset$ for $i=1,2$.
\end{itemize}

\begin{figure}[t]
  \centering
  \begin{tikzpicture}[scale=0.95, every node/.style={transform shape}]
\node[tt] (a11) at (1,1.5) {};
\node[] (a12) at (1,2.0) {};
\node[tt] (a13) at (1,2.5) {};
\node (fa1) at (0,4.5) {};
\node (fa2) at (1,4.5) {};
\node (fa3) at (0,0) {};
\node (fa4) at (1,0) {};
\draw [-, thick] (a13) -- (fa1);
\draw [-, thick] (a13) -- (fa2);
\draw [-, thick] (a11) -- (fa3);
\draw [-, thick] (a11) -- (fa4);
\draw [-, thick] (a11) -- (a13);

\node[tt] (b11) at (2.5,1.5) {};
\node[] (b12) at (2.5,2.0) {};
\node[tt] (b13) at (2.5,2.5) {};
\node (fb1) at (2.5,4.5) {};
\node (fb2) at (3.5,4.5) {};
\node (fb3) at (2.5,0) {};
\node (fb4) at (3.5,0) {};
\draw [-, thick] (b13) -- (fb1);
\draw [-, thick] (b13) -- (fb2);
\draw [-, thick] (b11) -- (fb3);
\draw [-, thick] (b11) -- (b13);

\draw [-, thick] (b11) -- (a11) -- (b13) -- (a13) -- (b11);

\begin{pgfonlayer}{background}
\node [rectangle, rounded corners, fill=green!20,fit=(fa1) (fa2) (fa3) (fa4), inner sep=1.5mm] {};
\node [rectangle, rounded corners, fill=green,fit=(a11) (a12) (a13), inner sep=1.5mm] {};
\node [rectangle, rounded corners, fill=red!20,fit=(fb1) (fb2) (fb3) (fb4), inner sep=1.5mm] {};
\node [rectangle, rounded corners, fill=red,fit=(b11) (b12) (b13), inner sep=1.5mm] {};
\end{pgfonlayer}
\node at (4,0) {};
\end{tikzpicture}
\begin{tikzpicture}[scale=0.95, every node/.style={transform shape}]
\node[tt] (a11) at (1,0.5) {};
\node[tt] (a13) at (1,1.5) {};
\node (fa1) at (0,4.5) {};
\node (fa2) at (1,4.5) {};
\node (fa3) at (0,0) {};
\node (fa4) at (1,0) {};
\node (fa5) at (0,1.5) {};
\draw [-, thick] (a11) -- (fa3);
\draw [-, thick] (a11) -- (fa4);
\draw [-, thick] (a13) -- (fa5);
\draw [-, thick] (a11) -- (a13);

\node[tt] (b11) at (2.5,0.5) {};
\node[tt] (b13) at (2.5,1.5) {};
\node (fb1) at (2.5,4.5) {};
\node (fb2) at (3.5,4.5) {};
\node (fb3) at (2.5,0) {};
\node (fb4) at (3.5,0) {};
\node (fb5) at (3.5,1.5) {};
\draw [-, thick] (b11) -- (fb3);
\draw [-, thick] (b11) -- (fb4);
\draw [-, thick] (b13) -- (fb5);
\draw [-, thick] (b11) -- (b13);

\draw [-, thick] (b11) -- (a11) -- (b13) -- (a13) -- (b11);

\node[tt] (a21) at (1,3.0) {};
\node[tt] (a23) at (1,4.0) {};
\node (fa25) at (0,3.0) {};
\draw [-, thick] (a23) -- (fa1);
\draw [-, thick] (a23) -- (fa2);
\draw [-, thick] (a21) -- (fa25);
\draw [-, thick] (a21) -- (a23);

\node[tt] (b21) at (2.5,3.0) {};
\node[tt] (b23) at (2.5,4.0) {};
\node (fb25) at (3.5,3.0) {};
\draw [-, thick] (b23) -- (fb1);
\draw [-, thick] (b23) -- (fb2);
\draw [-, thick] (b21) -- (fb25);
\draw [-, thick] (b21) -- (b23);

\draw [-, thick] (b21) -- (a21) -- (b23) -- (a23) -- (b21);

\begin{pgfonlayer}{background}
\node [rectangle, rounded corners, fill=green!20,fit=(fa1) (fa2) (fa3) (fa4), inner sep=1.5mm] {};
\node [rectangle, rounded corners, fill=green,fit=(a11) (a13), inner sep=1.5mm] {};
\node [rectangle, rounded corners, fill=red!20,fit=(fb1) (fb2) (fb3) (fb4), inner sep=1.5mm] {};
\node [rectangle, rounded corners, fill=red,fit=(b11) (b13), inner sep=1.5mm] {};
\node [rectangle, rounded corners, fill=green,fit=(a21) (a23), inner sep=1.5mm] {};
\node [rectangle, rounded corners, fill=red,fit=(b21) (b23), inner sep=1.5mm] {};
\end{pgfonlayer}
\node at (4,0) {};
\end{tikzpicture}
\begin{tikzpicture}[scale=0.95, every node/.style={transform shape}]
\node[tt] (a11) at (1,0.5) {};
\node[tt] (a13) at (1,1.5) {};
\node (fa1) at (0,4.5) {};
\node (fa2) at (1,4.5) {};
\node (fa3) at (0,0) {};
\node (fa4) at (1,0) {};
\node (fa5) at (0,1.5) {};
\draw [-, thick] (a11) -- (fa3);
\draw [-, thick] (a11) -- (fa4);
\draw [-, thick] (a13) -- (fa5);
\draw [-, thick] (a11) -- (a13);

\node[tt] (b11) at (3.5,0.5) {};
\node[tt] (b13) at (3.5,1.5) {};
\node (fb1) at (3.5,4.5) {};
\node (fb2) at (4.5,4.5) {};
\node (fb3) at (3.5,0) {};
\node (fb4) at (4.5,0) {};
\node (fb5) at (4.5,1.5) {};
\draw [-, thick] (b11) -- (fb3);
\draw [-, thick] (b11) -- (fb4);
\draw [-, thick] (b13) -- (fb5);
\draw [-, thick] (b11) -- (b13);

\draw [-, thick] (b11) -- (a11) -- (b13) -- (a13) -- (b11);

\node[tt] (a21) at (1,3.0) {};
\node[tt] (a23) at (1,4.0) {};
\node (fa25) at (0,3.0) {};
\draw [-, thick] (a23) -- (fa1);
\draw [-, thick] (a23) -- (fa2);
\draw [-, thick] (a21) -- (fa25);
\draw [-, thick] (a21) -- (a23);

\node[tt] (b21) at (3.5,3.0) {};
\node[tt] (b23) at (3.5,4.0) {};
\node (fb25) at (4.5,3.0) {};
\draw [-, thick] (b23) -- (fb1);
\draw [-, thick] (b23) -- (fb2);
\draw [-, thick] (b21) -- (fb25);
\draw [-, thick] (b21) -- (b23);

\draw [-, thick] (b21) -- (a21) -- (b23) -- (a23) -- (b21);

\node[tt] (v01) at (2.25,2.5) {};
\node[tt] (v02) at (2.25,2.0) {};

\draw [-, thick] (b21) -- (v01) -- (b23) -- (v02) -- (b21);
\draw [-, thick] (a21) -- (v01) -- (a23) -- (v02) -- (a21);
\draw [-, thick] (b11) -- (v01) -- (b13) -- (v02) -- (b11);
\draw [-, thick] (a11) -- (v01) -- (a13) -- (v02) -- (a11);
\draw [-, thick] (v01) -- (v02);

\begin{pgfonlayer}{background}
\node [rectangle, rounded corners, fill=green!20,fit=(fa1) (fa2) (fa3) (fa4), inner sep=1.5mm] {};
\node [rectangle, rounded corners, fill=green,fit=(a11) (a13), inner sep=1.5mm] {};
\node [rectangle, rounded corners, fill=red!20,fit=(fb1) (fb2) (fb3) (fb4), inner sep=1.5mm] {};
\node [rectangle, rounded corners, fill=red,fit=(b11) (b13), inner sep=1.5mm] {};
\node [rectangle, rounded corners, fill=green,fit=(a21) (a23), inner sep=1.5mm] {};
\node [rectangle, rounded corners, fill=red,fit=(b21) (b23), inner sep=1.5mm] {};
\node [circle, fill=purple,fit=(v01) (v02), inner sep=1.5mm] {};
\end{pgfonlayer}
\end{tikzpicture}
  \caption{The left panel shows a graph that admits a (pseudo-)$1$-join.
    The sets $A_1,A_2$ that define the (pseudo-)$1$-join are highlighted in dark green and dark red respectively.
    The sets $V_1,V_2$ are green and red respectively.
    The middle panel shows a graph that admits a (pseudo-)$2$-join.
    The sets $A_1,A_2,B_1,B_2$ that define the \mbox{(pseudo-)}$2$-join are highlighted in dark green and dark red respectively. 
    The sets $V_1,V_2$ are green and red respectively.
    The right panel shows a graph that admits a generalized $2$-join and a pseudo-$2$-join.
    The sets $A_1,A_2,B_1,B_2$ that define the generalized $2$-join and pseudo-$2$-join are highlighted in dark green and dark red respectively.
    The sets $V_0,V_1,V_2$ are purple, green, and red respectively.}
\label{fig:joins}
\end{figure}
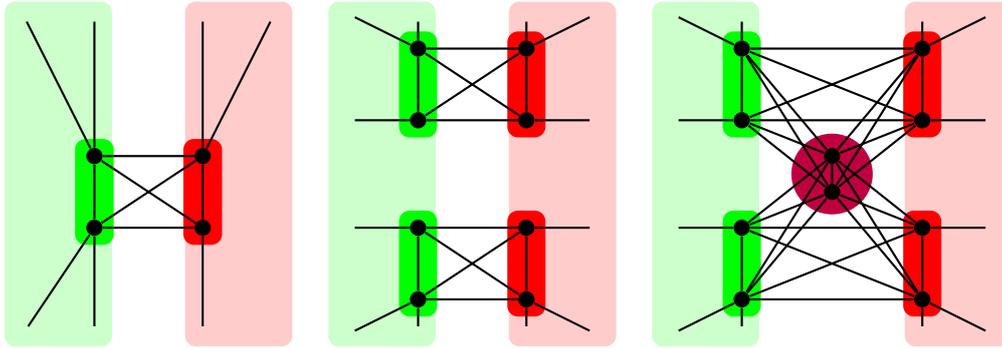

A partition $(V_{1}, V_{2})$ of $V(G)$ is a \emph{pseudo-$1$-join} if there are sets $A_{1} \subseteq V_{1}$, $A_{2} \subseteq V_{2}$ such that
\begin{itemize}
  \item $A_{1} \cup A_{2}$ is a strong clique,
  \item $V_{1}\setminus A_{1}$ is strongly anticomplete to $V_{2}$, and $V_{2}\setminus A_{2}$ is strongly anticomplete to $V_{1}$,
  \item neither $V_{1}$ nor $V_{2}$ is a strong stable set.
\end{itemize}
See the left panel of Fig.~\ref{fig:joins} for an example.

A trigraph admitting a $1$-join and no $0$-join admits a pseudo-$1$-join.
We will denote (pseudo-) 1-joins either by the partition $(V_{1},V_{2})$ of the vertices, or the `connecting subsets' $(A_{1},A_{2})$.
Note that $V_{1},V_{2}$ can be easily determined if we just know $A_{1}, A_{2}$, and vice versa.

A partition $(V_{0}, V_{1}, V_{2})$ of $V(G)$ forms a \emph{generalized $2$-join} if for $i=1,2$ there are disjoint sets $A_{i}, B_{i} \subseteq V_{i}$ such that
\begin{itemize}
  \item $V_{0}$, $V_{1}$, and $V_{2}$ are pairwise strongly anticomplete, except that $V_{0} \cup A_{1} \cup A_{2}$ and $V_{0} \cup B_{1} \cup B_{2}$ form a strong clique,
  \item $A_{i}, B_{i}, V_{i}\setminus (A_{i} \cup B_{i}) \not= \emptyset$ for $i=1,2$.
\end{itemize}
If $V_{0} = \emptyset$, then we call it a \emph{$2$-join}.

A partition $(V_{0}, V_{1}, V_{2})$ of $V(G)$ forms a \emph{pseudo-$2$-join} if for $i=1,2$ there are disjoint sets $A_{i}, B_{i} \subseteq V_{i}$ such that
\begin{itemize}
  \item $V_{0}$, $V_{1}$, and $V_{2}$ are pairwise strongly anticomplete, except that $V_{0} \cup A_{1} \cup A_{2}$ and $V_{0} \cup B_{1} \cup B_{2}$ form a strong clique,
  \item neither $V_{1}$ nor $V_{2}$ is a strong stable set.
\end{itemize}
See the middle and right panels of Fig.~\ref{fig:joins} for examples.

A graph admitting a (generalized) $2$-join and no $0$-join admits a pseudo-$2$-join.
We will use both the notation $(V_{0},V_{1},V_{2})$ or $(V_{1},V_{2})$ and the notation $(A_{1},A_{2},B_{1},B_{2})$ for (generalized/pseudo-) 2-joins.

A \emph{three-cliqued trigraph} $(G;A,B,C)$ consists of a trigraph $G$ and three pairwise disjoint strong cliques $A,B,C$ in $G$ such that $V(G) = A \cup B \cup C$.

A \emph{hex-join} of two three-cliqued trigraphs $(G_{1};A_{1},B_{1},C_{1})$ and $(G_{2};A_{2},B_{2},C_{2})$ is the three-cliqued trigraph $(G;A,B,C)$, where $A = A_{1} \cup A_{2}$, $B = B_{1} \cup B_{2}$, $C = C_{1} \cup C_{2}$, and $G$ is the trigraph with vertex set $V(G) = V(G_{1}) \cup V(G_{2})$ and adjacency as follows:
\begin{itemize}
  \item $G[V(G_{1})] = G_{1}$ and $G[V(G_{2})] = G_{2}$,
  \item $A_{1}$ is strongly complete to $V(G_{2}) \setminus B_{2}$, $B_{1}$ is strongly complete to $V(G_{2}) \setminus C_{2}$, and $C_{1}$ is strongly complete to $V(G_{2}) \setminus A_{2}$,
  \item the pairs $(A_{1},B_{2})$, $(B_{1},C_{2})$, $(C_{1},A_{2})$ are strongly anticomplete.
\end{itemize}
A trigraph $G$ admits a \emph{hex-join} if $G$ has three strong cliques $A,B,C$ such that $(G;A,B,C)$ is a three-cliqued trigraph expressible as a hex-join.

We often implicitly use the following important observation, which is immediate from the above definitions.
\begin{proposition}
\label{prp:join-thicken}
  Let $G$ be a trigraph that is a thickening of a trigraph $G'$.
  If $G'$ admits twins, a (proper) W-join, a $0$-join, a (pseudo-) $1$-join, a (generalized/pseudo-) $2$-join, or a hex-join, then $G$ admits twins, a (proper) W-join, a $0$-join, a (pseudo-) $1$-join, a (generalized/pseudo-) $2$-join, or a hex-join, respectively.
\end{proposition}

\label{def:indecomp}
The following definition is crucial to the structure theorem of Chudnovsky and Seymour for claw-free graphs.
A trigraph $G$ is \emph{indecomposable} if $G$ neither admits twins, nor a W-join, nor a $0$-join, nor a $1$-join, nor a generalized $2$-join, nor a hex-join.

\subsection{Strips and Stripes}
The notions in this section are of central importance to our algorithmic structure theorem, as well as to its applications.

\begin{definition}[strip-graph]
  A \emph{strip-graph} $H$ consists of disjoint finite sets $V(H)$ and $E(H)$, and an incidence relation between $V(H)$ and $E(H)$ (\ie a subset of $V(H) \times E(H)$). 
\end{definition}
Given a strip-graph $H$, for any $F \in E(H)$, let $\overline{F}$ denote the set of $h \in V(H)$ incident with $F$.
Note that the definition of strip-graphs is close to the definition of hypergraphs, except that we allow multiple edges and empty edges here.

Let $G$ be a trigraph and let $Y \subseteq V(G)$.
Then a family $(X_{1},\ldots,X_{k})$ of subsets of $Y$ is a \emph{circus} in $Y$ if
\begin{itemize}
  \item for $1 \leq i \leq k$ and $x \in X_{i}$, the set of neighbors of $x$ in $Y \setminus X_{i}$ is a strong clique,
  \item for $1 \leq i < j \leq k$, $X_{i} \cap X_{j}$ is strongly anticomplete to $Y \setminus (X_{i} \cup X_{j})$,
  \item for $1 \leq h < i < j \leq k$, $X_{h} \cap X_{i} \cap X_{j} = \emptyset$.
\end{itemize}

\begin{definition}[strip-structure]
  A \emph{strip-structure} $(H,\eta)$ of a trigraph $G$ is a strip-graph $H$ with $E(H) \not= \emptyset$ and a function $\eta$ that takes one or two parameters, an $F \in E(H)$ or an $F \in E(H)$ and a $h \in \overline{F}$, satisfying:
  \begin{itemize}
    \item $\eta(F) \subseteq V(G)$ and $\eta(F,h) \subseteq \eta(F)$ for each $F \in E(H)$ and each $h \in \overline{F}$.
    \item The sets $\eta(F)$ ($F \in E(H)$) are nonempty, pairwise disjoint, and have union $V(G)$.
    \item For each $h \in V(H)$, the union of the sets $\eta(F,h)$ for all $F \in E(H)$ with $h \in \overline{F}$ is a strong clique of $G$.
    \item For all distinct $F_{1},F_{2} \in E(H)$, if $v_{1} \in \eta(F_{1})$ and $v_{2} \in \eta(F_{2})$ are adjacent in $G$, then there exists $h \in \overline{F_{1}} \cap \overline{F_{2}}$ such that $v_{1} \in \eta(F_{1},h)$ and $v_{2} \in \eta(F_{2},h)$.
    \item For each $F \in E(H)$, the family $\eta(F,h)$ ($h \in \overline{F}$) is a circus in $\eta(F)$.
  \end{itemize}
  To simplify notation, we define $\eta(h) = \bigcup_{F \mid h \in \overline{F}} \eta(F,h)$ for all $h \in V(H)$.
\end{definition}

\begin{definition}[strip]
  Let $(H,\eta)$ be a strip-structure of a trigraph $G$ and let $F \in E(H)$, where $\overline{F} = \{h_{1},\ldots,h_{k}\}$.
  Let $z_{1},\ldots,z_{k}$ be new vertices and let $J$ be the trigraph obtained from $G[\eta(F)]$ by adding $z_{1},\ldots,z_{k}$, and for each $i$ making $z_{i}$ strongly complete to $\eta(F,h_{i})$ and strongly anticomplete to $J \setminus \eta(F,h_{i})$. Then $(J, \{z_{1},\ldots,z_{k}\})$ is the \emph{strip} corresponding to $F$. 
\end{definition}
Note that $Z = \{z_1,\ldots,z_k\}$ is a strong stable set in $J$ and that $Z \cap V(G) = \emptyset$.

Observe that if $G$ is claw-free, stating that the family $\eta(F,h)$ ($h \in \overline{F}$) is a circus in $\eta(F)$ is equivalent to stating that the strip corresponding to $F$ is claw-free.

\begin{definition}[spot, stripe]
  A strip $(J,Z)$ is a \emph{spot} if $J$ consists of three vertices $v$, $z_{1}$, $z_{2}$ such that $v$ is strongly adjacent to $z_{1}$ and $z_{2}$, and $z_{1}$ is strongly antiadjacent to $z_{2}$, and $Z = \{z_{1},z_{2}\}$.

  A strip $(J,Z)$ is a \emph{stripe} if $J$ is a claw-free trigraph and $Z \subseteq V(J)$ is a set of strongly simplicial vertices, such that no vertex of $V(J) \setminus Z$ is adjacent to more than one vertex of $Z$.
\end{definition}

We say that a stripe $(J,Z)$ is a \emph{thickening} of a stripe $(J',Z')$ if $J$ is a thickening of $J'$ with sets $W_{v}$ ($v \in V(J')$) such that $|W_{z}| = 1$ for each $z \in Z'$ and $Z = \bigcup_{z \in Z'} W_{z}$.

The \emph{nullity} of a strip-structure $(H,\eta)$ is the number of pairs $(F, h)$ with $F \in E(H)$, $h \in \overline{F}$, and $\eta(F,h) = \emptyset$.

Given a strip-structure $(H,\eta)$ of a trigraph, we say that $F \in E(H)$ is \emph{purified} if either the sets $\eta(F,h)$ ($h \in \overline{F}$) are pairwise disjoint, or $\overline{F} = \{h_{1},h_{2}\}$, $|\eta(F)| = 1$, and $\eta(F,h_{1}) = \eta(F,h_{2}) = \eta(F)$.
A strip-structure $(H,\eta)$ is \emph{purified} if each $F \in E(H)$ is purified.

Observe that saying that $F \in E(H)$ is purified is equivalent to saying that the strip corresponding to $F$ is either a stripe or a spot.

\tikzset{vertex/.style={minimum size=2mm,circle,fill=black,draw,inner sep=0pt},
        decoration={markings,mark=at position .5 with
{\arrow[black,thick]{stealth}}}}

\begin{figure}[t]
 \centering
 \begin{tikzpicture}[scale=0.5]
%this is the claw-free graph
   \draw[fill=green,fill opacity=0.2] (0,-4.5) ellipse (3.2cm and 1.4cm);
   \draw[fill=green] (2,-4.5) ellipse (0.5cm and 1cm);
   \draw[fill=green] (-2,-4.5) ellipse (0.5cm and 1cm);
   \draw[fill=red, fill opacity=0.2] (-3,-2) ellipse (0.6cm and 0.6cm);
   \draw[fill=red] (-3,-2) ellipse (0.35cm and 0.35cm);
   \draw[fill=orange, fill opacity=0.2] (0,-2) ellipse (1.8cm and 0.6cm);
   \draw[fill=orange] (-1,-2) ellipse (0.35cm and 0.35cm);
   \draw[fill=orange] (1,-2) ellipse (0.35cm and 0.35cm);
   \draw[fill=purple, fill opacity=0.2] (3,-2) ellipse (0.6cm and 0.6cm);
   \draw[fill=purple] (3,-2) ellipse (0.35cm and 0.35cm);
   \draw[fill=blue, fill opacity=0.2] (-3,0) ellipse (0.6cm and 0.6cm);
   \draw[fill=blue] (-3,0) ellipse (0.35cm and 0.35cm);
   \draw[fill=gray,fill opacity=0.2] (-0.5,1.5) ellipse (1.2cm and 1.2cm);
   \draw[fill=gray] (-1,1) ellipse (0.45cm and 0.45cm);
   \draw[fill=gray] (0,1) ellipse (0.45cm and 0.45cm);
  \node (a) at (-3.95,-2){$a$};
  \node (b) at (-3.95,0){$b$};
  \node (c) at (-2.3,1.75){$c$};
  \node (d) at (0,-1){$d$};
 \node (g) at (3.95,-2){$e$};
  \node (h) at (0,-6.5){$f$};
   \node (6) at (-2,-4) [vertex]{};
   \node (7) at (2,-4) [vertex]{};
   \node (8) at (-2,-5) [vertex]{};
   \node (9) at (2,-5) [vertex]{};
   \draw[thick] (7)--(9);
   \draw[thick] (6)--(8);
   \node (10) at (-1,-4) [vertex]{};
   \node (11) at (-0.5,-5) [vertex]{};
   \node (12) at (0.5,-4) [vertex]{};
   \node (12a) at (0.3,-5) [vertex]{};
   \draw[thick] (6)--(10);
   \draw[thick] (6)--(11);
   \draw[thick] (8)--(11);
   \draw[thick] (11)--(12);
   \draw[thick] (10)--(11);
   \draw[thick] (10)--(12);
   \draw[thick] (12)--(7);
   \draw[thick] (12)--(9);
   \draw[thick] (12a)--(10);
   \draw[thick] (12a)--(12);
   \draw[thick] (12a) -- (7);
   \draw[thick] (12a) -- (9);
   \node (13) at (3,-2) [vertex]{};
   \node (14) at (-3,-2) [vertex]{};
   \draw[thick] (13)--(7);
   \draw[thick] (13)--(9);
   \draw[thick] (14)--(6);
   \draw[thick] (14)--(8);
   \node (15a) at (-1,-2) [vertex]{};
   \draw[thick] (15a)--(13);
   \draw[thick] (15a)--(14);
   \node (15b) at (1,-2) [vertex]{};
   \node (16) at (-1,1) [vertex]{};
   \node (17) at (0,1) [vertex]{};
   \draw[thick] (16)--(17);
   \draw[thick] (17)--(13);
   \draw[thick] (17)--(15b);
   \draw[thick] (16)--(15a);
   \draw[thick] (16)--(14);
   \node (19) at (-3,0) [vertex]{};
   \draw[thick] (19)--(14);
   \draw[thick] (19)--(16);
   \draw[thick] (13)--(14);
   \node (20) at (-1,2) [vertex]{};
   \node (21) at (0,2) [vertex]{};
   \draw[thick] (21)--(17);
   \draw[thick] (20)--(16);
   \draw[thick] (20)--(21);
   \draw[thick] (15a)--(19);
 \end{tikzpicture}
 \tikzset{every loop/.style={min distance=20mm,looseness=10}}
 \begin{tikzpicture}
%this is the strip graph
  \node (0) at (-3,2) [vertex]{};
  \node (1) at (0,2) [vertex]{};
  \node (2) at (-3,3.5) [vertex]{};
  \node (3) at (0,3.5) [vertex]{};
  \draw[thick,green] (0)--(1);
  \draw[thick,purple] (1)--(3);
  \draw[thick,red] (2)--(0);
  \draw[thick,orange] (2)--(3);
  \node (hh) at (-4,4.5){};
  \draw[thick,blue] (2) -- (hh);
  \draw[thick,gray] (2) to[bend right=-50] (3);
  \node[red] (a) at (-3.25,2.75){$a$};
  \node[blue] (b) at (-3.5,3.75){$b$};
  \node[gray] (c) at (-1.5,4.45){$c$};
  \node[orange] (d) at (-1.5,3.25){$d$};
  \node[purple] (g) at (0.25,2.75){$e$};
  \node[green] (h) at (-1.5,1.75){$f$};
  \node at (-3,1){};
 \end{tikzpicture}
 \qquad
 \begin{tikzpicture}[scale=0.5]
 %these are the examples of strips and spot
   \draw[fill=red, fill opacity=0.2] (0,0) ellipse (1cm and 1cm);
   \draw[fill=red] (0,0) ellipse (0.5cm and 0.5cm);
   \node (1) at (0,0) [vertex] {};
   \node (2) at (-2,0) [vertex] {};
   \node (3) at (2,0) [vertex] {};
   \draw[thick] (1)--(2);
   \draw[thick] (1)--(3);
   \node at (0,-1.5){spot};
   \node at (-2,-0.8){$Z$};
   \node at (2,-0.8){$Z$};
   %------
   \draw[fill=green, fill opacity=0.2] (0,-4.5) ellipse (2.5cm and 2cm);
   \draw[fill=green] (2,-4.5) ellipse (0.5cm and 1cm);
   \draw[fill=green] (-2,-4.5) ellipse (0.5cm and 1cm);
   \node (4) at (-3,-4.5) [vertex]{};
   \node (5) at (3,-4.5) [vertex]{};
   \node (6) at (-2,-4) [vertex]{};
   \node (7) at (2,-4) [vertex]{};
   \node (8) at (-2,-5) [vertex]{};
   \node (9) at (2,-5) [vertex]{};
   \draw[thick] (7)--(9);
   \draw[thick] (6)--(8);
   \draw[thick] (4)--(6);
   \draw[thick] (4)--(8);
   \draw[thick] (5)--(7);
   \draw[thick] (5)--(9);
   \node (10) at (-1,-4) [vertex]{};
   \node (11) at (-0.5,-5) [vertex]{};
   \node (12) at (0.5,-4) [vertex]{};
   \node (12a) at (0.3,-5) [vertex]{};
   \draw[thick] (6)--(10);
   \draw[thick] (6)--(11);
   \draw[thick] (8)--(11);
   \draw[thick] (11)--(12);
   \draw[thick] (10)--(11);
   \draw[thick] (10)--(12);
   \draw[thick] (12)--(7);
   \draw[thick] (12)--(9);
   \draw[thick] (12a)--(10);
   \draw[thick] (12a)--(12);
   \draw[thick] (12a) -- (7);
   \draw[thick] (12a) -- (9);
   \node at (0,-7){stripe};
   \node at (-3.2,-5.3){$Z$};
   \node at (3.0,-5.3){$Z$};
 \end{tikzpicture}
 \quad
 \quad
  \caption{This figure is inspired by~\cite[Fig.~1]{HermelinMvL2014}.
    The left panel shows a claw-free graph $G$.
    A strip-structure $(H,\eta)$ has been marked in the panel.
    The middle panel shows the strip-graph $H$; note that $E(H) = \{a,b,c,d,e,f\}$ and that $|\overline{b}|=1$.
    The colored ellipses in the left panel show $\eta(F)$ for each $F \in E(H)$.
    The darker ellipses show the sets $\eta(F,h)$ for each $F \in E(H)$ and $h \in \overline{F}$. 
    The right panel shows the two types of strips $(J,Z)$: spots and stripes.
    The colored ellipses show $\eta(F) = V(J) \setminus Z$ and the darker ellipses show $\eta(F,h)$ for $F \in E(H)$ and $h \in \overline{F}$.
    Note that strips $a$, $b$, and $e$ in the left panel are spots; these always look as pictured.
    Conversely, strips $c$, $d$, and $f$ are stripes and might look different depending on $\eta(F)$; the stripe in the right panel corresponds to $f$ (the stripes corresponding to $c$ and $d$ are not pictured).
    From the picture, it is clear that the set $Z \subseteq V(J)$ is not part of $G$.}
\label{fig:stripstructure}
\end{figure}
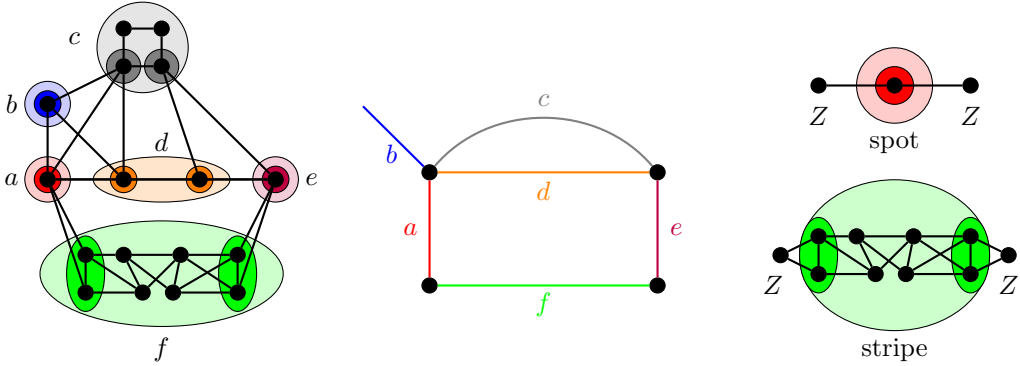

\subsection{Special Trigraphs}
In the definitions below, whenever adjacency between two vertices is not specified, they are strongly antiadjacent.
Moreover, if two vertices are said to be adjacent, they can be either strongly adjacent or semiadjacent, unless otherwise specified.
Similarly, if two vertices are said to be antiadjacent, they can be either strongly antiadjacent or semiadjacent, unless otherwise specified.

A \emph{line trigraph} $G$ of some graph $H$ is a trigraph where $V(G) = E(H)$ and $e,f \in E(H)$ are adjacent in $G$ if and only if $e$ and $f$ share an endpoint in $H$.
Moreover, $e,f$ are strongly adjacent if $e$ and $f$ share an endpoint of degree at least three.
The class of all line trigraphs is denoted \emph{$\mc{S}_{0}$}.
If $G$ is a graph (\ie $G$ has no semi-edges) and $G$ is the line trigraph of some graph $H$, then we call $G$ the \emph{line graph} of $H$, a classic notion in graph theory.
In this case, we call $H$ the \emph{pre-image} of $G$.
Note that the line graph of a graph $H$ is unique, while the pre-image of a line graph $G$ might not be unique (this happens only when $G$ is a triangle, in which case the pre-image is either a triangle or a claw).

The \emph{icosahedron} is the (planar) graph $G$ with $V(G) = \{v_{1},\ldots,v_{12}\}$ such that
\begin{itemize}
  \item for $i = 1,\ldots, 10$, $v_{i}$ is adjacent to $v_{i+1}$ and $v_{i+2}$ (indices modulo $10$),
  \item $v_{11}$ is adjacent to $v_{1},v_{3},v_{5},v_{7},v_{9}$,
  \item $v_{12}$ is adjacent to $v_{2},v_{4},v_{6},v_{8},v_{10}$.
\end{itemize}
This graph regarded as a trigraph is denoted by $G_{0}$.
Let $G_{1} = G_{0} \setminus \{v_{12}\}$.
Let $G_{2}$ be obtained from $G_{1} \setminus\{v_{10}\}$ by possibly making $v_{1}$ semiadjacent to $v_{4}$ or making $v_{6}$ semiadjacent to $v_{9}$.
The class of trigraphs denoted by \emph{$\mc{S}_{1}$} consists precisely of $G_{0}$, $G_{1}$, and the four possibilities for $G_{2}$.

\begin{definition}[XX-trigraph, $\mc{S}_{2}$]
  Let $G$ be the trigraph with $V(G) = \{v_{1},\ldots,v_{13}\}$ such that
  \begin{itemize}
    \item $v_{i}$ is adjacent to $v_{i+1}$ for $i = 1,\ldots,5$ and $v_{6}$ is adjacent to $v_{1}$; also $v_{i}$ is antiadjacent to $v_{j}$ for each $i = 1,\ldots,4$ and each $i+2 \leq j \leq 6$,
    \item $v_{7}$ is strongly adjacent to $v_{1}$ and $v_{2}$,
    \item $v_{8}$ is strongly adjacent to $v_{4}$, $v_{5}$, and possibly adjacent to $v_{7}$,
    \item $v_{9}$ is strongly adjacent to $v_{1}$, $v_{2}$, $v_{3}$, and $v_{6}$,
    \item $v_{10}$ is strongly adjacent to $v_{3}$, $v_{4}$, $v_{5}$, and $v_{6}$, and adjacent to $v_{9}$,
    \item $v_{11}$ is strongly adjacent to $v_{1}$, $v_{3}$, $v_{4}$, $v_{6}$, $v_{9}$, and $v_{10}$,
    \item $v_{12}$ is strongly adjacent to $v_{2}$, $v_{3}$, $v_{5}$, $v_{6}$, $v_{9}$, and $v_{10}$,
    \item $v_{13}$ is strongly adjacent to $v_{1}$, $v_{2}$, $v_{4}$, $v_{5}$, $v_{7}$, and $v_{8}$.
  \end{itemize}
  Then $G \setminus X$ for any $X \subseteq \{v_{7},v_{11},v_{12},v_{13}\}$ is an \emph{XX-trigraph}.
  The class of all XX-trigraphs is denoted~\emph{$\mc{S}_{2}$}.
\end{definition}

\begin{figure}[t]
  \begin{center}
    \begin{tikzpicture}[scale=1.5]
      \node (v1) at (1,1)[vertex]{};
      \node (v1l) at (0.75,1){$v_{1}$};
      \node (v2) at (2,2)[vertex]{};
      \node (v2l) at (2,2.25){$v_{2}$};
      \node (v3) at (3,2)[vertex]{};
      \node (v3l) at (3,2.25){$v_{3}$};
      \node (v4) at (4,1)[vertex]{};
      \node (v4l) at (4.25,1){$v_{4}$};
      \node (v5) at (3,0)[vertex]{};
      \node (v5l) at (3,-0.25){$v_{5}$};
      \node (v6) at (2,0)[vertex]{};
      \node (v6l) at (2,-0.25){$v_{6}$};
      \draw[thick](v1)--(v2);
      \draw[thick](v2)--(v3);
      \draw[thick](v3)--(v4);                  
      \draw[thick](v4)--(v5);
      \draw[thick](v5)--(v6);
      \draw[thick](v6)--(v1);
      \node (v7) at (1,2)[vertex]{};
      \node (v7l) at (0.75,2){$v_{7}$};
      \node (v8) at (4,0)[vertex]{};
      \node (v8l) at (4.25,0){$v_{8}$};
      \node (v9) at (2,1)[vertex]{};
      \node (v9l) at (1.8,1.2){$v_{9}$};
      \node (v10) at (3,1)[vertex]{};
      \node (v10l) at (3.25,1.2){$v_{10}$};
      \draw[thick](v7)--(v1);
      \draw[thick](v7)--(v2);
      \draw[thick](v8)--(v4);
      \draw[thick](v8)--(v5);
      \draw[thick](v9)--(v1);
      \draw[thick](v9)--(v2);
      \draw[thick](v9)--(v3);
      \draw[thick](v9)--(v6);
      \draw[thick](v10)--(v3);
      \draw[thick](v10)--(v4);
      \draw[thick](v10)--(v5);
      \draw[thick](v10)--(v6);
      \draw[thick](v10)--(v9);
    \end{tikzpicture}  
  \end{center}
  \caption{A claw-free XX-trigraph with $X = \{v_{11},v_{12},v_{13}\}$ and where $v_{7}$ and $v_{8}$ are strongly antiadjacent. A solid line illustrates strong adjacency and no line indicates strong antiadjacency. Note that there are no semiadjacencies per Proposition~\ref{prp:xx:claw-free}.}
\label{fig:xxtrigraphs}
\end{figure}
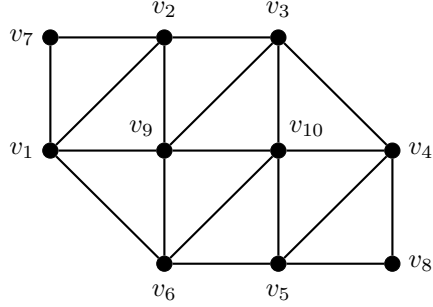
See Fig.~\ref{fig:xxtrigraphs} for an example of a member of $\mc{S}_{2}$.

\begin{proposition} \label{prp:xx:claw-free}
Let $G$ be the trigraph defined above on vertices $v_{1},\ldots,v_{13}$ and let $X \subseteq\{v_{11},v_{12},v_{13}\}$. If $G$ is claw-free, then the (anti)adjacencies between $v_{1},\ldots,v_{6}$ are strong.
\end{proposition}
\begin{proof}
First, suppose that:
\begin{itemize}
\item $v_{1}$ is semiadjacent to $v_{3}$. Then $v_{3}$ is strongly antiadjacent to $v_{6}$ by the definition of a trigraph. Hence, $v_{1},v_{3},v_{6},v_{7}$ forms a claw with center $v_{1}$, a contradiction.
\item $v_{1}$ is semiadjacent to $v_{4}$. Then $v_{1},v_{4},v_{8},v_{10}$ forms a claw with center $v_{4}$, a contradiction.
\item $v_{1}$ is semiadjacent to $v_{5}$. Then $v_{1},v_{5},v_{8},v_{10}$ forms a claw with center $v_{5}$, a contradiction.
\item $v_{2}$ is semiadjacent to $v_{4}$. Then $v_{2},v_{4},v_{8},v_{10}$ forms a claw with center $v_{4}$, a contradiction.
\item $v_{2}$ is semiadjacent to $v_{5}$. Then $v_{2},v_{5},v_{8},v_{10}$ forms a claw with center $v_{5}$, a contradiction.
\item $v_{2}$ is semiadjacent to $v_{6}$. Then $v_{3}$ is strongly antiadjacent to $v_{6}$ by the definition of a trigraph. Hence, $v_{2},v_{3},v_{6},v_{7}$ is a claw with center $v_{2}$, a contradiction.
\item $v_{3}$ is semiadjacent to $v_{5}$. Then $v_{3}$ is strongly antiadjacent to $v_{6}$ by the definition of a trigraph. Hence, $v_{3},v_{5},v_{6},v_{8}$ is a claw with center $v_{5}$, a contradiction.
\item $v_{3}$ is semiadjacent to $v_{6}$. Then $v_{6}$ is strongly antiadjacent to $v_{2}$ and to $v_{4}$ by the definition of a trigraph. Moreover, we already established that $v_{2}$ is strongly antiadjacent to $v_{4}$. Hence, $v_{2},v_{3},v_{4},v_{6}$ is a claw with center $v_{3}$, a contradiction.
\item $v_{4}$ is semiadjacent to $v_{6}$. Then $v_{3}$ is strongly antiadjacent to $v_{6}$ by the definition of a trigraph. Hence, $v_{3},v_{5},v_{6},v_{8}$ is a claw with center $v_{4}$, a contradiction.
\end{itemize}
Now it remains to remark that Hermelin \etal\cite[{Proposition~1}]{HermelinMvL2014} already showed that $v_{i}$ is strongly adjacent to $v_{i+1}$ for $i = 1,\ldots,5$ and $v_{6}$ is strongly adjacent to $v_{1}$.
\end{proof}

Consider the sphere $\mathbb{S}_{1}$ and a set $\mc{I} = \{I_{1}, \ldots, I_{\ell}\}$ of subsets of $\mathbb{S}_{1}$, such that no distinct $I_{i}, I_{j}$ share an endpoint and no three members of $\mc{I}$ have union $\mathbb{S}_{1}$.
Let $P$ be a finite subset of $\mathbb{S}_{1}$ and let $G$ be the trigraph with $V(G) = P$ such that distinct $u,v \in P$ are adjacent in $G$ if and only if $u,v \in I_{i}$ for some $i=1,\ldots,\ell$.
Moreover, $u,v$ are strongly adjacent if at least one of $u,v$ is in the interior of $I_{i}$.
Call such trigraphs \emph{circular interval trigraphs} and denote the class of all circular interval trigraphs by \emph{$\mc{S}_{3}$}.

Recall that the \emph{intersection graph} $G$ of a set $\mc{F}$ of subsets of some given universe is the graph where each vertex corresponds to a set in $\mc{F}$ and there is an edge between two vertices if and only if the corresponding sets of $\mc{F}$ intersect.
If $G$ is the intersection graph of $\mc{F}$, then we call $\mc{F}$ a \emph{model} for $G$. 

\begin{definition}[(Proper) Circular-arc graph]
  A graph is a \emph{circular-arc graph} if it is the intersection graph of some set of arcs of the sphere $\mathbb{S}_{1}$.
  A graph is a \emph{proper circular-arc graph} if it is the intersection graph of some set of arcs of the sphere $\mathbb{S}_{1}$ such that no arc of the set properly contains another.
\end{definition}
Note that this definition allows arcs to be the same.

\begin{proposition}
\label{prp:proper-circ}
  Any circular interval trigraph without semi-edges is a proper circular-arc graph.
\end{proposition}
\begin{proof}
  Let $G$ be a circular interval trigraph that has no semiadjacent edges.
  Then, in fact, $G$ is a graph.
  Consider the sphere $\mathbb{S}_{1}$.
  Let $\mc{I} = \{I_{1}, \ldots, I_{\ell}\}$ be subsets of $\mathbb{S}_{1}$ and $P$ be a finite subset of $\mathbb{S}_{1}$ such that $P = V(G)$ and distinct $u,v \in P$ are adjacent in $G$ if and only if $u,v \in I_{i}$ for some $i=1,\ldots,\ell$.
  Note that any adjacency must in fact be strong adjacency, by assumption.
  Define some orientation on $\mathbb{S}_{1}$ that we call clockwise.
  For each $u \in P$, let $I_{u}$ denote the subset of $\mathbb{S}_{1}$ in $\mc{I}$ that extends furthest clockwise from $u$ on $\mathbb{S}_{1}$.
  Let $I'_{u}$ denote the subset of $I_{u}$ that comes clockwise after~$u$ (including $u$ itself).
  Consider the set $\mc{I'} = \{I'_{u} \mid u \in P\}$.
  It is not hard to see that no arc of~$\mc{I'}$ contains another (possible after infinitesimal extensions of some arcs, \ie of $I'_v$ if $I_u = I_v$ and $v$ comes after $u$), and that $G$ is the intersection graph of $\mc{I'}$.
\end{proof}
Although one can prove a converse of this proposition if one makes more assumptions on the arcs of a proper circular-arc graph, this will not be relevant to this paper, and we thus omit it.

Let $H$ be a graph with $V(H) = \{h_{1},\ldots,h_{7}\}$ such that
\begin{itemize}
  \item $\{h_{1},\ldots,h_{5}\}$ is a cycle with vertices in this order,
  \item $h_{6}$ is adjacent to at least three of $h_{1},\ldots,h_{5}$,
  \item $h_{7}$ is adjacent to $h_{6}$ and no other vertices,
  \item all adjacencies not specified so far can be arbitrary.
\end{itemize}
Let $G$ be the graph obtained from the line graph of $H$ by adding a new vertex adjacent to those edges of $E(H)$ not incident with $h_{6}$, and then regarding it as a trigraph.
If $h_{4},h_{5}$ both have degree two in $H$, possibly make the vertices of $G$ corresponding to edges $h_{3}h_{4}$ and $h_{1}h_{5}$ semiadjacent.
The class of trigraphs containing precisely these trigraphs is denoted by \emph{$\mc{S}_{4}$}.

Let $G$ be a trigraph that is the disjoint union of three $n$-vertex strong cliques $A = \{a_{1},\ldots,a_{n}\}$, $B = \{b_{1},\ldots,b_{n}\}$, and $C = \{c_{1},\ldots,c_{n}\}$ for $n \geq 2$ and five vertices $\{d_{1},\ldots,d_{5}\}$ such that for some $X \subseteq A \cup B \cup C$ with $|X \cap A|,|X \cap B|,|X \cap C| \leq 1$,
\begin{itemize}
  \item for $1 \leq i, j \leq n$, $a_{i}$ and $b_{j}$ are adjacent if and only if $i=j$, and $c_{i}$ is strongly adjacent to $a_{j}$ if and only if $i \not= j$, and $c_{i}$ is strongly adjacent to $b_{j}$ if and only if $i \not= j$,
  \item $a_{i}$ is semiadjacent to $b_{i}$ for at most one value of $i \in \{1,\ldots,n\}$, and if so then $c_{i} \in X$,
  \item $a_{i}$ is semiadjacent to $c_{i}$ for at most one value of $i \in \{1,\ldots,n\}$, and if so then $b_{i} \in X$,
  \item $b_{i}$ is semiadjacent to $c_{i}$ for at most one value of $i \in \{1,\ldots,n\}$, and if so then $a_{i} \in X$,
  \item no two of $A\setminus X$, $B\setminus X$, $C\setminus X$ are strongly complete to each other,
  \item $d_{1}$ is strongly complete to $A\cup B \cup C$,
  \item $d_{2}$ is strongly complete to $A \cup B$ and adjacent to $d_{1}$,
  \item $d_{3}$ is strongly complete to $A \cup \{d_{2}\}$,
  \item $d_{4}$ is strongly complete to $B \cup \{d_{2},d_{3}\}$,
  \item $d_{5}$ is strongly adjacent to $d_{3}$ and $d_{4}$.
\end{itemize}
The class of all trigraphs $G\setminus X$ is denoted by \emph{$\mc{S}_{5}$}.

The following trigraphs are called \emph{near-antiprismatic} or \emph{antihat} trigraphs.
Let $G$ be a trigraph that is the disjoint union of three $n$-vertex strong cliques $A = \{a_{1},\ldots,a_{n}\}$, $B = \{b_{1},\ldots,b_{n}\}$, and $C = \{c_{1},\ldots,c_{n}\}$ for $n \geq 2$ and two vertices $a_{0}, b_{0}$ such that for some $X \subseteq A \cup B \cup C$ with $|C \setminus X| \geq 2$,
\begin{itemize}
  \item for $1 \leq i,j \leq n$ with $i \not= j$, $a_i$ is strongly antiadjacent to $b_j$ and $a_i$ and $b_i$ are strongly adjacent to $c_j$,
  \item for $1 \leq i \leq n$, $a_i$ and $b_i$ are adjacent, $c_i$ and $a_i$ are antiadjacent, $c_i$ and $b_i$ are antiadjacent.
  All such pairs are strongly (anti)adjacent, except possibly
  \begin{itemize}
    \item $a_{i}$ is semiadjacent to $b_{i}$ for at most one value of $i \in \{1,\ldots,n\}$, and if so then $c_{i} \in X$,
    \item $a_{i}$ is semiadjacent to $c_{i}$ for at most one value of $i \in \{1,\ldots,n\}$, and if so then $b_{i} \in X$,
    \item $b_{i}$ is semiadjacent to $c_{i}$ for at most one value of $i \in \{1,\ldots,n\}$, and if so then $a_{i} \in X$.
  \end{itemize}
  \item $a_{0}$ is strongly complete to $A$,
  \item $b_{0}$ is strongly complete to $B$,
  \item $a_{0}$ is antiadjacent to $b_{0}$.
\end{itemize}
The class of all trigraphs $G\setminus X$ is denoted by \emph{$\mc{S}_{6}$}.

A trigraph $G$ is \emph{antiprismatic} if for every $X \subseteq V(G)$ with $|X| = 4$, $X$ is not a claw, and at least two pairs of vertices in $X$ are strongly adjacent.
The class of antiprismatic trigraphs is denoted by \emph{$\mc{S}_{7}$}.

The following is the main result from the work of Chudnovsky and Seymour on claw-free graphs.
Recall the definition of indecomposable trigraphs from page~\pageref{def:indecomp}.

\begin{theorem}[Chudnovsky and Seymour~\cite{ChudnovskyS2008-4,ChudnovskyS2008-5}]
\label{thm:chudsey-main}
  Every indecomposable trigraph belongs to $\mc{S}_{0} \cup \cdots \cup \mc{S}_{7}$.
\end{theorem}

\subsection{Special Stripes}
\label{sec:defs:stripes}
In the definitions below, whenever adjacency between two vertices is not specified, they are strongly antiadjacent.
Moreover, if two vertices are said to be adjacent, they can be either strongly adjacent or semiadjacent, unless otherwise specified.
Similarly, if two vertices are said to be antiadjacent, they can be either strongly antiadjacent or semiadjacent, unless otherwise specified.

Consider a trigraph $J$ with vertex set $\{v_{1},\ldots,v_{n}\}$ with $n \geq 2$ such that for $1 \leq i < j < k \leq n$, if~$v_{i}$ and $v_{k}$ are adjacent in $J$, then $v_{j}$ is strongly adjacent to both $v_{i}$ and $v_{k}$.
This is a \emph{linear interval trigraph}. 

\begin{definition}[proper interval graph]
  A graph is a \emph{proper interval graph} if it is the intersection graph of a set of intervals of the real line, such that no interval of the set properly contains another. 
\end{definition}
The following observation is implied by a result of Looges and Olariu~\cite[Proposition~1, Theorem~1]{LoogesO1993} (the ``umbrella property'').
\begin{proposition}
\label{prp:proper-int}
  Any linear interval trigraph without semi-edges is a proper interval graph.
\end{proposition}

Let $J$ be a linear trigraph with vertex $\{v_{1},\ldots,v_{n}\}$ as per the definition, let $v_{1},v_{n}$ be strongly antiadjacent, let no vertex be adjacent to both $v_{1}$ and $v_{n}$, and let no vertex be semiadjacent to either~$v_{1}$ or $v_{n}$.
Let $Z = \{v_{1},v_{n}\}$.
The class of all such stripes $(J,Z)$ is denoted by $\mathcal{Z}_{1}$.

Let $J \in \mc{S}_{6}$, let $a_{0},b_{0}$ be as in the definition of $\mc{S}_{6}$, with $a_{0},b_{0}$ strongly antiadjacent, and let $Z = \{a_{0},b_{0}\}$.
The class of all such stripes $(J,Z)$ is denoted by $\mathcal{Z}_{2}$.

Let $H$ be a graph, and let $h_{1},\ldots,h_{5}$ be a path in $H$ such that $h_{1}$ and $h_{5}$ have degree one and every edge of $H$ is incident with one of $h_{2},h_{3},h_{4}$.
Let $J$ be obtained from a line trigraph of $H$ by making the vertices corresponding to edges $h_{2}h_{3}$ and $h_{3}h_{4}$ either semiadjacent or strongly antiadjacent, and let $Z = \{h_{1}h_{2},h_{4}h_{5}\}$.
The class of all such stripes $(J,Z)$ is denoted by $\mathcal{Z}_{3}$.
See Fig.~\ref{fig:z3:H} for an example.

Let $J$ be the trigraph with vertices $\{a_{0},a_{1},a_{2},b_{0},b_{1},b_{2},b_{3},c_{1},c_{2}\}$ such that $\{a_{0},a_{1},a_{2}\}$, $\{b_{0},b_{1},b_{2},b_{3}\}$, $\{a_{2},c_{1},c_{2}\}$, and $\{a_{1},b_{1},c_{2}\}$ are strong cliques, $b_{2},c_{1}$ are strongly adjacent, $b_{2},c_{2}$ are semiadjacent, and $b_{3},c_{1}$ are semiadjacent. Let $Z = \{a_{0},b_{0}\}$. The class of all such stripes $(J,Z)$ is denoted by $\mathcal{Z}_{4}$.
See Fig.~\ref{fig:z4} for an example.

\begin{definition}[$\mc{Z}_{5}$]
  Let $J \in \mc{S}_{2}$, let $v_{1},\ldots,v_{13}, X$ be as in the definition of $\mc{S}_{2}$, let $v_{7},v_{8}$ be strongly antiadjacent in $J$, and let $Z = \{v_{7},v_{8}\}\setminus X$.
  The class of all such stripes $(J,Z)$ is denoted by~$\mathcal{Z}_{5}$.
\end{definition}
See Fig.~\ref{fig:z5} for an example.

Let $J \in \mc{S}_{3}$, let $I_{1},\ldots,I_{n}$ be as in the definition of $\mc{S}_{3}$, let $z \in V(G)$ belong to at most one of $I_{1},\ldots,I_{n}$ and not be an endpoint of one of $I_{1},\ldots,I_{n}$.
Then $z$ is a strongly simplicial vertex of $J$. Let $Z = \{z\}$.
The class of all such stripes $(J,Z)$ is denoted by $\mathcal{Z}_{6}$.

Let $J \in \mc{S}_{4}$, let $H, h_{1},\ldots,h_{7}$ be as in the definition of $\mc{S}_{4}$, let $e$ be the edge in $H$ between $h_{6}$ and~$h_{7}$, and let $Z = \{e\}$.
The class of all such stripes $(J,Z)$ is denoted by $\mathcal{Z}_{7}$.

Let $J \in \mc{S}_{5}$, let $d_{1},\ldots,d_{5},A,B,C$ be as in the definition of $\mc{S}_{5}$, and let $Z = \{d_{5}\}$.
The class of all such stripes $(J,Z)$ is denoted by $\mathcal{Z}_{8}$.

Let $J$ be the trigraph where the vertex set is the union of five sets $\{z\},A,B,C,D$ with $A = \{a_{1},\ldots,a_{n}\}$ and $B = \{b_{1},\ldots,b_{n}\}$ for some $n \geq 1$, such that
\begin{itemize}
  \item $\{z\} \cup D$ is a strong clique,
  \item $A \cup C$ and $B \cup C$ are strong cliques,
  \item for $1 \leq i \leq n$, $a_{i}$ and $b_{i}$ are antiadjacent and every vertex of $D$ is strongly adjacent to precisely one of $a_{i},b_{i}$,
  \item for $1 \leq i < j \leq n$, $\{a_{i},b_{i}\}$ is strongly complete to $\{a_{j},b_{j}\}$, and
  \item the adjacency between $C$ and $D$ is arbitrary.
\end{itemize}
Note that $J$ is antiprismatic.
Let $Z = \{z\}$.
The class of all such stripes $(J,Z)$ is denoted by $\mathcal{Z}_{9}$.

\section{Algorithms to Find Twins and Joins}
\label{sec:joins}
If we want to prove an algorithmic decomposition theorem for claw-free graphs, then Theorem~\ref{thm:chudsey-main} suggests that we develop algorithms that find twins and the joins defined in Sect.~\ref{sec:def:joins} in trigraphs.
Actually, we will see later that we only apply these algorithms to graphs, and so we restrict our attention to graphs instead of trigraphs.
Moreover, we will see later that we do not need an algorithm to find a hex-join in a graph.
Therefore, in this section, we propose algorithms to find twins, proper W-joins, $0$-joins, (pseudo-) $1$-joins, and (pseudo-/generalized) $2$-joins in graphs in polynomial time.

We observe here that in the literature $1$-joins and $2$-joins are usually defined in a different manner than they were defined in this paper, and several fast, polynomial-time algorithms are known to find these alternate $1$- and $2$-joins (see \eg\cite{CornuejolsC1985,Dahlhaus2004,CharbitMR2010,CharbitHTV2010}).
It is not difficult to show that on claw-free graphs that do not admit twins and have a stable set of size at least five the alternate definitions of $1$- and $2$-joins coincide with those defined in this paper.
Under these conditions we can thus use the fast algorithms available in the literature.
However, we emphasize that on general (claw-free) graphs, the alternate definitions diverge from those in our paper, which led us to study them separately.
Our algorithms work on general graphs, and thus may be of independent interest. 

Throughout this section, let $G$ be a graph, let $n = |V(G)|$ and $m = |E(G)|$, and let $\Delta(G)$ be the maximum degree of any vertex in $G$.

\subsection{Finding \texorpdfstring{$0$}{0}-joins, Proper W-joins, and Twins}
The results of this section follow from known results or were known before.
First, finding $0$-joins corresponds to standard connectivity testing.

\begin{proposition}
  In $O(n+m)$ time, one can find a $0$-join in a graph $G$, or report that $G$ has no $0$-join.
\end{proposition}

King and Reed~\cite{KingR2008} showed that proper W-joins can be found in $O(n^{2}m)$ time.
\begin{theorem}[\cite{KingR2008}]
\label{thm:proper-wjoin-algo}
  In $O(n^{2}m)$ time, one can find a proper W-join in a graph $G$, or report that $G$ does not admit a proper W-join.
\end{theorem}
Other algorithms for finding (proper) W-joins and applications of such algorithms are considered in several papers (see~\eg\cite{ChudnovskyK2012,FaenzaOS2011,KingR2008,OrioloPS2008}).

Twins can be found in linear time.
The algorithm is actually implicitly given by Habib \etal~\cite{HabibPV1998}.
We provide it here only for completeness.

\begin{theorem}
\label{thm:twins-algo}
  In $O(n+m)$ time, one can find twins in a graph $G$, or report that $G$ does not admit twins.
\end{theorem}
\begin{proof}
  We use a technique called \emph{partition refinement} (see \eg\cite{HabibMPV2000}).
  Given a universe $\mathbb{U}$, a partition $\mc{P}$ of $\mathbb{U}$, and a subset $S$ of $\mathbb{U}$, the \emph{refinement} of $\mc{P}$ splits each partition class $P$ of $\mc{P}$ into a class $P \cap S$ and a class $P \setminus S$.
  This operation takes $O(|S|)$ time~\cite{HabibMPV2000}.
  Now observe that $u,v \in V(G)$ are twins if and only if for any $x \in V(G)$, either $u,v \in N[x]$ or $u,v \not\in N[x]$.
  Hence if we set $\mathbb{U} = V(G)$, initialize $\mc{P} = \{V(G)\}$, and iteratively refine $\mc{P}$ with $N[x]$ for each $x \in V(G)$, then any two vertices in any non-singleton class of the final partition $\mc{P}'$ are twins.
  Moreover, if $\mc{P}'$ consists only of singleton classes, then $G$ does not admit twins.
  The total run time of this algorithm is $O(\sum_{x \in V(G)} |N[x]|) = O(n+m)$.
\end{proof}

\subsection{Finding \texorpdfstring{$1$}{1}-joins}
In this subsection, we describe how to find a $1$-join in polynomial time. We first need some auxiliary lemmas.

\begin{lemma}
\label{lem:1-join-connected}
  Let $(A_{1}, A_{2})$ be a $1$-join of a connected graph $G$. Then $G[V_{1}]$ and $G[V_{2}]$ are connected.
\end{lemma}
\begin{proof}
  Suppose that $G[V_{1}]$ has two connected components $C$ and $C'$.
  Because $G[A_{1}]$ is a clique, at most one of $C, C'$ contains vertices of $A_{1}$.
  Suppose that $C$ contains no vertices of $A_{1}$.
  Because $(A_{1}, A_{2})$ is a $1$-join, $C$ is not connected to $V_{2}$.
  This contradicts that $G$ is connected.
\end{proof}
A spanning tree $T$ of a graph $G$ is a \emph{BFS-spanning tree} if it is obtained from a breadth-first search on $G$.
The \emph{root} of a BFS-spanning tree is the vertex that is the origin of the breadth-first search.
Throughout, for $a,b \in V(G)$, we use $\dist(a,b)$ to denote the length of a shortest path between $a$ and $b$ in $G$, and $\dist_T(a,b)$ to denote the length of a shortest path between $a$ and $b$ in $T$.

The following lemma is inspired by an observation of Cornu{\'e}jols and Cunningham~\cite{CornuejolsC1985}.
\begin{lemma}
  \label{lem:1-join-crossingedge}
  Let $T$ be a BFS-spanning tree of a connected graph $G$ and let $(A_{1}, A_{2})$ be a $1$-join.
  Suppose that the root $r$ of $T$ is in $V_{1}$.
  Then there is a vertex $v \in A_{1}$ such that all vertices of $A_{2}$ are neighbors of $v$ in $T$.
\end{lemma}
\begin{proof}
  This is immediate from the fact that $T$ is a BFS-spanning tree.
\end{proof}

For a tree $T$, call a pair of vertices $(u,v)$ \emph{diametral} if the length of the $u$--$v$-path in $T$ is maximum among all pairs.
\begin{lemma}
\label{lem:diametralpair-find}
  Given a tree $T$ on $n$ vertices, a diametral pair of vertices can be found in $O(n)$ time.
\end{lemma}
\begin{proof}
  Let $w$ be an arbitrary leaf of $T$ and let $u$ be a vertex furthest away in $T$ from $w$.
  Let $v$ be a vertex furthest away from $u$ in $T$.
  Note that all these vertices can be found in $O(n)$ time.
  If~$u$ is in a diametral pair, then $(u,v)$ is clearly a diametral pair.
  Otherwise, let $(a,b)$ be a diametral pair.
  For vertices $s,t$, let $P_{st}$ denote the unique path between $s$ and $t$ in $T$.
  Let $x$ be a vertex of~$P_{uv}$ that is closest to a vertex of $P_{ab}$.
  By the choice of $v$, $\dist(x,v) \geq \max\{\dist(x,a), \dist(x,b)\}$.
  Recall that a \emph{subtree} of $T$ is any connected subgraph of $T$. 
  Consider the subtree $T_{x}$ of $T$ induced by $x$, the vertices in the subtree of $T\setminus\{x\}$ containing $u$, and the vertices in the subtree of $T\setminus\{x\}$ containing~$v$. 
  Suppose that $w$ is in $V(T) \setminus V(T_{x})$.
  Since $x$ is contained in both the shortest path from~$w$ to~$u$ and the shortest path from $w$ to $v$, the choice of $u$ implies that $\dist(x,u) \geq \dist(x,v)$.
  Hence $\dist(x,u) \geq \max\{\dist(x,a),\dist(x,b)\}$, implying that
  \begin{equation*}
    \dist(u,v)    = \dist(u,x) + \dist(x,v)
               \geq 2 \max\{\dist(x,a), \dist(x,b)\}
               \geq \dist(x,a) + \dist(x,b) = \dist(a,b) \enspace .
  \end{equation*}
  Hence $(u,v)$ is a diametral pair.
  Suppose then that $w$ is in $V(T_{x})$ and let $y$ be the vertex of $P_{uv}$ closest to $w$.
  Note that any shortest path from $w$ to $u$, $v$, $a$, or $b$ must contain $y$, and that any shortest path from $w$ to $a$ or $b$ must contain $x$.
  The choice of $u$ implies that $\dist(y,u) \geq \max\{\dist(y,a),\dist(y,b)\}$, and thus $\dist(x,u) \geq \max\{\dist(x,a),\dist(x,b)\}$.
  As before, this implies that $\dist(u,v) \geq \dist(a,b)$.
  Hence $(u,v)$ is a diametral pair.
\end{proof}
We require some auxiliary notions on trees.
Given any rooted tree $T$, the \emph{nearest common ancestor} of any two vertices $a,b \in V(T)$, denoted by $\nca(a,b)$, is the vertex $c$ in $T$ that is an ancestor of both~$a$ and $b$ and no child of $c$ is an ancestor of both $a$ and $b$.
If $a$ is not an ancestor of $b$ and vice versa, define the \emph{$a$-nearest almost-common ancestor} of $a$ and $b$, or $a$-$\nca(a,b)$, as the child of $\nca(a,b)$ that is an ancestor of $a$.

\begin{lemma}
\label{lem:diametralpair}
  Let $T$ be a BFS-spanning tree of a connected graph $G$.
  Then for any diametral pair $(u,v)$ of $T$, at most one of $u,v$ is in $A_{1} \cup A_{2}$ for any $1$-join $(A_{1}, A_{2})$ of $G$.
\end{lemma}
\begin{proof}
  Let $T$ be a BFS-spanning tree of a graph $G$, let $(u,v)$ be any diametral pair of $T$, and let $(A_{1}, A_{2})$ be a $1$-join of $G$ with a corresponding partition $(V_1,V_2)$ of $G$.
  Suppose that the root $r$ of $T$ is in $V_{1}$. We distinguish three cases.
  \begin{quote}
    (1) $u,v \in A_{2}$
  \end{quote}
  Note that $u,v$ must be leafs of $T$.
  From Lemma~\ref{lem:1-join-crossingedge}, there is a vertex $w \in A_{1}$ adjacent in $T$ to all vertices of $A_{2}$, and so $\dist_T(u,v)=2$.
  Since $G[V_{2}]$ is connected by Lemma~\ref{lem:1-join-connected}, there is a path in $T \cap (V_{2} \cup \{w\})$ from $w$ to a vertex $x \in V_{2}\setminus A_{2}$ that is a leaf of $T$.
  But then $\dist_{T}(x,v) > \dist_{T}(u,v) = 2$, contradicting that $(u,v)$ is a diametral pair.

  \begin{quote}
    (2) $v \in A_{1}$, $u \in A_{2}$
  \end{quote}
  Consider the vertex $w \in A_{1}$ adjacent to all vertices of $A_{2}$.
  Since $G[V_{2}]$ is connected, there is a path from $w$ to a vertex $x \in V_{2}\setminus A_{2}$ that is a leaf of $T$.
  But then $\dist_{T}(x,v) > \dist_{T}(u,v)$, contradicting that $(u,v)$ is a diametral pair.

  \begin{quote}
    (3) $u,v \in A_{1}$
  \end{quote}
  Consider the vertex $w \in A_{1}$ neighboring all vertices in $A_{2}$.
  Again there is a vertex $x \in V_{2}\setminus A_{2}$ that is a leaf of $T$ with $\dist_{T}(w,x) \geq 2$ and thus $\dist_{T}(r,x) \geq \dist_{T}(r,w) + 2$.
  Note that $w \not= u,v$, as $u$ and $v$ are leafs.
  Moreover, $w$ is not adjacent to $u$ or $v$ in $T$; otherwise, $x$ has larger distance to $u$ or~$v$ than $v$ or $u$ respectively, contradicting that $(u,v)$ is a diametral pair.

  Let $a = \nca(u,v)$, $b = \nca(u,w)$ and $c = \nca(v,w)$. Note that $\dist_{T}(u,v) = \dist_{T}(r,u) + \dist_{T}(r,v) - 2\cdot\dist_{T}(r,a)$.
  Also, $\dist_{T}(u,x) = \dist_{T}(r,u) + \dist_{T}(r,x) -2\cdot\dist_{T}(r,b)$.
  Because $(v,w) \in E(G)$ and $T$ is a BFS-spanning tree, $|\dist_{T}(r,v) - \dist_{T}(r,w)| \leq 1$.
  But then $\dist_{T}(r,x) > \dist_{T}(r,v)$.
  Suppose that $b$ and $c$ are equal to $a$ or an ancestor of $a$.
  Since $\dist_{T}(r,a) \geq \dist_{T}(r,b)$, this implies that $\dist_{T}(u,x) > \dist_{T}(u,v)$.
  This contradicts that $(u,v)$ forms a diametral pair.
  Hence~$b$ or~$c$ is a descendant of $a$.
  Assume, \wloge that $c$ is a descendant of $a$.
  Then
  \begin{equation*}
    \begin{array}{rcl}
      \dist_{T}(c,x) &    = & \dist_{T}(r,x) - \dist_{T}(r,c)\\
                     & \geq & \dist_{T}(r,w) + 2 - \dist_{T}(r,c)\\
                     &    > & \dist_{T}(r,v) - \dist_{T}(r,c)\\
                     &    = & \dist_{T}(c,v).
    \end{array}
  \end{equation*}
  This means that $\dist_{T}(u,x) > \dist_{T}(u,v)$, contradicting that $(u,v)$ is diametral.
\end{proof}

A $1$-join $(A_{1}, A_{2})$ is \emph{minimal} if there is no $A \subset A_{1}$ such that $(A_{1}\setminus A, A_{2} \cup A)$ is $1$-join as well.

\begin{lemma}
\label{lem:1-join-minimal}
  If $(A_{1}, A_{2})$ is a minimal $1$-join, then $G[V_{1}]-E(G[A_{1}])$ is connected.
\end{lemma}
\begin{proof}
  For suppose not and let $\mc{C}$ be the set of connected components of $G[V_{1}]-E(G[A_{1}])$.
  For any $C \in \mc{C}$, observe that $C \cap A_{1} \not= \emptyset$, because $G[V_{1}]$ is connected.
  But since the components are pairwise disjoint, this implies that the components induce a partition of $A_{1}$ into nonempty subsets.
  As $|\mc{C}| \geq 2$, $(A_{1} \setminus (A_{1} \cap C), A_{2} \cup (A_{1} \cap C))$ is a $1$-join for any $C \in \mc{C}$, contradicting the minimality of $(A_{1}, A_{2})$.
\end{proof}

\begin{theorem} \label{thm:1-join-algo-easy}
  In $O(n (n+m))$ time, one can find a $1$-join in a connected graph $G$, or report that $G$ does not have such a join.
\end{theorem}
\begin{proof}
  Consider any BFS spanning tree of $G$ and some diametral pair $(r,r')$ of this tree.
  If $G$ has a $1$-join $(A_{1}, A_{2})$, we know from Lemma~\ref{lem:diametralpair} that at most one of $r,r'$ is in $A_{1} \cup A_{2}$.
  Assume, \wloge that $r$ is not in $A_{1} \cup A_{2}$ (algorithmically we will actually try both $r$ and~$r'$).
  Construct a BFS spanning tree $T$ with $r$ as its root.

  Let $e = (u,v)$ be any edge of the spanning tree not incident with $r$.
  Assume that $u \in A_{1}$, $v \in A_{2}$, and $v$ has a child $w$ in $T$.
  By Lemma~\ref{lem:1-join-crossingedge} and \ref{lem:1-join-connected}, we can assume that such an edge exists.
  Moreover, $w \in V_{2} \setminus A_{2}$.
  Now find the set of all vertices in $N[u] \cap N[v]$.
  This set forms the candidate set $A := A_{1} \cup A_{2}$.
  This can be done in $O(\Delta(G))$ time.

  What remains is to verify that we have indeed found a $1$-join.
  As a first step, we verify that $A$ is a clique.
  This can be done in $O(|A|^{2})$ or $O(m)$ time.
  Next, we find the partition $V_{1}, V_{2}$.
  To this end, collect the set $R$ of all vertices reachable from $u$ in $G \setminus A$, using say a breadth-first search.
  Let $A' = N(R)$.
  Note that $A' \subseteq A$.
  Then consider $N(A') \setminus A$ and continue the breadth-first search in $G \setminus A$ from those vertices.
  Iteratively apply this procedure.
  If the search visits all vertices of $G \setminus A$, then $G$ has no $1$-join with $(u,v)$ as its basis.
  Otherwise, the set $R$ of visited vertices forms $V_{1}\setminus A_{1}$ and the vertices in $N(R) \cap A$ form $A_{1}$. It is now easy to find $A_{2}$ and $V_{2}$.
  The correctness of this procedure follows from Lemma~\ref{lem:1-join-minimal}.

  The run time for each edge of $T$ is bounded by $O(n+m)$.
  Since we need to consider at most $n-2$ edges of the spanning tree and at most two possible roots, the total run time of the algorithm is $O(n(n+m))$.
\end{proof}
We will now speed-up the part of the algorithm responsible for finding $A_{1}$ and $A_{2}$, which is the most expensive part of the above lemma.
We assume a random access machine model with logarithmic costs.
Let $\alpha(i,j)$ denote the inverse Ackermann function.

\begin{theorem} \label{thm:1-join-algo}
  In $O(m (\Delta(G) + \alpha(m, \Delta(G))))$ time, one can find a $1$-join in a connected graph $G$, or report that $G$ does not have such a join.
\end{theorem}
\begin{proof}
  Consider again the rooted tree $T$ and edge $e = (u,v)$ from the previous lemma.
  As before, we assume that $u \in A_{1}$, $v \in A_{2}$ for some $1$-join $(A_{1}, A_{2})$.
  Moreover, we may assume that $u$ is closer to the root $r$ of $T$ than $v$.
  This in turn implies that all vertices of $A_{2}$ are further from $r$ than $u$.

  Given any rooted tree $T$ and vertex $t \in V(T)$, define $T_t$ as the subtree of $T$, rooted at $t$, containing~$t$ and all of its descendants.
  If $T$ is a spanning tree of a graph $G$, define $h(t)$ as the $\nca(a,b)$ closest to $r$ for any edge $(a,b) \in E(G)$ for which one of $a,b$ is in $V(T_t)\setminus\{t\}$.
  Clearly, $h(t)$ is either $t$ or an ancestor of $t$.

  Consider a graph $G$ and a rooted spanning tree $T$ of $G$.
  For any vertex $p$ and its set of children~$C$, we say that $c, c' \in C$ are \emph{linked} if there is an edge $(a,b) \in E(G)$ for which $\dist_{T}(a,b) \geq 3$, $a \in V(T_c)$, and $b \in V(T_{c'})$.
  We then say that $c, c' \in C$ are \emph{sequentially linked} if there is a sequence $c = c_{1}, \ldots, c_{i} = c'$ of children of $C$ such that $c_{j}$ is linked to $c_{j+1}$ for any $j = 1,\ldots,i-1$.
  Observe that if $c, c' \in A_{1} \cup A_{2}$ for some $1$-join $(A_{1}, A_{2})$ and $c$ and $c'$ are sequentially linked, then either both~$c$ and $c'$ must be in~$A_{1}$ or both must be in $A_{2}$. 

  Now think back on the algorithm of Theorem~\ref{thm:1-join-algo-easy} and let $A$ be the candidate set for the join.
  All vertices for which we have not yet decided whether they should be in $A_{1}$ or in $A_{2}$ must be children of $u$.
  We say a child $c$ of $u$ is of \emph{type $1$} if $h(c) \not\in V(T_u)$ or it is incident with an edge $e \not\in E(A)$ for which $\nca(e) \not\in V(T_u)$.
  Here $\nca(e)$ is a shorthand for the nearest common ancestor of the two endpoints of $e$.

  Now consider the following observation.
  \begin{quote}
    (1) A child $c$ is in $A_{1}$ for some minimal $1$-join $(A_{1}, A_{2})$ if and only if it is of type $1$ or sequentially linked to a child of type $1$.
  \end{quote}
  This follows immediately from Lemma~\ref{lem:1-join-minimal}.

  Using this observation, we can split $A$ into sets $A_{1}$ and $A_{2}$.
  Note that this choice is only difficult for children of $u$.
  Other vertices must belong to $A_{1}$.
  For children of $u$, we use the above observation.
  If every child of $u$ is, or is sequentially linked to, a vertex of type $1$, then $A$ cannot be split to form a $1$-join.

  To use these ideas algorithmically, we should be able to compute $\nca$ and $h$ efficiently.
  By preprocessing $T$ in linear time, we can compute $\nca$ and $a$-$\nca$ in constant time for any pair of vertices using the algorithms of Lu and Yeh~\cite{LuY2008} (note that finding the $a$-$\nca$ can be simulated by finding the $\nca$, then its depth, and then the appropriate level ancestor of $a$, all in constant time).
  Define $x(t)$ for any $t \in V(T)$ as the $\nca(t,t')$ closest to $r$ for any $(t,t') \in E(G)$.
  By preprocessing in linear time, we know the height of each vertex in the tree.
  Then we can compute $x(t)$ in $O(m)$ time.
  Now set $h(t) = t$ for any leaf of $T$, and for any nonleaf $t$, set $h(t)$ equal to the highest vertex among $x(t')$ and $h(t')$ over all children $t'$ of $t$.
  This takes $O(n)$ time.

  To determine whether two children $c, c'$ of a vertex $p$ are sequentially linked, we use preprocessing with a union-find data structure.
  For each vertex $p$, associate one such data structure, initially with each child $c$ of $p$ in a separate set.
  Now for any edge $(a,b) \in E(G)$, we determine $x = \nca(a,b)$.
  If $x = a$ or $x = b$, then this edge is not relevant.
  Otherwise, determine $y =$ $a$-$\nca(a,b)$ and $z=$ $b$-$\nca(a,b)$.
  If $y=a$ and $z=b$, then again this edge is not relevant.
  Otherwise, we perform a union in the data structure associated with $x$ on $y$ and $z$.
  This takes $O(m \cdot \alpha(m, \Delta(G)))$ time in total~\cite{Tarjan1975}.
  Afterwards, we process all structures such that find-operations will take constant time.
  This takes $O(n)$ time. Using this data structure, we can answer in $O(1)$ time whether two children are sequentially linked.

  This concludes the analysis of all preprocessing.
  Preprocessing takes $O(m \cdot \alpha(m, \Delta(G)))$ time in total.

  We now analyze the algorithm itself.
  Finding the candidate set $A$ can be done in $O(\Delta(G))$ time.
  Since we do this $n-2$ times, this contributes $O(n\cdot\Delta(G))$ to the run time.
  To determine the contribution of the total time it takes to verify that candidate sets form a clique, we note that we check a nonedge at most once per candidate set, so this uses $O(n)$ time.
  An edge is checked only if both endpoints are in the closed neighborhood of the edge $(u,v)$ that is the basis of the candidate set.
  Hence an edge is checked $O(\Delta(G))$ times during the course of the algorithm.
  This gives a contribution of $O(m\cdot\Delta(G))$ to the run time.

  It remains to analyze the run time for splitting the candidate set $A$ into $A_{1}$ and $A_{2}$.
  For this, we only need to check whether a child $c$ of $u$ is of type $1$ or sequentially linked to a child of type $1$. Determining whether $h(c) \not\in V(T_u)$ takes $O(1)$ time.
  Hence we spend no more than $O(n\cdot\Delta(G))$ time on this in the course of the algorithm.
  We determine $\nca(e)$ for an edge at most $\Delta(G)$ times if one of its endpoints is a child of the current vertex $u$.
  Since this happens at most twice, we spend no more than $O(m\cdot\Delta(G))$ time on this.
  If we determine that a vertex is of type $1$, we indicate in the set of the union-find data structure containing this vertex that it contains a vertex of type $1$.
  Finally, the algorithm checks for each child whether it is or is sequentially linked to a vertex of type $1$.
  This takes another $O(n\cdot\Delta(G))$ time over the course of the algorithm. 

  The conclusion of the analysis is a run time for the algorithm of $O(m\cdot\Delta(G) + m\cdot\alpha(m, \Delta(G)))$.
\end{proof}

\begin{corollary}
\label{cor:pseudo-1-join-algo}
  In $O(m (\Delta(G) + \alpha(m,\Delta(G))))$ time, one can find a pseudo-$1$-join in a connected graph $G$, or report that $G$ does not have such a join.
\end{corollary}
\begin{proof}
  Let $(A_{1}, A_{2})$ be a pseudo-$1$-join.
  Suppose that $V_{1}\setminus A_{1} = \emptyset$.
  Since $|A_{1}|$ must be at least two in this case, $A_{1}$ contains twins.
  Hence we should first check whether $G$ has twins whose neighborhood is a clique.
  If no such twins exist, any pseudo-$1$-join of $G$ also is a $1$-join, which can be found using Theorem~\ref{thm:1-join-algo}. Using Theorem~\ref{thm:twins-algo}, we can find all twins in $O(n+m)$ time.
  Using similar ideas as in the proof of Theorem~\ref{thm:1-join-algo}, we can test all neighborhoods of twins for being a clique in $O(m\cdot\Delta(G))$ time in total.
\end{proof}

\subsection{Finding \texorpdfstring{$2$}{2}-joins}
We now describe algorithms to find the various $2$-joins.

\begin{lemma} \label{lem:2-join-connected}
  Let $(A_{1}, A_{2}, B_{1}, B_{2})$ be a (generalized) $2$-join of a connected graph $G$ that does not admit a $1$-join.
  Then $G[V_{1}]$ and $G[V_{2}]$ are connected.
\end{lemma}
\begin{proof}
  Suppose that $G[V_{1}]$ has at least two connected components.
  Note that any such connected component must contain a vertex from $A_{1}$ or $B_{1}$, or it would be a connected component of $G$, contradicting that $G$ is connected.
  Hence $G[V_{1}]$ has a connected component $C$ such that $A_{1} \subseteq C$ and $C \cap B_{1} = \emptyset$.
  But then $(A_{1}, A_{2} \cup V_{0}$) is a $1$-join, a contradiction.
\end{proof}

\begin{theorem}
  \label{thm:2-join-algo}
  In $O(n m (\Delta(G) + \alpha(m,\Delta(G))))$ time, one can find a $2$-join in a connected graph $G$ that does not admit a $1$-join, or report that $G$ does not have such a join.
\end{theorem}
\begin{proof}
  Observe that if $(A_{1}, A_{2}, B_{1}, B_{2})$ is a $2$-join of $G$, then one of the connected components of $G \setminus (A_{1} \cup A_{2})$ admits a $1$-join.
  Furthermore, if $T$ is a spanning tree of $G$, then there is an edge $e = (u,v)$ of $T$ such that $u \in A_{1}, v \in A_{2}$ or $u \in B_{1}, v \in B_{2}$.

  We now proceed as follows.
  Find a spanning tree of $G$.
  For each edge $(u,v) \in E(T)$, we remove all edges between vertices in $N[u] \cap N[v]$ from $G$ and try to find a $1$-join in one of the connected components of this graph that does not use any vertices of $N[u] \cap N[v]$.
  If no such join exists, then $(u,v)$ cannot be a basis for a $2$-join.
  Otherwise, let $(B_{1}, B_{2})$ be this $1$-join.
  Now remove all edges between $B_{1}$ and $B_{2}$ from $G$ and try to find a $1$-join $(A_{1}, A_{2})$ in the remaining graph that does not use any vertices of $B_{1} \cup B_{2}$ and for which neither $V_{1}\setminus A_{1}$ nor $V_{2}\setminus A_{2}$ equals $B_{1}$ or $B_{2}$.
  If no such $1$-join exists, then $G$ does not have a $2$-join with $(u,v)$ as a basis.
  Otherwise, let $(A_{1}, A_{2})$ be the $1$-join we just found. Then $(A_{1}, A_{2}, B_{1}, B_{2})$ is a $2$-join.

  Using the algorithm of Theorem~\ref{thm:1-join-algo} as a subroutine, the algorithm described above takes $O(n m (\Delta(G) + \alpha(m,\Delta(G))))$ time.
\end{proof}

\begin{theorem} \label{thm:gen-2-join-algo}
  In $O(n m (\Delta(G) + \alpha(m,\Delta(G))))$ time, one can find a generalized $2$-join in a connected graph $G$ that does not admit a $1$-join, or report that $G$ does not have such a join.
\end{theorem}
\begin{proof}
  We first try to find a $2$-join (\ie a generalized $2$-join with $V_{0} = \emptyset$), using Theorem~\ref{thm:2-join-algo}. This takes $O(n m (\Delta(G) + \alpha(m,\Delta(G))))$ time.
  If no $2$-join exists, then only a generalized $2$-join with $V_0 \not= \emptyset$ might exist. Guess a vertex $v \in V_{0}$.
  Then $V_{0}$ is equal to the set of vertices $u \in N[v]$ for which $N[u] = N[v]$.
  Finding $V_0$ takes $O(n+m) = O(m)$ time by Theorem~\ref{thm:twins-algo}.
  Observe that $A_1 \cup A_2 \cup B_1 \cup B_2 = N(V_0)$. Let $a$ be an arbitrary vertex in $N(V_0)$ and without loss of generality assume that $a \in A_1$. Let $B'$ be the set of vertices in $N(V_0)$ that is antiadjacent to $a$; it takes $O(\Delta(G))$ time to find $B'$. Then $B' \subseteq B_1 \cup B_2$ and $B_2 \subseteq B'$ by definition. We now aim to identify $A_1 \cup A_2$ and $B_1 \cup B_2$. After that, we will test whether these two sets indeed form the basis of a generalized $2$-join with the given set $V_0$. We consider four cases:
  
  Suppose $B' \cap B_1 \not= \emptyset$. Then all vertices in $N(V_0) \setminus B'$ that are complete to $B'$ are in $B_1$, and all other vertices in $N(V_0) \setminus B'$ are in $A_1 \cup A_2$. Hence, we have identified $A_1 \cup A_2$ and $B_1 \cup B_2$. This takes $O(m)$ time.
  
  Suppose $B' \cap B_1 = \emptyset$. Then $B' = B_2$ and $a$ is complete to $A_1 \cup A_2 \cup B_1$. Let $b$ be an arbitrary vertex in $B'$. Let $A'$ be the set of vertices in $N(V_0)$ that are antiadjacent to $b$. Then $A' \subseteq A_1 \cup A_2$ and $A_1 \subseteq A'$ by definition. If $A' \cap A_2 \not= \emptyset$, then we can identify $A_1 \cup A_2$ and $B_1 \cup B_2$ as before. This takes $O(m)$ time. Otherwise, $A' \cap A_2 = \emptyset$ and thus, $A' = A_1$. Then $N(V_0) \setminus (A' \cup B')$ is precisely equal to a disjoint union of two cliques by definition; one clique is $A_2$, the other is $B_1$. We consider both possibilities; in each, we have identified $A_1 \cup A_2$ and $B_1 \cup B_2$. This takes $O(m)$ time too.
  
  Now we have identified $A_1 \cup A_2$ and $B_1 \cup B_2$ (to be precise, four possible cases).  Call two vertices $x, y \in V(G) \setminus V_0$ equivalent if there is a path between them in the graph $G'$ obtained from $G$ by removing all edges between vertices in $V_0 \cup A_1 \cup A_2$ and between vertices in $V_0 \cup B_1 \cup B_2$. Let $\mathcal{Q} = \{Q_1,\ldots,Q_k\}$ denote the set of equivalence classes of this relation. We can find these equivalence classes in $O(m)$ time using a depth-first search. Note that $G[V_1]$ and $G[V_2]$ are both connected by Lemma~\ref{lem:2-join-connected}. Moreover, $V_1$ and $V_2$ are antiadjacent, and any path between them in $G$ uses an edge of $E(G) \setminus E(G')$.
  Hence, there exists a partition of $\mathcal{Q}$ into two nonempty sets $\mathcal{Q}_1$ and $\mathcal{Q}_2$ such that for every $i=1,2$, $V_i = \bigcup_{Q \in \mathcal{Q}_i} Q$ and there exists a set in $\mathcal{Q}_i$ (say $Q_i$) such that $A_i \cap Q_i \not= \emptyset$ and $B_i \cap Q_i \not= \emptyset$.
  
  In fact, for any partition of $\mathcal{Q}$ into two nonempty sets $\mathcal{Q}_1'$ and $\mathcal{Q}_2'$ such that for every $i=1,2$ there exists a set in $\mathcal{Q}_i$ (say $Q_i$) such that $(A_1 \cup A_2) \cap Q_i \not= \emptyset$ and $(B_1 \cap B_2) \cap Q_i \not= \emptyset$, there exists a generalized $2$-join with sets $A'_i = (A_1 \cup A_2) \cap \bigcup_{Q \in \mathcal{Q}_i} Q$, $B'_i = (B_1 \cup B_2) \cap \bigcup_{Q \in \mathcal{Q}_i} Q$, and $V'_i = \bigcup_{Q \in \mathcal{Q}_i} Q$ for $i=1,2$. It takes $O(m)$ to identify such a partition, and verify whether the resulting sets indeed constitute a generalized $2$-join.
  
  The running time of this algorithm is $O(m)$ for each vertex $v$ we guess, and thus $O(nm)$ in total. Combined with the test for a $2$-join, this brings the total running time to $O(n m (\Delta(G) + \alpha(m,\Delta(G))))$, as claimed.
\end{proof}

\begin{theorem}
\label{thm:pseudo-2-join-algo}
  In $O(nm (\Delta(G) + \alpha(m,\Delta(G))))$ time, one can find a pseudo-$2$-join in a connected graph $G$ that does not admit a pseudo-$1$-join, or report that $G$ does not have such a join.
\end{theorem}
\begin{proof}
  We first try to find a generalized $2$-join using Theorem~\ref{thm:gen-2-join-algo}.
  If one exists, $G$ has a pseudo-$2$-join.
  Otherwise, suppose that $G$ has a pseudo-$2$-join ($V_{0},V_{1},V_{2})$ for which $V_{1}\setminus(A_{1} \cup B_{1}) = \emptyset$. Suppose that there is no edge in $G$ between $A_{1}$ and $B_{1}$.
  Then, \wloge $|A_{1}| \geq 2$.
  But then $(A_{1}, V(G)\setminus A_{1})$ would form a pseudo-$1$-join of $G$, a contradiction.
  Hence we can just apply the same idea as in the proof of Theorem~\ref{thm:gen-2-join-algo}, but when using Theorem~\ref{thm:2-join-algo} inside of it, we do not insist that neither $V_{1}\setminus A_{1}$ nor $V_{2}\setminus A_{2}$ equals $B_{1}$ or $B_{2}$.
\end{proof}

From the fact that the maximum degree of a claw-free graph is $O(\sqrt{m})$~\cite[Lemma~4]{KloksKM2000}, we immediate obtain the following corollary.

\begin{corollary}
\label{cor:join-finder}
  \begin{itemize}
    \item In $O(n+m)$ time, one can find a $0$-join in a graph $G$, or report that $G$ has no $0$-join.
    \item In $O(m^{3/2})$ time, one can find a pseudo-$1$-join in a connected graph $G$, or report that $G$ does not have such a join.
    \item In $O(nm^{3/2})$ time, one can find a pseudo-$2$-join in a connected graph $G$ that does not admit a pseudo-$1$-join, or report that $G$ does not have such a join.
  \end{itemize}
\end{corollary}

\section{Recognizing Thickenings of \texorpdfstring{$\mc{Z}_{2}$, $\mc{Z}_{3}$, $\mc{Z}_{4}$, and $\mc{Z}_{5}$}{Z2, Z3, Z4, and Z5}}
\label{sec:recog}
The algorithmic decomposition theorem that we prove later is based on a strip-structure of the graph where most of its stripes are the special stripes defined in Sect.~\ref{sec:defs:stripes}.
However, several classes of stripes defined in Sect.~\ref{sec:defs:stripes} play a more prominent role than others.
In particular, we will need algorithms that recognize stripes that are thickenings of members of $\mc{Z}_{2}$, $\mc{Z}_{3}$, $\mc{Z}_{4}$, and $\mc{Z}_{5}$.

We first make the following observation.
It is important to note that $J'$ might be a trigraph.
\begin{proposition}
\label{prp:recog:verify}
  Let $(J,Z)$ be a stripe such that $J$ is a graph, let $(J',Z')$ be a stripe, and let $\mc{W} = \{W_{v'} \subseteq V(J) \mid v' \in V(J') \}$.
  Then it can be verified in linear time whether $(J,Z)$ is the thickening~$\mc{W}$ of $(J',Z')$.
\end{proposition}
\begin{proof}
  We verify that:
  \begin{itemize}
    \item $W_{v'}$ is a nonempty clique for each $v' \in V(J')$;
    \item $\bigcup_{v' \in V(J')} W_{v'} = V(J)$;
    \item $|W_{z'}| = 1$ for each $z' \in Z'$;
    \item $\bigcup_{z' \in Z'} W_{z'} = Z$;
    \item for each $u',v' \in V(J')$, if the edge between $u'$ and $v'$ in $J'$ is:
      \begin{itemize}
        \item a nonedge, then $W_{u'}$ is anticomplete to $W_{v'}$;
        \item a (regular) edge, then $W_{u'}$ is complete to $W_{v'}$;
        \item a semi-edge, then $W_{u'}$ is neither complete nor anticomplete to $W_{v'}$.
      \end{itemize}
  \end{itemize}
  The first four items are straightforward to verify in linear time.
  For the fifth item (and its subitems), we group the neighbors of each vertex of $W_{v'}$ of each $v' \in V(J')$ according to which $W_{u'}$ it belongs to.
  Since $v \in W_{v'}$ is only supposed to have neighbors that belong to $W_{u'}$ for which $u'$ and $v'$ are adjacent in $J'$, the grouping takes linear time using bucket sort.
  After that, the last three items take linear time to verify as well by simply counting.
\end{proof}

We now describe linear-time recognition algorithms for thickenings of members $\mc{Z}_{2}$, $\mc{Z}_{3}$, $\mc{Z}_{4}$, and $\mc{Z}_{5}$ in turn.

\subsection{Recognizing Thickenings of \texorpdfstring{$\mc{Z}_{2}$}{Z2}}

\begin{lemma}
\label{lem:recog:z2}
  Let $(J,Z)$ be a stripe such that $J$ is a connected graph that does not admit twins.
  Then we can decide in linear time whether $(J,Z)$ is a thickening of a member of $\mc{Z}_{2}$.
  If so, then we can find such a member and its thickening to $(J,Z)$ as well in the same time.
\end{lemma}
\begin{proof}
  The proof of this lemma consists of two parts.
  In the first part, we argue about the structure of (thickenings) of members of $\mc{Z}_{2}$ and show that each (thickening of such a) member has, without loss of generality, a particular structure.
  In the second part, we give an algorithm to determine whether $(J,Z)$ is a thickening of such a structured member.

  Let $(J',Z')$ be a member of $\mc{Z}_{2}$. Let $n$, $a_0$, $b_0$, $A = \{a_{1},\ldots,a_{n}\}$, $B=\{b_{1},\ldots,b_{n}\}$, $C=\{c_{1},\ldots,c_{n}\}$, and $X$ be as in the definition of $\mc{Z}_{2}$.
  Without loss of generality, we may assume that $a_{1},b_{1}$ are semiadjacent (and thus $c_{1} \in X$), that $a_{2},c_{2}$ are semiadjacent (and thus $b_{2} \in X$), and that $b_{3},c_{3}$ are semiadjacent (and thus $a_{3} \in X$). 
  In particular, this implies that $n \geq 3$ and that if $J'$ has no semi-edges, then $a_i,b_i,c_i \in X$ for $i \in \{1,2,3\}$. 

  Suppose that $(J,Z)$ is a thickening $\mc{W}$ of $(J',Z')$.
  We show that without loss of generality $(J',Z')$ and $\mc{W}$ have specific properties.
  Let $A_{J} = \bigcup_{i \in \{1,\ldots,n\} \mid a_{i} \not\in X} W_{a_{i}}$, $B_{J} = \bigcup_{i \in \{1,\ldots,n\} \mid b_{i} \not\in X} W_{b_{i}}$, and $C_{J} = \bigcup_{i \in \{1,\ldots,n\} \mid c_{i} \not\in X} W_{c_{i}}$.

  \begin{cclaim}
  \label{c:z2:wj1-a}
    If $a \in A_{J}$ is complete to $C_{J}$ and anticomplete to $B_{J}$, then without loss of generality $a \in W_{a_{i}}$ for some $i \in \{4,\ldots,n\}$.
  \end{cclaim}
  \begin{cproof}
    Suppose that $a \in W_{a_{1}}$ or $a \in W_{a_{2}}$ (the case $a \in W_{a_{3}}$ is excluded by the previous assumption that $a_{3} \in X$).
    By assumption, this implies that $c_1 \in X$ respectively that $b_2 \in X$.

    We provide a modified member $(J'',Z'')$ of $\mc{Z}_{2}$ and a modified thickening $\mc{W}''$ to $(J,Z)$.
    Initially, $(J'',Z'')$ is equal to $(J',Z')$ and $\mc{W}'' = \mc{W}$, and in particular $n'',a_0'',b_0'',A'',B'',C'',X''$ are the same as $n,a_0,b_0,A,B,C,X$.
    Now add $1$ to $n''$ (effectively adding a $(n+1)$-th vertex to $A'',B'',C''$), add $b_{n+1},c_{n+1}$ to $X''$, remove $a$ from $W''_{a_{1}}$ or $W''_{a_{2}}$ respectively, and add $a$ to $W''_{a_{n+1}}$.
    At present, $(J'',Z'')$ is still a member of $\mc{Z}_{2}$, and $(J,Z)$ is the thickening $W''$ of $(J'',Z'')$, unless $b_{1} \not\in X$ and $W''_{a_{1}}$ is complete or anticomplete to $W''_{b_{1}}$ respectively $c_{2} \not\in X$ and $W''_{a_{2}}$ is complete or anticomplete to $W''_{b_{2}}$.

    Suppose first that $a \in W_{a_{1}}$ and $b_{1} \not\in X$.
    By the definition of a thickening, $W_{a_{1}}$ is not anticomplete to $W_{b_{1}}$.
    Hence, $W_{a_{1}} \setminus \{a\}$ is not anticomplete to $W_{b_{1}}$.
    Suppose that $W_{a_{1}} \setminus \{a\}$ is complete to $W_{b_{1}}$.
    Then add $1$ to $n''$ (effectively adding a $(n+2)$-th vertex to $A'',B'',C''$), add $a_1,b_1,c_{n+2}$ to $X''$, blank~$W''_{a_{1}}$ and $W''_{b_{1}}$, and set $W''_{a_{n+2}} = W_{a_{1}} \setminus \{a\}$ and $W''_{b_{n+2}} = W_{b_{1}}$.
    Then $(J'',Z'')$ is still a member of $\mc{Z}_{2}$, and $(J,Z)$ is the thickening $W''$ of $(J'',Z'')$.

    The case that $a \in W_{a_{2}}$ can be argued similarly as the previous case.
  \end{cproof}
  \noindent Similarly, one can argue the following.

  \begin{cclaim}
  \label{c:z2:wj1-b}
    If $b \in B_{J}$ is complete to $C_{J}$ and anticomplete to $A_{J}$, then without loss of generality $b \in W_{b_{i}}$ for some $i \in \{4,\ldots,n\}$.
  \end{cclaim}
  \medskip

  \noindent Let $A^{1}_{J}$ denote the set of vertices in $A_{J}$ that are complete to $C_{J}$ but not anticomplete to $B_{J}$; let $B^{1}_{J}$ denote the set of vertices in $B_{J}$ that are complete to $C_{J}$ but not anticomplete to $A_{J}$.

  \begin{cclaim}
  \label{c:z2:wj1-c}
    If $a_{1},b_{1} \not\in X$, then without loss of generality $W_{a_{1}} \subseteq A^{1}_{J}$ and $W_{b_{1}} \subseteq B^{1}_{J}$. 
  \end{cclaim}
  \begin{cproof}
    Suppose that $a_{1},b_{1} \not\in X$.
    By assumption, $c_1 \in X$.
    Hence, $a_{1}$ and $b_1$ are strongly complete to $C \setminus X$ in $J'$.
    Therefore, $W_{a_{1}}$ and $W_{b_{1}}$ are strongly complete to $C_{J}$ in $J$.
    By Claim~\ref{c:z2:wj1-a} and~\ref{c:z2:wj1-b}, without loss of generality every vertex in $W_{a_{1}}$ is not anticomplete to $B_{J}$ and every vertex in $W_{b_{1}}$ is not anticomplete to $A_{J}$.
    It follows that $W_{a_{1}} \subseteq A^{1}_{J}$ and $W_{b_{1}} \subseteq B^{1}_{J}$.
  \end{cproof}

  \begin{cclaim}
  \label{c:z2:wj1}
    If $|A^{1}_{J}|, |B^{1}_{J}| \geq 1$ and $|A^{1}_{J}| + |B^{1}_{J}| \geq 3$, then without loss of generality $a_{1},b_{1} \not\in X$, $W_{a_{1}} = A^{1}_{J}$, and $W_{b_{1}} = B^{1}_{J}$.
    Otherwise, without loss of generality $a_1,b_1 \in X$.
  \end{cclaim}
  \begin{cproof}
    Suppose that $|A^{1}_{J}|, |B^{1}_{J}| \geq 1$ and $|A^{1}_{J}| + |B^{1}_{J}| \geq 3$.
    Observe first that if $a_2 \not\in X$, then $W_{a_{2}} \cap A^{1}_{J} = \emptyset$, because $b_2 \in X$ and thus $a_{2}$ is strongly anticomplete to $B \setminus X$ in $J'$, and therefore every vertex of~$W_{a_{2}}$ is anticomplete to $B_{J}$ in $J$. 
    Next, we observe that if $W_{a_{i}} \cap A^{1}_{J} \not= \emptyset$ for some $i \geq 4$, then $W_{a_{i}} \subseteq A^{1}_{J}$. 
    This is immediate from the fact that $J$ does not admit twins and thus $|W_{a_{i}}| = 1$ for $i \geq 4$.
    Mutatis mutandis, we can argue that if $b_3 \not\in X$, then $W_{b_{3}} \cap B^{1}_{J} = \emptyset$, and if $W_{b_{i}} \cap B^{1}_{J} \not= \emptyset$ for some $i \geq 4$, then $W_{b_{i}} \subseteq B^{1}_{J}$. 

    Let $I_A = \{i \in \{4,\ldots,n\} \mid W_{a_{i}} \subseteq A^{1}_{J}\}$ and let $I_B = \{i \in \{4,\ldots,n\} \mid W_{b_{i}} \subseteq B^{1}_{J}\}$. 
    If $i \in I_A$, then $c_i \in X$, because each vertex in $A^{1}_{J}$ is complete to $C_{J}$, and $a_i$ is strongly antiadjacent to $c_i$ for $i \geq 4$ by the definition of $\mc{Z}_{2}$. Similarly, if $i \in I_B$, then $c_i \in X$. 
    We claim that $I_A = I_B$.
    If $i \in I_A$, then $b_i \not\in X$, because each vertex in $A^{1}_{J}$ has a neighbor in $B_{J}$ and $a_i$ is strongly adjacent to $b_i$ and strongly antiadjacent to $B \setminus\{b_{i}\}$ in $J'$ by the definition of $\mc{Z}_{2}$.
    Therefore, $W_{b_i}$ is complete to $W_{a_i}$, and in particular, not anticomplete to $A_{J}$ in $J$.
    Moreover, $c_i \in X$ as observed before; therefore, $b_i$ is strongly complete to $C \setminus X$ in $J'$, and thus $W_{b_i}$ is complete to $C_{J}$ in $J$.
    Hence, $i \in I_B$, and thus $I_A \subseteq I_B$.
    A similar argument shows that $I_B \subseteq I_A$, proving the claim. 

    Now note that $a_{1},b_{1}\not\in X$ or $|I_A| \geq 2$.
    Indeed, if $a_1,b_1 \in X$, then $|I_A| \geq 1$ since $|A^{1}_{J}| \geq 1$, $W_{a_{2}} \cap A^{1}_{J} = \emptyset$, and $a_3 \in X$.
    If $|I_A| = 1$, then let $I_A = \{i\}$.
    Since $I_B = \{i\}$ by the claim in the previous paragraph and since $W_{a_{2}} \cap A^{1}_{J} = \emptyset$ and $W_{b_{3}} \cap B^{1}_{J} = \emptyset$, $A^{1}_{J} = W_{a_{i}}$ and $B^{1}_{J} = W_{b_{i}}$.
    But then the assumption that $|A^{1}_{J}| + |B^{1}_{J}| \geq 3$ implies that $J$ admits twins, a contradiction.

    If $|I_A| = 0$, then without loss of generality $a_1,b_1 \not\in X$, $W_{a_{1}} = A^{1}_{J}$, and $W_{b_{1}} = B^{1}_{J}$ by Claim~\ref{c:z2:wj1-c}.
    So suppose otherwise.
    We provide a modified member $(J'',Z'')$ of $\mc{Z}_{2}$ and a modified thickening $\mc{W}''$ to $(J,Z)$.
    Initially, $(J'',Z'')$ is equal to $(J',Z')$ and $\mc{W}'' = \mc{W}$, and in particular $n'',a_0'',b_0'',A'',B'',C'',X''$ are the same as $n,a_0,b_0,A,B,C,X$.
    Remove $a_1,b_1$ from $X''$ if either is in~$X$.
    Now add $W_{a_{i}}$ to $W''_{a_{1}}$ and $W_{b_{i}}$ to $W''_{b_{1}}$ for each $i \in I_A$, and reduce $n''$ by $|I_A|$ (effectively removing $a_i,b_i,c_i$ from $A'',B'',C''$ respectively and from $X''$ for each $i \in I_A$).
    Since $c_i \in X$ for each $i \in I_A$, $|C''\setminus X''| = |C \setminus X| \geq 2$, and thus $(J'',Z'')$ is indeed still a member of $\mc{Z}_{2}$.
    Now recall that $I_A = I_B$, and that $a_{1},b_{1}\not\in X$ or $|I_A| \geq 2$ thus implies that $W_{a_{1}} \cup\bigcup_{i \in I_A} W_{a_{i}}$ is neither complete nor anticomplete to $W_{b_{1}} \cup\bigcup_{i \in I_A} W_{b_{i}}$.
    Therefore, $W''_{a_{1}}$ is neither complete nor anticomplete to $W''_{b_{1}}$.
    Moreover, $a_{1},b_{1}\not\in X$ or $|I_A| \geq 2$ implies that $W''_{a_{1}}, W''_{b_{1}} \not= \emptyset$.
    Hence, $(J,Z)$ is the thickening $\mc{W}''$ of $(J'',Z'')$.

    For the second part of the claim, suppose that $|A^{1}_{J}| = 0$, $|B^{1}_{J}| = 0$, or $|A^{1}_{J}| + |B^{1}_{J}| < 3$.
    If $a_{1},b_{1} \not\in X$, then by Claim~\ref{c:z2:wj1-c} without loss of generality $W_{a_{1}} \subseteq A^{1}_{J}$ and $W_{b_{1}} \subseteq B^{1}_{J}$.
    However, since $a_1$ and $b_1$ are semiadjacent, this is not possible by the definition of a thickening.
    If $a_1 \not\in X$ but $b_1 \in X$, then without loss of generality we contradict Claim~\ref{c:z2:wj1-a}.
    Similarly, if $a_1 \in X$ but $b_1 \not\in X$, then without loss of generality we contradict Claim~\ref{c:z2:wj1-b}.
    Hence, $a_1,b_1 \in X$.
  \end{cproof}
  The above claim will be sufficient to localize $W_{a_{1}}$ and $W_{b_{1}}$ in $J$, as we show later.
  We now turn our attention to $W_{a_{2}}$ and $W_{c_{2}}$.

  \begin{cclaim}
  \label{c:z2:wj2-a}
    No vertex in $A_{J}$ is anticomplete to $C_{J}$.
  \end{cclaim}
  \begin{cproof}
    Suppose that $a \in W_{a_{i}}$ for some $i \in \{1,\ldots,n\}$ is anticomplete to $C_{J}$.
    Observe that $a_{i}$ is strongly adjacent to $c_{j}$ for each $j \not =i$, and thus $W_{a_{i}}$ (and particularly $a$) is strongly complete to $W_{c_{j}}$. Since $|C \setminus X| \geq 2$ by the definition of $\mc{Z}_{2}$, such a $j$ indeed exists.
    Therefore, $a$ has an edge to a vertex in $C_{J}$, a contradiction.
  \end{cproof}
  The remaining claims and particularly their proofs are similar to Claim~\ref{c:z2:wj1-b}, \ref{c:z2:wj1-c}, and~\ref{c:z2:wj1}. We include them here for sake of completeness.

  \begin{cclaim}
  \label{c:z2:wj2-b}
    If $c \in C_{J}$ is complete to $A_{J}$ and $B_{J}$, then without loss of generality $c \in W_{c_{i}}$ for some $i \in \{4,\ldots,n\}$.
  \end{cclaim}
  \begin{cproof}
    Suppose that $c \in W_{c_{2}}$ or $c \in W_{c_{3}}$ (the case $c \in W_{c_{1}}$ is excluded by the previous assumption that $c_{1} \in X$).
    By assumption, this implies that $b_2 \in X$ respectively that $a_3 \in X$.

    We provide a modified member $(J'',Z'')$ of $\mc{Z}_{2}$ and a modified thickening $\mc{W}''$ to $(J,Z)$.
    Initially, $(J'',Z'')$ is equal to $(J',Z')$ and $\mc{W}'' = \mc{W}$, and in particular $n'',a_0'',b_0'',A'',B'',C'',X''$ are the same as $n,a_0,b_0,A,B,C,X$.
    Now add $1$ to $n''$ (effectively adding a $(n+1)$-th vertex to $A'',B'',C''$), add $a_{n+1},b_{n+1}$ to $X''$, remove $c$ from $W''_{c_{2}}$ or $W''_{c_{3}}$ respectively, and add $c$ to $W''_{c_{n+1}}$.
    At present, $(J'',Z'')$ is still a member of $\mc{Z}_{2}$, and $(J,Z)$ is the thickening $W''$ of $(J'',Z'')$, unless $a_2 \not\in X$ and $W''_{c_{2}}$ is complete or anticomplete to $W''_{a_{2}}$ respectively $b_3 \not\in X$ and $W''_{c_{3}}$ is complete or anticomplete to $W''_{b_{3}}$.

    Suppose first that $c \in W_{c_{2}}$ and $a_2 \not\in X$.
    By the definition of a thickening, $W_{c_{2}}$ is not anticomplete to $W_{a_{2}}$.
    Hence, $W_{c_{2}} \setminus \{c\}$ is not anticomplete to $W_{a_{2}}$.
    Suppose that $W_{c_{2}} \setminus \{c\}$ is complete to $W_{a_{2}}$.
    Then add $1$ to $n''$ (effectively adding a $(n+2)$-th vertex to $A'',B'',C''$), add $a_2,c_2,b_{n+2}$ to $X''$, blank $W''_{a_{2}}$ and $W''_{c_{2}}$, and set $W''_{a_{n+2}} = W_{a_{2}}$ and $W''_{c_{n+2}} = W_{c_{2}} \setminus \{c\}$.
    Then $(J'',Z'')$ is still a member of $\mc{Z}_{2}$, and $(J,Z)$ is the thickening $W''$ of $(J'',Z'')$.
  \end{cproof}

  Let $A^{2}_{J}$ denote the set of vertices in $A_{J}$ that are anticomplete to $B_{J}$ but not complete to $C_{J}$; let $C^{2}_{J}$ denote the set of vertices in $C_{J}$ that are complete to $B_{J}$ but not complete to $A_{J}$.

  \begin{cclaim}
  \label{c:z2:wj2-c}
    If $a_2,c_2 \not\in X$, then without loss of generality $W_{a_{2}} \subseteq A^{2}_{J}$ and $W_{c_{2}} \subseteq C^{2}_{J}$.
  \end{cclaim}
  \begin{cproof}
    Suppose that $a_2,c_2 \not\in X$.
    By assumption, $b_2 \in X$.
    Hence, $a_2$ is strongly anticomplete to $B \setminus X$ in $J'$ and $c_2$ is strongly complete to $B \setminus X$ in $J'$. Therefore, $W_{a_{2}}$ is anticomplete to $B_{J}$ and $W_{c_{2}}$ is complete to $B_{J}$ in $J$.
    By Claim~\ref{c:z2:wj1-a}, without loss of generality no vertex of $W_{a_{2}}$ is complete to $C_{J}$.
    By Claim~\ref{c:z2:wj2-b}, without loss of generality no vertex of $W_{c_{2}}$ is complete to $A_{J}$.
  \end{cproof}

  \begin{cclaim}
  \label{c:z2:wj2}
    If $|A^{2}_{J}|, |C^{2}_{J}| \geq 1$ and $|A^{2}_{J}| + |C^{2}_{J}| \geq 3$, then without loss of generality $a_{2},c_{2} \not\in X$, $W_{a_{2}} = A^{2}_{J}$, and $W_{c_{2}} = C^{2}_{J}$.
    Otherwise, without loss of generality $a_2,c_2 \in X$.
  \end{cclaim}
  \begin{cproof}
    Suppose that $|A^{2}_{J}|, |C^{2}_{J}| \geq 1$ and $|A^{2}_{J}| + |C^{2}_{J}| \geq 3$.
    Observe first that if $a_1 \not\in X$, then $W_{a_{1}} \cap A^{2}_{J} = \emptyset$, because $c_1 \in X$ and thus $a_{1}$ is strongly complete to $C \setminus X$ in $J'$, and therefore every vertex of $W_{a_{1}}$ is complete to $C_{J}$ in $J$.
    Next, we observe that if $W_{a_{i}} \cap A^{1}_{J} \not= \emptyset$ for some $i \geq 4$, then $W_{a_{i}} \subseteq A^{2}_{J}$. 
    This is immediate from the fact that $J$ does not admit twins and thus $|W_{a_{i}}| = 1$ for $i \geq 4$.
    If $c_3 \not \in X$, then $W_{c_{3}} \cap C^{2}_{J} = \emptyset$, because $a_3 \in X$ and thus $c_3$ is strongly complete to $A \setminus X$ in $J'$, and therefore every vertex of $W_{c_{3}}$ is complete to $A_{J}$ in $J$.
    If $W_{c_{i}} \cap C^{2}_{J} \not= \emptyset$ for some $i \geq 4$, then $W_{c_{i}} \subseteq C^{2}_{J}$, because $J$ does not admit twins and thus $|W_{c_{i}}| = 1$ for $i \geq 4$. 

    Let $I_A = \{i \in \{4,\ldots,n\} \mid W_{a_{i}} \subseteq A^{2}_{J}\}$ and let $I_C = \{i \in \{4,\ldots,n\} \mid W_{c_{i}} \subseteq B^{1}_{J}\}$. 
    If $i \in I_A$, then $b_i \in X$, because each vertex in $A^{2}_{J}$ is anticomplete to $B_{J}$, and $a_i$ is strongly adjacent to $b_i$ for $i \geq 4$ by the definition of $\mc{Z}_{2}$.
    If $i \in I_{C}$, then $b_i \in X$, because each vertex in $C^{2}_{J}$ is complete to $B_{J}$, and $c_i$ is strongly antiadjacent to $b_i$ for $i \geq 4$ by the definition of $\mc{Z}_{2}$. 
    We claim that $I_A = I_C$. If $i \in I_A$, then $c_i \not\in X$, because each vertex in $A^{2}_{J}$ has a nonneighbor in $C_{J}$ and $a_i$ is strongly antiadjacent to $c_i$ and strongly adjacent to $C \setminus\{c_{i}\}$ in $J'$ by the definition of $\mc{Z}_{2}$.
    Therefore, $W_{c_i}$ is anticomplete to $W_{a_i}$, and in particular, not complete to $A_{J}$.
    Moreover, $b_i \in X$ as observed before; therefore, $c_i$ is strongly complete to $B \setminus X$ in $J'$, and thus $W_{c_i}$ is complete to $B_{J}$ in $J$.
    Hence, $i \in I_C$ and $I_A \subseteq I_C$.
    A similar argument argument shows that $I_C \subseteq I_A$, proving the claim. 

    Now note that $a_{2},c_{2}\not\in X$ or $|I_A| \geq 2$.
    Indeed, if $a_2,c_2 \in X$, then $|I_A| \geq 1$ since $|A^{2}_{J}| \geq 1$, $W_{a_{1}} \cap A^{2}_{J} = \emptyset$, and $a_3 \in X$.
    If $|I_A| = 1$, then let $I_A = \{i\}$. Since $I_C = \{i\}$ by the claim in the previous paragraph and since $W_{a_{1}} \cap A^{2}_{J} = \emptyset$ and $W_{c_{3}} \cap C^{2}_{J} = \emptyset$, $A^{2}_{J} = W_{a_{i}}$ and $C^{2}_{J} = W_{c_{i}}$.
    But then the assumption that $|A^{2}_{J}| + |C^{2}_{J}| \geq 3$ implies that $J$ admits twins, a contradiction.

    If $|I_A| = 0$, then without loss of generality $a_2,c_2 \not\in X$, $W_{a_{2}} = A^{2}_{J}$, and $W_{c_{2}} = C^{2}_{J}$ by Claim~\ref{c:z2:wj2-c}.
    So suppose otherwise.
    We provide a modified member $(J'',Z'')$ of $\mc{Z}_{2}$ and a modified thickening $\mc{W}''$ to $(J,Z)$.
    Initially, $(J'',Z'')$ is equal to $(J',Z')$ and $\mc{W}'' = \mc{W}$, and in particular $n'',a_0'',b_0'',A'',B'',C'',X''$ are the same as $n,a_0,b_0,A,B,C,X$. Remove $a_2,c_2$ from $X''$ if either is in~$X$.
    Now add $W_{a_{i}}$ to $W''_{a_{2}}$ and $W_{c_{i}}$ to $W''_{c_{2}}$ for each $i \in I_A$, and reduce $n''$ by $|I_A|$ (effectively removing $a_i,b_i,c_i$ from $A'',B'',C''$ respectively and from $X''$ for each $i \in I_A$).
    Since $b_i \in X$ for each $i \in I_A$, $|B''\setminus X''| = |B \setminus X| \geq 1$, and thus $(J'',Z'')$ is indeed still a member of $\mc{Z}_{2}$.
    Now recall that $I_A = I_C$, and that $a_{2},c_{2}\not\in X$ or $|I_A| \geq 2$ thus implies that $W_{a_{2}} \cup\bigcup_{i \in I_A} W_{a_{i}}$ is neither complete nor anticomplete to $W_{c_{2}} \cup\bigcup_{i \in I_A} W_{c_{i}}$.
    Therefore, $W''_{a_{2}}$ is neither complete nor anticomplete to $W''_{c_{2}}$.
    Moreover, $a_{2},c_{2}\not\in X$ or $|I_A| \geq 2$ implies that $W''_{a_{2}}, W''_{c_{2}} \not= \emptyset$. Hence, $(J,Z)$ is the thickening $\mc{W}''$ of $(J'',Z'')$.

    For the second part of the claim, suppose that $|A^{2}_{J}| = 0$, $|C^{2}_{J}| = 0$, or $|A^{2}_{J}| + |C^{2}_{J}| < 3$. If $a_{2},c_{2} \not\in X$, then by Claim~\ref{c:z2:wj2-c} without loss of generality $W_{a_{2}} \subseteq A^{2}_{J}$ and $W_{c_{2}} \subseteq C^{2}_{J}$. However, since $a_2$ and $c_2$ are semiadjacent, this is not possible by the definition of a thickening. If $a_2 \not\in X$ but $c_2 \in X$, then without loss of generality we contradict Claim~\ref{c:z2:wj2-a}. Similarly, if $a_2 \in X$ but $c_2 \not\in X$, then without loss of generality we contradict Claim~\ref{c:z2:wj2-b}. Hence, $a_2,c_2 \in X$.
  \end{cproof}
  The above claim will be sufficient to localize $W_{a_{2}}$ and $W_{c_{2}}$ in $J$, as we show later. We now turn our attention to $W_{b_{3}}$ and $W_{c_{3}}$.

  Let $B^{3}_{J}$ denote the set of vertices in $B_{J}$ that are anticomplete to $A_{J}$ but not complete to $A_{J}$; let $C^{3}_{J}$ denote the set of vertices in $C_{J}$ that are complete to $A_{J}$ but not complete to $B_{J}$.
  Observe that the definition of $\mc{Z}_{2}$ is symmetric with respect to $A$ and $B$.
  Hence, similarly as to Claim~\ref{c:z2:wj2}, we can prove the following.
  \begin{cclaim}
  \label{c:z2:wj3}
    If $|B^{3}_{J}|, |C^{3}_{J}| \geq 1$ and $|B^{3}_{J}| + |C^{3}_{J}| \geq 3$, then without loss of generality $b_{3},c_{3} \not\in X$, $W_{b_{3}} = B^{3}_{J}$, and $W_{c_{3}} = C^{3}_{J}$.
    Otherwise, without loss of generality $b_3,c_3 \in X$.
  \end{cclaim}
  \medskip

  We are now ready to describe the recognition algorithm.
  Observe that the definition of $\mc{Z}_{2}$ is symmetric for the choice of $a_{0}, b_{0}$.
  That is, if $Z = \{z_{1},z_{2}\}$, then we can choose $a_{0} = z_{1}$ and $b_{0} = z_{2}$, or vice versa.
  Initialize a member of $\mc{Z}_{2}$ by setting $n=3$ and $X = \{a_i,b_i,c_i \mid i=1,2,3\}$, and initialize a thickening $\mc{W}$ such that $W_{a_{0}} = \{z_{1}\}$ and $W_{b_{0}} = \{z_{2}\}$. 

  Now let $A_{J} = N(z_{1})$, $B_{J} = N(z_{2})$, and $C_{J} = V(J)\setminus(A_{J} \cup B_{J} \cup Z)$.
  We verify that $A_{J}$, $B_{J}$, and $C_{J}$ are cliques; if not, then $(J,Z)$ is not a thickening of a member of $\mc{Z}_{2}$. Find the sets $A^{1}_{J}, A^{2}_{J}, B^{1}_{J}, B^{3}_{J}, C^{2}_{J}, C^{3}_{J}$.
  This all takes linear time.

  We now deal with the sets $A^{1}_{J}, A^{2}_{J}, B^{1}_{J}, B^{3}_{J}, C^{2}_{J}, C^{3}_{J}$ as prescribed by Claim~\ref{c:z2:wj1}, \ref{c:z2:wj2}, and~\ref{c:z2:wj3}.
  Observe that these sets are pairwise disjoint by definition.
  If $|A^{1}_{J}|, |B^{1}_{J}| \geq 1$ and $|A^{1}_{J}| + |B^{1}_{J}| \geq 3$, then by Claim~\ref{c:z2:wj1}, we can set $a_{1},b_{1} \not\in X$ (note that $c_{1} \in X$ still), $W_{a_{1}} = A^{1}_{J}$, $W_{b_{1}} = B^{1}_{J}$, we make $a_{1},b_{1}$ semiadjacent, and we discard $A^{1}_{J}$ and $B^{1}_{J}$.
  If $|A^{2}_{J}|, |C^{2}_{J}| \geq 1$ and $|A^{2}_{J}| + |C^{2}_{J}| \geq 3$, then by Claim~\ref{c:z2:wj2}, we can set $a_{2},c_{2} \not\in X$, (note that $b_{2} \in X$ still), $W_{a_{2}} = A^{2}_{J}$, $W_{c_{2}} = C^{2}_{J}$, we make $a_{2},c_{2}$ semiadjacent, and we discard $A^{2}_{J}$ and $C^{2}_{J}$.
  If $|B^{3}_{J}|, |C^{3}_{J}| \geq 1$ and $|B^{3}_{J}| + |C^{3}_{J}| \geq 3$, then by Claim~\ref{c:z2:wj3}, we can set $b_{3},c_{3} \not\in X$, (note that $a_{3} \in X$ still), $W_{b_{3}} = B^{3}_{J}$, $W_{c_{3}} = C^{3}_{J}$, we make $b_{3},c_{3}$ semiadjacent, and we discard~$B^{3}_{J}$ and $C^{3}_{J}$.

  Following Claim~\ref{c:z2:wj1}, \ref{c:z2:wj2}, and~\ref{c:z2:wj3}, all vertices not discarded so far should be in $W_{a_{i}}$, $W_{b_{i}}$, or~$W_{c_{i}}$ for some $i \geq 4$.
  But these are straightforward to recognize using the fact that $J$ does not admit twins. For each undiscarded vertex $a \in A_J$, increase $n$ by $1$, set $W_{a_{n}} = \{a\}$, and
  \begin{itemize}
    \item if $a$ is adjacent to a discarded vertex, or $a$ is adjacent to more than one vertex in $B_{J}$, or $a$ is antiadjacent to more than one vertex in $C_J$, or $a$ is adjacent to a single vertex $b$ in $B_{J}$ and $a$ and $b$ are not antiadjacent only to the same vertex $c$ in $C_{J}$, then answer that $(J,Z)$ is not a thickening of a member of $\mc{Z}_{2}$;
    \item if $a$ is anticomplete to $B_{J}$, then add $b_{n}$ to $X$;
    \item if $a$ is complete to $C_{J}$, then add $c_{n}$ to $X$;
    \item if $a$ is adjacent to a single vertex $b \in B_{J}$, then set $W_{b_{n}} = \{b\}$, make $a_{n}$ strongly adjacent to $b_{n}$, and discard $b$;
    \item if $a$ is antiadjacent to a single vertex $c \in C_{J}$, then set $W_{c_{n}} = \{c\}$, make $a_{n}$ strongly antiadjacent to $c_{n}$, and discard $c$.
  \end{itemize}
  Finally, we discard $a$.
  The correctness of the first step follows from Claim~\ref{c:z2:wj1}, \ref{c:z2:wj2}, and~\ref{c:z2:wj3}, which imply that $a \in W_{a_{i}}$ for $i \geq 4$, as well as from the fact that $J$ does not admit twins and the definition of $\mc{Z}_{2}$.
  We can perform similar steps for each undiscarded vertex in $B_{J}$ and $C_{J}$.
  At the end, we verify that $|C \setminus X| \geq 2$; otherwise, answer that $(J,Z)$ is not a thickening of a member of $\mc{Z}_{2}$.

  By construction, the resulting stripe $(J',Z')$ is indeed a member of $\mc{Z}_{2}$, and the resulting sets $W$ indeed form a thickening of $(J',Z')$ to $(J,Z)$.
\end{proof}

\subsection{Recognizing Thickenings of \texorpdfstring{$\mc{Z}_{3}$}{Z3}}
\begin{lemma}
\label{lem:recog:z3}
  Let $(J,Z)$ be a stripe such that $J$ is a connected graph that does not admit twins.
  Then we can decide in linear time whether $(J,Z)$ is a thickening of a member of $\mc{Z}_{3}$.
  If so, then we can find such a member, its underlying graph $H$, and its thickening to $(J,Z)$ as well in the same time.
\end{lemma}
\begin{proof}
  Let $(J',Z')$ be a member of $\mc{Z}_{3}$.
  Let $H$, $h_1,\ldots,h_5$ be as in the definition of $\mc{Z}_{3}$, and let $P = \{h_1,\ldots,h_5\}$.
  Recall that $J'$ is a line trigraph of $H$, where the vertex corresponding to the edge $h_2h_3$ and the vertex corresponding to the edge $h_3h_4$ are made strongly antiadjacent or semiadjacent.

  Let $Z = \{z_1,z_2\}$.
  Suppose that $(J,Z)$ is a thickening $\mc{W}$ of $(J',Z')$.
  \begin{cclaim}
  \label{c:z3:complete}
    $N(z_1)$ is not complete to $N(z_2)$.
  \end{cclaim}
  \begin{cproof}
    Consider $W_{h_2h_3}$ and $W_{h_3h_4}$.
    By the definition of $\mc{Z}_{3}$, $W_{h_2h_3} \subseteq N(z_1)$ and $W_{h_3h_4} \subseteq N(z_2)$ or vice versa.
    Recall that the vertex corresponding to the edge $h_2h_3$ and the vertex corresponding to the edge $h_3h_4$ are made strongly antiadjacent or semiadjacent in $J'$.
    Hence, by the definition of a thickening, $W_{h_2h_3}$ is not complete to $W_{h_3h_4}$.
    Therefore, $N(z_1)$ is not complete to $N(z_2)$.
  \end{cproof}
  \begin{figure}
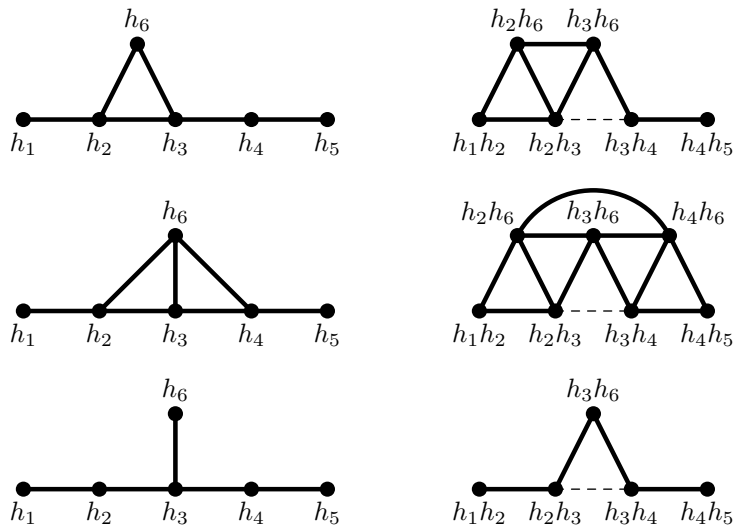

    \begin{center}
      \ig{z3-one-ex1.mps}
    \end{center}
    \begin{center}
      \ig{z3-one-ex3.mps}
    \end{center}
    \begin{center}
      \ig{z3-one-ex2.mps}
    \end{center}
    \caption{Three examples of graphs $H$ with exactly one edge incident to $h_{3}$ and the stripes they induce when used in the definition of $\mc{Z}_{3}$. Observe that the top two graphs in the right column can each be seen as a thickening of the bottom one if $h_2h_3$ is semiadjacent to $h_3h_4$.}
  \label{fig:z3:H}
  \end{figure}

  \begin{cclaim}
  \label{c:z3:one}
    If $h_3$ has degree~$3$, then without loss of generality $H$ is isomorphic to the graph on the bottom left of Fig.~\ref{fig:z3:H} with the possible addition of vertices adjacent to $h_2$ or $h_4$.
  \end{cclaim}
  \begin{cproof}
    Recall that by the definition of $H$, every edge of $H$ must be incident on $h_2$, $h_3$, or $h_4$.
    Consider the graph $H'$ obtained from $H$ by removing any vertices of $V(H) \setminus P$ not adjacent to $h_3$.
    Then the three graphs on the left of Fig.~\ref{fig:z3:H} are the only nonisomorphic possibilities for $H'$.
    However, the top two graphs on the right of Fig.~\ref{fig:z3:H} are thickenings of the bottom graph on the right of Fig.~\ref{fig:z3:H} if $h_2h_3$ is semiadjacent to $h_3h_4$. 
  \end{cproof}

  We are now ready to describe the recognition algorithm.
  Suppose that $V(J) \setminus N[Z] = \emptyset$.
  Then $J$ is a disjoint union of two cliques, $N[z_1]$ and $N[z_2]$.
  We verify that $J$ is indeed a disjoint union of two cliques and that $N(z_1)$ is not complete to $N(z_2)$; otherwise, answer that $(J,Z)$ is not a thickening of a member of $\mc{Z}_{3}$.
  The correctness of this step follows from Claim~\ref{c:z3:complete}.
  Since $J$ is connected, $N(z_1)$ is not anticomplete to $N(z_2)$.
  Hence, we can set $(J',Z')$ as a four-vertex path where the middle edge is a semi-edge.
  Clearly, $(J',Z') \in \mc{Z}_{3}$ and $(J,Z)$ is a thickening of $(J',Z')$.

  Suppose that $V(J) \setminus N[Z] \not= \emptyset$.
  We verify that $N[z_1]$, $N[z_2]$, and $V(J) \setminus N[Z]$ form cliques; otherwise, answer that $(J,Z)$ is not a thickening of a member of $\mc{Z}_{3}$.
  We now argue about a hypothetical member of $\mc{Z}_{3}$ with associated graph $H$, and a thickening $\mc{W}$ of it to $(J,Z)$. 
  Observe that without loss of generality:
  \begin{itemize}
    \item [1.] $e$ is incident on $h_2$ in $H$ if and only if $W_e \subseteq N[z_1]$;
    \item [2.] $e$ is incident on $h_3$ in $H$ if and only if $W_e \cap N[Z] = \emptyset$;
    \item [3.] $e$ is incident on $h_4$ in $H$ if and only if $W_e \subseteq N[z_2]$.
  \end{itemize}
  We note that for each $e \in E(H)$, either $W_e \subseteq N[z_1]$, $W_e \subseteq N[z_2]$, or $W_e \cap N[Z] = \emptyset$ by the definition of a thickening.

  We now identify $W_{h_2h_3}$ and $W_{h_3h_4}$.
  Suppose that $|V(J) \setminus N[Z]| = 1$ and let $v$ be this single vertex.
  By Claim~\ref{c:z3:one} and the above observations, we can focus on $H$ resembling the graph on the bottom left of Fig.~\ref{fig:z3:H}.
  Then, without loss of generality $W_{h_2h_3} = N(z_1) \cap N(v)$ and $W_{h_3h_4} = N(z_2) \cap N(v)$.
  We verify that $N(z_1) \cap N(v)$ and $N(z_2) \cap N(v)$ are both nonempty and not complete to each other; otherwise, answer that $(J,Z)$ is not a thickening of a member of $\mc{Z}_{3}$.
  The correctness follows from the previous discussion.

  Suppose that $|V(J) \setminus N[z]| \geq 2$.
  Let $u,v \in V(J) \setminus N[Z]$.
  Then $u$ and $v$ are from $W_e$ and $W_f$ respectively, where $e \not=f$ because $J$ does not admit twins, and $e$ and $f$ are incident on $h_3$ by Observation~2 above.
  Then, from the definition of $\mc{Z}_{3}$, it follows without loss of generality that $W_{h_2h_3} = N(z_1) \cap N(u) \cap N(v)$ and $W_{h_3h_4} = N(z_2) \cap N(u) \cap N(v)$.
  We verify that $N(z_1) \cap N(u) \cap N(v)$ and $N(z_2) \cap N(u) \cap N(v)$ are both nonempty cliques and that they are not complete to each other; otherwise, answer that $(J,Z)$ is not a thickening of a member of $\mc{Z}_{3}$.
  The correctness follows from the previous discussion.

  The previous paragraphs have identified $W_{h_2h_3}$ and $W_{h_3h_4}$ and that they are not complete to each other.
  If they are anticomplete to each other, then the edge between $h_2h_3$ and $h_3h_4$ is a nonedge; otherwise, this edge is a semi-edge.

  We are now ready to identify the vertices of $V(H) \setminus P$ and their incident edges.
  Let $L_2 = N(z_1) \setminus W_{h_2h_3}$, let $L_3 = V(J) \setminus N[Z]$, and let $L_4 = N(z_2) \setminus W_{h_3h_4}$.
  Note that $L_2$, $L_3$, and $L_4$ correspond to the yet unidentified vertices of $J$ and must belong to $W_e$ for an edge $e$ incident on $h_2$, $h_3$, and $h_4$ respectively by Observation~1, 2, and~3 above.

  \begin{cclaim}
  \label{c:z3:adj}
    Let $v_i \in L_i$ for some $i \in \{2,3,4\}$ be adjacent to a vertex $v_j \in L_j$ for each $j \in \{2,3,4\} \setminus \{i\}$.
    Then these three vertices are pairwise adjacent.
    Moreover, vertex $v_k$ has degree~$2$ outside $L_k \cup Z \cup W_{h_2h_3} \cup W_{h_3h_4}$ for each $k \in \{2,3,4\}$.
  \end{cclaim}
  \begin{cproof}
    Suppose that $i =2$.
    By the above observations, $v_2$, $v_3$, and $v_4$ belong to $W_{e_2}$, $W_{e_3}$, and $W_{e_4}$ respectively, where $e_2$, $e_3$, and $e_4$ are edges incident on $h_2$, $h_3$, and $h_4$ respectively.
    Let $x$ denote the end of $e_2$ that is not $h_2$.
    Since $v_2$ is adjacent to $v_3$, $e_2$ and $e_3$ are both incident on $x$. Since $v_2$ is adjacent to $v_4$, $e_2$ and $e_4$ are incident on $x$. But then $e_2$, $e_3$, $e_4$ are all incident on $x$, and $\{e_2,e_3,e_4\}$ induces a strong clique in $J'$.
    Hence, $v_3$ and $v_4$ are adjacent.
    Since $J$ does not admit twins, $W_{e_k} = \{v_k\}$ for each $k \in \{2,3,4\}$, and thus $v_k$ has degree~$2$ outside $L_k \cup Z \cup W_{h_2h_3} \cup W_{h_3h_4}$.
    The cases that $i$ is $3$ or $4$ follow similarly.
  \end{cproof}
  We first identify the set $D$ of vertices of $V(H) \setminus P$ of degree~$3$.
  Observe that the three vertices corresponding to the edges incident to a vertex of $D$ yield a strong clique in any line trigraph of $H$.
  In particular, the three vertices have no incident semi-edges, because $h_2$, $h_3$, and $h_4$ have degree at least~$3$.
  It follows that every vertex of $D$ creates a triangle of vertices in $J$ with one vertex from each of $L_2$, $L_3$, and $L_4$; moreover, the three vertices of the triangle are complete to $L_2$, $L_3$, and $L_4$ respectively.
  Finally, the vertex of this triangle from $L_2$ is anticomplete to $L_3$ except for the vertex from $L_3$ of the triangle, etc. By Claim~\ref{c:z3:adj}, the converse is also true, that is, if there is a triangle in $J$ with one vertex from each of $L_2$, $L_3$, and $L_4$, then this must be due a vertex of $D$.
  Hence, $|D|$ is equal to the number of triangles in $J$ with one vertex from each of $L_2$, $L_3$, and $L_4$.
  Using this characterization, we can construct $D$ and the sets $W$ corresponding to its incident edges in linear time.
  We first find the vertices in $L_i$ with degree~$1$ to $L_j$ for $i \not= j$.
  These vertices induce an auxiliary graph where every vertex has degree~$2$ outside their own $L_i$, and thus it becomes straightforward to detect the requested triangles in linear time.

  We remove vertices identified in the previous step from $L_2$, $L_3$, and $L_4$.
  By Claim~\ref{c:z3:adj}, every vertex that has not been identified so far is anticomplete to $L_i$ for at least one value of $i \in \{2,3,4\}$.
  Otherwise, answer that $(J,Z)$ is not a thickening of a member of $\mc{Z}_{3}$.
  We now consider $L_2$.
  Suppose that $L_2 \not= \emptyset$.
  Let $L_{2i}$ denote the set of vertices in $L_2$ adjacent to a vertex of $L_i$ for $i = \{3,4\}$.
  By Claim~\ref{c:z3:adj} and the previous step, $L_{23} \cap L_{24} = \emptyset$.
  Otherwise, answer that $(J,Z)$ is not a thickening of a member of $\mc{Z}_{3}$.
  For each $i \in \{3,4\}$, let $N_{2i}$ denote the set of neighbors of $L_{2i}$ in $L_i$.
  Verify that~$N_{23}$ is anticomplete to $N_{24}$; otherwise, answer that $(J,Z)$ is not a thickening of a member of $\mc{Z}_{3}$. The correctness of this step follows from Claim~\ref{c:z3:adj}.
  Add a single edge $e$ to $H$ incident on $h_2$ and a new vertex of $V(H) \setminus P$, set $W_e = L_2 \setminus (L_{23} \cup L_{24})$.
  Suppose that $L_{23} \not= \emptyset$.
  Then $N_{23} \not= \emptyset$.
  Add a new vertex to $V(H) \setminus P$ and make it adjacent to $h_2$ by edge $e_{23}$ and to $h_3$ by edge $f_{23}$.
  Set $W_{e_{23}} = L_{23}$ and $W_{f_{23}} = N_{23}$.
  Since $L_{23}$ is not anticomplete to $N_{23}$ by definition, $W_{e_{23}}$ is either complete to $W_{f_{23}}$ or neither complete nor anticomplete to $W_{f_{23}}$.
  In the first case, we make $e_{23}$ strongly adjacent to $f_{23}$ in~$J'$; in the second case, we make them semiadjacent.
  If $L_{24} \not= \emptyset$, then we apply a similar operation.
  Now remove $N_{23}$ from $L_3$ and $N_{24}$ from $L_{4}$.
  It remains to apply similar operations as above with respect to $L_3$ and $L_4$.

  This completes the description of the recognition algorithm.
  It is easy to see that it runs in linear time.
\end{proof}

\subsection{Recognizing Thickenings of \texorpdfstring{$\mc{Z}_{4}$}{Z4}}

\begin{figure}
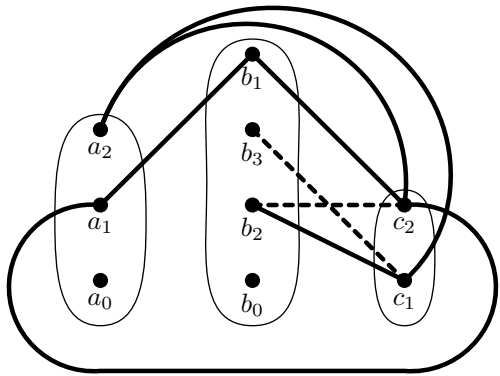

  \begin{center}
    \ig{z4.mps}
  \end{center}
  \caption{The stripe $\mc{Z}_{4}$. The thick dashed lines represent semi-edges. The thick lines represent edges, where the edges of the strong cliques $\{a_0,a_1,a_2\}$, $\{b_0,b_1,b_2,b_3\}$, and $\{c_1,c_2\}$ are not drawn. All other pairs of vertices are strongly antiadjacent.}
\label{fig:z4}
\end{figure}

\begin{lemma}
\label{lem:recog:z4}
  Let $(J,Z)$ be a stripe such that $J$ is a connected graph that does not admit twins.
  Then we can decide in linear time whether $(J,Z)$ is a thickening of a member of $\mc{Z}_{4}$.
  If so, then we can find such a member and its thickening to $(J,Z)$ as well in the same time.
\end{lemma}
\begin{proof}
  We argue about a hypothetical member $(J',Z')$ of $\mc{Z}_{4}$ and its hypothetical thickening~$\mc{W}$ to $(J,Z)$.
  By the definition of $\mc{Z}_{4}$, $z_1$ corresponds to $a_0$ and $z_2$ corresponds to $b_0$ or vice versa.
  Consider the stripe $(J',Z')$ of $\mc{Z}_{4}$ in Fig.~\ref{fig:z4}.
  Then $|N(a_{0})| = 2$, while $|N(b_{0})| \geq 3$. In a thickening of $(J',Z')$ that does not admit twins, these (in)equalities still hold.
  Hence, we can determine the correct correspondence of $z_1$ and $z_2$ in linear time by counting.
  Without loss of generality, $z_1$ corresponds to~$a_0$ and $z_2$ corresponds to $b_0$.
  We verify that $z_1$ indeed has precisely two neighbors and that these are adjacent; otherwise, answer that $(J,Z)$ is not a thickening of a member of $\mc{Z}_{4}$.
  Let $u,v$ denote the neighbors of $z_{1}$.
  Then $W_{a_{1}} = \{u\}$ and $W_{a_{2}} = \{v\}$, or vice versa.
  We determine this more precisely later.
  By definition, we observe that $\{c_{2}\} = (N(a_1) \cap N(a_2))\setminus \{a_{0}\}$ in $J'$.
  Hence, $W_{c_{2}} = (N(u) \cap N(v))\setminus \{z_1\}$ in $J$.
  Similarly, note that $\{b_{1}\} = N(b_{0}) \cap (N(a_1) \cup N(a_2))$ and $\{a_{1}\} = N(b_{1})\setminus(N[b_{0}] \cup \{c_{2}\})$.
  Hence, $W_{b_{1}} = N(z_2) \cap (N(u) \cup N(v))$ and $W_{a_{1}} = N(W_{b_{1}})\setminus(N[z_2] \cup W_{c_{2}})$.
  This enables us to identify whether $W_{a_{1}} = \{u\}$ and $W_{a_{2}} = \{v\}$, or vice versa.
  Without loss of generality, $W_{a_{1}} = \{u\}$ and $W_{a_{2}} = \{v\}$. 

  Now observe that $\{c_{1}\} = N(a_{2}) \setminus \{a_{0},a_{1},c_2\}$ in $J'$.
  Hence, $W_{c_{1}} = N(v) \setminus (\{z_{1},u\} \cup W_{c_{2}})$.
  We have now identified all sets of $\mc{W}$ except $W_{b_2}$ and $W_{b_3}$.
  Let $W_{b_{2}}$ be the subset of the remaining vertices that are complete to $W_{c_{1}}$ and put the other vertices in $W_{b_{3}}$.
  If $W_{b_{3}}$ is anticomplete to $W_{c_{1}}$, then there must be a vertex $v$ we put into $W_{b_{2}}$ that is anticomplete to $W_{c_{2}}$.
  Otherwise, answer that $(J,Z)$ is not a thickening of a member of $\mc{Z}_{4}$.
  Add $v$ to the set $W_{b_{3}}$; then $W_{b_{3}}$ is neither complete nor anticomplete to $W_{c_1}$.

  If the sets we identified are not pairwise disjoint, or not all nonempty, or not all cliques, or do not exhaust $V(J)$, then answer $(J,Z)$ is not a thickening of a member of $\mc{Z}_{4}$.
  Otherwise, by Proposition~\ref{prp:recog:verify}, it can be checked in linear time that $(J,Z)$ is the thickening that we identified of a member of $\mc{Z}_{4}$.
\end{proof}

\subsection{Recognizing Thickenings of \texorpdfstring{$\mc{Z}_{5}$}{Z5}}
\begin{figure}
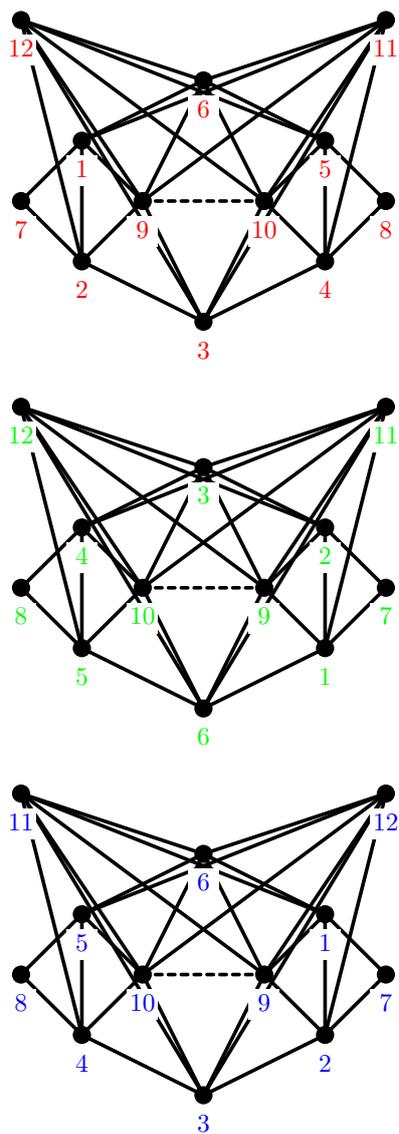

  \begin{center}
    \ig[scale=0.8]{z5-lab1.mps}
  \end{center}
  \begin{center}
    \ig[scale=0.8]{z5-lab2.mps}
  \end{center}
  \begin{center}
    \ig[scale=0.8]{z5-lab3.mps}
  \end{center}
  \caption{The stripe in $\mc{Z}_{5}$ with $|Z|=2$ with three possible labelings. The thick dashed edge is a semi-edge or a nonedge.}
\label{fig:z5}
\end{figure}

\begin{lemma}
\label{lem:recog:z5}
  Let $(J,Z)$ be a stripe such that $J$ is a connected claw-free graph that does not admit twins and $|Z|=2$.
  Then we can decide in linear time whether $(J,Z)$ is a thickening of a member of $\mc{Z}_{5}$.
  If so, then we can find such a member and its thickening to $(J,Z)$ as well in the same time.
\end{lemma}
\begin{proof}
  Suppose that $(J,Z)$ is indeed a thickening of $(J',Z') \in \mc{Z}_{5}$ and let $v_{1},\ldots,v_{13},X$ be as in the definition of $\mc{Z}_{5}$.
  Note that $v_{13}$ is (strongly) adjacent to both $v_{7}$ and $v_{8}$.
  Since $|Z|=2$ and thus $Z' = \{v_{7},v_{8}\}$, and since no vertex of a stripe can be adjacent to more than one vertex of~$Z'$, $v_{13} \in X$.
  There is also quite a bit of symmetry in the labeling of $V(J')$.
  In fact, we can freely swap the labels of $v_{7}$ and $v_{8}$, or of $v_{11}$ and $v_{12}$, by appropriately relabeling the other vertices of $J'$ (\ie relabeling those vertices with label $1,\ldots,6,9,10$).
  See Fig.~\ref{fig:z5}. We also recall Proposition~\ref{prp:xx:claw-free}, which states that the adjacencies along the cycle $v_{1},\ldots,v_{6}$ are strong and the chords of this cycle are nonedges.

  Let $Z = \{z_{1},z_{2}\}$.
  We now argue about a hypothetical member of $\mc{Z}_{5}$ and a thickening $\mc{W}$ of it to $(J,Z)$.
  By the aforementioned symmetry, without loss of generality $W_{v_{7}} = \{z_{1}\}$ and that $W_{v_{8}} = \{z_{2}\}$.
  Observe that $N(v_{7}) = \{v_{1},v_{2}\}$ and $N(v_{8}) = \{v_{4},v_{5}\}$ in $J'$, and that $v_{9}$ is the only vertex adjacent to both $v_{1}$ and $v_{2}$ in $J'$.
  Hence, $W_{v_{9}}$ can be found even though we do not yet know exactly which vertices of $N(v_{7})$ correspond to $W_{v_{1}}$ and which to $W_{v_{2}}$.
  Indeed, since $J$ does not admit twins, $z_1$ has two neighbors, and the common neighbors of these vertices form $W_{v_9}$. Similarly, observe that $\{v_{10}\} = (N(v_{4}) \cap N(v_{5}))\setminus\{v_{8}\}$ in $J'$.
  Hence, $W_{v_{10}}$ can be found in the same way as $W_{v_9}$.
  Now either $W_{v_{9}}$ and $W_{v_{10}}$ are anticomplete to each other and thus $|W_{v_{9}}| = |W_{v_{10}}| = 1$ (because $J$ does not admit twins), or $|W_{v_{9}}|, |W_{v_{10}}| \geq 1$ and $W_{v_{9}}$ and $W_{v_{10}}$ are neither complete nor anticomplete to each other.
  In the first case, the edge between $v_9$ and $v_{10}$ is a nonedge; in the second case, the edge is a semi-edge.

  Now observe that $v_{9}$ and $v_{10}$ share exactly two neighbors in $J'$, namely $v_{3}$ and $v_{6}$, \ie $\{v_{3},v_{6}\} = (N(v_{9})\setminus \{v_{10}\}) \cap (N(v_{10}) \setminus \{v_{9}\})$.
  Hence, $W_{v_3} \cup W_{v_6} = (N(W_{v_{9}})\setminus W_{v_{10}}) \cap (N(W_{v_{10}}) \setminus W_{v_{9}})$.
  Consider the set of remaining vertices.
  These must be split among $W_{v_{11}}$ and $W_{v_{12}}$.
  Recall that $J$ does not admit twins, and since $v_{11}$ and $v_{12}$ have no semi-neighbors in $J'$, $|W_{v_{11}}|, |W_{v_{12}}| \leq 1$.
  Hence, there are three cases.
  If there are no remaining vertices, then $v_{11},v_{12} \in X$.
  If there are two remaining vertices, then they must be antiadjacent in $J$.
  By the aforementioned symmetry, we can make these vertices $W_{v_{11}}$ and $W_{v_{12}}$ arbitrarily.
  Finally, if there is one remaining vertex, then by the aforementioned symmetry, we can make this vertex $W_{v_{11}}$ and add $v_{12}$ to $X$.
  Using the symmetry and the structure of $J'$ (see Fig.~\ref{fig:z5}), it is now straightforward to determine $W_{v_{1}},\ldots,W_{v_{6}}$.

  If the sets we identified are not pairwise disjoint, or not all nonempty, or not all cliques, or do not exhaust $V(J)$, then answer $(J,Z)$ is not a thickening of a member of $\mc{Z}_{5}$.
  Otherwise, by Proposition~\ref{prp:recog:verify}, it can be checked in linear time that $(J,Z)$ is the thickening that we identified of a member of $\mc{Z}_{5}$.
\end{proof}

\section{Towards a Modified Decomposition Theorem for Claw-Free Graphs}
\label{sec:structure}
In this section, we present many supporting lemmas that enable us to slightly modify the decomposition theorem for claw-free graphs proposed by Chudnovsky and Seymour~\cite{ChudnovskyS2008-5}.
Our main contribution compared to their decomposition theorem is that we completely avoid or sidestep certain structures known as bicliques, as well as hex-joins.
In particular, we will not need to find these structures algorithmically later.
We also prove several novel structural results that make it easier for our algorithms to find the decomposition.

The lemmas and theorems in this section draw heavily on the work of Chudnovsky and Seymour~\cite{ChudnovskyS2008-5}, and we explicitly point out these connections.
Hence, lemmas and theorems where we do not mention such a connection are novel and independent of their work.

Throughout the remainder of this section, we only consider trigraphs unless explicitly specified otherwise. 

The following definition of almost-unbreakable stripes $(J,Z)$ is crucial to our structural analysis of claw-free graphs in this section.
It is similar to the definition of unbreakable stripes by Chudnovsky and Seymour~\cite{ChudnovskyS2008-5}, except that we still allow $J$ to contain bicliques (see~\cite{ChudnovskyS2008-5} for their definition) and that $Z$ does not need to be the set of all simplicial vertices in $J$.

\begin{definition}
\label{def:almost-unbreakable}
  We call a stripe $(J,Z)$ \emph{almost-unbreakable} if
  \begin{itemize}
    \item $J$ neither admits a $0$-join, nor a pseudo-$1$-join, nor a pseudo-$2$-join,
    \item there are no twins $u,v \in V(J)\setminus Z$,
    \item there is no $W$-join $(A,B)$ in $J$ such that $Z \cap A, Z \cap B = \emptyset$.
  \end{itemize}
\end{definition}

\subsection{Twins and Proper W-Joins}
We prove some lemmas related to twins and proper W-joins.
The following lemma is inspired by and strengthens~\cite[Theorem~10.1]{ChudnovskyS2008-5}; in particular, we avoid the assumption that $Z$ contains all simplicial vertices of $J$.
\begin{lemma}
\label{lem:twins}
  Let $(J,Z)$ be an almost-unbreakable stripe.
  Then either $J$ does not admit twins, or $J$ is a strong clique with $|V(J)| = 2$ and $|Z| = 1$.
\end{lemma}
\begin{proof}
  Suppose that $J$ admits twins $u,v$.
  Since by the definition of almost-unbreakable $J \setminus Z$ does not admit twins and since $Z$ is strongly stable, exactly one of $u,v$ is in $Z$, and thus $|Z| \geq 1$.
  But then $N[u]=N[v]$ is a strong clique, because each $z \in Z$ is strongly simplicial.
  As $(J,Z)$ is an almost-unbreakable stripe, $(\{u,v\}, V(J)\setminus\{u,v\})$ is not a pseudo-$1$-join of $J$, and thus $V(J) \setminus\{u,v\}$ must be strongly stable.
  Since $J$ does not admit a $0$-join, each vertex of $V(J) \setminus\{u,v\}$ is strongly adjacent to $u$ and $v$. Hence, $J$ is a strong clique, and thus $|Z| = 1$, as $Z$ is strongly stable.
  Moreover, as $J \setminus Z$ does not admit twins, and therefore, $|V(J)| = 2$.
\end{proof}
The following lemma is used implicitly by Chudnovsky and Seymour~\cite{ChudnovskyS2008-5}.
\begin{lemma}
\label{lem:properWjoin}
  Let $(J,Z)$ be a thickening of an almost-unbreakable stripe.
  Then $J$ admits a proper W-join if and only if $J$ admits a proper W-join $(A,B)$ such that $Z \cap A, Z \cap B = \emptyset$.
\end{lemma}
\begin{proof}
  Suppose that $J$ admits a proper W-join $(A,B)$.
  Suppose, \wloge that $Z \cap A \not= \emptyset$.
  As $Z$ is strongly stable, $Z \cap A = \{z\}$ for some $z \in Z$.
  Since $(A,B)$ is a proper W-join, $N[z] \cap B \not= \emptyset$.
  Say $b \in N[z] \cap B$.
  As $A$ is a strong clique and $z \in A$, $A \subseteq N[z]$.
  Since $z$ is strongly simplicial, this implies that $b$ is strongly complete to $A$, contradicting that $(A,B)$ is a proper W-join.
  The converse is trivial.
\end{proof}

\subsection{Circular and Linear Interval Trigraphs}
We show that if a thickening of a circular or linear interval graph has no proper W-join, then it is a proper circular-arc or proper interval graph, respectively.

\begin{lemma}
\label{lem:circular}
  Let graph $G$ be a thickening of a circular interval trigraph.
  Then $G$ is a proper circular-arc graph or $G$ contains a proper W-join.
\end{lemma}
\begin{proof}
  We follow an idea of Eisenbrand~\etal\cite[Lemma~2]{EisenbrandOSV2008}.
  Let $G'$ be a circular interval trigraph with $|V(G')|$ maximum such that there is a set $\mathcal{X} = \{X_{v} \mid v \in V(G')\}$ that is a thickening of $G'$ to~$G$.
  Let $\mc{I}$ be a set of subsets of $\mathbb{S}_{1}$ that defines $G'$.
  Suppose that $G$ does not admit a proper W-join, but that there exist vertices $u',v' \in V(G')$ such that $u',v'$ are semiadjacent.
  Then $(X_{u'},X_{v'})$ is a W-join in $G$.
  Since this is not a proper W-join in $G$, there exists (say) a vertex $w \in X_{u'}$ that is strongly complete or strongly anticomplete to $X_{v'}$ in $G$.
  Create a new vertex $w'$ in $G'$ such that $w'$ is strongly adjacent (if $w$ is compete to $X_{v'}$ in $G$) or strongly antiadjacent (otherwise) to $v'$ and that $w'$ is strongly complete to $N[u'] \setminus \{v'\}$.
  Observe that the resulting graph $G''$ is a circular interval trigraph (by duplicating and slightly moving the point corresponding to $u'$, and possibly adding another subset of $\mathbb{S}_{1}$ to $\mc{I}$).
  Moreover, $G$ is a thickening of $G''$ and $|V(G'')| > |V(G')|$.
  This contradicts the choice of $G'$. Hence, $G'$ has no semi-edges.
  Then it follows from Proposition~\ref{prp:proper-circ} that $G'$ is a proper circular-arc graph.
  As proper circular-arc graphs are closed under inserting twins, the lemma follows.
\end{proof}

\begin{lemma}
\label{lem:interval}
  Let graph $G$ be a thickening of a linear interval trigraph. Then $G$ is a proper interval graph or $G$ contains a proper W-join.
\end{lemma}
\begin{proof}
  We follow the same idea as in the above lemma.
  Let $G'$ be a linear interval trigraph with $|V(G')|$ maximum such that there is a set $\mathcal{X} = \{X_{v} \mid v \in V(G')\}$ that is a thickening of $G'$ to $G$. 
  Let $v_{1},\ldots,v_{\ell}$ be an ordering for the vertices $G'$ as prescribed by the definition of a linear interval trigraph.
  Suppose that $G$ does not admit a proper W-join, but that there exist vertices $v_{i},v_{j} \in V(G')$ with $i < j$ such that $v_{i},v_{j}$ are semiadjacent.
  Then $(X_{v_{i}},X_{v_{j}})$ is a W-join in $G$.
  Since this is not a proper W-join in $G$, there exists (say) a vertex $w \in X_{v_{i}}$ such that $w$ is strongly complete or strongly anticomplete to $X_{v_{j}}$ in $G$.
  Create a new vertex $w'$ in $G'$ such that $w'$ is strongly adjacent (if $w$ is compete to $X_{v_{j}}$ in $G$) or strongly antiadjacent (otherwise) to $v'$ and that $w'$ is strongly complete to $N[v_{i}] \setminus \{v_{j}\}$. 
  Since $v_{i}$ and $v_{j}$ are only semiadjacent, there do not exist $v_{k}$ and $v_{l}$ with $k \leq i$, $j \leq l$, $\{k,l\} \not= \{i,j\}$ such that $v_{k}$ and $v_{l}$ are adjacent.
  Hence, if we insert $w'$ between $v_{i}$ and $v_{i+1}$ when~$w$ is strongly complete to $X_{v_{j}}$, or $w'$ between $v_{i-1}$ and $v_{i}$ when $w$ is strongly anticomplete to $X_{v_{j}}$, then it follows that the resulting graph $G''$ is a linear interval trigraph.
  Moreover, $G$ is a thickening of $G''$ and $|V(G'')| > |V(G')|$.
  This contradicts the choice of $G'$.
  Hence, $G'$ has no semi-edges.
  Then it follows from Proposition~\ref{prp:proper-int} that $G'$ is a proper interval graph.
  As proper interval graphs are closed under inserting twins, the lemma follows.
\end{proof}
The following lemma shows that~\cite[Theorem~12.1]{ChudnovskyS2008-5} goes through under the weaker assumption that $(J,Z)$ is almost-unbreakable.

\begin{lemma}
\label{lem:indecomposable-circular}
  Let $(J,Z)$ be an almost-unbreakable stripe such that $Z \not= \emptyset$.
  If $J$ is a thickening of a circular interval trigraph, then $(J,Z) \in \mathcal{Z}_{1} \cup \mathcal{Z}_{6}$.
\end{lemma}
\begin{proof}
  Let $J$ be a thickening of a circular interval trigraph $J'$.
  By Lemma~\ref{lem:twins}, either $J$ does not admit twins or $J$ is a strong clique with $|V(J)| = 2$ and $|Z|=1$.
  In the latter case, $(J,Z) \in \mc{Z}_{1}$ trivially.
  Hence, we assume that $J$ does not admit twins.

  It remains to follow the proof of~\cite[Theorem~12.1]{ChudnovskyS2008-5}, which relies only on the fact that $Z \not=\emptyset$ and that $J$ does not admit twins, $0$-joins, pseudo-$1$-joins, nor pseudo-$2$-joins.
  This holds by the above paragraph and by the assumption that $(J,Z)$ is almost-unbreakable.
\end{proof}

\subsection{The Union of Two Strong Cliques}
We aim to understand the structure of almost-unbreakable stripes $(J,Z)$ for which $J$ is the union of two strong cliques.

The following lemma shows that~\cite[Theorem~10.2]{ChudnovskyS2008-5} goes through under the weaker assumption that $(J,Z)$ is almost-unbreakable; in particular, we avoid bicliques.

\begin{lemma}
\label{lem:strong-cliques}
  Let $(J,Z)$ be an almost-unbreakable stripe with $|V(J)| > 2$ such that $J$ is the union of two strong cliques.
  Then $|V(J)| \leq 4$, and $(J,Z) \in \mathcal{Z}_{1} \cup \mathcal{Z}_{6}$. 
\end{lemma}
\begin{proof}
  We start with some observations about the structure of $(J,Z)$.
  Since $(J,Z)$ is almost-unbreakable, $J \setminus Z$ does not admit twins.
  As $|V(J)| > 2$, Lemma~\ref{lem:twins} implies that $J$ does not admit twins.
  Let $A,B$ be disjoint strong cliques in $J$ with $A \cup B = V(J)$ such that $A$ is maximal (\ie no vertex of $B$ is strongly complete to $A$).
  Let $X$ denote the set of vertices of $A$ that are strongly complete to $B$. Since $J$ does not admit twins and $|V(J)| > 2$, $|X| \leq 1$ and $|A\setminus X|,|B| \geq 1$.
  As $Z$ is strongly stable, $|Z \cap A|, |Z \cap B| \leq 1$.
  Because vertices of $Z$ are strongly simplicial but no member of $A\setminus X$ is strongly complete to $B$ and vice versa, $X \cap Z = \emptyset$.
  For the same reason, each vertex of $Z$ is either strongly anticomplete to $A \setminus X$ or strongly anticomplete to $B$. Since $(J,Z)$ is almost-unbreakable, $Z$ must contain a vertex of any W-join of $J$.
  Observe, however, that $(A,B)$ is a homogeneous pair and one of $|A|,|B| > 1$.
  Therefore, $Z \not= \emptyset$.

  We now prove that $|V(J)| \leq 4$.
  Let $A' = A \setminus X$.
  Since $(A'\setminus Z, B\setminus Z)$ is a homogeneous pair that is not a W-join (as $(J,Z)$ is almost-unbreakable) and $J$ does not admit twins, $|A'\setminus Z|,|B\setminus Z| \leq 1$.
  Hence $|V(J)| \leq 5$. 
  Suppose that $|V(J)| = 5$.
  Since $|A'\setminus Z| + |B\setminus Z| \leq 2$ and $|Z| \leq 2$, $|A'|+|B| \leq 4$ and thus $|X| = 1$.
  As $(J,Z)$ is a stripe, no vertex of $V(J)$ is adjacent to more than one vertex of $Z$.
  As~$X$ is strongly complete to $A'$ and $B$, $|Z| = 1$.
  Recall that $|A'\setminus Z| + |B\setminus Z| \leq 2$, and thus $|A'| + |B| \leq 3$ and $|V(J)| = |A| + |B| \leq 4$, a contradiction.
  It follows that $|V(J)| \leq 4$.

  It remains to prove that $(J,Z) \in \mathcal{Z}_{1} \cup \mathcal{Z}_{6}$.
  We consider two cases, depending on $|V(J)|$.

  Consider first the case that $|V(J)| = 4$.
  Suppose that $|A| = 3$.
  As $J$ has no $0$-join, a vertex of~$A$ is adjacent to $b$, where $B = \{b\}$.
  If more than one vertex of $A$ is strongly adjacent to $b$, or if more than one vertex of $A$ is strongly antiadjacent to $b$, then $J$ admits twins, a contradiction.
  If all three vertices of $A$ are semiadjacent to $b$, then $Z = \emptyset$, a contradiction.
  If two vertices of $A$ are semiadjacent to $b$, then neither is in $Z$, and those two vertices of $A$ together with $b$ form a W-join in $J \setminus Z$, contradicting that $(J,Z)$ is almost-unbreakable.
  Hence, there is a vertex of $A$ that is strongly adjacent to $b$, a vertex of $A$ that is semiadjacent to $b$, and a vertex of $A$ that is strongly antiadjacent to $b$. Then $(J,Z) \in \mc{Z}_{6}$.
  Suppose that $|A| = 1$, where $A = \{a\}$. 
  Following the same reasoning as before (mutatis mutandis), there is a vertex of $B$ that is strongly adjacent to $a$.
  This, however, contradicts the maximality of $A$.
  Therefore, $|A| = 2$.
  Suppose that $B \cap Z \not= \emptyset$ and $|X| = 1$.
  Then $A \cap Z = \emptyset$ and the vertex in $A'$ must be semiadjacent to the vertex in $B \setminus Z$, as $J$ does not admit twins. Hence, $(J,Z) \in \mc{Z}_{6}$.
  Suppose that $B \cap Z \not= \emptyset$ and $X = \emptyset$.
  As $A$ is neither strongly complete nor strongly anticomplete to $B\setminus Z$, $(A,B\setminus Z)$ is a homogeneous pair, and as $(J,Z)$ is almost-unbreakable, $A \cap Z \not= \emptyset$.
  But then $J$ is a four-vertex path and $(J,Z) \in \mc{Z}_{1}$.
  Suppose that $B \cap Z = \emptyset$.
  Then $A \cap Z \not= \emptyset$.
  As $J$ does not admit twins, $X = \emptyset$.
  Since $(A\setminus Z,B)$ is a homogeneous pair, $A\setminus Z$ is neither strongly complete nor strongly anticomplete to $B$, and as $(J,Z)$ is almost-unbreakable, $B \cap Z \not= \emptyset$, a contradiction.
  This completes the analysis in the case that $|V(J)| = 4$.

  Finally, consider the case that $|V(J)| = 3$.
  Since $Z \not= \emptyset$ and $J$ does not admit twins, $J$ contains at least one strongly antiadjacent pair of vertices. 
  But then $J$ must be a three-vertex path, and thus $(J,Z)$ is a member of $\mc{Z}_{6}$.
\end{proof}

\begin{proposition}
\label{prp:union-two}
  Let $G$ be a trigraph that is a thickening of a trigraph $G'$.
  Then $G$ is the disjoint union of two strong cliques if and only if $G'$ is.
\end{proposition}
\begin{proof}
  Suppose that $G'$ is the disjoint union of two strong cliques.
  Then by the definition of a thickening, so is $G$.

  Suppose that $G$ is the disjoint union of two strong cliques $A, B$. Let $\mc{W} = \{W_{v'} \mid v' \in V(G')\}$ be a thickening of $G$ to $G'$.
  We use induction on the number of vertices $v' \in V(G')$ for which $W_{v'} \cap A \not=\emptyset$ and $W_{v'} \cap B \not=\emptyset$.
  In the base case, for each $v' \in V(G')$ it holds that either $W_{v'} \subseteq A$ or $W_{v'} \subseteq B$. Then $G'$ is the disjoint union of two strong cliques by the definition of a thickening.
  In the inductive step, let $v' \in V(G)$ be such that $W_{v'} \cap A \not=\emptyset$ and $W_{v'} \cap B \not=\emptyset$. Since all vertices of $A \setminus W_{v'}$ are strongly complete to $W_{v'} \cap A$ and all vertices of $B \setminus W_{v'}$ are strongly complete to $W_{v'} \cap B$, it follows that each vertex of $G'$ is adjacent to $v'$.
  Moreover, by the definition of a trigraph, at most one vertex of $G'$ is semiadjacent to $v'$ and all other vertices of $G'$ must be strongly adjacent to $v'$.
  Let $u'$ denote the vertex of $G'$ that is semiadjacent to $v'$ if it exists, and let $u'$ be any other vertex of $V(G')\setminus\{v'\}$ otherwise.
  Since, $G \setminus W_{v'}$ is the union of two strong cliques $A \setminus W_{v'}$ and $B \setminus W_{v'}$ and $G \setminus W_{v'}$ is a thickening $\mc{W} \setminus \{W_{v'}\}$ of $G' \setminus v'$, it follows inductively that $G' \setminus v'$ is the disjoint union of two strong cliques $A'$ and $B'$.
  We then add $v'$ to $A'$ or $B'$ when $A'$ or $B'$ respectively does not contain $u'$.
  Hence, $G'$ is the disjoint union of two strong cliques.
\end{proof}

\begin{corollary}
\label{cor:strong-cliques}
  Let $(J,Z)$ be a thickening of an almost-unbreakable stripe $(J',Z')$ such that $J$ is the union of two strong cliques and $Z \not= \emptyset$. Then $(J,Z)$ is a thickening of a member of $\mc{Z}_{1} \cup \mc{Z}_{6}$.
\end{corollary}
\begin{proof}
  Since $J$ is the union of two strong cliques, $J'$ also is a union of two strong cliques by Proposition~\ref{prp:union-two}.
  If $|V(J')| = 2$, then $J'$ consists of a strongly adjacent pair of vertices.
  Then $(J',Z') \in \mc{Z}_{6}$.
  If $|V(J')| > 2$, then it follows from Lemma~\ref{lem:strong-cliques} that $(J',Z') \in \mc{Z}_{1} \cup \mc{Z}_{6}$.
  The corollary follows.
\end{proof}
The following proposition is inspired by and strengthens~\cite[Claim~2 in Theorem~13.2]{ChudnovskyS2008-5}.
\begin{proposition}
  \label{prp:indecomposable-helper}
  Let $G$ be a trigraph that does not admit a $0$-join nor a pseudo-$2$-join and let $(A,B)$ be a W-join of $G$ such that no vertex of $V(G) \setminus (A \cup B)$ is strongly complete to both $A$ and $B$, that there is a vertex $a \in A$ that is strongly simplicial in $G \setminus B$, and that there is a vertex $b \in B$ that is strongly simplicial in $G \setminus A$.
  Then $G$ is the union of two strong cliques.
\end{proposition}
\begin{proof}
  Let $V_{0} = \emptyset$, $V_{1} = A \cup B$, and $V_{2} = V(G) \setminus (A \cup B)$. Let $A_{1} = A$ and $B_{1} = B$.
  Let $A_{2}$ be the set of neighbors of $a$ in $V_{2}$; since $a$ is strongly simplicial in $G \setminus B$ and every $v \in V_2$ is either strongly complete or strongly anticomplete to $A$ by the definition of a W-join, $A_{1} \cup A_{2}$ is a strong clique. Let $B_{2}$ be the set of neighbors of $b$ in $V_{2}$; using a similar argument as for $A_1 \cup A_2$, $B_{1} \cup B_{2}$ is a strong clique. 
  Note that each vertex of $V_{2} \setminus (A_{2} \cup B_{2})$ is strongly anticomplete to $V_{1}$, because $(A,B)$ is a W-join and thus each vertex is either strongly complete or strongly anticomplete to $A$ or~$B$, and $A_2$ and $B_2$ thus contain all neighbors of $A$ and $B$ in $V_2$.
  Moreover, $A_{2} \cap B_{2} = \emptyset$, as no vertex of $V(G) \setminus (A \cup B)$ is strongly complete to both $A$ and $B$. Since $G$ does not admit a pseudo-$2$-join, $(V_{0},V_{1},V_{2})$ is not a pseudo-$2$-join.
  By the definition of a W-join, $V_{1}$ is not a strong stable set.
  Hence, $V_{2}$ must be a strong stable set.
  Since $G$ does not admit a $0$-join, every vertex in $V_{2}$ is strongly complete to either $A$ or $B$.
  Moreover, at most one vertex of $V_2$ is strongly complete to $A$, because~$V_2$ is a strong stable set and $a$ is strongly simplicial.
  Similarly, at most one vertex of $V_2$ is strongly complete to $B$.
  Therefore, $G$ is the union of two strong cliques ($A_1 \cup A_2$ and $B_1 \cup B_2$).
\end{proof}
The following lemma is an explicit statement of a result that is implicitly obtained in \cite[Theorem~10.3]{ChudnovskyS2008-5}.

\begin{lemma} \label{lem:linetrigraph-join2}
  Let $G$ be a thickening of a line trigraph of a graph $H$ such that $G$ admits no $0$-join, pseudo-$1$-join, or pseudo-$2$-join. Then $G$ is a thickening of the line graph of $H$, or $G$ is the union of two strong cliques.
\end{lemma}
\begin{proof}
  Let $G'$ denote a line trigraph of $H$ such that $G$ is a thickening of $G'$.
  The first paragraph of \cite[Theorem~10.3]{ChudnovskyS2008-5} then proves that $H$ has no vertex of degree two, or $G$ is the union of two strong cliques.
  If $H$ has no vertex of degree two, then the definition of a line trigraph implies that~$G'$ has no semiadjacent vertices.
  In particular, $G'$ is the line graph of $H$.
\end{proof}

\subsection{The Union of Three Strong Cliques}
We aim to understand the structure of (almost-unbreakable) stripes $(J,Z)$ for which $J$ is the union of three strong cliques.

The following lemma is similar to~\cite[Theorem~13.1]{ChudnovskyS2008-5}; in particular, we observe that the proof goes through under the weaker assumption that $(J,Z)$ is any stripe.
We repeat the proof only to be more self-contained.

\begin{lemma}
\label{lem:hexjoin-z}
  Let $(J,Z)$ be a stripe such that $J$ admits a hex-join.
  Then $|Z| \leq 2$.
\end{lemma}
\begin{proof}
  First, we claim that if a three-cliqued trigraph $(G; A,B,C)$ is a hex-join of $(G_{1};A_{1},B_{1},C_{1})$ and $(G_{2};A_{2},B_{2},C_{2})$, then any stable set of size three in $G$ is also contained in either $G_{1}$ or $G_{2}$.
  Suppose that $\{x,y,z\}$ is a stable set.
  By symmetry, we can assume that $x,y \in V(G_{1})$ and $z \in V(G_{2})$, and that $x \in A_{1}$ and $y \in B_{1}$.
  However, in $G$, $A_{1}$ is strongly complete to $V(G_{2}) \setminus B_{2}$ and $B_{1}$ is strongly complete to $V(G_{2}) \setminus C_{2}$.
  In particular, one of $x,y$ is strongly adjacent to $z$ in $G$, a contradiction.
  The claim follows.

  To prove the lemma, suppose for sake of contradiction that $|Z| \geq 3$, and let $z_{1},z_{2},z_{3}$ be three distinct vertices from $Z$.
  Let $J$ be a hex-join of $J_{1}$ and $J_{2}$.
  Since $\{z_{1},z_{2},z_{3}\}$ is a (strong) stable set, it follows from the claim that, \wloge $\{z_{1},z_{2},z_{3}\} \subseteq V(J_{1})$.
  Consider any $v \in V(J_{2})$.
  Since $(J,Z)$ is a stripe, $v$ is (strongly) antiadjacent to at least two of $z_{1},z_{2},z_{3}$.
  This yields a stable set of size three.
  Again, following the claim, the stable set and the two vertices of $z_{1},z_{2},z_{3}$ in particular are in $J_{2}$, a contradiction.
  The lemma follows.
\end{proof}

\begin{lemma}
\label{lem:threecliques}
  Let $(J,Z)$ be a stripe with $|Z| = 2$.
  Then $J$ is the union of three (nonempty) strong cliques if and only if $V(J)\setminus N[Z]$ is a (nonempty) strong clique.
\end{lemma}
\begin{proof}
  Suppose that $V(J)\setminus N[Z]$ is a strong clique.
  Since $(J,Z)$ is a stripe, the vertices in $Z$ are strongly simplicial.
  It follows that $J$ is the union of three strong cliques.

  Let $Z = \{z_{1},z_{2}\}$.
  Suppose that $J$ is the union of three strong cliques $A$, $B$, $C$.
  Since $z_{1}$ and $z_{2}$ are strongly antiadjacent, we can assume that $z_{1} \in A$ and $z_{2} \in B$.
  But then $V(J) \setminus N[Z] \subseteq C$.
\end{proof}
The following lemma and its corollary show that~\cite[Theorem~13.2]{ChudnovskyS2008-5} goes through under the weaker assumption that $(J,Z)$ is almost-unbreakable.
We repeat the proof only to be more self-contained, and to make the connection to our earlier lemmas explicit.

\begin{lemma}
\label{lem:threecliques-zis2}
  Let $(J,Z)$ be an almost-unbreakable stripe with $|Z| = 2$, such that $J$ is the union of three strong cliques.
  Then $(J,Z)$ is a member of $\mc{Z}_{1} \cup \mc{Z}_{2} \cup \mc{Z}_{3} \cup \mc{Z}_{4}$ or $J$ is a line trigraph.
\end{lemma}
\begin{proof}
  Let $Z = \{z_{1},z_{2}\}$. Since $(J,Z)$ is a stripe, no vertex is adjacent to both $z_{1}$ and $z_{2}$.
  Moreover, $|V(J)| > 2$.

  If $J$ is the union of two strong cliques, then it follows from Lemma~\ref{lem:strong-cliques} that $(J,Z) \in \mc{Z}_{1}$ (note that $(J,Z) \not\in \mc{Z}_{6}$, as $|Z| = 2$).
  Hence, we may assume that $J$ is not the union of two strong cliques.
  By Lemma~\ref{lem:twins}, this implies that $J$ does not admit twins.

  Let $J'$ be the trigraph that is obtained from $J$ by making $z_{1},z_{2}$ semiadjacent.
  Note that $J'$ is claw-free, as $z_{1}$ and $z_{2}$ are strongly simplicial in $J$. 

  We claim that $J'$ does not admit a W-join.
  Suppose otherwise and let $(A,B)$ be a W-join of $J'$.
  Since $(J,Z)$ is almost-unbreakable, at least one of $z_{1},z_{2}$ must be in $A \cup B$.
  As $z_{1},z_{2}$ are semiadjacent in $J'$, this implies by the definition of a W-join that (without loss of generality) $z_{1} \in A$ and $z_{2} \in B$.
  It follows that $A$ and $B$ satisfy all the conditions of Proposition~\ref{prp:indecomposable-helper} in $J'$ (adding a semi-edge cannot create a $0$-join or pseudo-$2$-join), and thus Proposition~\ref{prp:indecomposable-helper} implies that $J'$ (and thus also $J$) is the union of two strong cliques, a contradiction.
  This proves the claim.

  Using Lemma~\ref{lem:threecliques}, we may observe that $J'$ and $z_{1},z_{2}$ satisfy all conditions of~\cite[Theorem~11.1]{ChudnovskyS2008-4}.
  Consider the six possible conclusions from applying~\cite[Theorem 11.1]{ChudnovskyS2008-4}.
  In the first case, note that $J'$ does not admit twins as $J$ does not and that $J'$ does not admit a W-join by the claim, a contradiction.
  In the second case, $(J,Z) \in \mc{Z}_{1}$.
  In the third case, $J'$ is a line trigraph of a graph $H$ such that $z_{1}$ and $z_{2}$ (as edges in $H'$) are both incident on the same vertex of degree two in $H'$.
  By `splitting' this vertex, we obtain a graph $H$ where $z_{1}$ and $z_{2}$ have no common end and are each incident on a pendant vertex of $H$.
  It follows that $J$ is a line trigraph of $H$.
  In the fourth case, it can be seen that $(J,Z) \in \mc{Z}_{3}$.
  In the fifth case, either $(J,Z) \in \mc{Z}_{4}$ or $J$ admits a generalized $2$-join, where the latter contradicts that $J$ does not admit a pseudo-$2$-join.
  In the sixth case, $(J,Z) \in \mc{Z}_{2}$.
  The lemma follows.
\end{proof}

\begin{corollary}
\label{cor:threecliques-zis2}
  Let $(J,Z)$ be a thickening of an almost-unbreakable stripe with $|Z| = 2$, such that~$J$ is the union of three strong cliques. 
  Then $(J,Z)$ is a thickening of a member of $\mc{Z}_{1} \cup \mc{Z}_{2} \cup \mc{Z}_{3} \cup \mc{Z}_{4}$ or $J$ is a thickening of a line trigraph.
\end{corollary}
\begin{proof}
  Suppose that $(J,Z)$ is a thickening of an almost-unbreakable stripe $(J',Z')$ with $|Z| = 2$ such that $J$ is the union of three strong cliques.
  Since each $z' \in Z'$ is strongly simplicial, it follows that $V(J) \setminus N[Z]$ is a thickening of $V(J') \setminus N[Z']$. As $V(J) \setminus N[Z]$ is empty or a strong clique by Lemma~\ref{lem:threecliques}, so is $V(J') \setminus N[Z']$.
  Hence, $J'$ is the union of three strong cliques.
  It follows from Lemma~\ref{lem:threecliques-zis2} that $(J',Z')$ is a member of $\mc{Z}_{1} \cup \mc{Z}_{2} \cup \mc{Z}_{3} \cup \mc{Z}_{4}$ or $J'$ is a line trigraph.
\end{proof}

\begin{lemma}
\label{lem:hexjoin}
  Let $(J,Z)$ be a stripe such that $J$ admits a hex-join.
  Then $J$ is the union of two strong cliques, or $V(J)\setminus N[Z] \not= \emptyset$ and $\alpha(J) \leq 3$.
\end{lemma}
\begin{proof}
  We may assume that $J$ is not the union of two strong cliques.
  Since the neighborhood of each $z \in Z$ is a strong clique and $|Z| \leq 2$ by Lemma~\ref{lem:hexjoin-z}, $V(J)\setminus N[Z] \not= \emptyset$.
  Moreover, $J$ is the union of three strong cliques by the definition of a hex-join, and thus $\alpha(J) \leq 3$.
\end{proof}

\subsection{Almost-Unbreakable Stripes and Indecomposable Members}
We prove strong relations between almost-unbreakability and indecomposability.

The following lemma is similar to~\cite[Theorem~10.4]{ChudnovskyS2008-5}, and essentially shows that it goes through under the weaker assumption that $(J,Z)$ is almost-unbreakable.
\begin{lemma}
\label{lem:stripe-indecomposable}
  Let $(J,Z)$ be a thickening of an almost-unbreakable stripe.
  Then either $J$ is a thickening of an indecomposable member of $\mathcal{S}_{0},\ldots,\mathcal{S}_{7}$ or $J$ admits a hex-join.
\end{lemma}
\begin{proof}
  Let $(J,Z)$ be a thickening of an almost-unbreakable stripe $(J',Z')$ (where possibly $J = J'$).
  Since $(J',Z')$ is almost-unbreakable, $J'$ does not admit a $0$-join, a pseudo-$1$-join, nor a pseudo-$2$-join.
  Following the observations near the definitions of pseudo-$1$-join and pseudo-$2$-join, this means that~$J'$ does not admit a $0$-join, a $1$-join, nor a generalized $2$-join.
  Then, following \cite[Theorem~10.4]{ChudnovskyS2008-5}, either $J'$ is a thickening of an indecomposable member of $\mathcal{S}_{0},\ldots,\mathcal{S}_{7}$ or $J'$ admits a hex-join.
  The lemma follows using Proposition~\ref{prp:join-thicken}.
\end{proof}
We prove the following corollary.

\begin{corollary}
  Let $(J,Z)$ be a thickening of an almost-unbreakable stripe.
  Then either $(J,Z)$ is a thickening of an almost-unbreakable stripe $(J',Z')$ such that $J'$ is an indecomposable member of one of $\mc{S}_{0},\ldots,\mc{S}_{7}$, or $J$ is the union of three strong cliques and $|Z| \leq 2$.
\end{corollary}
\begin{proof}
  Following Lemma~\ref{lem:stripe-indecomposable}, either $J$ is a thickening of an indecomposable member of $\mathcal{S}_{0},\ldots,\mathcal{S}_{7}$, or $J$ admits a hex-join.
  In the latter case, it follows from the definition of a hex-join and Lemma~\ref{lem:hexjoin-z} that $J$ is the union of three strong cliques and $|Z| \leq 2$.
\end{proof}

A vertex $v$ of a trigraph $G$ is \emph{near-simplicial} if $v$ is semiadjacent to some vertex and the set of strong neighbors of $v$ is a strong clique~\cite{ChudnovskyS2008-5}.

\begin{lemma}[{Chudnovsky and Seymour~\cite[Theorem~11.1]{ChudnovskyS2008-5}}]
\label{lem:indecomposable-simplicial}
  Let $J \in \mathcal{S}_{i}$ for some $i \in \{1,2,4,5,6,7\}$ and suppose that $J$ is indecomposable and not the union of two strong cliques.
  \begin{enumerate}
    \item\label{lem:indecomposable-simplicial:1} If $z \in V(J)$ is a simplicial vertex, let $Z$ be the set of all simplicial vertices of $J$. Then $|Z| \leq 2$ and $(J,Z) \in \mc{Z}_{j}$ for some $j \in \{2,5,7,8,9\}$.
    \item\label{lem:indecomposable-simplicial:2} If $z \in V(J)$ is a near-simplicial vertex semiadjacent to $z'$, let $Z = \{z,z'\}$.
      Then $(J',Z) \in \mathcal{Z}_{2} \cup \mathcal{Z}_{5}$, where $J'$ is the trigraph obtained from $G$ by making $z,z'$ strongly antiadjacent.
  \end{enumerate}
\end{lemma}
The following lemma and its corollary show that~\cite[Theorem~12.2]{ChudnovskyS2008-5} goes through under the weaker assumption that $(J,Z)$ is almost-unbreakable.
We repeat the proof only to be more self-contained, and to make the connection to our earlier lemmas explicit.

\begin{lemma}
\label{lem:indecomposable-indecomposable}
  Let $(J,Z)$ be an almost-unbreakable stripe with $Z \not= \emptyset$.
  If $J$ is a thickening of an indecomposable member of $\mathcal{S}_{i}$, where $i \in \{1,\ldots,7\}$, then $(J,Z) \in \mathcal{Z}_{j}$, where $j \in \{1,2,5,6,7,8,9\}$.
\end{lemma}
\begin{proof}
  Let $J'$ be a trigraph such that $J'$ is an indecomposable member of $\mathcal{S}_{i}$ for some $i \in \{1,\ldots,7\}$ and that there is a set $\mathcal{W} = \{W_{v'} \mid v' \in V(G')\}$ that is a thickening of $J'$ to $J$.
  By Lemma~\ref{lem:twins}, either $J$ does not admit twins or $J$ is a strong clique with $|V(J)| = 2$ and $|Z|=1$.
  In the latter case, $(J,Z) \in \mc{Z}_{6}$ trivially.
  Hence, we may assume that $J$ does not admit twins.
  If $J'$ is the union of two strong cliques, then so is $J$ by Proposition~\ref{prp:union-two}.
  If $J$ is the union of two strong cliques, then the result follows from Lemma~\ref{lem:strong-cliques}.
  If $i=3$, then $(J,Z) \in \mathcal{Z}_{1} \cup \mathcal{Z}_{6}$ by Lemma~\ref{lem:indecomposable-circular}.
  So assume that neither $J'$ nor $J$ is the union of two strong cliques and that $i\not=3$.

  We claim that $|W_{v'}| = 1$ for each $v' \in V(J')$.
  Suppose otherwise.
  Since $J$ does not admit twins, there exist $u',v'\in V(J')$ such that $u',v'$ are semiadjacent and $(W_{u'},W_{v'})$ is a W-join in $J$.
  As $(J,Z)$ is almost-unbreakable, without loss of generality, $Z \cap W_{u'} \not= \emptyset$.
  Because each $z \in Z$ is strongly simplicial, $u'$ is near-simplicial in $J'$.
  This, combined with the assumption that $J'$ is not the union of two strong cliques, $i\not=3$, and $J'$ is indecomposable, implies by Lemma~\ref{lem:indecomposable-simplicial}:\ref{lem:indecomposable-simplicial:2} that $(J'',Z'')$, where~$J''$ is obtained from $J'$ by making $u',v'$ strongly antiadjacent and $Z'' = \{u',v'\}$, is a stripe.
  Hence, $v'$ is also near-simplicial in $J'$ and no vertex of $J'$ is adjacent to both $u'$ and $v'$.
  It follows that $A = W_{u'}$ and $B = W_{v'}$ satisfy all the conditions of Proposition~\ref{prp:indecomposable-helper} in $J$, and thus Proposition~\ref{prp:indecomposable-helper} implies that $J$ is the union of two strong cliques, a contradiction.
  This proves the claim.

  Following the claim, $J$ and $J'$ are isomorphic.
  Hence, $J \in \mathcal{S}_{i}$. Since $Z \not= \emptyset$, $J$ has a simplicial vertex.
  By Lemma~\ref{lem:indecomposable-simplicial}:\ref{lem:indecomposable-simplicial:1}, $(J,Z) \in \mathcal{Z}_{j}$ where $j \in \{2,5,7,8,9\}$.
  The lemma follows.
\end{proof}
The following corollary is immediate.
\begin{corollary}
  Let $(J,Z)$ be an almost-unbreakable stripe such that $Z \not= \emptyset$ and $J$ is a thickening of an indecomposable member of one of $\mathcal{S}_{1},\ldots,\mathcal{S}_{7}$.
  Then $|Z| \leq 2$.
\end{corollary}

\subsection{Stability Numbers}
We prove bounds on the stability number of certain trigraphs.
We start with the following observation.

\begin{proposition}[{Hermelin~\etal\cite[Proposition~2]{HermelinMvL2014}}] \label{prp:independent:observation}
  Let $G$ be a trigraph that is a thickening of a trigraph $G'$.
  Then $\alpha(G) = \alpha(G')$.
\end{proposition}

\begin{lemma}
  \label{lem:independent}
  Let $G$ be a trigraph.
  If $G \in \mc{S}_{1} \cup \mc{S}_{4} \cup \mc{S}_{5} \cup \mc{S}_{6} \cup \mc{S}_{7}$, then $\alpha(G) \leq 3$.
  If $G \in \mc{S}_{2}$, then $\alpha(G) \leq 4$.
\end{lemma}
\begin{proof}
  The result for the case that $G \in \mc{S}_{2}$ was shown by Hermelin~\etal\cite[Proposition~1]{HermelinMvL2014}.

  Throughout, let $G'$ be the graph obtained from $G$ by removing all semi-edges.
  Note that deleting (semi-)edges can only increase $\alpha$, and thus $\alpha(G) \leq \alpha(G')$.
  Therefore, we show an upper bound on~$\alpha(G')$.
  Let $I$ be any stable set of $G'$ and suppose that $|I| > 3$.
  Consider the various cases depending on to which class of graphs $G$ belongs:
  \begin{itemize}
    \item[$\mathcal{S}_{1}$:] Consider the definition of $\mc{S}_{1}$.
      Since deleting vertices or adding (semi)edges can only reduce~$\alpha$, it suffices to show that $\alpha(G) \leq 3$ if $G = G_{0}$.
      So let $v_{1},\ldots,v_{12}$ be as in the definition of~$G_{0}$.
      If $v_{11}, v_{12} \in I$, then $|I| = 2$, a contradiction.
      Let $I' = I \cap \{v_{1},\ldots,v_{10}\}$ and let $J$ denote the set of indices of the vertices in $I'$.
      If $J$ contains only even or only odd integers, then $|I'| \leq 2$ as $v_{i}$ is adjacent to $v_{i+2}$ (indices modulo $10$) for $1 \leq i \leq 10$.
      Hence, $|I| \leq 3$, a contradiction.
      But then $J$ contains both odd and even integers and thus $v_{11},v_{12} \not\in I$. If $j \in J$, then $j-2, j-1, j+1, j+2 \not\in J$ (integers modulo $10$).
      Hence, $|I| = |J| \leq \lfloor 10/3 \rfloor = 3$, a contradiction.
      Therefore, $\alpha(G) \leq \alpha(G') \leq 3$.
    \item[$\mathcal{S}_{4}$:] Let $H, h_{1},\ldots,h_{7}$ be as in the definition of $\mathcal{S}_{4}$.
      Let $E_{6}$ denote the set of edges of $H$ incident with $h_{6}$ and let $x$ be the vertex added to $L(H)$ to obtain $G$. Note that $x$ is strongly adjacent to the edges and chords of the cycle $C = \{h_{1},\ldots,h_{5}\}$; denote this set of edges and chords by $E(C)$.
      Observe that $V(G') = E_{6} \cup E(C) \cup \{x\}$.
      Since $E_{6}$ is a strong clique in $G'$, $|I \cap E_{6}| \leq 1$.
      Hence, if $x \in I$, then $|I| \leq 2$, a contradiction.
      Otherwise, as $G'[E(C)]$ is can be covered by five strong cliques such that each vertex is in at least two of these cliques, $|I \cap E(C)| \leq 2$ and thus $|I| \leq 3$, a contradiction.
      Therefore, $\alpha(G) \leq \alpha(G') \leq 3$.
    \item[$\mathcal{S}_{5}$:] Let $A, B, C, d_{1},\ldots,d_{5},X$ be as in the definition of $\mathcal{S}_{5}$.
      Since $A,B,C$ are strong cliques, $|I \cap (A \cup B \cup C)| \leq 3$.
      Suppose that $|I \cap (A \cup B \cup C)| = 3$.
      Then $I$ has precisely one vertex from each of $A,B,C$, say $a_{i}, b_{j}, c_{k}$ respectively.
      Since $c_{k}$ is strongly adjacent to $b_{j}$ if and only if $j \not= k$, $j = k$.
      Similarly, $i = k$.
      But then $i = j$, thus $a_{i}$ and $b_{j}$ are adjacent in $G$.
      By the construction of $G'$, $a_{i}$ and $b_{j}$ must be semiadjacent in $G$ and thus  $c_{k} \in X$.
      Hence, $c_{k} \not\in V(G)$ and thus $c_{k} \not\in V(G')$, a contradiction.
      Hence, $|I \cap (A \cup B \cup C)| \leq 2$.
      Note that $d_{3},d_{4},d_{5}$ are pairwise strongly adjacent, so $|I \cap \{d_{3},d_{4},d_{5}\}| \leq 1$.
      Therefore, $d_{1} \in I$ or $d_{2} \in I$.
      As $d_{1}$ is strongly complete to $A \cup B \cup C$, $d_{1} \in I$ implies that $|I| \leq 3$, a contradiction.
      Hence, $d_{1} \not\in I$ and $d_{2} \in I$.
      As $d_{2}$ is strongly complete to $A \cup B$, $|I \cap (A \cup B \cup C)| \leq 1$, implying that $|I| \leq 3$, a contradiction. Therefore, $\alpha(G) \leq \alpha(G') \leq 3$.
    \item[$\mathcal{S}_{6}$:] Observe that $G$ is the union of three strong cliques, and thus so is $G'$.
      Therefore, $|I| \leq 3$ and $\alpha(G) \leq \alpha(G') \leq 3$.
    \item[$\mathcal{S}_{7}$:] Since for any $X \subseteq V(G)$ with $|X| = 4$, at least two pairs of vertices in $X$ are strongly adjacent, the same holds with respect to $G'$.
      Applying this to any subset of $I$ of size four, we obtain a contradiction to the assumption that $I$ is stable.
\end{itemize}
This completes the proof.
\end{proof}
From this lemma and Proposition~\ref{prp:independent:observation}, we immediately obtain the following corollary.

\begin{corollary}
\label{cor:independent}
  Let $G$ be a trigraph that is a thickening of a trigraph $G'$.
  If $G' \in \mc{S}_{1} \cup \mc{S}_{4} \cup \mc{S}_{5} \cup \mc{S}_{6} \cup \mc{S}_{7}$, then $\alpha(G) \leq 3$.
  If $G' \in \mc{S}_{2}$, then $\alpha(G) \leq 4$.
\end{corollary}

\subsection{Supporting Lemma}
Finally, we present a crucial supporting lemma.
Although the lemma is ours, the proof is somewhat inspired by the proof of~\cite[Theorem~7.2]{ChudnovskyS2008-5}.

\begin{lemma}
\label{lem:tos2}
  Let $G$ be a connected claw-free graph with $\alpha(G) > 3$ such that $G$ does not admit twins or proper W-joins.
  If $G$ is not a line graph nor a proper circular-arc graph, then $G$ is a thickening of a member of $\mathcal{S}_{2}$, or $G$ admits a pseudo-$1$-join or a pseudo-$2$-join.
\end{lemma}
\begin{proof}
  Suppose that $G$ does not admit a pseudo-$1$-join or a pseudo-$2$-join.
  Hence, $G$ does not admit a $1$-join or a generalized $2$-join.
  As $G$ is connected, $G$ does not admit a $0$-join neither.
  Since $\alpha(G) > 3$, $G$ is not the union of three strong cliques, and thus $G$ does not admit a hex-join.
  Let $G'$ be a trigraph with $|V(G')|$ minimum among all trigraphs of which $G$ is a thickening (\ie $G$ is a thickening of $G'$). Observe that $G'$ does not admit a $0$-join, a $1$-join, a generalized $2$-join, nor a hex-join, as~$G$ does not.
  Moreover, $G'$ does not admit twins nor W-joins, as $|V(G')|$ is minimum.
  Hence, $G'$ is indecomposable, and thus $G' \in \mathcal{S}_{0} \cup \cdots \cup \mathcal{S}_{7}$ following Theorem~\ref{thm:chudsey-main}.
  From Corollary~\ref{cor:independent} and the assumption that $\alpha(G) > 3$, it follows that $G' \in \mathcal{S}_{0} \cup \mathcal{S}_{2} \cup \mathcal{S}_{3}$ and thus $G$ is a thickening of a member of $\mathcal{S}_{0} \cup \mathcal{S}_{2} \cup \mathcal{S}_{3}$.

  Suppose that $G$ is a thickening of a member of $\mathcal{S}_{0}$.
  Following Lemma~\ref{lem:linetrigraph-join2} and the fact that $\alpha(G) > 3$, $G$ is a thickening of a line graph.
  As $G$ does not admit twins, $G$ is in fact a line graph, a contradiction.

  Suppose that $G$ is a thickening of a member of $\mathcal{S}_{3}$.
  Since $G$ does not admit a proper W-join, it follows from Lemma~\ref{lem:circular} that $G$ is a proper circular-arc graph, a contradiction.

  Therefore, $G$ is a thickening of a member of $\mathcal{S}_{2}$, and the lemma follows.
\end{proof}

\section{An Algorithmic Decomposition for Claw-Free Graphs}
\label{sec:decomp}
In this section, we obtain our algorithmic decomposition theorem for claw-free graphs following from the decomposition approach for claw-free graphs by Chudnovsky and Seymour~\cite{ChudnovskyS2008-5}. 
The main idea is to decompose the graph along $0$-joins, pseudo-$1$-joins, and pseudo-$2$-joins, so that we obtain a strip-structure of the graph where each strip is a spot or a thickening of an almost-unbreakable stripe.
Almost-unbreakable stripes that are line graphs are decomposed even further into trivial line graph strips (spots or two-vertex stripes).
We then apply the supporting structural lemmas of Sect.~\ref{sec:structure} to show that the other strips indeed belong to a few basic graph classes, \ie the ones described in Sect.~\ref{sec:defs:stripes}.

As before, given a graph $G$, we use $n = |V(G)|$ and $m = |E(G)|$.

\subsection{Finding a Strip-Structure}
We give an algorithm that finds a strip-structure of a claw-free graph such that the resulting strips are spots or thickenings of almost-unbreakable stripes (recall Definition~\ref{def:almost-unbreakable}).
The proof of this theorem is a combination of the observation that~\cite[Theorem~8.1 and Theorem 9.1]{ChudnovskyS2008-5} go through if only we need to decompose to spots and almost-unbreakable stripes, and the algorithms to find twins and joins that we developed in Sect.~\ref{sec:joins}.
We give (most of) the proof only to be self-contained, and to make the dependence on the algorithms that find twins and joins explicit.

\begin{theorem}
\label{thm:strip-structure}
  Every claw-free trigraph $G$ admits a purified strip-structure with nullity zero such that all its strips are either spots or thickenings of almost-unbreakable stripes.
  If $G$ is a graph, such a strip-structure can be found in $\stripTime$ time.
\end{theorem}
\begin{proof}
  Let $H$ be such that $V(H) = \emptyset$, $E(H) = \{F\}$, and the incidence relation is empty, and let $\eta$ be such that $\eta(F) = V(G)$.
  Then $(H,\eta)$ is a purified strip-structure for $G$ with nullity zero.

  We now define five conditions on the strips of $H$; when a particular condition is met for a strip $F\in E(H)$, we apply a transformation to the strip-structure.
  We call a condition and its transformation a \emph{rule}.
  We iteratively apply the rules: we do not apply rule $i$ until rule $i-1$ cannot be applied to any strip.
  Moreover, after applying a rule, we check for rule $1$ again, etc., until no more rules can be applied to any strip.
  The rules closely follow those proposed by Chudnovsky and Seymour~\cite[Theorem~8.1 and Theorem~9.1]{ChudnovskyS2008-5}.

  Consider any $F \in E(H)$ and its corresponding strip $(J,Z)$.
  Assume that $(J,Z)$ is not a spot.
  Let $\overline{F} = \{h_{1}, \ldots,h_{k}\}$ and let $Z = \{z_{1},\ldots,z_{k}\}$.
  \begin{quote}
    (1) \emph{$F$ is not purified}
  \end{quote}
  If $k \leq 1$, then $F$ is purified.
  So suppose, \wloge that $k \geq 2$ and $\eta(F,h_{1}) \cap \eta(F,h_{2}) \not= \emptyset$.
  Let $W = \eta(F,h_{1}) \cap \eta(F,h_{2})$.
  If $k \geq 3$, then $W \cap \eta(F,h_{i}) = \emptyset$ for $3 \leq i \leq k$ by the definition of a circus.
  Moreover, again by the definition of a circus, $W$ is strongly anticomplete to $\eta(F)\setminus(\eta(F,h_{1}) \cup \eta(F,h_{2}))$.
  Finally, by the definition of a strip-structure, $W$ is a strong clique of $G$.

  Let $W = \{w_{1},\ldots,w_{\ell}\}$.
  Create a new edge $F'_{i}$ for each $w_{i}$ and construct a new strip-structure $(H',\eta')$ for $G$ as follows:
  \begin{itemize}
    \item $V(H') = V(H)$ and $E(H') = E(H) \cup \{F'_{1},\ldots,F'_{\ell}\}$,
    \item for each $F_{0} \in E(H)$ and $h \in V(H)$, $F_{0}$ is incident with $h$ in $H'$ if and only if they are incident in $H$,
    \item each $F'_{i}$ is incident only with $h_{1},h_{2}$ for all $1 \leq i \leq \ell$,
    \item for each $F_{0} \in E(H) \setminus\{F\}$, $\eta'(F_{0}) = \eta(F_{0})$ and $\eta'(F_{0},h) = \eta(F_{0},h)$ for all $h \in \overline{F_{0}}$,
    \item $\eta'(F) = \eta(F)\setminus W$, $\eta'(F,h_{1}) = \eta(F,h_{1})\setminus W$, $\eta'(F,h_{2}) = \eta(F,h_{2})\setminus W$, and $\eta'(F,h_{i}) = \eta(F,h_{i})$ for all $3 \leq i \leq k$,
    \item $\eta'(F'_{i}) = \{w_{i}\}$, $\eta'(F'_{i},h_{1}) = \{w_{i}\}$, and $\eta'(F'_{i},h_{2}) = \{w_{i}\}$ for all $1 \leq i \leq \ell$.
  \end{itemize}
  If $\eta(F) = W$, we remove $F$ from $H'$.
  It can be quickly verified that $(H',\eta')$ is a strip-structure for $G$.
  Furthermore, $F'_{1},\ldots,F'_{\ell}$ are all purified.
  Finally, note that $(H',\eta')$ can be constructed in linear time.

  \medskip
  Observe that the nullity of $(H',\eta')$ might be higher than the nullity of $(H,\eta)$.
  This happens only if $\eta(F,h_{1})=W$ or $\eta(F,h_{2}) = W$.
  Hence, the nullity of $(H',\eta')$ is at most two more than the nullity of $(H,\eta)$.
  Thus, we need a rule to reduce the nullity.

  \begin{quote}
    (2) \emph{$\eta(F,h_{i}) = \emptyset$ for some $h_{i} \in \overline{F}$}
  \end{quote}
   Construct a new strip-structure $(H',\eta')$ for $G$ as follows.
  \begin{itemize}
    \item $V(H') = V(H)$ and $E(H') = E(H)$,
    \item for each $F_{0} \in E(H)\setminus\{F\}$ and $h \in V(H)$, $F_{0}$ is incident with $h$ in $H'$ if and only if they are incident in $H$,
    \item for each $h \in V(H) \setminus\{h_{i}\}$, $F$ is incident with $h$ in $H'$ if and only if they are adjacent in $H$ (note that this implies that $F$ is not incident with $h_{i}$ in $H'$),
    \item for each $F_{0} \in E(H) \setminus\{F\}$, $\eta'(F_{0}) = \eta(F_{0})$ and $\eta'(F_{0},h) = \eta(F_{0},h)$ for all $h \in \overline{F_{0}}$,
    \item $\eta'(F) = \eta(F)$ and $\eta'(F,h) = \eta(F,h)$ for all $h \in \overline{F}\setminus\{h_{i}\}$.
  \end{itemize}
  Clearly, $(H',\eta')$ is a strip-structure for $G$.
  Note that $(H',\eta')$ can be constructed in linear time.

  \medskip
  After exhaustively applying Rule (1) and (2), $(H,\eta)$ is a purified strip-structure of nullity zero. In particular, $(J,Z)$ is a stripe.

  \begin{quote}
    (3) \emph{$J$ admits a $0$-join $(V_{1},V_{2})$}
  \end{quote}
  For $j = 1,2$, let $Z_{j} = Z \cap V_{j}$ and let $P_{j} = \{h_{i} \mid 1 \leq i \leq k \mbox{\ and\ } z_{i} \in Z_{j}\}$. Clearly, $P_{1} \cap P_{2} = \emptyset$ and $P_{1} \cup P_{2} = \overline{F}$.
  As $(H,\eta)$ has nullity zero, each $z_{i}$ has a neighbor in $J \setminus Z$.
  Hence $V_{1} \setminus Z_{1} \not= \emptyset$ and $V_{2} \setminus Z_{2} \not= \emptyset$.

  Now create two new edges $F_{1}'$ and $F_{2}'$ and let $(H',\eta')$ be obtained from $(H,\eta)$ as follows.
  \begin{itemize}
    \item $V(H') = V(H)$ and $E(H') = (E(H) \setminus\{F\}) \cup \{F_{1}', F_{2}'\}$,
    \item for each $F_{0} \in E(H) \setminus\{F\}$ and $h \in V(H)$, $F_{0}$ is incident with $h$ in $H'$ if and only if they are incident in $H$,
    \item for $j=1,2$ and $h \in V(H)$, $F_{j}'$ is incident with $h$ in $H'$ if and only if $h \in P_{j}$,
    \item for all $F_{0} \in E(H) \setminus\{F\}$, $\eta'(F_{0}) = \eta(F_{0})$ and $\eta'(F_{0},h) = \eta(F_{0},h)$ for all $h \in \overline{F_{0}}$,
    \item for $j=1,2$, $\eta'(F_{j}') = V_{j}\setminus Z_{j}$ and $\eta'(F_{j}',h) = \eta(F,h)$ for all $h \in P_{j}$.
  \end{itemize}
  Observe that $(H',\eta')$ is again a strip-structure for $G$ and that it has the same nullity as $(H,\eta)$.
  Note that $(H',\eta')$ can be constructed in linear time.

  \begin{quote}
    (4) \emph{$J$ admits a pseudo-$1$-join $(V_{1},V_{2})$}
  \end{quote}
  As $J$ has no $0$-join, $A_{i} \not= \emptyset$ for $i=1,2$.
  For $j=1,2$, since $V_{j}$ is not strongly stable, $V_{j} \setminus Z \not= \emptyset$.
  For $j=1,2$, let $Z_{j} = V_{j} \cap Z$ and let $P_{j} = \{h_{i} \mid 1 \leq i \leq k \mbox{\ and\ } z_{i} \in Z_{j}\}$.
  Clearly, $P_{1} \cap P_{2} = \emptyset$ and $P_{1} \cup P_{2} = \overline{F}$.
  Moreover, $V_{j} \setminus Z_{j} \not= \emptyset$.

  If $Z \cap (A_{1} \cup A_{2}) \not= \emptyset$, suppose that $z_{1} \in A_{1}$.
  Then since $A_{1} \cup A_{2}$ is a strong clique and $(H,\eta)$ is purified, $Z \cap (A_{1} \cup A_{2}) = \{z_{1}\}$ and $Z \setminus\{z_{1}\}$ is strongly anticomplete to $A_{1} \cup A_{2}$.
  Furthermore, $z_{1}$ is strongly anticomplete to $V(J) \setminus (A_{1} \cup A_{2})$, since every vertex in this set has an antineighbor in $A_{1} \cup A_{2}$ and $z_{1}$ is strongly simplicial.

  Now create two new edges $F_{1}'$ and $F_{2}'$.
  If $Z \cap (A_{1} \cup A_{2}) = \emptyset$, then create a new vertex $h'$ as well; otherwise, assume that $z_{1} \in A_{1}$ and let $h' = h_{1}$.
  Let $(H',\eta')$ be obtained from $(H,\eta)$ as follows:
  \begin{itemize}
    \item $V(H') = V(H) \cup \{h'\}$, and $E(H') = (E(H) \setminus\{F\}) \cup \{F_{1}', F_{2}'\}$,
    \item for each $F_{0} \in E(H) \setminus\{F\}$ and $h \in V(H)$, $F_{0}$ is incident with $h$ in $H'$ if and only if they are incident in $H$,
    \item for $j=1,2$ and $h \in V(H)$, $F_{j}'$ is incident with $h$ in $H'$ if and only if $h \in P_{j}$,
    \item $F_{1}'$ and $F_{2}'$ are incident with $h'$,
    \item for all $F_{0} \in E(H) \setminus\{F\}$, $\eta'(F_{0}) = \eta(F_{0})$ and $\eta'(F_{0},h) = \eta(F_{0},h)$ for all $h \in \overline{F_{0}}$,
    \item for $j=1,2$, $\eta'(F_{j}') = V_{j}\setminus Z_{j}$ and $\eta'(F_{j}',h) = \eta(F,h)$ for all $h \in P_{j} \setminus \{h'\}$,
    \item for $j=1,2$, $\eta'(F_{j}',h') = A_{j}\setminus Z_{j}$.
  \end{itemize}
  Observe that since $J$ does not admit a $0$-join, $A_{j}\setminus Z_{j} \not= \emptyset$ for $j = 1,2$.
  It can now be readily verified that $(H',\eta')$ is a strip-structure of $G$ and that it has the same nullity as $(H,\eta)$.
  Note that $(H',\eta')$ can be constructed in linear time.
  \begin{quote}
    (5) \emph{$J$ admits a pseudo-$2$-join}
  \end{quote}
  We again compute a new strip-structure of $G$ in a manner similar to Rule~(4).
  To be precise, we apply the procedure as described in Chudnovsky and Seymour~\cite[Theorem~9.1:~Condition~(3)]{ChudnovskyS2008-5}.
  Note that this procedure takes linear time and does not increase the nullity.

  Let $(H,\eta)$ be the resulting strip-structure and let $(J,Z)$ be any stripe of the strip-decomposition.
  Choose $(J', Z')$ with $|V(J')|$ minimum such that $(J,Z)$ is a thickening of $(J',Z')$.
  Since $|V(J')|$ is minimum, no two vertices in $V(J') \setminus Z'$ are twins in $J'$ and there is no W-join $(A,B)$ in $J'$ with $Z' \cap A, Z' \cap B = \emptyset$.
  Moreover, by Rule~(3)-(5), $J$ does not admit a $0$-join, a pseudo-$1$-join, or a pseudo-$2$-join, and thus neither does $J'$ by Proposition~\ref{prp:join-thicken}.
  This shows that $(J,Z)$ is a thickening of an almost-unbreakable stripe.

  Observe that in applying one of Rules~(1),(3)-(5), the number of edges of the strip-structure increases by at least one.
  By definition, a strip-structure can have at most $|V(G)|$ edges, and thus we need only apply these rules at most $|V(G)|$ times. 
  Rule~(2) only needs to be applied after Rule~(1) has been applied, as all other rules do not increase the nullity.
  Since Rule~(1) can increase the nullity by at most two, Rule~(2) is applied at most twice as often as Rule~(1).
  Furthermore, Rule~(1) only needs to be applied after one of Rules (3)-(5) have been applied.
  Hence, Rule~(1) is applied at most $|V(G)|/2$ times and Rule~(2) at most $|V(G)|$ times.

  Note that applying the transformation of a rule takes at most linear time.
  If $G$ is a graph, then by Corollary~\ref{cor:join-finder}, it takes $\joinTime$ time to test the condition of a rule. Hence, the total run time of the algorithm is $\stripTime$.
\end{proof}

\subsection{Decomposing Line Graphs}
We first show that we can find a good strip-structure of line graphs in a fast and easy way, without using Theorem~\ref{thm:strip-structure}.
The arguments here resemble those in~\cite[Theorem~8]{HermelinMvL2014}. For this, we need the following definition.

\begin{definition}
\label{def:linestrip}
  Let $(J,Z)$ be a strip such that $J$ is a line graph.
  Then this is a \emph{trivial line graph strip} if $(J,Z)$ is a spot or if $|V(J)| = 2$ and $|Z|=1$. Otherwise, we call $(J,Z)$ a \emph{nontrivial line graph strip}.
\end{definition}
Observe that if $(J,Z)$ is a trivial line graph strip, then $|V(J) \setminus Z| = 1$.

\begin{lemma}
\label{lem:line-structure}
  Let $G$ be a connected line graph.
  Then in linear time, we can find a strip-structure of $G$ such that each strip is a trivial line graph strip.
\end{lemma}
\begin{proof}
  Recall that a vertex of $G$ is \emph{pendant} if it has degree~$1$. 
  Compute the pre-image $G'$ of $G$ (\ie a graph $G'$ such that $G$ is the line graph of $G'$); this takes linear time~\cite{Roussopoulos1973}.
  Then consider the following strip-graph $H$: $V(H)$ is equal to the set of vertices of $G'$ that are not pendant; $E(H)$ is obtained from $E(G')$ be removing a vertex from an edge of $E(G')$ if that vertex is pendant. 
  For each $F \in E(H)$, let $\eta(F)$ be the set containing just the vertex of $J$ that corresponds to the edge of $G$ that corresponds to $F$.
  For each $F \in E(H)$ and each $h \in \overline{F}$, let $\eta(F,h) = \eta(F)$.
  We then output the strip-structure $(H,\eta)$.
  Observe that $(H,\eta)$ is indeed a valid strip-structure, and that each strip is trivial line graph strip.
\end{proof}
We then show that if we are given a strip-structure of a claw-free graph, then we can decompose each nontrivial line graph strip into trivial line graph strips using Lemma~\ref{lem:line-structure}.

\begin{lemma}
\label{lem:strips:line}
  Let $G$ be a graph and let $(H,\eta)$ be a purified strip-structure of nullity zero of $G$ such that each strip is connected and either a spot or a stripe.
  Then, in linear time, we can find a purified strip-structure of nullity zero of $G$ such that if a strip $(J,Z)$ of $(H,\eta)$ is not a nontrivial line graph strip, then $(J,Z)$ is a strip of the resulting strip-structure, and all other strips of the resulting strip-structure are trivial line graph strips.
\end{lemma}
\begin{proof}
  Consider each strip $(J,Z)$ of $(H,\eta)$ in turn.
  If $(J,Z)$ is a trivial line graph strip, then we proceed to the next strip.
  Since spots are trivial line graph strips, we may thus assume that $(J,Z)$ is a stripe.
  We then check whether $J$ is a line graph; this takes linear time~\cite{Roussopoulos1973}.
  If $J$ is indeed a line graph, then $(J,Z)$ is a nontrivial line graph strip.
  We now aim to `shatter' $(J,Z)$ into several trivial line graph strips.

  Suppose that $J$ is the line graph of $J'$.
  Apply Lemma~\ref{lem:line-structure} to compute a strip-structure $(H',\eta')$ of $J$ such that each strip is a trivial line graph strip. We now incorporate $(H',\eta')$ into $(H,\eta)$. 
  For each $z \in Z$, let $F_{z}$ denote the edge of $H'$ such that $\eta'(F_{z}) = \{z\}$ and let $h_z$ denote the vertex of $H$ that corresponds to $z$.
  Let $F$ denote the edge of $H$ that corresponds to the strip $(J,Z)$. Now construct a new strip-structure $(H'',\eta'')$ by merging $(H,\eta)$ and $(H',\eta')$ as follows:
  \begin{itemize}
    \item $V(H'') = V(H) \cup (V(H') \setminus (\bigcup_{z \in Z} \overline{F_{z}}))$ and $E(H'') = (E(H') \setminus (\bigcup_{z \in Z} \{F_{z}\})) \cup (E(H) \setminus \{F\})$,
    \item for each $F_{0} \in E(H) \setminus \{F\}$ and $h \in V(H)$, $F_{0}$ is incident with $h$ in $H''$ if and only if they are incident in $H$,
    \item for each $F_{0} \in E(H') \setminus (\bigcup_{z \in Z} \{F_{z}\})$ and $h \in V(H')\setminus (\bigcup_{z \in Z} \overline{F_{z}})$, $F_{0}$ is incident with $h$ in $H''$ if and only if they are incident in $H'$,
    \item for each $F_{0} \in E(H') \setminus (\bigcup_{z \in Z} \{F_{z}\})$, $F_{0}$ is incident with $h_{z}$ in $H''$ if and only if $F_{0}$ is incident with a vertex of $\overline{F_{z}}$ in $H'$,
    \item for each $F_{0} \in E(H) \setminus \{F\}$, $\eta''(F_{0}) = \eta(F_{0})$ and $\eta''(F_{0},h) = \eta(F_{0},h)$ for each $h \in \overline{F_{0}}$,
    \item for each $F_{0} \in E(H') \setminus (\bigcup_{z \in Z} \{F_{z}\})$, $\eta''(F_{0}) = \eta'(F_{0})$ and $\eta''(F_{0},h) = \eta'(F_{0},h)$ for each $h \in \overline{F_{0}} \setminus (\bigcup_{z \in Z} \overline{F_{z}})$,
    \item for each $z \in Z$ and each $F_{0} \in E(H') \setminus (\bigcup_{z \in Z} \{F_{z}\})$ such that $\overline{F_{z}} \cap \overline{F_{0}} \not= \emptyset$, $\eta''(F_{0},h_{z}) = \eta'(F_{0}, h'_{z})$, where $\overline{F_{z}} \cap \overline{F_{0}} = \{h'_{z}\}$.
  \end{itemize}
  Since $(J,Z)$ is a stripe, no vertex of $J$ is adjacent to two or more vertices of $Z$ and $Z$ is strongly stable. Therefore, $(H'',\eta'')$ is indeed a valid strip-structure of $G$.
  Moreover, it is purified and has nullity zero.
  Observe also that each strip that corresponds to an edge of $E(H') \setminus (\bigcup_{z \in Z} F_{z})$ is a trivial line graph strip of $(H'',\eta'')$, and that each strip that corresponds to an edge of $E(H) \setminus \{F\}$ is still a strip of $(H'',\eta'')$.

  Proceeding iteratively, we indeed find a strip-structure of $G$ such that no strip is a nontrivial line graph strip.
  Moreover, it follows from the description of the algorithm that if a strip of $(H,\eta)$ was not a nontrivial line graph strip, it is still a strip of the resulting strip-structure. 
  Since each of the strips resulting from the `shattering' of a nontrivial line graph strip is a trivial line graph stripe and the `shattering' procedure takes linear time, the total run time of the algorithm is linear.
\end{proof}

\subsection{Auxiliary Algorithm}
Let $\omega$ denote the matrix-multiplication constant; currently $\omega < 2.373$~\cite{Williams2012,LeGall2014}.

\begin{lemma}
\label{lem:zis1:runtime}
  Let $(J,Z)$ be a stripe such that $J$ is a graph.
  We can check in $O(n^{\omega})$ time whether $|Z| = 1$, $V(J) \setminus N[Z] \not= \emptyset$, and $\alpha(J) \leq 3$.
\end{lemma}
\begin{proof}
  It can be checked in linear time whether $|Z|=1$ and $V(J) \setminus N[Z] \not= \emptyset$.
  Then there is a largest stable set of $J$ that contains the lone element $z$ of $Z$, as $z$ is simplicial.
  Hence, it remains to verify that there is no triangle in the complement of $J \setminus N[Z]$, which takes $O(n^{\omega})$ time~\cite{ItaiR1978}.
\end{proof}

\subsection{Main Theorems}
We now state our main structural and algorithmic results. 

\begin{theorem}
\label{thm:main-base}
  Let $G$ be a connected claw-free graph with $\alpha(G) > 3$ such that $G$ does not admit twins or proper W-joins. Then
  \begin{itemize}
    \item $G$ is a thickening of a member of $\mathcal{S}_{2}$ (\ie $G$ is a thickening of an XX-trigraph),
    \item $G$ is a proper circular-arc graph, or
    \item $G$ admits a strip-structure such that for each strip $(J,Z)$
    \begin{itemize}
      \item $(J,Z)$ is a trivial line graph strip, or
      \item $(J,Z)$ is a stripe for which $J$ is connected and
      \begin{itemize}
        \item $|Z|=1$, $\alpha(J) \leq 3$, and $V(J) \setminus N[Z] \not= \emptyset$,
        \item $|Z|=1$, $J$ is a proper circular-arc graph, and either $J$ is a strong clique or $\alpha(J) > 3$,
        \item $|Z|=2$ and $J$ is a proper interval graph, or
        \item $(J,Z)$ is a thickening of a member of $\mc{Z}_{2} \cup \mc{Z}_{3} \cup \mc{Z}_{4} \cup \mathcal{Z}_{5}$.
      \end{itemize}
    \end{itemize}
  \end{itemize}
  Moreover, we can distinguish the cases and find the strip-structure in $\stripTime$ time.
\end{theorem}
\begin{proof}
  We first check whether $G$ is a line graph; this takes linear time~\cite{Roussopoulos1973}.
  If $G$ is indeed a line graph, then we use Lemma~\ref{lem:line-structure} to compute, in linear time, a strip-structure of $G$ such that each strip is a trivial line graph strip.
  Note that such strips are either spots or proper circular-arc graphs with $|Z|=1$.
  Hence, we can output this strip-structure.

  We then check whether $G$ is a proper circular-arc graph; this takes linear time~\cite{DengHH1996}.
  If it is, then we output $G$.

  We may thus assume that $G$ is neither a line graph nor a proper circular-arc graph.
  By Lemma~\ref{lem:tos2}, this implies that $G$ is a thickening of a member of $\mathcal{S}_{2}$, or $G$ admits a pseudo-$1$-join or a pseudo-$2$-join.
  Following Corollary~\ref{cor:join-finder}, we can decide in $\joinTime$ time whether $G$ admits a pseudo-$1$-join or a pseudo-$2$-join.
  If not, then we output $G$, as $G$ must be a thickening of a member of $\mathcal{S}_{2}$.

  We may thus assume that $G$ admits a pseudo-$1$-join or a pseudo-$2$-join.
  By Theorem~\ref{thm:strip-structure}, we can find in $\stripTime$ time a purified strip-structure $(H,\eta)$ of nullity zero of $G$ such that all strips are spots or thickenings of almost-unbreakable stripes.
  Since $G$ is connected and thus does not admit a $0$-join, by inspecting the proof of Theorem~\ref{thm:strip-structure}, it follows that each edge of $H$ is incident on at least one vertex of $H$.
  Hence, each strip of the strip-structure has $|Z| \geq 1$.

  Now apply the algorithm of Lemma~\ref{lem:strips:line} in linear time and (by abuse of notation) call the resulting strip-structure $(H,\eta)$ as well.
  Observe that $(H,\eta)$ is purified, has nullity zero, and all its strips are trivial line graph strips or thickenings of almost-unbreakable stripes.
  Moreover, still each strip $(J,Z)$ of the strip-structure has $|Z| \geq 1$.

  Consider each strip $(J,Z)$ of the strip-structure in turn.
  We first check whether $(J,Z)$ is a trivial line graph strip; this takes constant time.
  If it is, then we can proceed to the next strip.
  Hence, we may assume that $(J,Z)$ is a thickening of an almost-unbreakable stripe.
  Since the stripe is almost-unbreakable, it follows that by the definition of thickenings that $J$ is connected.

  We then check whether $|Z| = 1$, $\alpha(J) \leq 3$, and $V(J)\setminus N[Z]\not= \emptyset$; this takes $O(|V(J)|^{\omega})$ time by Lemma~\ref{lem:zis1:runtime}.
  If it is, then we proceed to the next strip.

  We then check whether $|Z|=1$ and $J$ is a proper circular-arc graph; this takes linear time~\cite{DengHH1996}.
  If it is, then note that either $V(J) \setminus N[Z] = \emptyset$ and thus $J$ is a strong clique, or $\alpha(J) > 3$, and we proceed to the next strip.
  We then check whether $|Z|=2$ and $J$ is a proper interval graph; this takes linear time~\cite{DengHH1996}.
  If it is, then we proceed to the next strip. 
  \begin{cclaim}
  \label{c:main-base}
    At this stage, $(J,Z)$ is a thickening of a member of $\mc{Z}_{2} \cup \mc{Z}_{3} \cup \mc{Z}_{4} \cup \mc{Z}_{5}$.
  \end{cclaim}
  \begin{cproof}
    Observe that $(J,Z)$ is a thickening of some almost-unbreakable stripe $(J',Z')$.
    Note that $|V(J)| \geq 2$, as $|Z| \geq 1$ and each strip of $(J,Z)$ has nullity zero.
    If $|V(J)| = 2$ or $|V(J')| = 2$, then $|Z|=1$ and $J$ is a proper circular-arc graph, a contradiction at this stage.
    Hence, $|V(J)| > 2$ and $|V(J')| > 2$.

    Suppose that $J$ contains a proper W-join $(A,B)$ such that $Z \cap A, Z \cap B = \emptyset$.
    Then $(A,B)$ is also a proper W-join in $G$, a contradiction.
    Hence, by Lemma~\ref{lem:properWjoin}, $J$ does not admit a proper W-join.

    Suppose that $(J,Z)$ is a thickening of a member of $\mc{Z}_{1}$ or of $\mc{Z}_{6}$.
    Hence, by Lemma~\ref{lem:interval} respectively Lemma~\ref{lem:circular}, $J$ is a proper interval graph and $|Z|=2$ respectively $J$ is a proper circular-arc graph and $|Z|=1$.
    This is a contradiction at this stage, and thus $(J,Z)$ is not a thickening of a member of $\mc{Z}_{1} \cup \mc{Z}_{6}$. In particular, $(J',Z')$ is not a member of $\mc{Z}_{1} \cup \mc{Z}_{6}$.

    Suppose that $J$ is the union of two strong cliques.
    By Corollary~\ref{cor:strong-cliques}, $(J,Z)$ is a thickening of a member of $\mc{Z}_{1} \cup \mc{Z}_{6}$, a contradiction. Hence, $J$ is not the union of two strong cliques.
    Therefore, by Proposition~\ref{prp:union-two}, $J'$ is not the union of two strong cliques. 

    Suppose that $J$ is a thickening of a line trigraph.
    Since $J$ is not the union of two strong cliques, it follows from Lemma~\ref{lem:linetrigraph-join2} that $J$ is a thickening of a line graph.
    As $|V(J)| > 2$, it follows from Lemma~\ref{lem:twins} that $J$ does not admit twins.
    Hence, $J$ is in fact a line graph, a contradiction (at this stage, we can have neither trivial nor nontrivial line graph strips).
    Hence, $J$ is not a thickening of a line trigraph.
    In particular, $J'$ is not a line trigraph and thus not a line graph.
    We can then repeat the same argument to show that $J'$ is not even a thickening of a line trigraph.

    Suppose that $J$ admits a hex-join.
    Then $|Z| \leq 2$ by Lemma~\ref{lem:hexjoin-z}.
    Suppose that $|Z|=1$.
    Since $J$ is not the union of two strong cliques, it follows from Lemma~\ref{lem:hexjoin} that $V(J) \setminus N[Z] \not= \emptyset$ and $\alpha(J) \leq 3$, a contradiction at this stage.
    Hence, $|Z|=2$.
    Then using the fact that $J$ is not a thickening of a line trigraph and that $(J,Z)$ is not a thickening of a member of $\mc{Z}_{1}$, it follows from Corollary~\ref{cor:threecliques-zis2} that $(J,Z)$ is a thickening of a member of $\mc{Z}_{2} \cup \mc{Z}_{3} \cup \mc{Z}_{4}$ and the claim would follow.

    Suppose that $J$ does not admit a hex-join. 
    Then $J'$ does not admit a hex-join either by Proposition~\ref{prp:join-thicken}.
    Hence, by Lemma~\ref{lem:stripe-indecomposable}, the fact that $(J',Z')$ is a (trivial) thickening of itself, and the fact that $J'$ is not a thickening of a line trigraph, $J'$ is a thickening of an indecomposable member of $\mc{S}_{1},\ldots,\mc{S}_{7}$. 
    It follows from Lemma~\ref{lem:indecomposable-indecomposable} and the fact that $(J',Z')$ is not a member of $\mc{Z}_{1} \cup \mc{Z}_{6}$, that $(J',Z')$ is a member of $\mathcal{Z}_{i}$, where $i \in \{2,5,7,8,9\}$.

    Suppose that $(J',Z')$ is a member of $\mathcal{Z}_{i}$, where $i \in \{7,8,9\}$. We show that $\alpha(J) \leq 3$ and $V(J) \setminus N[Z] \not= \emptyset$.
    For suppose that $(J',Z')$ is a member of
    \begin{itemize}
      \item[$\mathcal{Z}_{7}$:] Since $J' \in \mathcal{S}_{4}$, $\alpha(J) \leq 3$ by Corollary~\ref{cor:independent}.
      Since the edge $h_{6}h_{7}$ is not incident with any edges of the cycle $h_{1}\ldots h_{5}$, $V(J) \setminus N[Z] \not= \emptyset$.
      \item[$\mathcal{Z}_{8}$:] Since $J' \in \mathcal{S}_{5}$, $\alpha(J) \leq 3$ by Corollary~\ref{cor:independent}.
        Since $d_{5}$ is not incident with the nonempty strong cliques $A,B,C$, it follows that $V(J) \setminus N[Z] \not= \emptyset$.
      \item[$\mathcal{Z}_{9}$:] Since $J' \in \mathcal{S}_{7}$ ($J'$ is antiprismatic), $\alpha(J) \leq 3$ by Corollary~\ref{cor:independent}.
        Moreover, $z$ is strongly antiadjacent to the nonempty strong cliques $A,B,C$, so $V(J) \setminus N[Z] \not= \emptyset$.
    \end{itemize}
    In each of these three cases, trivially $|Z| = 1$.
    We thus reach a contradiction at this stage. Hence, $(J',Z') \in \mc{Z}_{2} \cup \mc{Z}_{5}$, proving the claim.
  \end{cproof}
  Observe that the total run time of the described algorithm is dominated by the run time of the algorithm of Theorem~\ref{thm:strip-structure}, which is $\stripTime$.
  This proves the theorem.
\end{proof}
By slightly adapting the algorithm and the analysis, we obtain the following results.

\begin{theorem}
\label{thm:main-base-kernel}
  Let $G$ be a connected claw-free graph with $\alpha(G) > 3$ such that $G$ does not admit twins or proper W-joins. Then
  \begin{itemize}
    \item $G$ is a thickening of a member of $\mathcal{S}_{2}$ (\ie $G$ is a thickening of an XX-trigraph),
    \item $G$ is a proper circular-arc graph, or
    \item $G$ admits a strip-structure such that for each strip $(J,Z)$
    \begin{itemize}
      \item $(J,Z)$ is a trivial line graph strip, or
      \item $(J,Z)$ is a stripe for which $J$ is connected and
      \begin{itemize}
        \item $|Z|=1$, $\alpha(J) \leq 3$, and $V(J) \setminus N[Z] \not= \emptyset$,
        \item $|Z|=1$, $J$ is a proper circular-arc graph, and either $J$ is a strong clique or $\alpha(J) > 3$,
        \item $|Z|=2$ and $J$ is a proper interval graph,
        \item $|Z|=1$, $\alpha(J) = 4$, and $V(J) \setminus N[Z] \not= \emptyset$, or
        \item $|Z|=2$ and $(J,Z)$ is a thickening $\mc{W}$ of a member $(J',Z')$ of $\mc{Z}_{2} \cup \mc{Z}_{3} \cup \mc{Z}_{4} \cup \mathcal{Z}_{5}$.
          Moreover, we know $\mc{W}$, $(J',Z')$, and the class that $(J',Z')$ belongs to.
      \end{itemize}
    \end{itemize}
  \end{itemize}
  Moreover, we can distinguish the cases and find the strip-structure in $\stripTime$ time.
\end{theorem}
\begin{proof}
  We apply the same algorithm as in the proof of Theorem~\ref{thm:main-base}.
  However, we modify it at the stage that we consider each $(J,Z)$ of a strip-structure of $G$.
  At the end (\ie at Claim~\ref{c:main-base}), we prove that $(J,Z)$ is a thickening $\mc{W}$ of a member $(J',Z')$ of $\mc{Z}_{2} \cup \mc{Z}_{3} \cup \mc{Z}_{4} \cup \mathcal{Z}_{5}$.
  We now extend the algorithm as follows.

  If $|Z| = 1$, then it follows from the definitions of $\mc{Z}_{2}$, $\mc{Z}_{3}$, $\mc{Z}_{4}$, and $\mathcal{Z}_{5}$ that $(J,Z)$ is a thickening of a member of $\mc{Z}_{5}$.
  By Corollary~\ref{cor:independent}, this implies that $\alpha(J) \leq 4$.
  If $V(J) \setminus N[Z] = \emptyset$, then $J$ is a (strong) clique, and in particular, $J$ is a proper circular-arc graph, a contradiction at this stage.
  Hence, $V(J) \setminus N[Z] \not= \emptyset$. Then $\alpha(J) \leq 3$ would form a contradiction at this stage.
  Therefore, $\alpha(J) = 4$.

  If $|Z| = 2$, then we run the recognition algorithms of Lemma~\ref{lem:recog:z2}, Lemma~\ref{lem:recog:z3}, Lemma~\ref{lem:recog:z4}, and Lemma~\ref{lem:recog:z5}.
  Then, in linear time, we know $\mc{W}$, $(J',Z')$, and the class that $(J',Z')$ belongs to.
\end{proof}

\begin{theorem}
\label{thm:main-base2}
  Let $G$ be a connected claw-free graph with $\alpha(G) > 3$ such that $G$ does not admit twins or proper W-joins.
  Then
  \begin{itemize}
    \item $G$ is a thickening of a member of $\mathcal{S}_{2}$ (\ie $G$ is a thickening of an XX-trigraph),
    \item $G$ is a proper circular-arc graph, or
    \item $G$ admits a strip-structure such that for each strip $(J,Z)$
    \begin{itemize}
      \item $(J,Z)$ is a trivial line graph strip, or
      \item $(J,Z)$ is a stripe for which $J$ is connected and
      \begin{itemize}
        \item $1\leq |Z| \leq 2$, $\alpha(J) \leq 3$, and $V(J) \setminus N[Z] \not= \emptyset$,
        \item $|Z|=1$, $J$ is a proper circular-arc graph, and either $J$ is a strong clique or $\alpha(J) > 3$,
        \item $|Z|=2$, $J$ is a proper interval graph, and $\alpha(J) > 3$, or
        \item $(J,Z)$ is a thickening of a member of $\mathcal{Z}_{5}$.
      \end{itemize}
    \end{itemize}
  \end{itemize}
  Moreover, we can distinguish the cases and find the strip-structure in $\stripTime$ time.
\end{theorem}
\begin{proof}
  We apply the same algorithm as in the proof of Theorem~\ref{thm:main-base}.
  However, we modify it at the stage that we consider each $(J,Z)$ of a strip-structure of $G$.
  At the stage where we test whether $|Z|=1$, $\alpha(J) \leq 3$ and $V(J) \setminus N[Z] \not= \emptyset$, we additionally test for each strip $(J,Z)$ whether $|Z|=2$, $\alpha(J) \leq 3$ and $V(J) \setminus N[Z] \not= \emptyset$.
  Since there is a largest stable set of $J$ that contains both elements of~$Z$, it remains to verify that there is no nonedge in $J \setminus N[Z]$, which takes linear time.
  Hence, the run time of the algorithm remains $\stripTime$.

  We then modify the final part of the analysis, where $(J',Z')$ has the same properties as in the final part of the proof of Theorem~\ref{thm:main-base}.
  Recall that in that part of the analysis, $J$ is not the union of two strong cliques, and thus neither is $J'$.
  We show that $\alpha(J) \leq 3$ and $V(J)\setminus N[Z] \not=\emptyset$.
  Suppose that $(J',Z')$ is a member of
  \begin{itemize}
    \item[$\mathcal{Z}_{2}$:] Since $J' \in \mathcal{S}_{6}$, $\alpha(J) \leq 3$ by Corollary~\ref{cor:independent}.
      As the strong clique $C$ in the definition of $\mathcal{S}_{6}$ has $|C\setminus X| \geq 2$ and both $a_{0}$ and $b_{0}$ are strongly antiadjacent to $C$, it follows that $V(J) \setminus N[Z] \not= \emptyset$.
    \item[$\mathcal{Z}_{3}$:] Note that $J$ is the union of three nonempty strong cliques.
      Hence, $\alpha(J) \leq 3$ and $V(J) \setminus N[Z] \not= \emptyset$.
    \item[$\mathcal{Z}_{4}$:] Note that $J$ is the union of three nonempty strong cliques.
      Hence, $\alpha(J) \leq 3$ and $V(J) \setminus N[Z] \not= \emptyset$.
  \end{itemize}
  Since $|Z|=2$ in all three cases, we would obtain a contradiction at this stage.
\end{proof}
Using Corollary~\ref{cor:independent}, the definition of a trivial line graph strip, and Theorem~\ref{thm:main-base2}, we obtain the following.

\begin{theorem}
  Let $G$ be a connected claw-free graph with $\alpha(G) > 4$ such that $G$ does not admit twins or proper W-joins.
  Then
  \begin{itemize}
    \item $G$ is a proper circular-arc graph, or
    \item $G$ admits a strip-structure such that for each strip $(J,Z)$
    \begin{itemize}
      \item $(J,Z)$ is a spot, or
      \item $(J,Z)$ is a stripe for which $J$ is connected and
      \begin{itemize}
        \item $1\leq |Z| \leq 2$, $\alpha(J) \leq 4$, and $V(J) \setminus N[Z] \not= \emptyset$,
        \item $|Z|=1$, $J$ is a proper circular-arc graph, and either $J$ is a strong clique or $\alpha(J) > 3$, or
        \item $|Z|=2$, $J$ is a proper interval graph, and either $J$ is a strong clique or $\alpha(J) > 3$.
      \end{itemize}
    \end{itemize}
  \end{itemize}
  Moreover, we can distinguish the cases and find the strip-structure in $\stripTime$ time.
\end{theorem}

\section*{Part II -- Algorithmic Applications on Claw-Free Graphs}

\section{Fixed-Parameter Algorithm for Dominating Set}
\label{sec:ds}
In this section, we show that {\sc Dominating Set} parameterized by solution size is fixed-parameter tractable on claw-free graphs.
The general idea of how to establish this is as follows.
We first show how to remove twins and proper W-joins from $G$ without changing the size of its smallest dominating set (see Sect.~\ref{sec:fpt-ds:twin-join}).
Moreover, if $\alpha(G) \leq 3$, then we can find a smallest dominating set of $G$ by exhaustive enumeration.
Then we can apply Theorem~\ref{thm:main-base2}, and $G$ either belongs to some basic class, or it can be decomposed into strips that each belong to a basic class.
If $G$ belongs to a basic class, then we can again find a smallest dominating set of $G$ in polynomial time (see Sect.~\ref{sec:fpt-ds:basic}).
If $G$ can be decomposed into strips, then we solve {\sc Dominating Set} separately on each strip in polynomial time (see Sect.~\ref{sec:fpt-ds:basic}), and then we present a fixed-parameter algorithm to stitch the solutions of the strips together (see Sect.~\ref{sec:fpt-ds:stitch}).

Throughout the section, we rely on the following notation.
Let $G$ be a graph.
We let $\ds{G}$ denote the smallest size of a dominating set of $G$.
More generally, for each subset $A\subseteq V(G)$, we let $\dsg{G}{A}$ denote the size of a smallest subset of $V(G)$ dominating all vertices in $V(G) \setminus A$.
An \emph{independent dominating set} is a subset of $V(G)$ that is both an independent set and a dominating set of $G$.
We need the following fact proved by Allan and Laskar~\cite{AllanL1978}.
\begin{lemma}[Allan and Laskar~\cite{AllanL1978}]
   \label{lem:ds:is}
  There is a (polynomial-time) algorithm that, given a claw-free graph $G$ and a dominating set $D$ of $G$, outputs an independent dominating set of $G$ of size at most $|D|$.
\end{lemma}
As a consequence of this lemma, we can assume throughout that any smallest dominating set that we consider is also an independent set.

We also need the following folklore fact.
\begin{proposition}
\label{prp:ds:is-size}
  Let $G$ be a graph.
  Then $\ds{G} \leq \alpha(G)$.
\end{proposition}

\subsection{Removing Twins and W-joins}
\label{sec:fpt-ds:twin-join}
We first show how to remove twins and (proper) W-joins from a graph $G$ without changing the size of its smallest dominating set. The reductions are powerful enough to operate on general graphs, while still maintaining claw-freeness.

\begin{lemma}
\label{lem:ds:twins}
  Let $a,b$ be twins of a graph $G$, and let $G' = G \setminus a$. Then $\ds{G} = \ds{G'}$. Moreover, if $G$ is claw-free, then so is $G'$.
\end{lemma}
\begin{proof}
  Let $D$ be a smallest dominating set of $G$.
  Since $N[a] = N[b]$ (in particular, $a$ and $b$ are adjacent) and $D$ is a smallest dominating set of $G$, at most one of $a,b$ belongs to $D$.
  If $a \in D$, then replace $a$ by $b$. Then the resulting set is still a dominating set of $G$ of the same size as $D$, and thus also a dominating set of $G' = G \setminus a$.

  Let $D'$ be a smallest dominating set of $G'$.
  Then $D' \cap N[b] \not= \emptyset$.
  Since $N[a] = N[b]$, $D'$ is a dominating set of $G$ as well.
\end{proof}

We remark here that a more general reduction exists that is even more powerful. The reduction removes vertex $a$ if there is an adjacent vertex $b$ for which $N[a] \subseteq N[b]$. Note that twins indeed satisfy the conditions of this reduction. Using this rule, all W-joins in the graph would be proper, as proven by Martin~\etal\cite[Lemma~14]{MartinPvL2018} (see also~\cite[Lemma~12]{MartinPvL2018-arxiv}). The absence of general W-joins would simplify several parts of the structural decomposition theorem presented in this work. Martin~\etal\cite{MartinPvL2018,MartinPvL2018-arxiv} provide a first example of such simplifications. We leave further explorations to future work.

\begin{lemma}
\label{lem:ds:Wjoins}
  Let $(A,B)$ be a W-join of a graph $G$.
  Construct a graph $G'$ from $G$ as follows:
  \begin{enumerate}
    \item if some $a_{0} \in A$ is complete to $B$ and some $b_{0} \in B$ is complete to $A$, then remove $A\setminus\{a_{0}\}$ and $B\setminus\{b_{0}\}$;
    \item otherwise, if some $a_{0} \in A$ is complete to $B$ and some $a_{1} \in A$ is antiadjacent to some $b_{0} \in B$, then remove all vertices of $A \setminus\{a_{0},a_{1}\}$ and all vertices of $B\setminus\{b_{0}\}$;
    \item otherwise, if some $b_{0} \in B$ is complete to $A$ and some $b_{1} \in B$ is antiadjacent to some $a_{0} \in A$, then remove all vertices of $B \setminus\{b_{0},b_{1}\}$ and all vertices of $A\setminus\{a_{0}\}$;
    \item otherwise, let $a_{0} \in A$ be a vertex that is antiadjacent to some $b_{0} \in B$, and remove $A\setminus\{a_{0}\}$ and $B\setminus\{b_{0}\}$.
  \end{enumerate}
  Then $\ds{G} = \ds{G'}$.
  Moreover, if $G$ is claw-free, then so is $G'$.
\end{lemma}
\begin{proof}
  Since $G'$ is obtained from $G$ by removing vertices, $G'$ is claw-free if $G$ is.
  It remains to show that $\ds{G} = \ds{G'}$. We do that by showing that, in each of the above four cases, we can construct from a (smallest) dominating set of $G$ a dominating set of $G'$ of equal or smaller size, and vice versa.

  \ccase{1} Some $a_{0} \in A$ is complete to $B$ and some $b_{0} \in B$ is complete to $A$.\\
  Let $D' \subseteq V(G')$ be a dominating set of $G'$.
  Consider any vertex $a \in A\setminus\{a_{0}\}$.
  If $b_{0} \in D'$ or $a_{0} \in D'$, then $D'$ dominates $a$, since $b_{0}$ is complete to $A$ by assumption and $A$ is a clique by the definition of a W-join respectively.
  If $a_{0},b_{0} \not\in D'$, then there is a vertex $v$ that is both in $N(a_{0}) \cap D'$ and in $V(G') \setminus \{a_{0},b_{0}\} = V(G) \setminus (A \cup B)$.
  Since every vertex of $V(G)\setminus(A \cup B)$ is either $A$-complete or $A$-anticomplete by the definition of a W-join and $v \in N(a_{0})$, it follows that $v$ is $A$-complete, and thus $D'$ dominates $a$.
  By symmetric arguments, $D'$ dominates every $b \in B \setminus\{b_{0}\}$.
  Hence, $D'$ is also a dominating set of $G$.

  Let $D \subseteq V(G)$ be a smallest dominating set of $G$.
  If $D \cap A \not= \emptyset$, then we can assume that $D \cap A = \{a_{0}\}$, because $N[a_{0}] \supseteq N[a]$ for each $a \in A$ by the definition of a W-join and by the assumption that $a_{0}$ is complete to $B$.
  Hence, we may assume that $D \cap A \subseteq \{a_{0}\}$.
  Similarly, we may assume that $D \cap B \subseteq \{b_{0}\}$.
  Therefore, $D$ is a dominating set of $G'$.

  \medskip
  In the remainder, we may thus assume that each vertex of $A$ is not complete to $B$ or that each vertex of $B$ is not complete to $A$.
  
  \ccase{2} Some $a_{0} \in A$ is complete to $B$ and some $a_{1} \in A$ is antiadjacent to some $b_{0} \in B$.\\
  Let $D' \subseteq V(G')$ be a dominating set of $G'$.
  Consider any vertex $a \in A \setminus\{a_{0},a_{1}\}$.
  If $a_{0} \in D'$ or $a_{1} \in D'$, then $D'$ dominates $a$, since $A$ is a clique.
  If $a_{0},a_{1} \not\in D'$, then because $b_{0}$ is antiadjacent to $a_{1}$ (both in $G$ and $G'$), it follows that there is a vertex $v$ that is both in $N(a_{1}) \cap D'$ and in $V(G') \setminus \{a_{0},a_{1},b_{0}\} = V(G) \setminus (A \cup B)$. Since every vertex of $V(G)\setminus(A \cup B)$ is either $A$-complete or $A$-anticomplete by the definition of a W-join and $v \in N(a_{1})$, it follows that $v$ is $A$-complete, and thus $D'$ dominates $a$.
  Consider any vertex $b \in B \setminus\{b_{0}\}$.
  Observe that $N_{G'}[a_{0}] \supseteq N_{G'}[a_{1}]$; hence, we may assume that $D' \cap \{a_{0},a_{1}\} \subseteq \{a_{0}\}$. 
  Then we can use similar arguments as in Case~1 to show that $D'$ dominates $b$.
  Therefore, $D'$ is also a dominating set of $G$.

  Let $D \subseteq V(G)$ be a smallest dominating set of $G$.
  If $D \cap A \not= \emptyset$, then we can assume that $D \cap A = \{a_{0}\}$, because $N[a_{0}] \supseteq N[a]$ for each $a \in A$ by the definition of a W-join and by the assumption that $a_{0}$ is complete to $B$.
  Hence, we may assume that $D \cap A \subseteq \{a_{0}\}$.
  If $|D \cap B| > 1$, then we can replace a vertex of $D \cap B$ by $a_{0}$ to obtain another smallest dominating set of $G$, because $A$ and $B$ are cliques and $N(b) \setminus A = N(b') \setminus A$ for each $b,b' \in B$ by the definition of a W-join. 
  Hence, we may assume that $|D \cap B| \leq 1$ and $D \cap A \subseteq \{a_{0}\}$.
  If $D \cap B = \{b\}$ for some vertex $b$, then $b$ is not complete to $A$ by assumption.
  Hence, $a_{0} \in D$ or there is a vertex $v$ both in $D$ and in $V(G) \setminus (A \cup B)$ that is adjacent to a vertex of $A$ that is not adjacent to $b$.
  In the second case, $v$ is $A$-complete by the definition of a W-join.
  Then each vertex of $A$ is dominated by $a_{0}$ or $v$ respectively.
  Hence, the only responsibility of $b$ is to dominate $N[b] \setminus A$.
  Since $N[b] \setminus A = N[b'] \setminus A$ for each $b' \in B$ by the definition of a W-join, we may assume that $b = b_{0}$. Therefore, we may assume that $D \cap A \subseteq \{a_{0}\}$ and $D \cap B \subseteq \{b_{0}\}$.
  Hence, $D$ is a dominating set of $G'$.

  \ccase{3} Some $b_{0} \in B$ is complete to $A$ and some $b_{1} \in B$ is antiadjacent to some $a_{0} \in A$.\\
  This case is symmetric to the previous one.

  \medskip
  In the remainder, we may thus assume that each vertex of $A$ is not complete to $B$ and that each vertex of $B$ is not complete to $A$.

  \ccase{4} Let $a_{0} \in A$ be a vertex that is antiadjacent to some $b_{0} \in B$.\\
  Let $D' \subseteq V(G')$ be a dominating set of $G'$.
  Consider any vertex $a \in A \setminus\{a_{0}\}$.
  If $a_{0} \in D'$, then $D'$ dominates $a$, since $A$ is a clique.
  If $a_{0} \not\in D'$, then because $b_{0}$ is antiadjacent to $a_{0}$ (both in $G$ and in~$G'$), there is a vertex $v$ that is both in $N_{G'}(a_{0}) \cap D'$ and in $V(G') \setminus \{a_{0},b_{0}\} = V(G) \setminus (A \cup B)$.
  Since every vertex of $V(G)\setminus(A \cup B)$ is either $A$-complete or $A$-anticomplete by the definition of a W-join and $v \in N(a_{0})$, it follows that $v$ is $A$-complete, and thus $D'$ dominates $a$.
  By symmetric arguments, $D'$ dominates every $b \in B \setminus\{b_{0}\}$.
  Hence, $D'$ is also a dominating set of $G$.

  Let $D$ be a smallest dominating set of $G$.
  If $|D \cap A| > 1$, then we replace a vertex of $D \cap A$ by $b_{0}$ to obtain another smallest dominating set of $G$, because $A$ and $B$ are cliques and $N[a] \setminus B = N[a'] \setminus B$ for each $a,a' \in A$ by the definition of a W-join.
  Hence, we may assume that $|D \cap A| \leq 1$.
  If $D \cap A = \{a\}$ for some vertex $a$, then $a$ is not complete to $B$ by assumption.
  Hence, $b \in D$ for some $b \in B$ or there is a vertex $v$ both in $D$ and in $V(G) \setminus (A \cup B)$ that is adjacent to a vertex of $B$ that is not adjacent to $a$.
  In the second case, $v$ is $B$-complete by the definition of a W-join.
  Then each vertex of $B$ is dominated by $b$ or $v$ respectively.
  Hence, the only responsibility of $a$ is to dominate $N[a] \setminus B$.
  Since $N[a] \setminus B = N[a'] \setminus B$ for each $a' \in A$ by the definition of a W-join, we may assume that $a = a_{0}$. Therefore, we may assume that $D \cap A \subseteq \{a_{0}\}$ and, similarly, that $D \cap B \subseteq \{b_{0}\}$.
  Therefore, $D$ is a dominating set of $G'$.
\end{proof}
The final case of the lemma implies that we can remove proper W-joins, since in a proper W-join $(A,B)$ each vertex of $A$ is not complete to $B$ and each vertex of $B$ is not complete to $A$.

\begin{corollary}
\label{cor:ds:Wjoins}
  Let $(A,B)$ be a proper W-join of a graph $G$ and let $a_{0} \in A$ be a vertex that is antiadjacent to some $b_{0} \in B$. Create a graph $G'$ from $G$ by removing $A\setminus\{a_{0}\}$ and $B\setminus\{b_{0}\}$.
  Then $\ds{G} = \ds{G'}$.
  Moreover, if $G$ is claw-free, then so is $G'$.
\end{corollary}

\subsection{Dominating Set in Basic Classes}
\label{sec:fpt-ds:basic}
Let $G$ be a claw-free graph.
Through the reductions of Lemma~\ref{lem:ds:twins} and Corollary~\ref{cor:ds:Wjoins}, we may assume that $G$ admits no twins and proper W-joins.
Consider the following lemma.
\begin{lemma}
\label{lem:ds:constant}
  Let $G$ be a graph and $k$ an integer.
  Then in $O(n^{k+1})$ time we can compute $\ds{G}$ or correctly decide that $\ds{G} > k$.
\end{lemma}
\begin{proof}
  Use exhaustive enumeration to find a smallest set $D \subseteq V(G)$ with $|D| \leq k$ such that $|N[D]| = |V(G)|$, or report that no such set exists.
  This takes $O(n^{k+1})$ time.
\end{proof}

\begin{corollary}
\label{lem:ds:alpha}
  Let $G$ be a graph such that $\alpha(G) \leq 3$.
  Then we can compute $\ds{G}$ in $O(n^{4})$ time.
\end{corollary}
\begin{proof}
  By Proposition~\ref{prp:ds:is-size}, $\ds{G} \leq \alpha(G) \leq 3$, and the result follows from Lemma~\ref{lem:ds:constant}.
\end{proof}
Intuitively, we may now assume that the claw-free graph $G$ admits no twins, admits no proper W-joins, and satisfies $\alpha(G) > 3$.
Therefore, we can use the implications of Theorem~\ref{thm:main-base2} for $G$.
(A formal proof of these facts follows later.)

First, we show that if $G$ is a proper circular-arc graph or a thickening of an XX-trigraph, then we can compute $\ds{G}$ in polynomial time. 

\begin{theorem}[Hsu and Tsai~\cite{HsuT1991}]
\label{thm:ds:circular}
  Let $G$ be a circular-arc graph. Then $\ds{G}$ can be computed in linear time.
\end{theorem}

\begin{lemma}
\label{lem:ds:s2}
  Let $G$ be a graph that is a thickening of an XX-trigraph.
  Then $\ds{G}$ can be computed in $O(n^{4})$ time.
\end{lemma}
In order to show this lemma, we need the following auxiliary result.

\begin{lemma}
\label{lem:ds:tri}
  Let $G$ be a graph that is thickening of a trigraph $G'$, and let $G''$ be the graph obtained from $G'$ by removing all semi-edges from $G'$.
  Then $\ds{G} \leq \ds{G''}$.
\end{lemma}
\begin{proof}
  Let $\mathcal{W}$ be a thickening of $G'$ to $G$, and let $D'$ be any dominating set of $G''$.
  Construct a set $D \subseteq V(G)$ as follows: for each $v' \in D'$, pick an arbitrary vertex $v \in W_{v'}$. We claim that $D$ is a dominating set of $G$.
  Consider any $w \in V(G) \setminus D$ and let $w' \in V(G')$ be such that $w \in W_{w'}$.
  If $w' \in D'$, then because $W_{w'}$ is a (strong) clique and $W_{w'} \cap D \not= \emptyset$ by construction, $w$ is dominated by~$D$.
  Otherwise, there is a $u' \in D' \subseteq V(G'')$ such that $w'$ and $u'$ are adjacent.
  By the construction of~$G''$, this implies that $w'$ and $u'$ are strongly adjacent in $G'$.
  By the definition of a thickening, each vertex of $W_{u'}$ is (strongly) complete to $W_{w'}$. By construction, $D \cap W_{u'} \not= \emptyset$, and thus $w$ is dominated by $D$.
  The claim follows.
  Since $|D| = |D'|$, $\ds{G} \leq \ds{G''}$.
\end{proof}
It is now straightforward to prove Lemma~\ref{lem:ds:s2}.

\begin{proof}[Proof of Lemma~\ref{lem:ds:s2}]
  Consider an XX-trigraph $G'$ such that $G$ is a thickening of $G'$. Remove all semi-edges from $G'$ and call the resulting graph $G''$. By the definition of XX-trigraphs, it follows that $\{v_{2},v_{4},v_{6}\}$ is a dominating set of $G''$. Hence, by Lemma~\ref{lem:ds:tri}, $\ds{G} \leq \ds{G''} \leq 3$, and the result follows from Lemma~\ref{lem:ds:constant}.
\end{proof}

Intuitively, Theorem~\ref{thm:ds:circular} and Lemma~\ref{lem:ds:s2} imply that Theorem~\ref{thm:main-base2} yields a strip-structure.
Therefore, we turn to the basic classes of strips of Theorem~\ref{thm:main-base2}.
For reasons that will become clear later, we need stronger results for strips $(J,Z)$ than just being able to compute a smallest dominating set.
However, intuitively, if we compute $\ds{J}$, then we enforce that any dominating set that attains this bound contains a vertex of $N[Z]$.
Sometimes we might want to enforce this, but sometimes we do not.
Similarly, it might be that $N(Z)$ is already dominated by a vertex from another strip, and then we do not need to dominate it, but can of course still include a vertex of it in the dominating set.
Therefore, we want to compute $\dsg{J\setminus(Q \cup R)}{N[R]}$ for any disjoint $Q,R \subseteq Z$.
We now do this for each strip type of Theorem~\ref{thm:main-base2}.

\begin{lemma}
\label{lem:ds:circular-stripe}
  Let $(J,Z)$ be a stripe such that $J$ is a proper circular-arc graph and either $J$ is a (strong) clique or $\alpha(J) > 3$.
  For any disjoint $Q,R \subseteq Z$, $\dsg{J\setminus(Q \cup R)}{N[R]}$ can be computed in linear time.
\end{lemma}
\begin{proof}
  First, test whether $J$ is a (strong) clique.
  If so, then $\dsg{J\setminus(Q \cup R)}{N[R]}$ is trivial to compute for any disjoint $Q,R \subseteq Z$; this all takes linear time.
  So assume that $J$ is not a (strong) clique.
  Then find a set of arcs $I_{1},\ldots,I_{n}$ of the sphere $\mathbb{S}_{1}$ that forms a representation of $J$ as a proper circular-arc graph (that is, $I_{i} \not\subseteq I_{j}$ for each $i\not=j$).
  Such a set of arcs can be found in linear time~\cite{DengHH1996}.

  We now show this representation has the Helly property, that is, any three arcs that pairwise intersect have a common intersection point.
  To this end, it suffices to show that no two or three arcs cover the circle (\ie their union is equal to $\mathbb{S}_{1}$)~\cite[Theorem~7]{LinSS2013}.
  For sake of contradiction, suppose that there are three arcs, say $I_{1},I_{2},I_{3}$, that jointly cover the circle (we implicitly allow that $I_{1}$ and $I_{2}$ already cover the circle).
  Let $i \in \{1,2,3\}$ and consider any arc $I_{a}$ that intersects $I_{i}$ (possibly $a=i$).
  Since $J$ is a proper circular-arc graph and $I_{i}$ intersects $I_{i'}$ for each $i' \in \{1,2,3\}\setminus\{i\}$, $I_{a}$ covers an endpoint of $I_{i}$ as well as an endpoint of $I_{i'}$ for some $i' \in \{1,2,3\}\setminus\{i\}$.
  Hence, $I_{a}$ covers at least two endpoints of the arcs $I_{1},I_{2},I_{3}$.
  Since $I_{1},I_{2},I_{3}$ cover the circle, it follows that any arc corresponding to a vertex of an independent set of $G$ must cover at least two endpoints of the arcs $I_{1},I_{2},I_{3}$.
  Therefore, $\alpha(J) \leq 3$, a contradiction.
  Hence, the representation has the Helly property.

  From the definition of a stripe, each $z \in Z$ is strongly simplicial.
  Since the representation has the Helly property, there is a point $p_{z} \in \mathbb{S}_{1}$ for each $z \in Z$ such that the arcs containing $p_{z}$ are precisely those corresponding to $N[z]$.
  We can assume that $p_{z}$ is contained in the interior of each of these intervals.

  Now consider disjoint sets $Q,R \subseteq Z$. For each $z \in R$, remove the interval $[p_{z}-\epsilon,p_{z}+\epsilon]$ from each arc for some infinitesimally small $\epsilon > 0$.
  Let $I'_{1},\ldots,I'_{n'}$ be the resulting set of arcs and $J'$ the intersection graph of these arcs.
  Note that $n' = n + |N[R]|$ and that $J'$ is a circular-arc graph.

  Consider some $z \in R$. Observe that both copies of $z$ in $J'$ correspond to either a `leftmost' or a `rightmost' interval of the representation. For each copy of $z$, add a new vertex to $J'$ that is adjacent only to this copy of $z$. Let $J''$ be the graph obtained by adding these new vertices for all $z \in R$, and removing $Q$. By the preceding observation, $J''$ is still a circular-arc graph. 

  The vertices added to $J'$ ensure that both copies of $z$, for each $z \in R$, will belong to some smallest dominating set of $J''$. In particular, this ensures that there is a smallest dominating set of $J''$ that dominates $N[R]$. Hence, $\ds{J''} - 2|R| = \dsg{J \setminus (Q \cup R)}{N[R]}$. As $J''$ is a circular-arc graph, $\ds{J''}$ can be computed in linear time following Theorem~\ref{thm:ds:circular}.
\end{proof}

\begin{lemma}
\label{lem:ds:z5}
  Let $(J,Z)$ be a thickening of a stripe $(J',Z') \in \mc{Z}_{5}$ such that $J$ is a graph.
  For any disjoint $Q,R \subseteq Z$, $\dsg{J\setminus (Q \cup R)}{N[R]}$ can be computed in $O(n^{5})$ time.
\end{lemma}
\begin{proof}
  Consider disjoint sets $Q,R \subseteq Z$.
  For any $z \in R$, add a new vertex adjacent to $z$ only, and remove $Q$; call the resulting graph $H$.
  By construction, there is a smallest dominating set of $H$ that contains $R$, and thus $\ds{H} - |R| = \dsg{J \setminus (Q \cup R)}{N[R]}$.

  It remains to compute $\ds{H}$.
  Since $J$ is a thickening of $J'$, where $J'$ is an XX-trigraph, $H$ is a thickening of a trigraph $H'$ that is obtained from $J'$ by possibly adding a vertex that is strongly adjacent to $v_{7}$, possibly adding a vertex that is strongly adjacent to $v_{8}$, and possibly removing~$v_{7}$ or $v_{8}$.
  Let $H''$ be the graph obtained from $H$ by removing all semi-edges.
  From the definition of XX-trigraphs, one of $\{v_{4},v_{6},v_{7}\}$, $\{v_{2},v_{6},v_{8}\}$, $\{v_{3},v_{6},v_{7},v_{8}\}$, or $\{v_{2},v_{4},v_{6}\}$ is a dominating set of~$H''$.
  Hence, using Lemma~\ref{lem:ds:tri}, $\ds{H} \leq \ds{H''} \leq 4$, and the result follows from Lemma~\ref{lem:ds:constant}.
\end{proof}

\begin{lemma}
\label{lem:ds:stripe-is}
  Let $(J,Z)$ be a stripe such that $J$ is a graph and $1 \leq |Z| \leq 2$, $\alpha(J) \leq 3$, and $V(J)\setminus N[Z] \not= \emptyset$. For any disjoint $Q,R \subseteq Z$, $\dsg{J\setminus (Q \cup R)}{N[R]}$ can be computed in $O(n^{4})$ time.
\end{lemma}
\begin{proof}
  Consider disjoint sets $Q,R \subseteq Z$. For any $z \in R$, add a new vertex adjacent to $z$ only, and remove $Q$; call the resulting graph $J'$.
  By construction, there is a smallest dominating set of $J'$ that contains $R$, and thus $\ds{J'} - |R| = \dsg{J \setminus (Q \cup R)}{N[R]}$.

  It remains to compute $\ds{J'}$. Since each $z \in R$ is simplicial and $R$ is an independent set (recall that $(J,Z)$ is a stripe), there is a maximum independent set of $J$ that contains $R$.
  Hence, using Proposition~\ref{prp:ds:is-size}, $\ds{J'} \leq \alpha(J) \leq 3$.
  The result then follows from Lemma~\ref{lem:ds:constant}.
\end{proof}
Since trivial line graph strips have at most three vertices (recall Definition~\ref{def:linestrip}), the following lemma is immediate.
\begin{lemma}
\label{lem:ds:spot}
  Let $(J,Z)$ be a trivial line graph strip such that $J$ is a graph.
  For any disjoint $Q,R \subseteq Z$, $\dsg{J\setminus (Q \cup R)}{N[R]}$ can be computed in constant time.
\end{lemma}
We will also rely on the following observation.

\begin{proposition}
\label{prp:ds:strip-boundary}
  Let $(J,Z)$ be a strip such that $J$ is a graph.
  For any disjoint $Q,R \subseteq Z$, there is a set $D \subseteq V(J) \setminus Z$ of size $\dsg{J\setminus (Q \cup R)}{N[R]}$ that dominates $V(J) \setminus (Q \cup N[R])$.
\end{proposition}
\begin{proof}
  Let $D \subseteq V(J) \setminus (Q \cup R)$ be a set of size $\dsg{J\setminus (Q \cup R)}{N[R]}$ that dominates $V(J) \setminus (Q \cup N[R])$.
  Suppose that $D$ contains some $z \in Z$. By the definition of a strip, $z$ is a simplicial vertex that is not adjacent to any $z' \in Z \setminus \{z\}$.
  Hence, we can replace $z$ by any other vertex of $N(z)$, and the resulting set still is a subset of $V(J) \setminus (Q \cup R)$ that dominates $V(J) \setminus (Q \cup N[R])$.
\end{proof}

\subsection{Stitching Dominating Sets}
\label{sec:fpt-ds:stitch}
We now describe a method to stitch dominating sets for individual strips of a strip-structure together to form a dominating set of the entire graph. To this end, we need several supporting definitions and lemmas.

\begin{definition}
\label{def:ds:striped}
  Let $G$ be a graph and let $(H,\eta)$ be a purified strip-structure of nullity zero for $G$.
  We say $h \in V(H)$ is \emph{striped} if there is an $F \in E(H)$ such that $h \in \overline{F}$ and the strip corresponding to $F$ is a stripe.
  The subset of $V(H)$ that is striped is denoted $\mc{P}(H)$.
\end{definition}
We can show the following useful lemma.

\begin{lemma}
\label{lem:ds:striped}
  Let $G$ be a graph, let $(H,\eta)$ be a purified strip-structure of nullity zero for $G$, let $k$ be an integer, and let $d = \max_{F \in E(H)} |\overline{F}|$.
  If $\ds{G} \leq k$, then $|\mc{P}(H)| \leq dk$.
\end{lemma}
\begin{proof}
  Suppose that $\ds{G} \leq k$, but $|\mc{P}(H)| > dk$.
  Let $D$ be a dominating set of $G$ of size at most~$k$.
  Define a set $M \subseteq V(H)$ where $h \in M$ if and only if there is an $F \in E(H)$ such that $\eta(F) \cap D \not= \emptyset$ and $h \in \overline{F}$.
  Since $|D| \leq k$ and $d = \max_{F \in E(H)} |\overline{F}|$, $|M \cap \mc{P}(H)| \leq dk$.
  Let $h \in \mc{P}(H) \setminus M$; as $|\mc{P}(H)| > dk$, $h$ is properly defined.
  As $h \in \mc{P}(H)$, there is an $F \in E(H)$ such that $h \in \overline{F}$ and the strip $(J,Z)$ corresponding to $F$ is a stripe. Consider the vertices in $\eta(F,h)$.
  Note that $\eta(F) \cap D = \emptyset$, because $h \not\in M$.
  Hence, the vertices in $\eta(F,h)$ can only be dominated by a vertex in $\eta(h) \cap D$, since $\eta(F,h') \cap \eta(F,h) = \emptyset$ for any $h' \in \overline{F} \setminus \{h\}$, as $(J,Z)$ is a stripe.
  However, because $h \not\in M$, there is no $F' \in E(H)$ such that $\eta(F') \cap D \not= \emptyset$, and thus $\eta(h) \cap D \not= \emptyset$.
  Because $(H,\eta)$ has nullity zero, $\eta(F,h) \not= \emptyset$, and therefore, there is a vertex not dominated by $D$, a contradiction.
\end{proof}
We now define a set of auxiliary edge-weighted multigraphs (each possibly with parallel edges) with each an associated integer.
The goal will be to show that if $\ds{G} \leq k$, then there is at least one such an auxiliary multigraph with an edge dominating of weight bounded by the associated integer. 

Let $G$ be a graph, let $(H,\eta)$ be a purified strip-structure of nullity zero for $G$ such that $1 \leq |\overline{F}| \leq 2$ for each $F \in E(H)$, and let $P \subseteq \mc{P}(H)$. 
The idea of the construction is to ensure that $h \in P$ if and only if there is an $F \in E(H)$ that corresponds to a stripe for which the dominating set has a vertex in $\eta(F,h)$.
To this end, we define the triple $(K_{P}, w_{P}, k_{P})$ of a multigraph, an edge-weight function, and an integer as follows. Initially, $K_{P}$ is empty and $k_{P} = 0$.
For each $h \in V(H)$, add a vertex $v_h$ to~$K_{P}$.
If $h \in P$, then also add a vertex $v'_h$ to $K_{P}$ as well as an edge between $v_h$ and $v'_h$ of weight $\infty$.
For each $F \in E(H)$, there are several cases:

\ccase{1} $\overline{F} = \{h\}$ for some $h \in V(H)$.\\
Let $(J,Z)$ be the strip corresponding to $F$.
Note that $(J,Z)$ is in fact a stripe.
There are several cases:

\csubcase{1a} $h \in P$.\\
Add a vertex $v_F$ and an edge $e_F$ between $v_h$ and $v_F$ of weight $\ds{J} - \dsg{J\setminus Z}{N[Z]}$ to $K_{P}$; additionally, increase $k_P$ by $\dsg{J\setminus Z}{N[Z]}$.

\csubcase{1b} $h \not\in P$.\\
Increase $k_{P}$ by $\ds{J\setminus Z}$.

\ccase{2} $\overline{F} = \{h,h'\}$ for some $h, h' \in V(H)$. The strip $(J,Z)$ corresponding to $F$ is a stripe.\\
Let $Z = \{z,z'\}$, where $z$ corresponds to $h$ and $z'$ to $h'$. There are several cases:

\csubcase{2a}  $h,h' \in P$ and $\dsg{J \setminus Z}{N[Z]} = \dsg{J \setminus \{z'\}}{N[z']} = \dsg{J \setminus \{z\}}{N[z]} = \ds{J}-1$. \\
Add two vertices $v_{F}^{1},v_{F}^{2}$ to $K_{P}$ and three edges $e_{F}^{h},e_{F}^{h'},e_{F}^{1}$:
\begin{itemize}
  \item $e_{F}^{h}$ connects $v_{h}$ and $v_{F}^{1}$, and has weight $1$.
  \item $e_{F}^{h'}$ connects $v_{h'}$ and $v_{F}^{1}$, and has weight $1$.
  \item $e_{F}^{1}$ connects $v_{F}^{1}$ and $v_{F}^{2}$, and has weight $\infty$.
\end{itemize}
Increase $k_P$ by $\dsg{J \setminus Z}{N[Z]}-1$.

\csubcase{2b} $h,h' \in P$ and not $\dsg{J \setminus Z}{N[Z]} = \dsg{J \setminus \{z'\}}{N[z']} = \dsg{J \setminus \{z\}}{N[z]} = \ds{J}-1$.\\
Add two vertices $v_{F}^{h}, v_{F}^{h'}$ to $K_{P}$, and three edges $e_{F}^{h}, e_{F}^{h'}, e_{F}^{h,h'}$:
\begin{itemize}
  \item $e_{F}^{h,h'}$ connects $v_{h}$ and $v_{h'}$, and has weight $\ds{J} - \dsg{J \setminus Z}{N[Z]}$.
  \item $e_{F}^{h}$ connects $v_{F}^{h}$ and $v_{h}$, and has weight $\dsg{J \setminus \{z'\}}{N[z']} - \dsg{J \setminus Z}{N[Z]}$.
  \item $e_{F}^{h'}$ connects $v_{F}^{h'}$ and $v_{h'}$, and has weight $\dsg{J \setminus \{z\}}{N[z]} - \dsg{J \setminus Z}{N[Z]}$. 
\end{itemize}
Increase $k_P$ by $\dsg{J \setminus Z}{N[Z]}$.

\csubcase{2c} $h \in P$ and $h' \not\in P$.\\
Add a vertex $v_{F}$ to $K_{P}$ as well as an edge $e_{F}$ of weight $\ds{J \setminus \{z'\}} - \dsg{J \setminus Z}{N[z]}$ between $v_{F}$ and~$v_{h}$, and increase $k_P$ by $\dsg{J \setminus Z}{N[z]}$.

\csubcase{2d} $h' \in P$ and $h \not\in P$.\\
Add a vertex $v_{F}$ to $K_{P}$ as well as an edge $e_{F}$ of weight $\ds{J \setminus\{z\}} - \dsg{J \setminus Z}{N[z']}$ between $v_{F}$ and~$v_{h'}$, and increase $k_{P}$ by $\dsg{J \setminus Z}{N[z']}$.

\csubcase{2e} $h,h' \not\in P$. \\
Increase $k_{P}$ by $\ds{J \setminus Z}$.

\ccase{3} $\overline{F} = \{h,h'\}$ for some $h, h' \in V(H)$. The strip $(J,Z)$ corresponding to $F$ is a spot.\\
Add an edge $e_{F}$ between $v_{h}$ and $v_{h'}$ of weight $1$.

\medskip\noindent
Since $1 \leq |\overline{F}| \leq 2$ for each $F \in E(H)$ and $(H,\eta)$ is purified, the cases are exhaustive.

It is worth observing that $k_{P} \geq 0$.
Indeed, the only case where possibly a negative number could be added to $k_P$ is Case~2a.
However, we can argue that the number added to $k_P$ in Case~2a is nonnegative.
Let $h$, $h'$, $F$, and $(J,Z)$ be as in Case~2a.
Since $(J,Z)$ is a stripe, $N[z] \cap N[z'] = \emptyset$, and since the strip-structure has nullity zero, $N[z] \not= \emptyset$. Hence, any dominating set of $J \setminus \{z'\}$ that does not necessarily dominate $N[z']$ has at least one vertex in order to dominate $N[z]$.
Therefore, $\dsg{J \setminus \{z'\}}{N[z']} \geq 1$, and thus $\dsg{J \setminus Z}{N[Z]}-1 = \dsg{J \setminus \{z'\}}{N[z']} - 1 \geq 0$.
Then the number added to $k_P$ in this case is nonnegative.
Hence, $k_P \geq 0$.
Moreover, each edge of $K_{P}$ received a nonnegative weight.

Now recall that an \emph{edge dominating set} of a graph $G$ is a set $D \subseteq E(G)$ such that each edge of~$E(G)$ has an endpoint in common with an edge of $D$.

\begin{lemma}
\label{lem:ds:reduction}
  Let $G$ be a graph, let $(H,\eta)$ be a purified strip-structure of nullity zero for $G$ such that $1 \leq |\overline{F}| \leq 2$ for each $F \in E(H)$, and let $k$ be an integer.
  Then $\ds{G} \leq k$ if and only if there is a $P \subseteq \mc{P}(H)$ such that there is an edge dominating set $D$ of $K_{P}$ for which $w_{P}(D) \leq k - k_P$.
  Moreover, such a set $D$ satisfies $|D| \leq k$.
\end{lemma}
\begin{proof}
  Suppose that there is a $P \subseteq \mc{P}(H)$ such that there is an edge dominating set $D$ of $K_{P}$ for which $w_{P}(D) \leq k - k_P$.

  We first slightly modify $D$ without increasing its weight.
  Consider any $F \in E(H)$ where $\overline{F} = \{h,h'\}$ for which the corresponding strip $(J,Z)$ is a stripe.
  Let $Z = \{z,z'\}$, where $z$ corresponds to $h$ and $z'$ to $h'$.
  Suppose additionally that $h,h' \in P$ and not $\dsg{J \setminus Z}{N[Z]} = \dsg{J \setminus \{z'\}}{N[z']} = \dsg{J \setminus \{z\}}{N[z]} = \ds{J}-1$.
  In short, $F$ corresponds to a stripe in Case~2b.
  Suppose that $e^{h,h'}_{F} \in D$ and at least one of $e^{h}_{F},e^{h'}_{F}$ is in $D$.
  Then by the construction of $K_{P}$, we can remove $e^{h}_{F}$ and $e^{h'}_{F}$ from $D$, and the resulting set would still be an edge dominating set of $K_{P}$; clearly, the weight has not increased.
  Suppose that $e^{h,h'}_{F} \not\in D$ and $e^{h}_{F},e^{h'}_{F} \in D$.
  Then we claim that we can replace $e^{h}_{F}$ and $e^{h'}_{F}$ by just $e^{h,h'}_{F}$.
  Indeed, by the construction of $K_{P}$, the resulting set is still an edge dominating set of $K_{P}$.
  Moreover, note that
  \begin{eqnarray*}
                       \ds{J} - 1 & \leq & \dsg{J \setminus \{z'\}}{N[z']},\\
    \dsg{J \setminus \{z\}}{N[z]} & \leq & \ds{J},\\
    \dsg{J \setminus Z}{N[Z]}     & \leq & \dsg{J \setminus \{z'\}}{N[z']},\\
    \dsg{J \setminus\{z\}}{N[z]}  & \leq & \dsg{J \setminus Z}{N[Z]}+1,
  \end{eqnarray*}
  and thus, using the assumption of Case~2b,
  \begin{equation*}
    \ds{J} - \dsg{J \setminus Z}{N[Z]} \leq (\dsg{J \setminus\{z'\}}{N[z']} - \dsg{J \setminus Z}{N[Z]}) + (\dsg{J \setminus\{z\}}{N[z]} - \dsg{J \setminus Z}{N[Z]}),
  \end{equation*}
  and therefore, $w_{P}(e_{F}^{h,h'}) \leq w_{P}(e_{F}^{h}) + w_{P}(e_{F}^{h'})$.
  Hence, the weight of the resulting set is at most $w_{P}(D)$.

  Apply the above modification to $D$ for each $F \in E(H)$ that corresponds a stripe in Case~2b.
  By abuse of notation, we call the resulting set $D$ as well.
  Observe that $D$ is still an edge dominating set of $K_{P}$ for which $w_{P}(D) \leq k - k_P$, but now with the additional property that for each $F \in E(H)$ that is a stripe in Case~2b, at most one of $e_{F}^{h}$, $e_{F}^{h'}$, and $e_{F}^{h,h'}$ is in $D$.

  We construct a set $D' \subseteq V(G)$ of size at most $k$ as follows (and later show that it is in fact a dominating set of $G$).
  Initially, $D'=\emptyset$.
  For each $F \in E(H)$, we consider the same cases as before:

  \ccase{1} $\overline{F} = \{h\}$ for some $h \in V(H)$.\\
  Let $(J,Z)$ be the strip corresponding to $F$.
  Note that $(J,Z)$ is in fact a stripe.
  There are several cases:

  \csubcase{1a} $h \in P$.\\
  There are two cases:

  \csubcase{1a-i} $e_{F} \in D$.\\
  Add a smallest subset of $V(J) \setminus Z$ to $D'$ that dominates all vertices of $V(J)$---note that the weight of~$e_{F}$ is $\ds{J} - \dsg{J\setminus Z}{N[Z]}$ and $k_{P}$ was increased by $\dsg{J\setminus Z}{N[Z]}$ in the construction; the sum is equal to the size of the set added to $D'$ (using Proposition~\ref{prp:ds:strip-boundary}).

  \csubcase{1a-ii} $e_{F} \not\in D$.\\
  Add a smallest subset of $V(J) \setminus Z$ to $D'$ that dominates all vertices of $V(J) \setminus N[Z]$---note that $k_{P}$ was increased by $\dsg{J\setminus Z}{N[Z]}$ in the construction, which is equal to the size of the set added to~$D'$.

  \csubcase{1b} $h \not\in P$.\\
  Add a smallest subset of $V(J) \setminus Z$ to $D'$ that dominates all vertices of $V(J) \setminus Z$---note that $k_{P}$ was increased by $\ds{J \setminus Z}$ in the construction, which is equal to the size of the set added to $D'$.

  \ccase{2} $\overline{F} = \{h,h'\}$ for some $h, h' \in V(H)$. The strip $(J,Z)$ corresponding to $F$ is a stripe.\\
  Let $Z = \{z,z'\}$, where $z$ corresponds to $h$ and $z'$ to $h'$. There are several cases:

  \csubcase{2a} $h,h' \in P$ and $\dsg{J \setminus Z}{N[Z]} = \dsg{J \setminus \{z'\}}{N[z']} = \dsg{J \setminus \{z\}}{N[z]} = \ds{J}-1$.\\
  Then at least one of $e^{h}_{F}$ and $e_{F}^{h'}$ is in $D$ by the construction of $K_P$, because $w_{P}(e^1_F) = \infty$ and thus $e^1_F \not\in D$.
  This leads to several cases:

  \csubcase{2a-i} $e^{h}_{F},e^{h'}_{F} \in D$.\\
  Add a smallest subset of $V(J) \setminus Z$ to $D'$ that dominates all vertices of $V(J)$ --- note that the weight of $e^{h}_{F},e^{h'}_{F}$ combined is $2 = \ds{J} - (\dsg{J \setminus Z}{N[Z]}-1)$ and $k_{P}$ was increased by $\dsg{J\setminus Z}{N[Z]}-1$ in the construction; the sum is equal to the size of the set added to $D'$ (using Proposition~\ref{prp:ds:strip-boundary}).

  \csubcase{2a-ii} $e^{h}_{F} \in D$ and $e^{h'}_{F} \not\in D$.\\
  Add a smallest subset of $V(J) \setminus Z$ to $D'$ that dominates all vertices of $V(J) \setminus N[z']$ --- note that the weight of $e^{h}_{F}$ is $1 = \dsg{J \setminus\{z'\}}{N[z']} - (\dsg{J \setminus Z}{N[Z]}-1)$ and $k_{P}$ was increased by $\dsg{J\setminus Z}{N[Z]}-1$ in the construction; the sum is equal to the size of the set added to $D'$ (using Proposition~\ref{prp:ds:strip-boundary}).

  \csubcase{2a-iii} $e^{h'}_{F} \in D$ and $e^{h}_{F} \not\in D$.\\
  Add a smallest subset of $V(J) \setminus Z$ to $D'$ that dominates all vertices of $V(J) \setminus N[z]$ --- note that the weight of $e^{h'}_{F}$ is $1 = \dsg{J \setminus\{z\}}{N[z]} - (\dsg{J \setminus Z}{N[Z]}-1)$ and $k_{P}$ was increased by $\dsg{J\setminus Z}{N[Z]}-1$ in the construction; the sum is equal to the size of the set added to $D'$ (using Proposition~\ref{prp:ds:strip-boundary}).

  \csubcase{2b} $h,h' \in P$ and not $\dsg{J \setminus Z}{N[Z]} = \dsg{J \setminus \{z'\}}{N[z']} = \dsg{J \setminus \{z\}}{N[z]} = \ds{J}-1$.\\
  After the above modification of $D$, we can assume that at most one of $e_{F}^{h}$, $e_{F}^{h'}$, and $e_{F}^{h,h'}$ is in $D$. This leads to several cases:

  \csubcase{2b-i} $e_{F}^{h,h'} \in D$.\\
  Add a smallest subset of $V(J) \setminus Z$ to $D'$ that dominates all vertices of $V(J)$ --- note that the weight of $e_{F}^{h,h'}$ is $\ds{J} - \dsg{J \setminus Z}{N[Z]}$ and $k_{P}$ was increased by $\dsg{J \setminus Z}{N[Z]}$; the sum is equal to the size of the set added to $D'$ (using Proposition~\ref{prp:ds:strip-boundary}).

  \csubcase{2b-ii} $e_{F}^{h} \in D$.\\
  Add a smallest subset of $V(J) \setminus Z$ to $D'$ that dominates all vertices of $V(J) \setminus N[z']$ and that does not contain a vertex of $Z$---note that the weight of $e_{F}^{h}$ is $\dsg{J \setminus \{z'\}}{N[z']} - \dsg{J \setminus Z}{N[Z]}$ and~$k_{P}$ was increased by $\dsg{J \setminus Z}{N[Z]}$; the sum is equal to the size of the set added to $D'$ (using Proposition~\ref{prp:ds:strip-boundary}).

  \csubcase{2b-iii} $e_{F}^{h'} \in D$.\\
  Add a smallest subset of $V(J) \setminus Z$ to $D'$ that dominates all vertices of $V(J) \setminus N[z]$ and that does not contain a vertex of $Z$---note that the weight of $e_{F}^{h'}$ is $\dsg{J \setminus \{z\}}{N[z]} - \dsg{J \setminus Z}{N[Z]}$ and~$k_{P}$ was increased by $\dsg{J \setminus Z}{N[Z]}$; the sum is equal to the size of the set added to $D'$ (using Proposition~\ref{prp:ds:strip-boundary}).

  \csubcase{2b-iv} none of $e_{F}^{h}$, $e_{F}^{h'}$, and $e_{F}^{h,h'}$ is in $D$.\\
  Add a smallest subset of $V(J) \setminus Z$ to $D'$ that dominates all vertices of $V(J) \setminus N[Z]$ --- note that $k_{P}$ was increased by $\dsg{J\setminus Z}{N[Z]}$ in the construction, which is equal to the size of the set added to~$D'$.

  \csubcase{2c} $h \in P$ and $h' \not\in P$.\\
  There are two cases:

  \csubcase{2c-i} $e_{F} \in D$.\\
  Add a smallest subset of $V(J) \setminus Z$ to $D'$ that dominates all vertices of $V(J) \setminus \{z'\}$---note that the weight of $e_{F}$ is $\dsg{J \setminus\{z'\}}{N[z']} - \dsg{J \setminus Z}{N[z]}$ and that~$k_{P}$ was increased by $\dsg{J \setminus Z}{N[z]}$ in the construction; the sum is equal to the size of the set added to $D'$ (using Proposition~\ref{prp:ds:strip-boundary}).

  \csubcase{2c-ii} $e_{F} \not\in D$.\\
  Add a smallest subset of $V(J) \setminus Z$ to $D'$ that dominates all vertices of $V(J) \setminus (Z \cup N[z])$---note that~$k_{P}$ was increased by $\dsg{J\setminus Z}{N[z]}$ in the construction, which is equal to the size of the set added to~$D'$.

  \csubcase{2d} $h' \in P$ and $h \not\in P$.\\
  There are two cases:

  \csubcase{2d-i} $e_{F} \in D$.\\
  Add a smallest subset of $V(J) \setminus Z$ to $D'$ that dominates all vertices of $V(J) \setminus \{z\}$---note that the weight of $e_{F}$ is $\ds{J \setminus\{z\}} - \dsg{J \setminus Z}{N[z']}$ and that $k_{P}$ was increased by $\dsg{J \setminus Z}{N[z']}$ in the construction; the sum is equal to the size of the set added to $D'$ (using Proposition~\ref{prp:ds:strip-boundary}).

  \csubcase{2d-ii} $e_{F} \not\in D$.\\
  Add a smallest subset of $V(J) \setminus Z$ to $D'$ that dominates all vertices of $V(J) \setminus (Z \cup N[z'])$---note that~$k_{P}$ was increased by $\dsg{J\setminus Z}{N[z']}$ in the construction, which is equal to the size of the set added to~$D'$.

  \csubcase{2e} $h,h' \not\in P$.\\
  Add a smallest subset of $V(J) \setminus Z$ to $D'$ that dominates all vertices of $V(J) \setminus Z$ --- note that $k_{P}$ was increased by $\ds{J \setminus Z}$ in the construction, which is equal to the size of the set added to $D'$.

  \ccase{3} $\overline{F} = \{h,h'\}$ for some $h, h' \in V(H)$.
  The strip $(J,Z)$ corresponding to $F$ is a spot.\\
  If $e_{F} \in D$, then add the vertex of $J \setminus Z$ to $D'$ --- note that the weight of $e_F$ is $1$, which is equal to the size of the set added to $D'$.
  If $e_F \not\in D$, then do nothing.

  \medskip\noindent
  From the construction of $D'$ and the analysis in each of the above cases, it follows that $|D'| \leq w_P(D) + k_P = k$.

  We now show that $D'$ is a dominating set of $G$.
  Let $v$ be an arbitrary vertex of $G$.
  Since $(H,\eta)$ is a strip-structure for $G$, by definition there is an $F \in E(H)$ such that $v \in \eta(F)$.
  Let $(J,Z)$ be the strip corresponding to~$F$.
  Then $v \in N(Z)$ or $v \in V(J) \setminus N[Z]$. 
  If $v \in V(J) \setminus N(Z)$, then $(J,Z)$ is not a spot, and regardless of in which of the above subcases of Case~1 and~2 that $F$ falls, $v$ will be dominated by the dominating set that is added to $D'$ for $F$.
  So assume that $v \in N(z)$ for some $z \in Z$ and let $h \in V(H)$ correspond to $z$. If $(J,Z)$ is a spot, then we have a choice in $z$ (and thus also in $h$): we prefer that $h \in P$ if possible.

  If $h \not\in P$ and $F$ falls into Case~1 or Case~2, then regardless of in which of the above subcases $F$ falls (Case~1b, 2c, 2d, or 2e), $N(z)$ and thus $v$ will be dominated by the dominating set that is added to $D'$ for $F$.

  If $h \not\in P$ and $F$ falls into Case~3, then let $h'$ denote the single element of $\overline{F} \setminus\{h\}$.
  Since we preferred~$h$ to $h'$, $h' \not\in P$.
  By the construction of $K_P$, there is an edge between $v_h$ and $v_{h'}$ corresponding to~$F$.
  Moreover, since $h,h' \not \in P$ and by the construction of $K_P$, every edge incident with $v_h$ and $v_{h'}$ corresponds to an edge introduced for a spot in Case~3.
  Therefore, since $D$ is an edge dominating set of $K_P$, there is an edge in $D$ that is incident with $v_h$ or $v_h'$ which corresponds to an $F' \in E(H)$ (possibly $F = F'$) with $h \in \overline{F'}$ or $h' \in \overline{F'}$ such that $F'$ corresponds to a spot.
  By the construction of $D'$, the single vertex of $\eta(F')$ is in $D'$, that is, $\eta(F') \subseteq D'$.
  Assume without loss of generality that $h \in \overline{F'}$.
  Then $v$ is dominated, because $\eta(F',h) = \eta(F') \subseteq D'$ and thus $\eta(h) \cap D' \not= \emptyset$, $v \in \eta(F,h) \subseteq \eta(h)$, and~$\eta(h)$ is a clique.

  If $h \in P$, then due to the presence of the edge $e$ between $v_{h}$ and $v'_{h}$, there is an edge $e'$ incident with $v_h$ that is in $D$. Since $e$ has weight $\infty$, $e' \not=e$.
  Hence, there is an $F' \in E(H)$ (possibly $F' = F$) with $h \in \overline{F'}$ such that an edge associated with $F'$ in the construction is in $D$.
  Then, by inspecting the above cases (Case~1a-i, 2a-i, 2b-i, 2b-ii, 2b-iii, 2c-i, 2d-i, 3) for the strip $(J',Z')$ corresponding to $F'$, there is a vertex of $D'$ in $N(z')$, where $z'$ is the vertex of $Z'$ that corresponds to $h$.
  In other words, $\eta(F',h) \cap D' \not= \emptyset$, and thus $\eta(h) \cap D' \not= \emptyset$.
  Since $\eta(h)$ is a clique and $\eta(F',h) \subseteq \eta(h)$, each vertex in $\eta(h)$ is dominated, and in particular each vertex in $\eta(F,h)$ is dominated.
  As $v \in N[z] = \eta(F,h)$, $v$ is dominated.
  Hence, $D'$ is a dominating set of $G$ of size at most $k$.

  \bigskip
  Suppose that $\ds{G} \leq k$.
  Let $D'$ be a smallest dominating set of $G$; hence, $|D'| \leq k$.
  Let $P$ be such that $h \in P$ if and only if $h \in \mc{P}(H)$ and there is an $F \in E(H)$ with $h \in \overline{F}$ and $D' \cap \eta(F,h) \not= \emptyset$. 

  We construct a set $D \subseteq E(K_{P})$ of weight at most $k-k_{P}$ that is an edge dominating set of $K_{P}$. Initially, $D = \emptyset$.
  For each $F \in E(H)$, we consider the following cases.

  \ccase{1} $\overline{F} = \{h\}$ for some $h \in V(H)$.\\
  If $h \in P$ and $\eta(F,h) \cap D' \not= \emptyset$, then add $e_{F}$ to $D$ --- note that the weight of $e_{F}$ is $\ds{J} - \dsg{J\setminus Z}{N[Z]}$, which is at most $|D' \cap \eta(F)| - \dsg{J\setminus Z}{N[Z]}$.

  \ccase{2} $\overline{F} = \{h,h'\}$ for some $h, h' \in V(H)$.
  The strip $(J,Z)$ corresponding to $F$ is a stripe.\\
  Let $Z = \{z,z'\}$, where $z$ corresponds to $h$ and $z'$ to $h'$.
  There are several cases:

  \csubcase{2a}  $h,h' \in P$ and $\dsg{J \setminus Z}{N[Z]} = \dsg{J \setminus \{z'\}}{N[z']} = \dsg{J \setminus \{z\}}{N[z]} = \ds{J}-1$.
  \begin{itemize}
    \item If $\eta(F,h) \cap D' \not= \emptyset$ and $\eta(F,h') \cap D' \not= \emptyset$, then add $e_{F}^{h}$ and $e_{F}^{h'}$ to $D$ --- note that the weight of the added edges is $2=\ds{J} - (\dsg{J\setminus Z}{N[Z]}-1)$, which is at most $|D' \cap \eta(F)| - (\dsg{J\setminus Z}{N[Z]}-1)$.
    \item If $\eta(F,h) \cap D' \not= \emptyset$ and $\eta(F,h') \cap D' = \emptyset$, then add $e_{F}^{h}$ to $D$ --- note that the weight of the added edges is $1=\dsg{J \setminus \{z'\}}{N[z']} - (\dsg{J\setminus Z}{N[Z]}-1)$, which is at most $|D' \cap \eta(F)| - (\dsg{J\setminus Z}{N[Z]}-1)$.
    \item If $\eta(F,h') \cap D' \not= \emptyset$ and $\eta(F,h) \cap D' = \emptyset$, then add $e_{F}^{h'}$ to $D$ --- note that the weight of the added edges is $1=\dsg{J \setminus \{z\}}{N[z]} - (\dsg{J\setminus Z}{N[Z]}-1)$, which is at most $|D' \cap \eta(F)| - (\dsg{J\setminus Z}{N[Z]}-1)$.
    \item If $\eta(F,h) \cap D' = \emptyset$ and $\eta(F,h') \cap D' = \emptyset$, then add $e_{F}^{h}$ to $D$ --- note that the weight of the added edges is $1=\dsg{J \setminus Z}{N[Z]} - (\dsg{J\setminus Z}{N[Z]}-1)$, which is at most $|D' \cap \eta(F)| - (\dsg{J\setminus Z}{N[Z]}-1)$.
  \end{itemize}

  \csubcase{2b} $h,h' \in P$ and not $\dsg{J \setminus Z}{N[Z]} = \dsg{J \setminus \{z'\}}{N[z']} = \dsg{J \setminus \{z\}}{N[z]} = \ds{J}-1$.
  \begin{itemize}
    \item If $\eta(F,h) \cap D' \not= \emptyset$ and $\eta(F,h') \cap D' \not= \emptyset$, then add $e_{F}^{h,h'}$ to $D$ --- note that the weight of $e_{F}^{h,h'}$ is $\ds{J} - \dsg{J\setminus Z}{N[Z]}$, which is at most $|D' \cap \eta(F)| - \dsg{J\setminus Z}{N[Z]}$.
    \item If $\eta(F,h) \cap D' \not= \emptyset$ and $\eta(F,h') \cap D' = \emptyset$, then add $e_{F}^{h}$ to $D$ --- note that the weight of $e_{F}^{h}$ is $\dsg{J \setminus \{z'\}}{N[z']} - \dsg{J\setminus Z}{N[Z]}$, which is at most $|D' \cap \eta(F)| - \dsg{J\setminus Z}{N[Z]}$.
    \item If $\eta(F,h') \cap D' \not= \emptyset$ and $\eta(F,h) \cap D' = \emptyset$, then add $e_{F}^{h'}$ to $D$ --- note that the weight of $e_F^{h'}$ is $\dsg{J \setminus \{z\}}{N[z]} - \dsg{J\setminus Z}{N[Z]}$, which is at most $|D' \cap \eta(F)| - \dsg{J\setminus Z}{N[Z]}$.
    \item If $\eta(F,h) \cap D' = \emptyset$ and $\eta(F,h') \cap D' = \emptyset$, then do nothing --- note that $|D' \cap \eta(F)| - \dsg{J\setminus Z}{N[Z]} \geq 0$.
  \end{itemize}

  \csubcase{2c} $h \in P$ and $h' \not\in P$.\\
  If $\eta(F,h) \cap D' \not= \emptyset$, then add $e_{F}$ to $D$ --- note that the weight of $e_{F}$ is $\ds{J \setminus \{z'\}} - \dsg{J\setminus Z}{N[z]}$, which is at most $|D' \cap \eta(F)| - \ds{J \setminus Z}{N[z]}$.

  \csubcase{2d} $h' \in P$ and $h \not\in P$.\\
  If $\eta(F,h') \cap D' \not= \emptyset$, then add $e_{F}$ to $D$ --- note that the weight of $e_{F}$ is $\ds{J \setminus \{z\}} - \dsg{J\setminus Z}{N[z']}$, which is at most $|D' \cap \eta(F)| - \dsg{J \setminus Z}{N[z']}$.

  \csubcase{2e} $h,h' \not\in P$. \\
  Do nothing; note that $\ds{J \setminus Z} \leq |D' \cap \eta(F)|$.

  \ccase{3} $\overline{F} = \{h,h'\}$ for some $h, h' \in V(H)$. The strip $(J,Z)$ corresponding to $F$ is a spot.\\
  If the vertex of $J\setminus Z$ is in $D'$, then add $e_{F}$ to $D$ --- note that the weight of $e_{F}$ is $|D' \cap \eta(F)|$. 
  Otherwise, do nothing.

  \medskip\noindent
  From the construction of $D$ and the analysis in each of the above cases, it follows that $w_{P}(D) \leq k - k_{P}$.

  To see that $|D| \leq k$, observe that when we add an edge to $D$, we can map it to a unique vertex of~$D'$.
  Indeed, in Case~3, this is clear, because we add $e_F$ when the vertex of the corresponding spot is in $D'$.
  In the other two cases (and their subcases), whenever we add an edge to $D$, then this edge is incident on a vertex~$v_h$ with $h \in P$.
  Moreover, for the considered $F \in E(H)$, there is a vertex in $\eta(F,h) \cap D'$.
  As for stripes the sets $\eta(F,h')$ where $h' \in \overline{F}$ are pairwise disjoint and we add at most one edge per element of $\overline{F}$ to $D$, this indeed presents the claimed mapping from the edges added to~$D$ to the vertices of $D'$.
  Therefore, $|D| \leq k$, as claimed.

  We now show that $D$ is an edge dominating set of $K_{P}$.
  Let $e \in E(K_{P})$. 
  If $e$ is incident on a vertex $v_{h}$ with $h \in P$, then there is an $F \in E(H)$ with $h \in \overline{F}$ such that $\eta(F,h) \cap D' \not= \emptyset$.
  For this $F$, by going through the above cases, an edge incident on $v_h$ will be added to $D$ and $e$ is dominated.
  If $e$ is not incident on any vertex $v_{h}$ for $h \in V(H)$, then $e$ is the edge $e^1_F$ of the construction of Case 2a.
  Then, by the construction of $D$, $e$ is dominated. 
  We can now assume that $e$ is incident only on vertices of $v_h$ for $h \in V(H) \setminus P$.
  By the construction of $K_P$, this implies that $e$ was added in Case~3 due to an $F \in E(H)$ that corresponds to a spot.
  Hence, there is an $F' \in E(H)$ (possibly $F=F'$) such that $h \in \overline{F'}$ or $h' \in \overline{F'}$, and respectively $\eta(F',h) \cap D \not= \emptyset$ or $\eta(F',h') \cap D \not= \emptyset$.
  Assume it is the former.
  Note that $F'$ corresponds to a spot; otherwise, $h \in P$ because $\eta(F',h) \cap D \not= \emptyset$, a contradiction.
  It follows that the edge $e_{F'}$ is in $D$, and thus $e$ is dominated.
  Hence, $D$ is an edge dominating set.
\end{proof}
The previous lemma shows that it suffices to resolve certain instances of the \textsc{Weighted Edge Dominating Set} problem on multigraphs.
Therefore, we would like to use the following result:

\begin{theorem}[Iwaide and Nagamochi~\cite{IwaideN2015}]
  Let $G$ be a graph, $w : E(G) \rightarrow \mathbb{N}^{+}$ a weight function, and $k \in \mathbb{N}^{+}$. Then we can decide in $O^{*}(2.2351^{k})$ time whether $G$ has an edge dominating set~$D$ with $w(D) \leq k$.
\end{theorem}
This theorem extends immediately to multigraphs: parallel edges can be replaced by a single edge with weight equal to the minimum of the weights of the parallel edges, and a self-loop $e$ of a vertex $v$ can be replaced by adding a pendant vertex $v'$ and modifying $e$ so that it is between $v$ and~$v'$. 
Unfortunately, however, our weight functions $w_P$ might set the weight of certain edges of~$K_P$ to~$0$.
Hence, the above theorem might not be directly applicable, as it assumes positive edge weights.
However, there is a straightforward workaround, inspired by a similar construction by Fernau~\cite[Proposition~1]{Fernau2006} for a different problem.

\begin{lemma}
  Let $G$ be a graph, $w: E(G) \rightarrow \mathbb{Z}$ a weight function, and $k \in \mathbb{Z}$.
  Then in linear time, we can construct a graph $G'$, a weight function $w': E(G') \rightarrow \mathbb{N}^{+} \cup \{\infty\}$, and a $k' \in \mathbb{N}^{+}$ such that $G$ has an edge dominating set of weight at most $k$ (under $w$) if and only if $G'$ has an edge dominating set of weight at most $k'$ (under $w'$).
  In particular, $k' \leq \max\{1,1+k-w(N)\}$, where $N \subseteq V(G)$ denotes the set of edges in $G$ of nonpositive weight (under $w$).
\end{lemma}
\begin{proof}
  If $k < w(N)$, then $G$ has no edge dominating set of weight at most $k$, so we can return a trivial ``no''-instance (for example, a graph with a single edge of weight $2$ and $k' = 1$).
  So assume that $k \geq w(N)$.
  If $N = \emptyset$, then we return $G' = G$, $w' = w$, $k' = k$ if $k > 0$, a trivial ``yes''-instance if $E(G) = \emptyset$ and $k\geq 0$, and a trivial no-instance otherwise. 

  We then define $G'$ as follows.
  Let $X$ denote the set of vertices in $G$ incident on an edge of $N$.
  We identify all vertices of $X$ into a single vertex $x$, and remove any loops that arise (this includes all edges of nonpositive weight).
  We then add two vertices $s,t$ and edges $(x,s)$ of weight $1$ and $(s,t)$ of weight $\infty$. Note that this defines a weight function $w'$ on $E(G')$ that is the same for all edges of $E(G') \cap E(G)$.
  Finally, set $k' = 1+k-w(N)$. Since $k \geq w(N)$, indeed $k' \in \mathbb{N}^{+}$.

  Suppose that $G$ has an edge dominating set $D$ with $w(D) \leq k$.
  Since the edges of $N$ have nonpositive weight, we can assume that $N \subseteq D$. 
  Then, we can also assume that $D$ contains no edges of $E(G) \setminus N$ for which both endpoints are in $X$, which means that no edges of $D \setminus N$ will be removed by removing loops that arise in the construction of $G'$.
  We now claim that $D' = \{(x,s)\} \cup (D\setminus N)$ is an edge dominating set of $G'$ with $w'(D') \leq k'$.
  Suppose that $e \in E(G')$ is not dominated by $D'$.
  Then $e$ is not incident on $x$ or $s$ by the construction of $D'$, so $e \in E(G)$ and $e$ is incident on two vertices in $V(G) \setminus X$.
  As $D$ is an edge dominating set of $G$, there is an edge $f \in D$ incident on one of the endpoints of $e$.
  Moreover, $f \in E(G')$, since the endpoint shared with $e$ is not in $X$.
  Hence, $f \in D'$ and $e$ is dominated, a contradiction.
  Therefore, $D'$ is an edge dominating set of $G'$. To see that it has weight at most $k'$, recall that $k \geq w(N)$, $N \subseteq D$, and that $D \setminus N$ does not contain any edges for which both endpoints are in $X$, and thus $w'(D') = 1 + w'(D \setminus N) = 1 + w(D\setminus N) \leq 1 + k-W(N) = k'$.

  Suppose that $G'$ has an edge dominating set $D'$ with $w'(D') \leq k'$.
  By the construction of $G'$ and~$w'$, the edge $(x,s)$ is in $D'$, but the edge $(s,t)$ is not.
  We now claim that $D = (D'\setminus\{(x,s)\}) \cup N$ is an edge dominating set of $G$ with $w(D) \leq k$. Indeed, $w(D) = w(D' \setminus\{(x,s)\}) + w(N) = w'(D' \setminus\{(x,s)\}) + w(N) \leq k'-1 + w(N) = k$.
  Suppose that $e \in E(G)$ is not dominated by $D$.
  Then $e$ is not incident on a vertex of $X$, so $e \in E(G')$ and $e$ is incident on two vertices of $V(G') \setminus \{x,s,t\}$.   
  As~$D'$ is an edge dominating set of $G'$, there is an edge $f \in D'$ incident on one of the endpoints of~$e$.
  Moreover, $f \in E(G)$, since the endpoint shared with $e$ is not in $\{x,s,t\}$.
  Hence, $f \in D$ and $e$ is dominated, a contradiction.
  Therefore, $D$ is an edge dominating set of $G$.
\end{proof}

\begin{corollary}
\label{cor:weds}
  Let $G$ be a multigraph, $w: E(G) \rightarrow \mathbb{Z}$ a weight function, and $k \in \mathbb{Z}$.
  Then we can decide in $O^{*}(2.2351^{\max\{1,1+k-w(N)\}})$ time whether $G$ has an edge dominating set $D$ with $w(D) \leq k$, where $N \subseteq V(G)$ denotes the set of edges in $G$ of nonpositive weight (under $w$).
\end{corollary}
We present an alternative algorithm, which gives an explicit bound on the polynomial factor and (in the exponent) will yield an almost equally fast algorithm as the one through Corollary~\ref{cor:weds}.

\begin{lemma}
\label{lem:ds:compute-eds}
  Let $G$ be a connected graph, let $(H,\eta)$ be a purified strip-structure of nullity zero for $G$ such that $1 \leq |\overline{F}| \leq 2$ for each $F \in E(H)$, let $k$ be an integer, and let $P \subseteq \mc{P}(H)$.
  Then we can decide in $O(2^{2k-|P|}\, |V(K_P)|^{3})$ time whether there is an edge dominating set $D$ of $K_{P}$ of smallest weight such that $w_{P}(D) \leq k - k_P$ and $|D| \leq k$.
\end{lemma}
\begin{proof}
  We adapt an algorithm of Fernau~\cite{Fernau2006}.
  He essentially showed that a minimum-weight edge dominating set of a graph can be found by enumerating all minimal vertex covers $C$ of the graph, computing a minimum-weight generalized edge cover with respect to each $C$, and then take the generalized edge cover of smallest weight over all $C$ as the solution.
  Here, a \emph{generalized edge cover} with respect to $C$ is a subset of the edges of the graph such that each vertex of $C$ is incident on an edge of this subset.

  To see that this approach is correct, observe the following.
  Let $K$ be a graph and let $w$ be an edge-weight function $E(K) \rightarrow \mathbb{N} \cup \{\infty\}$.
  If $D$ is an edge dominating set of $K$ of minimum weight (with respect to $w$), then let $C'$ denote the set of endpoints of $D$.
  Since $D$ is an edge dominating set of $K$, $C'$ is a vertex cover of $K$, and thus there is a minimal vertex cover $C \subseteq C'$. Clearly, $D$ is a generalized edge cover with respect to $C$.
  Conversely, if $C$ is a (minimal) vertex cover of $K$, then any generalized edge cover with respect to $C$ is an edge dominating set of $K$.

  In our context, we modify Fernau's approach to limit the number of minimal vertex covers that need to be enumerated.
  Again, let $K$ be a graph and let $w$ be an edge-weight function $E(K) \rightarrow \mathbb{N} \cup \{\infty\}$.
  Suppose that $K$ has an edge-dominating set of weight not $\infty$ and that no two pendant vertices of $K$ are adjacent.
  Then we claim that it suffices to enumerate all minimal vertex covers $C$ of $K$ such that $C$ does not contain a pendant vertex. Indeed, if $D$ is an edge dominating set of $K$ of minimum weight (with respect to $w$), then let $C'$ denote the set of endpoints of $D$.
  Since $D$ is an edge dominating set of $K$, $C'$ is a vertex cover of $K$.
  Now remove all pendant vertices from $C'$ and call the resulting set $C''$.
  As pendant vertices are only present in $C'$ if the edge incident on them is in $D$ and no two pendant vertices are adjacent, $C''$ is still a vertex cover of $K$.
  Then there is a minimal vertex cover $C \subseteq C''$ of $K$ that contains all neighbors of all pendant vertices in $C'$.
  Hence, $D$ is still a generalized edge cover with respect to $C$.
  Moreover, by construction, $C$ does not contain a pendant vertex nor a vertex that is only incident on vertices of weight $\infty$. This proves the claim.

  Following the claim and the construction of $K_{P}$, we can construct a set $L$ of vertices that have to be in any minimal vertex cover that we enumerate.
  To be precise, $L$ contains all vertices $v_{h}$ where $h \in P$ (these vertices $v_h$ have a pendant neighbor, but are not pendant themselves, because they are in $P$), and all vertices $v_{F}^{1}$ that are added in Case~2a (note that $v_F^2$ is pendant).
  By the claim, it suffices to enumerate all minimal vertex covers of $K_{P}$ that contain $L$.

  Suppose that $D$ has an edge dominating set $D$ such that $w_{P}(D) \leq k - k_{P}$ and $|D| \leq k$. Let $C$ be the minimal vertex cover that corresponds to $D$ following the modified approach above.
  Then $|P|$ of the endpoints of $D$ are of the type $v_h$ for $h \in P$, due to the presence of the edge $(v_h,v'_h)$. It follows that $|C-L| \leq 2k-|P|$, because no vertices $v'_h$ for $h \in P$ are in $C$.

  The final algorithm then enumerates all minimal vertex covers of $K_{P} \setminus L$ of size at most $2k-|P|$.
  This takes $O(|E(K_{P})| + k^{2} 2^{2k-|P|})$ time using an algorithm of Damaschke~\cite{Damaschke2006}, and there are most $2^{2k-|P|}$ such vertex covers.
  If no such vertex cover is found, then there is no edge dominating set~$D$ of $K_{P}$ such that $w_{P}(D) \leq k - k_P$ and $|D| \leq k$.
  Otherwise, for each such minimal vertex cover $C$ found, we compute a minimum-weight generalized edge cover of $K_{P}$ with respect to $C \cup L$, which takes $O(|V(K_P)|^3)$ time using an algorithm of Fernau~\cite{Fernau2006} (this algorithm straightforwardly extends to multigraphs).
  We output the smallest generalized edge cover found.
  The correctness of this algorithm was argued above, and the run time is as claimed.
\end{proof}
Using this lemma, we can finally prove the main result of this section.

\begin{theorem}
\label{thm:ds:fpt}
  Let $G$ be a claw-free graph with $n$ vertices and $m$ edges, and let $k$ be an integer.
  Then there is an algorithm that runs in $9^k \cdot O(n^{5})$ time and that correctly reports that either $\ds{G} > k$ or returns a smallest dominating set of $G$.
\end{theorem}
\begin{proof}
  We first preprocess the graph.
  We iteratively apply the following procedure.
  We check whether $G$ admits twins; this can be done in $O(n+m)$ time by Theorem~\ref{thm:twins-algo}.
  If so, then following Lemma~\ref{lem:ds:twins}, we can find a claw-free graph $G'$ with $|V(G')| < |V(G)|$ such that $\ds{G} = \ds{G'}$ in linear time, and continue to the next iteration with $G=G'$.
  Otherwise, $G$ does not admit twins, and we check whether $G$ admits a proper W-join; this can be done in $O(n^{2}m)$ time by Theorem~\ref{thm:proper-wjoin-algo}.
  If so, then following Corollary~\ref{cor:ds:Wjoins}, we can find a claw-free graph $G'$ with $|V(G')| < |V(G)|$ (note that proper W-joins of twin-less graphs have at least three vertices by definition, and two vertices remain after applying the corollary) such that $\ds{G} = \ds{G'}$ in linear time, and continue to the next iteration with $G=G'$.
  Since the number of vertices decreases in every iteration, there are at most $n$ iterations, and the run time of the procedure is $O(n^5)$.
  After the procedure finishes, we end up with a claw-free graph that does not admit twins nor proper W-joins for which the size of its smallest dominating set is unchanged.
  By abuse of notation, we call this graph $G$ as well.

  Suppose for now that $G$ is connected.
  We test whether $\ds{G} \leq 3$ using the algorithm of Lemma~\ref{lem:ds:constant} in $O(n^{4})$ time.
  If so, then the algorithm actually gives a dominating set of size $\ds{G}$, and we are done.
  Now we may assume that $\ds{G} > 3$. If $\alpha(G) \leq 3$, then $\ds{G} \leq 3$ by Proposition~\ref{prp:ds:is-size}, a contradiction.
  Hence, $\alpha(G) > 3$.

  We then apply the algorithm of Theorem~\ref{thm:main-base2} in $\stripTime$ time.
  If it outputs that $G$ is a proper circular-arc graph, then we can compute $\ds{G}$ in linear time by Theorem~\ref{thm:ds:circular}.
  If it outputs that $G$ is a thickening of an XX-trigraph, then we can compute $\ds{G}$ in $O(n^{4})$ time by Lemma~\ref{lem:ds:s2}.
  It remains that it outputs that $G$ admits a strip-structure such that for each strip $(J,Z)$:
  \begin{itemize}
    \item $(J,Z)$ is a trivial line graph strip, or
    \item $(J,Z)$ is a stripe for which
    \begin{itemize}
      \item $1\leq |Z| \leq 2$, $\alpha(J) \leq 3$, and $V(J) \setminus N[Z] \not= \emptyset$,
      \item $|Z|=1$, $J$ is a proper circular-arc graph, and either $J$ is a (strong) clique or $\alpha(J) > 3$,
      \item $|Z|=2$, $J$ is a proper interval graph, and either $J$ is a (strong) clique or $\alpha(J) > 3$, or
      \item $(J,Z)$ is a thickening of a member of $\mathcal{Z}_{5}$.
    \end{itemize}
  \end{itemize}
  We now determine the set $\mc{P}(H)$, which takes linear time.
  Then for each $P \subseteq \mc{P}(H)$, compute the triple $(K_{P},w_P,k_{P})$.
  Using Proposition~\ref{prp:ds:strip-boundary} and Lemma~\ref{lem:ds:circular-stripe}, \ref{lem:ds:z5}, \ref{lem:ds:stripe-is}, and~\ref{lem:ds:spot}, we can compute a set $D \subseteq V(J) \setminus Z$ of size $\dsg{J\setminus (Q \cup R)}{N[R]}$ that dominates $J \setminus (Q \cup N[R])$ for any disjoint $Q,R \subseteq Z$ for the strips $(J,Z)$ of the strip-structure. Following these lemmas, we can compute the triple $(K_{P},w_P,k_{P})$ in $O(n^{5})$ time.
  Note that the number of vertices of $K_P$ is in the order of $|E(H)|$, which is $O(n)$, because each strip contains at least one vertex of $G$ that is unique to it.
  Now, using Lemma~\ref{lem:ds:compute-eds}, decide in $O(2^{2k-|P|} n^{3})$ time whether there is an edge dominating set $D$ of $K_{P}$ of smallest weight such that $w_{P}(D) \leq k - k_P$ and $|D| \leq k$.
  If so, then by Lemma~\ref{lem:ds:reduction}, $\ds{G} \leq k$.
  Moreover, since~$D$ is an edge dominating set of $K_P$ of smallest weight, the transformation of Lemma~\ref{lem:ds:reduction} gives a smallest dominating set of $G$.
  If Lemma~\ref{lem:ds:compute-eds} returns for all $P \subseteq \mc{P}(H)$ that there is no edge dominating set $D$ of $K_{P}$ such that $w_{P}(D) \leq k - k_P$ and $|D| \leq k$, then we can return a ``no''-instance by Lemma~\ref{lem:ds:reduction}.
  Since $|\mc{P}(H)| \leq 2k$ by Lemma~\ref{lem:ds:striped}, the run time is $$O\left(\sum_{i=0}^{|\mc{P}(H)|} {|\mc{P}(H)| \choose i} (n^5 + 2^{2k-i}n^3)\right) = O\left(n^{5}\  \sum_{i=0}^{2k} {2k \choose i} 2^{2k-i}\right) = 9^k \cdot O(n^5).$$

  If $G$ is not connected, then note that the above algorithm computes a smallest dominating set of an $n$-vertex connected graph in time $9^k \cdot O(n^5)$ if it has size at most $k$, or correctly decides that such a dominating set has size more than $k$. Therefore, we apply this algorithm to each connected component of $G$ in turn.
  If the algorithm answers that every dominating set of the component has size more than $k$, then this answer is also true for $G$.
  Otherwise, we return the union of the answers found for each component, which is indeed a smallest dominating set of $G$.
  The total run time is $9^k \cdot O(n^5)$.
\end{proof}

\begin{corollary}
\label{cor:ds:fpt}
  Let $G$ be a claw-free graph with $n$ vertices, and let $k$ be an integer.
  Then we can decide in $9^k \cdot O(n^{5})$ time whether $\ds{G} \leq k$.
\end{corollary}
Instead of relying on Lemma~\ref{lem:ds:compute-eds} in the above theorem, we could also apply the algorithm of Corollary~\ref{cor:weds}.
This yields a slight improvement in the exponential factor.

\begin{theorem}
  Let $G$ be a claw-free graph with $n$ vertices, and let $k$ be an integer.
  Then we can decide in $O^*(8.9404^k)$ time whether $\ds{G} \leq k$.
\end{theorem}
Note that the polynomial factor hidden in the $O^{*}$ is at least $n^5$.

\section{Fixed-Parameter Algorithm for Connected Dominating Set}
\label{sec:cds}
In this section, we prove that {\sc Connected Dominating Set} parameterized by solution size is fixed-parameter tractable on claw-free graphs.
The general idea of how to establish this is similar to the approach for {\sc Dominating Set} on claw-free graphs.

Throughout the section, we rely on the following notation.
Let $G$ be a graph.
We let $\cds{G}$ denote the size of a smallest connected dominating set of $G$.
Furthermore, for each subset $A\subseteq V(G)$, let $\cdsa{G}{A}$ denote the size of a smallest subset $D$ of $V(G)$ dominating $V(G) \setminus A$ such that each connected component of $G[D]$ contains at least one vertex of $A$.
We also need the following well-known relation between $\ds{G}$ and $\cds{G}$.

\begin{proposition}
\label{prp:cds:boundbyds}
  For any graph $G$, $\ds{G} \leq \cds{G} \leq 3 \cdot \ds{G} - 2$.
\end{proposition}

\subsection{Removing Twins and Proper W-joins}
We first show how to remove twins and proper W-joins from a graph $G$ without changing the size of its smallest connected dominating set.
The reductions are powerful enough to operate on general graphs, while maintaining claw-freeness and connectivity.

\begin{lemma}
\label{lem:cds:twins}
  Let $a,b$ be twins of a connected graph $G$, and let $G' = G \setminus a$. Then $G'$ is connected and $\cds{G} = \cds{G \setminus a}$.
  Moreover, if $G$ is claw-free, then so is $G'$.
\end{lemma}
\begin{proof}
  Let $D$ be a smallest connected dominating set of $G$.
  Since $N[a] = N[b]$ (in particular, $a$ and $b$ are adjacent) and $D$ is a smallest connected dominating set of $G$, at most one of $a,b$ belongs to $D$.
  If $a \in D$, then replace $a$ by $b$.
  Then the resulting set is still a connected dominating set of $G$ of the same size as $D$, and thus also a connected dominating set of $G'=G \setminus a$.

  Let $D'$ be a smallest connected dominating set of $G'$. Then $D' \cap N[b] \not= \emptyset$. Since $N[a] = N[b]$, $D'$ is a connected dominating set of $G$ as well.
\end{proof}
While removing twins has a similar effect for smallest connected dominating sets as for smallest dominating sets, the situation for proper W-joins is somewhat different in the two problems.
The reason is that removing all edges crossing a W-join might disconnect the graph.

\begin{lemma}
\label{lem:cds:properWjoins}
  Let $(A,B)$ be a proper W-join of a connected graph $G$, and let $G'$ be the graph obtained from $G$ as follows:
  \begin{enumerate}
    \item remove $A$ and $B$,
    \item add vertices $a,a',b,b'$, and
    \item connect $a,a'$ to $N(A) \setminus B$, $b,b'$ to $N(B)\setminus A$, $a$ to $a'$, $b$ to $b'$, and $a$ to $b$.
  \end{enumerate}
  Then $G'$ is connected and $\cds{G} = \cds{G'}$.
  Furthermore, if $G$ is claw-free, then so is $G'$.
\end{lemma}
\begin{proof}
  Let $D$ be a smallest connected dominating set of $G$.
  We first show that we can assume that $|D \cap A| \leq 1$.
  For suppose that $|D \cap A| \geq 2$.
  Note that $N[x]\setminus B$ is the same for each $x \in A$ and $N[y] \setminus A$ is the same for each $y \in B$.
  However, we cannot remove all but one vertex of $D \cap A$ from $D$, because this would lead to a smaller connected dominating set and thus a contradiction.
  Therefore, the vertices of $D \cap A$ are the only vertices dominating $B$.
  But then we can replace a vertex of $D \cap A$ by a vertex of $B$ adjacent to a second vertex of $D \cap A$.
  Such vertices exist because the W-join $(A,B)$ is proper.
  It follows that we can assume $|D \cap A|\leq 1$, and similarly, $|D \cap B| \leq 1$.
  
  Now, since the W-join $(A,B)$ is proper, no vertex of $A$ (resp.~$B$) dominates all of $B$ (resp.~$A$).
  Hence, $D \cap (N[A]\setminus B) \not= \emptyset$ and $D \cap (N[B]\setminus A) \not= \emptyset$.
  Let $D'$ be the set containing all vertices of $D$ that belong to $G'$.
  Additionally, if $D \cap A \not= \emptyset$, then add $a$ to $D'$, and if $D \cap B \not= \emptyset$, then add $b$ to $D'$.
  By construction, $D'$ is a connected dominating set of $G'$ and $|D'| \leq |D|$.

  Let $D'$ be a smallest connected dominating set of $G'$.
  From the definition of a W-join, recall that every vertex of $V(G)\setminus(A \cup B)$ is either $A$-complete or $A$-anticomplete, and is either $B$-complete or $B$-anticomplete.
  Furthermore, $A$ and $B$ are cliques.
  Let $D$ be the set containing all vertices of $D'$ that belong to $G$, and additionally,
  \begin{itemize}  
    \item if $D'\cap\{a,a'\}\not=\emptyset$ and $D'\cap\{b,b'\}\not=\emptyset$, then also add any two adjacent vertices, one of $A$ and one of $B$, to $D$.
    \item otherwise, if one of $a,a' \in D'$, then add any vertex of $A$ to $D$, and if one of $b,b \in D'$, then add any vertex of $B$ to $D$.
  \end{itemize}
  Then $D$ is a connected dominating set of $G$ and $|D| \leq |D'|$.

  Finally, consider the case that $G$ is claw-free.
  For the sake of contradiction, suppose that $G'$ has a claw $\{w;x,y,z\}$ with center $w$.
  If at most one vertex of $w,x,y,z$ belongs to $\{a,a',b,b'\}$, then $\{w;x,y,z\}$ directly corresponds to a claw in $G$, a contradiction.
  Because both $\{a,a'\}$ and $\{b,b'\}$ form a clique, we have $N[a]\setminus\{b\} = N[a']$ and $N[b]\setminus\{a\} = N[b']$, and thus at most two of $w,x,y,z$ belong to $\{a,a',b,b'\}$.
  If $w \in V(G) \cap V(G')$, then assume without loss of generality that $x = a'$ and $y\in\{b,b'\}$.
  But then we can replace $x$ and $y$ by any antiadjacent pair of vertices of $A$ and~$B$ to obtain a claw in $G$, a contradiction.
  Hence, $w \in \{a,a',b,b'\}$.
  Then assume, without loss of generality, that $w = a$ and $x = b$.
  But then we can replace $w$ and $x$ by any adjacent pair of vertices of $A$ and $B$ to obtain a claw in $G$, a contradiction. Therefore, $G'$ is indeed claw-free.
\end{proof}

\subsection{Connected Dominating Set in Basic Classes}
Let $G$ be a connected claw-free graph.
Through the reductions of Lemma~\ref{lem:cds:twins} and Lemma~\ref{lem:cds:properWjoins}, we may essentially assume that $G$ admits no twins and proper W-joins. Consider the following lemma.

\begin{lemma}
\label{lem:cds:constant}
  Let $G$ be a graph and $k$ an integer.
  Then in $O(n^{k+1})$ time we can compute $\cds{G}$ (resp.\ $\cdsa{G}{A}$ for given $A \subseteq V(G)$) or correctly decide that $\cds{G} > k$ (resp.\ $\cdsa{G}{A} > k$).
\end{lemma}
The proof of the above lemma is similar to the proof of Lemma~\ref{lem:ds:constant}.

\begin{corollary}
\label{cor:cds:constant}
  Let $G$ be a graph such that $\alpha(G) \leq 3$.
  Then we can compute $\cds{G}$ in $O(n^{8})$ time.
\end{corollary}
\begin{proof}
  By Proposition~\ref{prp:ds:is-size} and~\ref{prp:cds:boundbyds}, $\cds{G} \leq 3\ds{G} -2 \leq 3\alpha(G)-2 \leq 7$, and the result follows from Lemma~\ref{lem:cds:constant}.
\end{proof}
Intuitively, we may thus assume that the connected claw-free graph $G$ admits no twins, admits no proper W-joins, and satisfies $\alpha(G) > 3$. Therefore, we can use the implications of Theorem~\ref{thm:main-base2} for $G$.
(A formal proof of these facts follows later.)

First, we show that if $G$ is a proper circular-arc graph or a thickening of an XX-trigraph, then we can compute $\cds{G}$ in polynomial time.

\begin{theorem}[Chang \cite{Chang1998}]
\label{thm:cds:circular}
  Let $G$ be a circular-arc graph.
  Then $\cds{G}$ can be computed in linear time.
\end{theorem}

\begin{lemma}
\label{lem:cds:s2}
  Let $G$ be a graph that is a thickening of an XX-trigraph.
  Then $\cds{G}$ can be computed in $O(n^{5})$ time.
\end{lemma}
In order to show this lemma, we need the following auxiliary result, which is similar to Lemma~\ref{lem:ds:tri}.

\begin{lemma}
\label{lem:cds:tri}
  Let $G$ be a graph that is thickening of a trigraph $G'$, and let $G''$ be the graph obtained from $G'$ by removing all semi-edges from $G'$.
  Then $\cds{G} \leq \cds{G''}$.
\end{lemma}
\begin{proof}
  Let $\mathcal{W}$ be a thickening of $G'$ to $G$, and let $D'$ be any connected dominating set of $G''$.
  Construct a set $D \subseteq V(G)$ as follows: for each $v' \in D'$, pick an arbitrary vertex $v \in W_{v'}$.
  We claim that $D$ is a dominating set of $G$.
  Consider any $w \in V(G) \setminus D$ and let $w' \in V(G')$ be such that $w \in W_{w'}$.
  If $w' \in D'$, then because $W_{w'}$ is a (strong) clique and $W_{w'} \cap D \not= \emptyset$ by construction, $w$ is dominated by $D$.
  Otherwise, there is a $u' \in D' \subseteq V(G'')$ such that $w'$ and $u'$ are adjacent.
  By the construction of $G''$, this implies that $w'$ and $u'$ are strongly adjacent in $G'$. By the definition of a thickening, each vertex of $W_{u'}$ is (strongly) complete to $W_{w'}$.
  By construction, $D \cap W_{u'} \not= \emptyset$, and thus $w$ is dominated by $D$.
  The claim follows.
  It remains to show that $G[D]$ is connected.
  Observe that if $u',v' \in D'$ are adjacent in $G''$, then $u'$ and $v'$ are strongly adjacent in $G'$, and thus each vertex of $W_{u'}$ is (strongly) complete to $W_{v'}$.
  Therefore, the vertices $u$ and $v$ chosen in $D \cap W_{u'}$ and $D \cap W_{v'}$ respectively are (strongly) adjacent.
  Hence, $G[D]$ is indeed connected.
  Since $|D| = |D'|$, $\cds{G} \leq \cds{G''}$.
\end{proof}
It is now straightforward to prove Lemma~\ref{lem:cds:s2}.
\begin{proof}
  Consider an XX-trigraph $G'$ such that $G$ is a thickening of $G'$.
  Remove all semi-edges from~$G'$ and call the resulting graph $G''$.
  By the definition of XX-trigraphs, it follows that $\{v_{1},v_{2},v_{3},v_{4}\}$ forms a connected dominating set of $G''$. Hence, by Lemma~\ref{lem:cds:tri}, $\cds{G} \leq \cds{G''} \leq 4$.
  The result then follows from Lemma~\ref{lem:cds:constant}.
\end{proof}

Intuitively, Theorem~\ref{thm:cds:circular} and Lemma~\ref{lem:cds:s2} imply that Theorem~\ref{thm:main-base2} yields a strip-structure.
Therefore, we turn to the basic classes of strips of Theorem~\ref{thm:main-base2}.
For reasons that will become clear later, we need stronger results for strips $(J,Z)$ than just being able to compute a smallest connected dominating set.
Intuitively, if we compute $\cds{J}$, then we enforce that any connected dominating set that attains this bound contains a vertex in $N(z)$ for each $z \in Z$.
Sometimes we might want to enforce this, but sometimes we do not.
Hence, we should be able to compute $\cds{J \setminus Q}$ for each $Q \subseteq Z$.
Also, we sometimes do not want to enforce that the dominating set is connected, but instead we want every component of the dominating set to contain a vertex of $N(Z)$, and therefore, we need to compute $\cdsa{J}{N[Z]}$.
We now do this for each strip type of Theorem~\ref{thm:main-base2}.

\begin{lemma}
\label{lem:cds:circular-stripe}
  Let $(J,Z)$ be a stripe such that either $J$ is a proper circular-arc graph and $|Z| = 1$, or $J$ is a proper interval graph and $|Z| = 2$. 
  Then we can compute $\cds{J \setminus Q}$ for each $Q \subseteq Z$ and $\cdsa{J}{N[Z]}$ in $O(nm)$ time, where $n$ and $m$ are the number of vertices and edges of $J$ respectively.
\end{lemma}
\begin{proof}
  By Theorem~\ref{thm:cds:circular}, we can compute $\cds{J \setminus Q}$ for each $Q \subseteq Z$ in linear time.

  Suppose that $J$ is a proper circular-arc graph and $|Z| = 1$.
  Then $\cdsa{J}{N[Z]} = \cds{J}$ if $V(J) \not= N[Z]$, and $\cdsa{J}{N[Z]} = 0$ otherwise.
  Hence, we can compute $\cdsa{J}{N[Z]}$ in linear time by Theorem~\ref{thm:cds:circular}.

  Suppose that $J$ is a proper interval graph and $|Z|=2$.
  Let $Z = \{z,z'\}$. First, find a set of intervals $I_{1},\ldots,I_{n}$ of the line that represent $J$ as a proper interval graph; such a representation can be found in linear time~\cite{DengHH1996}.
  Suppose that $I_{z}$ and $I_{z'}$ are the intervals representing the vertices $z$ and $z'$ of $Z$.
  Since $(J,Z)$ is a stripe, $N[z] \cap N[z'] = \emptyset$.
  Hence, without loss of generality, the right endpoint $p_{r}$ of~$I_{z}$ is to the left of the left endpoint $p_{l}$ of $I_{z'}$.

  Consider a smallest dominating set $D$ of $V(J)$ dominating $V(J) \setminus N[Z]$ such that each connected component of $J[D]$ contains at least one vertex of $N[Z]$, \ie $|D| = \cdsa{J}{N[Z]}$.
  Assume that $V(J) \setminus N[Z] \not= \emptyset$.
  Moreover, since $N[z]$ and $N[z']$ are cliques, $D$ has at most two connected components: at most one (denoted $C$) that contains a vertex of $N[z]$ and at most one (denoted $C'$) that contains a vertex of $N[z']$.
  If $D \cap N[z] = \emptyset$, then $\cdsa{J}{N[Z]} = \cds{J \setminus N[z]}$.
  Similarly, if $D \cap N[z'] = \emptyset$, then $\cdsa{J}{N[Z]} = \cds{J \setminus N[z']}$.
  It remains the case that $D$ contains a vertex of both $N[z]$ and $N[z']$. 
  
  Then either $\cdsa{J}{N[Z]} = \cds{J}$, or there is a point between $p_{r}$ and $p_{l}$ that is not covered by any interval of $D$.
  In the latter case, $C$ and $C'$ exist, $C \not= C'$, and in particular, there is a point $p$ between $p_{r}$ and $p_{l}$ such that each interval starting to the right of $p$ is dominated by $C'$, and each interval starting to the left of $p$ is dominated by $C$.

  This analysis implies the following algorithm.
  For each point $p$ between $p_{l}$ and $p_{r}$, split $J$ into two induced subgraphs: one induced by the intervals with their left endpoint to the left of $p$ and one induced by the intervals with their left endpoint to the right of $p$.
  Using Theorem~\ref{thm:cds:circular}, we can compute a smallest connected dominating set of both induced subgraphs in linear time, and sum their sizes.
  Let $M$ denote the smallest value over all points $p$.
  Since there are only $O(n)$ points $p$ for which the split of $J$ we made yields different induced subgraphs, computing $M$ takes $O(nm)$ time.
  Then $\cdsa{J}{N[Z]}$ is simply the minimum of $M$, $\cds{J \setminus N[z]}$, and $\cds{J \setminus N[z']}$ if $V(J) \setminus N[Z] \not= \emptyset$ and $0$ otherwise.
  Using Theorem~\ref{thm:cds:circular}, this implies that $\cdsa{J}{N[Z]}$ can be computed in $O(nm)$ time.
\end{proof}

\begin{lemma}
\label{lem:cds:z5}
  Let $(J,Z)$ be a thickening of a stripe $(J',Z') \in \mc{Z}_{5}$ such that $J$ is a graph.
  Then we can compute $\cds{J \setminus Q}$ for each $Q \subseteq Z$ and $\cdsa{J}{N[Z]}$ in $O(n^{5})$ time.
\end{lemma}
\begin{proof}
  Let $J''$ be the graph obtained from $J'$ by removing all semi-edges. By the definition of XX-trigraphs $\{v_1,v_2,v_3,v_4\}$ forms a connected dominating set of $J''$.
  Hence, by Lemma~\ref{lem:cds:tri}, $\cds{J} \leq \cds{J''} \leq 4$. 
  Observe that $\cdsa{J}{N[Z]} \leq \cds{J}$ by definition.
  The result then follows by Lemma~\ref{lem:cds:constant}.
\end{proof}

\begin{lemma}
\label{lem:cds:stripe-is}
  Let $(J,Z)$ be a stripe such that $J$ is a graph with $1 \leq |Z| \leq 2$, $\alpha(J) \leq 3$, $V(J) \setminus N[Z] \not= \emptyset$.
  Then we can compute $\cds{J \setminus Q}$ for each $Q \subseteq Z$ and $\cdsa{J}{N[Z]}$ in $O(n^{8})$ time.
\end{lemma}
\begin{proof}
  By Proposition~\ref{prp:ds:is-size} and~\ref{prp:cds:boundbyds}, $\cds{J} \leq 3\ds{J} -2 \leq 3\alpha(J)-2 \leq 7$. Observe that $\cdsa{J}{N[Z]} \leq \cds{J}$ by definition.
  The result then follows by Lemma~\ref{lem:cds:constant}.
\end{proof}
The following lemma is immediate from the definition of trivial line graph strips (see Definition~\ref{def:linestrip}).

\begin{lemma}
\label{lem:cds:spot}
  Let $(J,Z)$ be a trivial line graph strip such that $J$ is a graph.
  Then we can compute $\cds{J \setminus Q}$ for each $Q \subseteq Z$ and $\cdsa{J}{N[Z]}$ in constant time.
\end{lemma}
We also rely on the following observation.

\begin{proposition}
\label{prp:cds:strip-boundary}
  Let $(J,Z)$ be a strip such that $J$ is a graph.
  There is a connected set $D \subseteq V(J) \setminus Z$ of size $\cds{J \setminus Q}$ for each $Q \subseteq Z$ or $\cdsa{J}{N[Z]}$ that respectively dominates $V(J) \setminus Q$ or dominates $V(J) \setminus N[Z]$ while each component contains a vertex of $N[Z]$.
\end{proposition}
\begin{proof}
  Consider any $z \in Z$.
  By the definition of a strip, $z$ is a simplicial vertex.
  Therefore, $z$ can be replaced by any vertex of $N(z)$ without destroying the domination property or connectivity of a connected set needed to attain $\cds{J \setminus Q}$ for each $Q \subseteq Z$ or $\cdsa{J}{N[Z]}$.
\end{proof}

\subsection{Stitching Connected Dominating Sets}
We now describe a method to stitch connected dominating sets for individual strips of a strip-structure together to form a connected dominating set of the entire graph. The main idea is essentially the same as the one we employed for dominating set in the previous section, except that the connectivity requirement makes the stitching significantly simpler.

Let $G$ be a connected graph that is not a complete graph, let $(H,\eta)$ be a purified strip-structure of nullity zero for $G$ such that $1 \leq |\overline{F}| \leq 2$ for each $F \in E(H)$.

\begin{definition} 
  We say that $P \subseteq \mc{P}(H)$ (recall Definition~\ref{def:ds:striped}) is \emph{covering}  if $\overline{F} \cap P \not= \emptyset$ for each $F \in E(H)$ that corresponds to a stripe.
\end{definition}

Now let $P \subseteq \mc{P}(H)$ be covering.
We again define the triple $(K_{P}, w_{P}, k_{P})$ of a multigraph, an edge-weight function, and an integer as follows.
Initially, $K_{P}$ is empty and $k_{P} = 0$.
For each $h \in V(H)$, add a vertex $v_h$ to $K_{P}$.
If $h \in P$, then also add a vertex $v'_h$ to $K_{P}$ as well as an edge between $v_h$ and~$v'_h$ of weight $\infty$.
For each $F \in E(H)$, there are several cases:

\ccase{1} $\overline{F} = \{h\}$ for some $h \in V(H)$.\\
Let $(J,Z)$ be the strip corresponding to $F$. Note that $(J,Z)$ is in fact a stripe. Since $P$ is covering, $h \in P$.
Add vertices $v_F, v_F'$, an edge $e_F$ between $v_h$ and $v_F$ of weight $\cdsa{J}{N[Z]}$, and an edge $e_F'$ between $v_F$ and $v_F'$ of weight $\infty$ to $K_P$.

\ccase{2} $\overline{F} = \{h,h'\}$ for some $h,h' \in V(H)$. The strip $(J,Z)$ corresponding to $F$ is a stripe.\\
Let $Z = \{z,z'\}$, where $z$ corresponds to $h$ and $z'$ to $h'$. Since $P$ is covering, $h \in P$ or $h' \in P$. There are several cases:

\csubcase{2a} $h,h' \in P$.\\
Add an edge $e_F$ between $v_h$ and $v_{h'}$ of weight $\cds{J} - \cdsa{J}{N[Z]}$; additionally, increase $k_P$ by $\cdsa{J}{N[Z]}$.

\csubcase{2b} $h \in P$, $h' \not\in P$.\\
Add vertices $v_F, v_F'$, an edge $e_F$ between $v_h$ and $v_F$ of weight $\cds{J \setminus z'}$, and an edge $e_F'$ between $v_F$ and $v_F'$ of weight $\infty$ to $K_P$.

\csubcase{2c} $h \not\in P$, $h' \in P$.\\
Add vertices $v_F, v_F'$, an edge $e_F$ between $v_h$ and $v_F$ of weight $\cds{J \setminus z}$, and an edge $e_F'$ between $v_F$ and $v_F'$ of weight $\infty$ to $K_P$.

\ccase{3} $\overline{F} = \{h,h'\}$ for some $h, h' \in V(H)$. The strip $(J,Z)$ corresponding to $F$ is a spot.\\
Add an edge $e_{F}$ between $v_{h}$ and $v_{h'}$ of weight $1$.

\medskip\noindent
Since $1 \leq |\overline{F}| \leq 2$ for each $F \in E(H)$ and $(H,\eta)$ is purified, the cases are exhaustive. Moreover, $k_P \geq 0$ by construction and each edge of $K_{P}$ has received nonnegative weight.

Now recall that a \emph{connected edge dominating set} of a graph $G$ is an edge dominating set $D \subseteq E(G)$ such that $G[D]$ is connected. Here, we use $G[D]$ as a shorthand for the subgraph $G'$ of $G$ with $V(G')$ equal to the set of endpoints of the edges in $D$ and with $E(G') = D$.

\begin{lemma}
\label{lem:cds:reduction}
  Let $G$ be a connected graph that is not a complete graph, let $(H,\eta)$ be a purified strip-structure of nullity zero for $G$ such that $1 \leq |\overline{F}| \leq 2$ for each $F \in E(H)$, and let $k$ be an integer.
  Then $\cds{G} \leq k$ if and only if there is a $P \subseteq \mc{P}(H)$ such that $P$ is covering and there is a connected edge dominating set $D$ of $K_{P}$ such that $w_{P}(D) \leq k - k_P$.
  Moreover, such a set $D$ satisfies $|D| \leq k$.
\end{lemma}
\begin{proof}
  Suppose that there is a $P \subseteq \mc{P}(H)$ such that $P$ is covering and there is a connected edge dominating set $D$ of $K_{P}$ such that $w_{P}(D) \leq k - k_P$.
  We construct a set $D' \subseteq V(G)$ of size at most $k$ as follows (and later show that it is in fact a connected dominating set of $G$). Initially, $D' = \emptyset$. For each $F \in E(H)$, we consider the same cases as before:

  \ccase{1} $\overline{F} = \{h\}$ for some $h \in V(H)$.\\
  Let $(J,Z)$ be the strip corresponding to $F$.
  By construction, $e_F \in D$.
  Add a smallest set $B \subseteq V(J) \setminus Z$ to $D'$ such that $B$ dominates all vertices of $V(J) \setminus N[Z]$ and each component of $J[B]$ contains a vertex of $N[Z]$---note that the weight of $e_F$ is equal to the size of the set added to $D'$ (using Proposition~\ref{prp:cds:strip-boundary}).

  \ccase{2} $\overline{F} = \{h,h'\}$ for some $h,h' \in V(H)$. The strip $(J,Z)$ corresponding to $F$ is a stripe.\\
  Let $Z = \{z,z'\}$, where $z$ corresponds to $h$ and $z'$ to $h'$. Since $P$ is covering, $h \in P$ or $h' \in P$. There are several cases:

  \csubcase{2a} $h,h' \in P$.\\
  If $e_F \in D$, then add a smallest connected subset of $V(J) \setminus Z$ to $D'$ that dominates all vertices of $V(J)$---note that the weight of $e_F$ is $\cds{J} - \cdsa{J}{N[Z]}$ and $k_P$ was increased by $\cdsa{J}{N[Z]}$ in the construction; the sum is equal to the size of the set added to $D'$ (using Proposition~\ref{prp:cds:strip-boundary}).
  If $e_F \not\in D$, then add a smallest set $B \subseteq V(J) \setminus Z$ to $D'$ such that $B$ dominates all vertices of $V(J) \setminus N[Z]$ and each component of $J[B]$ contains a vertex of $N[Z]$---note that $k_P$ was increased by $\cdsa{J}{N[Z]}$ in the construction, which is equal to the size of the set added to $D'$ (using Proposition~\ref{prp:cds:strip-boundary}).

  \csubcase{2b} $h \in P$, $h' \not\in P$.\\
  By construction, $e_F \in D$.
  Add a smallest connected subset of $V(J) \setminus Z$ to $D'$ that dominates all vertices of $V(J) \setminus z'$---note that the weight of $e_F$ is equal to the size of the set added to $D'$ (using Proposition~\ref{prp:cds:strip-boundary}).

  \csubcase{2c} $h \not\in P$, $h' \in P$.\\
  By construction, $e_F \in D$.
  Add a smallest connected subset of $V(J) \setminus Z$ to $D'$ that dominates all vertices of $V(J) \setminus z$ --- note that the weight of $e_F$ is equal to the size of the set added to $D'$ (using Proposition~\ref{prp:cds:strip-boundary}).

  \ccase{3} $\overline{F} = \{h,h'\}$ for some $h, h' \in V(H)$.
  The strip $(J,Z)$ corresponding to $F$ is a spot.\\
  If $e_{F} \in D$, then add the vertex of $J \setminus Z$ to $D'$---note that the weight of $e_F$ is $1$, which is equal to the size of the set added to $D'$.
  If $e_F \not\in D$, then do nothing.

  \medskip\noindent
  From the construction of $D'$ and the analysis in each of the above cases, it follows that $|D'| = w_P(D) + k_P \leq k$.

  We now show that $D'$ is a dominating set of $G$.
  Let $v$ be an arbitrary vertex of $G$.
  Since $(H,\eta)$ is a strip-structure for $G$, by definition there is an $F \in E(H)$ such that $v \in \eta(F)$. Let $(J,Z)$ be the strip corresponding to $F$. Then $v \in N(Z)$ or $v \in V(J) \setminus N[Z]$. 
  If $v \in V(J) \setminus N(Z)$, then $(J,Z)$ is not a spot, and regardless of in which of the above subcases of Case~1 and~2 that $F$ falls, $v$ will be dominated by the dominating set that is added to $D'$ for $F$.
  So assume that $v \in N(z)$ for some $z \in Z$ and let $h \in V(H)$ correspond to $z$.
  If $(J,Z)$ is a spot, then we have a choice in $z$ (and thus also in $h$): we prefer that $h \in P$ if possible.

  If $h \not\in P$, then $F$ cannot fall into Case~1, since $P$ is covering.
  If $F$ falls into Case~2, then regardless of in which of the above subcases of Case~2 $F$ falls, $N(z)$ and thus $v$ will be dominated by the dominating set that is added to $D'$ for $F$.
  If $F$ falls into Case~3, then let $h'$ denote the single element of $\overline{F} \setminus\{h\}$.
  Since we preferred $h$ to $h'$, $h' \not\in P$.
  By the construction of $K_P$, there is an edge between~$v_h$ and $v_{h'}$ corresponding to $F$.
  Moreover, since $h,h' \not \in P$ and by the construction of $K_P$, every edge incident with $v_h$ and $v_{h'}$ corresponds to an edge introduced for a spot in Case~3.
  Therefore, since $D$ is a (connected) edge dominating set of $K_P$, there is an edge in $D$ that is incident with $v_h$ or $v_h'$ which corresponds to an $F' \in E(H)$ (possibly $F = F'$) with $h \in \overline{F'}$ or $h' \in \overline{F'}$ such that $F'$ corresponds to a spot.
  By the construction of $D'$, the single vertex of $\eta(F')$ is in $D'$, that is, $\eta(F') \subseteq D'$.
  Assume without loss of generality that $h \in \overline{F'}$.
  Then $v$ is dominated, because $\eta(F',h) = \eta(F') \subseteq D'$ and thus $\eta(h) \cap D' \not= \emptyset$, $v \in \eta(F,h) \subseteq \eta(h)$, and $\eta(h)$ is a clique.

  If $h \in P$, then due to the presence of the edge $e$ between $v_{h}$ and $v'_{h}$, there is an edge $e'$ incident with $v_h$ that is in $D$. Since $e$ has weight $\infty$, $e' \not= e$.
  Hence, there is an $F' \in E(H)$ (possibly $F' = F$) with $h \in \overline{F'}$ such that an edge associated with $F'$ in the construction is in $D$.
  In particular, we can choose $F'$ such that $|\overline{F'}| = 1$ and $\eta(F') \setminus \eta(F',h) \not= \emptyset$ or $|\overline{F'}| = 2$, since $G$ is not a complete graph.
  Then, by inspecting the above cases (Case~1, 2a, 2b, 2c, 3) for the strip $(J',Z')$ corresponding to $F'$, there is a vertex of $D'$ in $N(z')$, where $z'$ is the vertex of $Z'$ that corresponds to $h$.
  In other words, $\eta(F',h) \cap D' \not= \emptyset$, and thus $\eta(h) \cap D' \not= \emptyset$.
  Since $\eta(h)$ is a clique and $\eta(F',h) \subseteq \eta(h)$, each vertex in $\eta(h)$ is dominated, and in particular each vertex in $\eta(F,h)$ is dominated. 
  As $v \in N(z) = \eta(F,h)$, $v$ is dominated.
  Hence, $D'$ is a dominating set of $G$ of size at most $k$.

  It remains to prove that $G[D']$ is connected.
  Observe that in the construction of $K_P$, we add exactly one edge $e_F$ to $K_P$ for each $F \in E(H)$ (and possibly an edge $e_F'$ of weight $\infty$).
  Moreover, if $e_F \in D$ for some $F \in E(H)$, then $G[D' \cap \eta(F)]$ is connected and $D'$ contains a vertex of $\eta(F,h)$ for each $h \in \overline{F}$ by the construction of $D'$.

  Now let $v,v' \in D'$, where $v \in \eta(F)$ and $v' \in \eta(F')$ for $F,F' \in E(H)$.
  If $e_F, e_{F'} \in D$, then $v$ and $v'$ are connected in $G[D']$ by the above observation and the fact that $D$ is a connected edge dominating set.
  Suppose that $e_F \in D$ and $e_{F'} \not\in D$.
  This only happens if $F'$ falls into Case~2a.
  Moreover, there is a connected component of the set of vertices added to $D'$ for $F'$ that contains both $v'$ and a vertex of $\eta(F',h)$ for some $h \in \overline{F'}$ by construction.
  Since $D$ is an edge dominating set and an edge $(v_h,v'_h)$ exists in $K_P$, there is an edge $e_{F''} \in D$ for some $F'' \in E(H)$.
  Hence, by the above observation, there is a vertex $w \in \eta(F'',h) \cap D'$.
  Thus, $w$ and $v$ are in the same connected component of $G[D']$.
  But then $v$ and $v'$ are in the same connected component of $G[D']$, as $\eta(h)$ is a clique.
  It follows that $G[D']$ is connected.
  Similar arguments can be made if $e_{F} \in D$ and $e_{F'} \not\in D$ or if $e_{F},e_{F'} \not\in D$.
  Hence, $D'$ is a connected dominating set of $G$ of size at most $k$.

  \bigskip
  Suppose that $\cds{G} \leq k$.
  Let $D'$ be a smallest connected dominating set of $G$; hence, $|D'| \leq k$.
  Let~$P$ be such that $h \in P$ if and only if $h \in \mc{P}(H)$ and there is an $F \in E(H)$ with $h \in \overline{F}$ and $D' \cap \eta(F,h) \not= \emptyset$. 
  Then $P$ is covering. 
  Indeed, suppose that $\overline{F} \cap P \not= \emptyset$ for some $F \in E(H)$ that corresponds to a stripe, then $\eta(h) \cap D' = \emptyset$ for each $h \in \overline{F}$, and thus there is a vertex of $\eta(F)$ that is not dominated by $D'$ or~$D'$ has at least two connected components.

  We construct a set $D \subseteq E(K_{P})$ of weight at most $k-k_{P}$ that is a connected edge dominating set of $K_{P}$. Initially, $D = \emptyset$.
  For each $F \in E(H)$, we consider the following cases.

  \ccase{1} $\overline{F} = \{h\}$ for some $h \in V(H)$.\\
  Let $(J,Z)$ be the strip corresponding to $F$.
  Add $e_F$ to $D$---note that the weight of $e_F$ is $\cdsa{J}{N[Z]}$, which is at most $|D' \cap \eta(F)|$.

  \ccase{2} $\overline{F} = \{h,h'\}$ for some $h,h' \in V(H)$.
  The strip $(J,Z)$ corresponding to $F$ is a stripe.\\
  Let $Z = \{z,z'\}$, where $z$ corresponds to $h$ and $z'$ to $h'$.
  Since $P$ is covering, $h \in P$ or $h' \in P$. There are several cases:

  \csubcase{2a} $h,h' \in P$.\\
  If $G[D' \cap \eta(F)]$ is connected, $\eta(F',h) \cap D' \not= \emptyset$, and $\eta(F',h') \cap D' \not= \emptyset$, then add $e_F$ to $D$---note that the weight of $e_F$ is $\cds{J} - \cdsa{J}{N[Z]}$, which is at most $|D' \cap \eta(F)| - \cdsa{J}{N[Z]}$.
  Otherwise, do nothing.

  \csubcase{2b} $h \in P$, $h' \not\in P$.\\
  Add $e_F$ to $D$ --- note that the weight of $e_F$ is $\cds{J \setminus z'}$, which is at most $|D' \cap \eta(F)|$, since $\eta(h) \setminus \eta(F,h) \not= \emptyset$ and $h' \not\in P$.

  \csubcase{2c} $h \not\in P$, $h' \in P$.\\
  Add $e_F$ to $D$---note that the weight of $e_F$ is $\cds{J \setminus z}$, which is at most $|D' \cap \eta(F)|$, since $\eta(h') \setminus \eta(F,h') \not= \emptyset$ and $h \not\in P$.

  \ccase{3} $\overline{F} = \{h,h'\}$ for some $h, h' \in V(H)$. The strip $(J,Z)$ corresponding to $F$ is a spot.\\
  If the vertex of $J\setminus Z$ is in $D'$, then add $e_{F}$ to $D$ --- note that the weight of $e_{F}$ is $|D' \cap \eta(F)|$.

  \medskip\noindent
  From the construction of $D$ and the analysis in each of the above cases, it follows that $w_P(D) \leq |D| - k_P \leq k - k_P$. 
  To see that $|D| \leq k$, observe that when we add an edge $e_F$ to $D$ for some $F \in E(H)$, we can map it to a unique vertex of $\eta(F) \cap D'$.

  We now show that $D$ is an edge dominating set of $K_{P}$.
  Let $e \in E(K_{P})$. 
  If $e$ is incident on a vertex $v_{h}$ with $h \in P$, then there is an $F \in E(H)$ with $h \in \overline{F}$ such that $\eta(F,h) \cap D' \not= \emptyset$.
  Note that it is possible to choose $F$ such that $G[D' \cap \eta(F)]$ is connected and $\eta(F,h') \cap D' \not= \emptyset$ for each $h' \in \overline{F}$.
  For this $F$, by going through the above cases, an edge incident on $v_h$ will be added to $D$ and $e$ is dominated.
  If~$e$ is not incident on any vertex $v_{h}$ for $h \in V(H)$, then $e$ is an edge $e_F'$ of the construction of Case 1 or 2.
  Then, by the construction of $D$, $e$ is dominated. 
  We can now assume that $e$ is incident only on vertices of $v_h$ for $h \in V(H) \setminus P$.
  By the construction of $K_P$, this implies that $e$ was added in Case~3 due to an $F \in E(H)$ that corresponds to a spot.
  Hence, there is an $F' \in E(H)$ (possibly $F=F'$) such that $h \in \overline{F'}$ or $h' \in \overline{F'}$, and respectively $\eta(F',h) \cap D \not= \emptyset$ or $\eta(F',h') \cap D \not= \emptyset$.
  Assume it is the former.
  Note that $F'$ corresponds to a spot; otherwise, $h \in P$ because $\eta(F',h) \cap D \not= \emptyset$, a contradiction.
  It follows that the edge $e_{F'}$ is in $D$, and thus $e$ is dominated.
  Hence, $D$ is an edge dominating set.

  It remains to show that $K_P[D]$ is connected.
  Define $L \subseteq V(H)$ such that $h \in L$ if and only if $\eta(h) \cap D' \not= \emptyset$.
  Note that $P \subseteq L$.
  By the construction of $D$, it follows that $e_F \in D$ for some $F \in E(H)$ if and only if $\eta(F,h) \cap D' \not= \emptyset$ for each $h \in \overline{F} \cap L$ and $G[D' \cap \eta(F)]$ is connected.
  Hence, $K_P[D]$ is connected, because $G[D']$ is connected.
\end{proof}
The previous lemma shows that it suffices to resolve certain instances of the \textsc{Weighted Connected Edge Dominating Set} problem on multigraphs.
Fortunately, this problem is fixed-parameter tractable.
We need the following auxiliary result.

\begin{theorem}
\label{thm:ds:compute-ceds}
  Let $G$ be a connected graph on $n$ vertices, $w: E(G) \rightarrow \mathbb{N}^{+}$ a weight function, and $k \in \mathbb{N}^{+}$. 
  Then we can find in $O^{*}(16^{k}W)$ time a connected edge dominating set $D$ of smallest weight with $|D| \leq k$, or correctly decide that no such set $D$ exists, where $W = \max_{e \in E(G)} \{w(e)\}$.
\end{theorem}
\begin{proof}
  We adapt the algorithm of Lemma~\ref{lem:ds:compute-eds}.
  Recall that, given a graph $G$ and a set $S \subseteq V(G)$, a \emph{Steiner tree} of $S$ is a connected subgraph of $G$ that contains all vertices of $S$.

  Suppose that $G$ has a connected edge dominating set $D \subseteq V(G)$ with $|D| \leq k$.
  Then $G$ has a minimal vertex cover $C$ with $|C| \leq 2k$, by taking (a subset of) the endpoints of the edges in $D$.
  Moreover, $D$ forms a Steiner tree of $C$.
  Conversely, if $C \subseteq V(G)$ is a vertex cover of $G$, then any Steiner tree of $C$ forms a connected edge dominating set of $G$.

  This suggests the following algorithm.
  Enumerate all minimal vertex covers of $G$ of size at most~$2k$.
  This takes $O(|E(G)| + k^{2} 2^{2k})$ time using an algorithm of Damaschke~\cite{Damaschke2006}, and there are most $2^{2k}$ such vertex covers.
  If no such vertex cover is found, then there is no connected edge dominating set $D$ of $G$ with $|D| \leq k$.
  Otherwise, for each such minimal vertex cover $C$ found, we compute a minimum-weight Steiner tree of $C$, which takes $O^{*}(2^{2k}W)$ time using an algorithm of Nederlof~\cite{Nederlof2013}.
  We output the smallest-weight Steiner tree found.
  The correctness of this algorithm was argued above, and the run time is as claimed.
\end{proof}
This theorem extends immediately to multigraphs: parallel edges can be replaced by a single edge with weight equal to the minimum of the weights of the parallel edges, and a self-loop $e$ of a vertex $v$ can be replaced by adding a pendant vertex $v'$ and modifying $e$ so that it is between $v$ and~$v'$. 
Unfortunately, our weight functions $w_P$ might set the weight of certain edges of $K_P$ to $0$.
Hence, the above theorem might not be directly applicable, as it assumes positive edge weights.
However, there is a straightforward workaround.

\begin{lemma}
\label{lem:wceds}
  Let $G$ be a graph, $w: E(G) \rightarrow \mathbb{Z}$ a weight function, and $k \in \mathbb{Z}$.
  Then in linear time, we can construct a graph $G'$, a weight function $w' : E(G') \rightarrow \mathbb{N}^{+} \cup \{\infty\}$, and a $k' \in \mathbb{N}^{+}$ such that $G$ has a connected edge dominating set of weight at most $k$ (under $w$) if and only if $G'$ has a connected edge dominating set of weight at most $k'$ (under $w'$).
  In particular, $k' \leq \max\{1,1 + k-w(N)\}$, where $N \subseteq V(G)$ denotes the set of edges in $G$ of nonpositive weight (under $w$).
\end{lemma}
\begin{proof}
  If $k < w(N)$, then $G$ has no connected edge dominating set of weight at most $k$, so we can return a trivial ``no''-instance (for example, a graph with a single edge of weight $2$ and $k' = 1$).
  So assume that $k \geq w(N)$.
  If $N = \emptyset$, then we return $G' = G$, $w' = w$, $k' = k$ if $k > 0$, a trivial ``yes''-instance if $E(G) = \emptyset$ and $k \geq 0$, and a trivial ``no''-instance otherwise.
  Now recall that the \emph{contraction} of an edge $(u,v) \in E(G)$ removes $u$ and $v$, adds a vertex $x$ that is adjacent to all neighbors of $u$ and of $v$, and removes self-loops (but not parallel edges). We call $u$ and $v$ \emph{contracted} vertices and $x$ a \emph{fused} vertex.
  Note that connectivity is maintained under contraction.

  We now define $G'$ as follows.
  Starting with $G$, contract all edges of $N$, and let $X \subseteq V(G')$ denote the set of fused vertices and $Y \subseteq V(G)$ the set of contracted vertices.
  For each $x \in X$, add two vertices $s_x,t_x$ and edges $(x,s_x)$ of weight $0$ and $(s_x,t_x)$ of weight $\infty$. For exactly one $x \in X$, we set the weight of $(x,s_x)$ to $1$.
  Note that this defines a weight function $w'$ on $E(G')$ that is the same for all edges of $E(G') \cap E(G)$. Finally, set $k' = 1+k-w(N)$.
  Since $k \geq w(N)$, indeed $k' \in \mathbb{N}^{+}$.

  Suppose that $G$ has a connected edge dominating set $D$ with $w(D) \leq k$.
  Since the edges of $N$ have nonpositive weight, we can assume that $N \subseteq D$. 
  Then, we can also assume that $D$ contains no edges of $E(G) \setminus N$ for which both endpoints are in the same connected component of $G[N]$, which means that no edges of $D$ will be removed by removing loops that arise in the construction of $G'$.
  We now claim that $D' = \{(x,s_x) \mid x \in X\} \cup (D\setminus N)$ is a connected edge dominating set of $G'$ with $w'(D') \leq k'$.
  Suppose that $e \in E(G')$ is not dominated by $D'$.
  Then $e$ is not incident on $x$ or $s_x$ for any $x \in X$ by the construction of $D'$, so $e \in E(G)$ and $e$ is incident on two vertices in $V(G) \setminus Y$.
  As $D$ is an edge dominating set of $G$, there is an edge $f \in D$ incident on one of the endpoints of $e$. Moreover, $f \in E(G')$, since the endpoint shared with $e$ is not in $Y$.
  Hence, $f \in D'$ and $e$ is dominated, a contradiction.
  Therefore, $D'$ is an edge dominating set of $G'$.
  To see that it has weight at most $k'$, recall that $k \geq w(N)$, $N \subseteq D$, and that $D \setminus N$ does not contain any edges for which both endpoints are in the same connected component of $G[N]$, and thus $w'(D') = 1 + w'(D \setminus N) = 1 + w(D\setminus N) \leq 1 + k-W(N) = k'$.
  Finally, $G'[D']$ is connected, because $N \subseteq D$ and $D'$ can be obtained from $D$ by contracting all edges of $N$ and by adding edges incident on the fused vertices, which retains connectivity.

  Suppose that $G'$ has an edge dominating set $D'$ with $w'(D') \leq k'$.
  By the construction of $G'$ and~$w'$, the edges $(x,s_x)$ for each $x \in X$ are in $D'$, but the edges $(s_x,t_x)$ are not.
  We now claim that $D = (D'\setminus\{(x,s_x) \mid x \in X\}) \cup N$ is a connected edge dominating set of $G$ with $w(D) \leq k$.
  Indeed, $w(D) = w(D' \setminus\{(x,s_x) \mid x \in X \}) + w(N) = w'(D' \setminus\{(x,s_x) \mid x \in X\}) + w(N) \leq k'-1 + w(N) = k$.
  Suppose that $e \in E(G)$ is not dominated by $D$.
  Then $e$ is not incident on $Y$, so $e \in E(G')$ and $e$ is incident on two vertices of $V(G') \setminus \{x,s_x,t_x \mid x \in X\}$.
  As $D'$ is an edge dominating set of $G'$, there is an edge $f \in D'$ incident on one of the endpoints of $e$. Moreover, $f \in E(G)$, since the endpoint shared with $e$ is not in $\{x,s_x,t_x \mid x \in X\}$. Hence, $f \in D$ and $e$ is dominated, a contradiction.
  Therefore, $D$ is an edge dominating set of $G$.
  Finally, $G[D]$ is connected, because $N \subseteq D$ and $D$ is obtained from $D'$ by uncontracting all edges of $N$, which retains connectivity since the uncontracted edge is added to $D$.
\end{proof}

\begin{corollary}
\label{cor:wceds}
  Let $G$ be a multigraph, $w: E(G) \rightarrow \mathbb{Z}$ a weight function, and $k \in \mathbb{Z}$.
  Then we can decide in $O^{*}(16^{\max\{1,1+k-w(N)\}}W)$ time whether $G$ has a connected edge dominating set $D$ with $w(D) \leq k$, where $N \subseteq V(G)$ denotes the set of edges in $G$ of nonpositive weight (under $w$) and $W = \max_{e \in E(G)} \{w(e)\}$.
\end{corollary}

\begin{theorem}
  Let $G$ be a connected claw-free graph with $n$ vertices and let $k$ be an integer.
  Then there is an algorithm that runs in $O^{*}(64^k)$ time and that either correctly reports that $\cds{G} > k$ or returns a smallest connected dominating set of $G$.
\end{theorem}
\begin{proof}
  We first preprocess the graph.
  We iteratively apply the following procedure.
  We check whether $G$ admits twins; this can be done in $O(n+m)$ time by Theorem~\ref{thm:twins-algo}.
  If so, then following Lemma~\ref{lem:cds:twins}, we can find a connected claw-free graph $G'$ with $|V(G')| < |V(G)|$ such that $\cds{G} = \cds{G'}$ in linear time, and continue to the next iteration with $G = G'$.
  Otherwise, $G$ does not admit twins, and we check whether $G$ admits a proper W-join; this can be done in $O(n^{2}m)$ time by Theorem~\ref{thm:proper-wjoin-algo}.
  If so, the following Lemma~\ref{lem:cds:properWjoins}, we can find a connected claw-free graph $G'$ with $|E(G')| < |E(G)|$ (note that proper W-joins have at least two edges by definition, and one remains after applying the lemma) such that $\cds{G} = \cds{G'}$ in linear time, and continue to the next iteration with $G=G'$.
  Since the number of edges decreases in every iteration, there are at most $m$ iterations, and the run time of the procedure is $O(n^6)$.
  After the procedure finishes, we end up with a connected claw-free graph that does not admit twins nor proper W-joins for which the size of its smallest dominating set is unchanged.
  By abuse of notation, we call this graph $G$ as well.

  We test whether $\cds{G} \leq 7$ using the algorithm of Lemma~\ref{lem:cds:constant} in $O(n^{8})$ time.
  If so, then the algorithm actually gives a connected dominating set of size $\cds{G}$, and we are done.
  Otherwise, if $\alpha(G) \leq 3$, then $\cds{G} \leq 7$ by Proposition~\ref{prp:ds:is-size} and Proposition~\ref{prp:cds:boundbyds}, a contradiction. Hence, $\alpha(G) > 3$.

  We then apply Theorem~\ref{thm:main-base2} in $\stripTime$ time.
  If it outputs that $G$ is a proper circular-arc graph, then we can compute $\cds{G}$ in linear time by Theorem~\ref{thm:cds:circular}.
  If it outputs that $G$ is a thickening of an XX-trigraph, then we can compute $\cds{G}$ in $O(n^{5})$ time by Lemma~\ref{lem:cds:s2}.
  It remains that it outputs that $G$ admits a strip-structure such that for each strip $(J,Z)$:
  \begin{itemize}
    \item $(J,Z)$ is a trivial line graph strip, or
    \item $(J,Z)$ is a stripe for which
    \begin{itemize}
      \item $1\leq |Z| \leq 2$, $\alpha(J) \leq 3$, and $V(J) \setminus N[Z] \not= \emptyset$,
      \item $|Z|=1$, $J$ is a proper circular-arc graph, and either $J$ is a (strong) clique or $\alpha(J) > 3$,
      \item $|Z|=2$, $J$ is a proper interval graph, and either $J$ is a (strong) clique or $\alpha(J) > 3$, or\item $(J,Z)$ is a thickening of a member of $\mathcal{Z}_{5}$.
    \end{itemize}
  \end{itemize}
  We now determine the set $\mc{P}(H)$, which takes linear time.
  Then for each $P \subseteq \mc{P}(H)$ such that $P$ is covering, compute the triple $(K_{P},w_P,k_{P})$.
  Using Proposition~\ref{prp:cds:strip-boundary} and Lemma~\ref{lem:cds:circular-stripe}, \ref{lem:cds:z5}, \ref{lem:cds:stripe-is}, and~\ref{lem:cds:spot}, we can to compute a set $D \subseteq V(J) \setminus Z$ of size $\cds{J \setminus Q}$ that dominates $J \setminus Q$ for any $Q \subseteq Z$ and a set $D \subseteq V(J) \setminus Z$ of size $\cdsa{J}{N[Z]}$ that dominates $J \setminus N[Z]$ and that has for each component a vertex of $N(Z)$ for the strips $(J,Z)$ of the strip-structure.
  Hence, we can compute the triple $(K_{P},w_P,k_{P})$ in $O(n^{8})$ time.
  Note that the number of vertices of $K_P$ is in the order of $|E(H)|$, which is $O(n)$, because each strip contains at least one vertex of $G$ that is unique to it. 
  Moreover, the weight of each edge of $K_P$ under $w_P$ is nonnegative and at most $k+O(1) = O(n)$, or we can answer ``no''.
  Then, using Corollary~\ref{cor:wceds}, we can decide in $O^{*}(16^{k})$ time whether $K_P$ has a connected edge dominating set $D$ of smallest weight such that $w_P(D) \leq k - k_P$.
  If so, then by Lemma~\ref{lem:cds:reduction}, $\cds{G} \leq k$.
  Moreover, since~$D$ is a connected edge dominating set of $K_P$ of smallest weight, the transformation of Lemma~\ref{lem:cds:reduction} gives a smallest dominating set of $G$.
  If Corollary~\ref{cor:wceds} returns for all $P \subseteq \mc{P}(H)$ that are covering that there is no connected edge dominating set $D$ of $K_{P}$ with $|D| \leq k$, then we can answer NO by Lemma~\ref{lem:cds:reduction}.
  Since $|\mc{P}(H)| \leq 2k$ by Lemma~\ref{lem:ds:striped}, the run time is $O^{*}(2^{2k} 16^k) = O^{*}(64^k)$.
\end{proof}
Note that the polynomial factor hidden in the $O^{*}$ is at least $n^8$.
Also observe that the above algorithm can be modified to take into account that in Corollary~\ref{cor:wceds}/Theorem~\ref{thm:ds:compute-ceds} we are only interested in minimal vertex covers that contain every vertex $v_h$ for $h \in P$.
This can be argued in a manner similar to Lemma~\ref{lem:ds:compute-eds}.
Hence, we only need to enumerate minimal vertex covers of size $2k-|P|$ in Theorem~\ref{thm:ds:compute-ceds}. Then the run time improves to
$$\sum_{i=0}^{2k} {2k \choose i} 2^{2k-i} 2^{2k} = 2^{2k} \sum_{i=0}^{2k} {2k \choose i} 2^{2k-i} = 2^{2k}3^{2k}.$$

\begin{corollary}
\label{cor:cds:fpt}
  Let $G$ be a connected claw-free graph and let $k$ be an integer.
  Then we can decide in $O^{*}(36^k)$ time whether $\cds{G} \leq k$.
\end{corollary}

\section{Polynomial Kernel for \textsc{Dominating Set}}
\label{sec:kernel}
We show that \textsc{Dominating Set} has a polynomial kernel on claw-free graphs.
The basic idea of our kernel is to replace each stripe of the strip-structure given in Theorem~\ref{thm:main-base-kernel} by a stripe of size at most constant or linear in $k$.
We show how to do this in Sect.~\ref{sec:kernel:stripes}.
We then reduce the strip-structure itself to have polynomial size in Sect.~\ref{sec:kernel:reduce}.
We combine these ideas in Sect.~\ref{sec:kernel:kernel} to give the actual kernel.

\subsection{Reducing Stripes} \label{sec:kernel:stripes}
In this subsection, we show how to reduce the stripes of a strip-structure for a claw-free graph to have small size, and essentially provide a kernel for individual stripes.
These kernels preserve the value of $\dsg{J\setminus(Q \cup R)}{N[R]}$ for any disjoint $Q,R \subseteq Z$, in line with the ideas behind the  algorithm described in Sect.~\ref{sec:ds}.
We show that each `stripe kernel' has linear size, and often even constant size. 

In order to actually find and use such kernels, we need the following lemmas.
The first lemma shows that we can replace one stripe with another that, when appropriately chosen, yields an equivalent instance of \textsc{Dominating Set}.

\begin{lemma}
\label{lem:kernel:equiv}
  Let $G$ be a graph, let $(H,\eta)$ be a purified strip-structure of nullity zero for $G$, and let $F \in E(H)$ correspond to a stripe $(J,Z)$.
  Let $(J',Z')$ be a claw-free stripe and $k' \geq 0$ be an integer such that $Z = Z'$ and $\dsg{J\setminus(Q \cup R)}{N[R]} = \dsg{J'\setminus(Q \cup R)}{N[R]} + k'$ for any disjoint $Q,R \subseteq Z$.
  Let $G'$ be obtained from $G$ by replacing the vertices of $\eta(F) = V(J) \setminus Z$ by those of $V(J') \setminus Z$ and let $k$ be an integer.
  Then $G'$ has a purified strip-structure $(H',\eta')$ of nullity zero such that $G'$ and $(H',\eta')$ have the following properties:
  \begin{itemize}
    \item $V(H) = V(H')$, $E(H) = E(H')$, $\overline{F'}$ in $H$ is equal to $\overline{F'}$ in $H'$ for each $F' \in E(H) = E(H')$, 
    \item $\eta'(F') = \eta(F')$ for each $F' \in E(H) \setminus \{F\}$ and $\eta'(F) = V(J') \setminus Z$,
    \item $\ds{G} \leq k$ if and only if $\ds{G'} \leq k - k'$, and
    \item if $G$ is claw-free, then so is $G'$.
  \end{itemize}
\end{lemma}
\begin{proof}
  The statements about the existence of $(H',\eta')$ are straightforward.
  Hence, we focus on proving the final two statements of the lemma.

  Let $D$ be a smallest dominating set of $G$, and assume $|D| \leq k$.
  We first construct two disjoint sets $Q,R \subseteq Z$.
  Initially, $Q=R=\emptyset$.
  Consider each $h \in \overline{F}$ and let $z \in Z$ correspond to $h$.
  If $D \cap \eta(h) = \emptyset$, then add $z$ to $Q$. If $D \cap \eta(h) \not= \emptyset$ and $D \cap \eta(F,h) = \emptyset$, then add $z$ to $R$.

  Observe that the construction of $Q,R$ implies that if a set $S \subseteq V(J) \setminus Z$ dominates $V(J) \setminus (Q \cup N[R])$, then $(D \setminus \eta(F)) \cup S$ is a dominating set of $G$.
  Indeed, let $h \in \overline{F}$ with corresponding $z \in Z$. If $D \cap \eta(h) = \emptyset$, then $S$ dominates $\eta(F,h) = N(z)$.
  If $D \cap \eta(h) \not= \emptyset$ and $D \cap \eta(F,h) = \emptyset$, then $S$ does not necessarily dominate $\eta(F,h) = N(z)$, because it is already dominated by a vertex in $D \cap (\eta(h) \setminus \eta(F,h)) \not= \emptyset$.
  Otherwise, the vertex $z$ ensures that $N(z)$ contains a vertex of $S$.
  All other vertices are straightforwardly dominated by $D \setminus \eta(F)$ or $S$.
  It is important to note that this observation holds regardless of the inner structure of $(J,Z)$.

  Since $D$ is a smallest dominating set of $G$, the above observation implies that $|D \cap \eta(F)| \leq \dsg{J\setminus(Q \cup R)}{N[R]}$.
  A similar argument as in the above observation also implies that if $S' \subseteq V(J') \setminus Z$ dominates $V(J') \setminus (Q \cup N[R])$, then $(D \setminus\eta(F)) \cup S'$ is a dominating set of $G'$.
  Using this with a set~$S'$ that attains $\dsg{J'\setminus(Q \cup R)}{N[R]}$ implies that 
  \begin{equation*}
    \ds{G'} \leq |(D \setminus \eta(F)) \cup S'|
            \leq \ds{G} - \dsg{J\setminus(Q \cup R)}{N[R]} + \dsg{J'\setminus(Q \cup R)}{N[R]}
            \leq \ds{G} - k' \leq k - k' \enspace .
  \end{equation*}
  For the converse, we apply the same argument.
  Since $\eta(h)$ is a clique for each $h \in V(H)$, it follows immediately from the fact that $J'$ is claw-free that $G'$ is claw-free if $G$ is.
\end{proof}
We also need the following definition in order to simplify the algorithms that find `stripe kernels'.

\begin{definition}
\label{def:kernel:semi}
  Let $(J,Z)$ be a stripe that is a thickening $\mc{W}$ of a stripe $(J',Z')$ such that $J$ is a graph.
  Then $(J,Z)$ is called \emph{semi-thickened} if $W_{v'} \cup W_{w'}$ does not contain twins in $J$ for any $v',w' \in V(J') \setminus Z'$ (possibly $v'=w'$) that are not incident on a semi-edge in $J'$.
\end{definition}
Note that this definition implies in particular that $v',w' \in V(J') \setminus Z'$ are not twins in $J'$ if $v',w'$ are not incident on a semi-edge.

We now show that we get this property essentially for free for stripes that arise in a graph that does not admit twins.

\begin{lemma}
\label{lem:kernel:semi}
  Let $G$ be a graph, let $(H,\eta)$ be a purified strip-structure of nullity zero for $G$, and let $F \in E(H)$ correspond to a stripe $(J,Z)$ such that $J$ is a graph and $(J,Z)$ is a thickening $\mc{W}$ of a stripe $(J',Z')$.
  If $G$ does not admit twins, then $(J,Z)$ is semi-thickened.
\end{lemma}
\begin{proof}
  Suppose that $G$ does not admit twins.
  Let $v',w' \in V(J') \setminus Z'$ (possibly $v'=w'$) that are not incident on a semi-edge in $J'$.
  Suppose that $W_{v'} \cup W_{w'}$ contains two vertices $v,w$ that form twins in $J$. Since $v',w' \not\in Z'$, $v,w \in V(G)$, and the definition of a strip-structure implies that $G$ admits twins $v,w$, a contradiction.
  Hence, $(J,Z)$ is semi-thickened.
\end{proof}
We now show a basic operation to reduce the size of certain stripes.
This operation preserves claw-freeness and semi-thickedness.

\begin{lemma}
\label{lem:kernel:zis2-base}
  Let $G$ be a graph, let $(H,\eta)$ be a purified strip-structure of nullity zero for $G$, and let $F \in E(H)$ correspond to a stripe $(J,Z)$ such that $J$ is a graph and $(J,Z)$ is a thickening $\mc{W}$ of a stripe $(J'',Z'')$.
  Then in linear time, we can compute a graph $G'$ and a purified strip-structure $(H',\eta')$ of nullity zero such that:
  \begin{itemize}
    \item $V(H) = V(H')$, $E(H) = E(H')$, $\overline{F'}$ in $H$ is equal to $\overline{F'}$ in $H'$ for each $F' \in E(H) = E(H')$, 
    \item $\eta'(F') = \eta(F')$ for each $F' \in E(H) \setminus \{F\}$, 
    \item $\ds{G} = \ds{G'}$, 
    \item the stripe $(J',Z')$ corresponding to $F$ under $\eta'$ is a thickening $\mc{W'}$ of $(J'',Z'')$,
    \item $|W'_{w''}| = 1$ if $w'' \in V(J'')$ is not incident on a semi-edge in $J''$, 
    \item $|W'_{v''}| + |W'_{w''}| \leq 3$ if $v'',w'' \in V(J'')$ are semiadjacent in $J''$,
    \item if $G$ is claw-free, then so is $G'$, and
    \item if $J$ is semi-thickened, then so is $J'$.
  \end{itemize}
\end{lemma}
It is important to note that this lemma crucially relies on knowing both $(J'',Z'')$ and the thickening $\mc{W}$, and we assume them to be given as input to the algorithm.
\begin{proof}
  Let $w'' \in V(J'') \setminus Z''$ and suppose that $w''$ is not incident on a semi-edge.
  If $|W_{w''}| \geq 2$, then any two vertices of $W_{w''}$ form twins in $G$. Using Lemma~\ref{lem:ds:twins}, we can remove one of these vertices from $G$ without changing the size of the smallest dominating set and without disturbing the absence of claws.

  Let $v'',w'' \in V(J'') \setminus Z''$ and suppose that $v''$ and $w''$ are semiadjacent.
  Since $J$ is a graph, $(W_{v''},W_{w''})$ is a W-join in $G$.
  Using Lemma~\ref{lem:ds:Wjoins}, we can remove vertices of this W-join from $G$ without changing the size of the smallest dominating set and without disturbing the absence of claws.

  By applying this to all $w'' \in V(J'') \setminus Z''$ that are not incident on a semi-edge, and to all $v'',w'' \in V(J'') \setminus Z''$ that are semiadjacent, we arrive at the requested graph $G'$.
  Using the definition of a thickening, it is straightforward to see that $G'$ has a purified strip-structure $(H',\eta')$ of nullity zero such that $G'$ and $(H',\eta')$ have the required properties.
  Moreover, Lemma~\ref{lem:ds:twins} and~\ref{lem:ds:Wjoins} imply that $G'$ and $(H',\eta')$ can be computed in linear time.
  Since we only delete vertices, if $G$ is claw-free, then so is $G'$. The construction implies straightforwardly that $(J',Z')$ is a thickening of $(J'',Z'')$.
  Hence, if $J$ is semi-thickened, then so is $J'$.
\end{proof}
This lemma leads to the following definition.

\begin{definition}
  Let $G$ be a graph, let $(H,\eta)$ be a purified strip-structure of nullity zero for $G$, and let $F \in E(H)$ correspond to a stripe $(J,Z)$ such that $J$ is a graph and $(J,Z)$ is a thickening $\mc{W}$ of a stripe $(J'',Z'')$.
  Then $(J,Z)$ is \emph{reduced} if $|W_{w''}| = 1$ when $w'' \in V(J'')$ is not incident on a semi-edge in $J''$ and $|W'_{v''}| + |W'_{w''}| \leq 3$ when $v'',w'' \in V(J'')$ are semiadjacent in $J''$.
\end{definition}
We now describe how to find `stripe kernels' and consider stripes $(J,Z)$ with $|Z| = 1$ and $|Z| = 2$ separately.

\subsubsection{\texorpdfstring{$|Z|=1$}{|Z|=1}}
Stripes $(J,Z)$ with $|Z| = 1$ are reasonable straightforward to reduce to an equivalent stripe (in the sense of Lemma~\ref{lem:kernel:equiv}) of constant size.

\begin{lemma}
\label{lem:kernel:zis1}
  Let $(J,Z)$ be a stripe such that $J$ is a graph and $|Z|=1$.
  Given the values of $\dsg{J\setminus(Q \cup R)}{N[R]}$ for any disjoint $Q,R \subseteq Z$, we can compute in constant time a claw-free stripe $(J',Z')$ and an integer $k' \geq 0$ such that $Z = Z'$, $|V(J')| \leq 4$, and $\dsg{J\setminus(Q \cup R)}{N[R]} = \dsg{J'\setminus(Q \cup R)}{N[R]} + k'$ for any disjoint $Q,R \subseteq Z$.
\end{lemma}
\begin{proof}
  Let $Z = \{z\}$.
  First, we observe that $\dsg{J \setminus Z}{N[Z]} \leq \ds{J \setminus Z} \leq \ds{J}$.
  The first inequality holds by definition.
  To see the second inequality, recall that $N[Z]$ is a (strong) clique.
  Consider any dominating set $D$ of $J$.
  If $z \in D$, then replace it by any vertex in $N(z)$.
  The resulting set is still a dominating set of $J$, and by construction, also of $J \setminus Z$.

  Second, observe that $0 \leq \ds{J} - \dsg{J \setminus Z}{N[Z]} \leq 1$, because any dominating set of $J \setminus Z$ where $N[Z]$ does not need to be dominated can be made into a dominating set of $J$ by simply adding $z$ to it.

  From the above observations, there are three cases to consider.
  First, suppose that $\dsg{J \setminus Z}{N[Z]} + 1 = \ds{J \setminus Z} = \ds{J}$.
  Let $J'$ be a two-vertex path with $z$ as one of its ends and let $Z' = \{z\}$.
  Let $k' = \dsg{J \setminus Z}{N[Z]}$.
  Note that $\ds{J'} = \ds{J' \setminus Z'} = 1$ and that $\dsg{J' \setminus Z'}{N[Z']} = 0$. 
  Hence, $(J',Z')$ is as required by the lemma statement.

  Suppose that $\dsg{J \setminus Z}{N[Z]} = \ds{J \setminus Z} = \ds{J} - 1$.
  Let $J'$ be a four-vertex path with $z$ as one of its ends and let $Z' = \{z\}$.
  Let $k' = \dsg{J \setminus Z}{N[Z]}-1$. Note that $\ds{J'} = 2$ and that $\dsg{J' \setminus Z'}{N[Z']} = \ds{J' \setminus Z'} = 1$.
  Moreover, $k' \geq 0$, because $\ds{J \setminus Z} < \ds{J}$, implying that any smallest dominating set $D$ for $J \setminus Z$ satisfies $D \cap N(Z) = \emptyset$, even though the vertices in the nonempty set $N(Z)$ have to be dominated.
  Therefore, $V(J) \setminus N[Z] \not= \emptyset$, and thus $\dsg{J \setminus Z}{N[Z]} > 0$. Hence, $(J,Z')$ is as required by the lemma statement.

  Finally, suppose that $\dsg{J \setminus Z}{N[Z]} = \ds{J \setminus Z} = \ds{J}$.
  Let $J'$ be a three-vertex path with $z$ as one of its ends and let $Z' = \{z\}$. Let $k' = \ds{J} - 1$.
  Note that $\dsg{J' \setminus Z'}{N[Z']} = \ds{J' \setminus Z'} = \ds{J'} = 1$.
  Moreover, $\ds{J} \geq 1$ as $N(z)$ is nonempty, and thus $k' \geq 0$. Hence, $(J,Z')$ is as required by the lemma statement.

  Given the values of $\dsg{J\setminus(Q \cup R)}{N[R]}$ for any disjoint $Q,R \subseteq Z$, we can trivially compute the $(J',Z')$ prescribed by the three above cases in constant time.
\end{proof}

\subsubsection{\texorpdfstring{$|Z|=2$}{|Z|=2}}
Stripes $(J,Z)$ with $|Z|=2$ are significantly harder to reduce, because it seems infeasible to provide the generic construction like we provided for the case $|Z|=1$: there would be too many different cases to deal with.
Therefore, we consider the stripes with $|Z|=2$ that can occur in Theorem~\ref{thm:main-base-kernel}, and find an explicit `stripe kernel' for each of them.

\begin{lemma}
\label{lem:kernel:interval}
  Let $(J,Z)$ be a stripe such that $J$ is a proper interval graph and $|Z| = 2$.
  Then we can compute in linear time a claw-free stripe $(J',Z')$ such that $Z = Z'$, $|V(J')| \leq 18\,\ds{J}+2$, and $\dsg{J\setminus(Q \cup R)}{N[R]} = \dsg{J'\setminus(Q \cup R)}{N[R]}$ for any disjoint $Q,R \subseteq Z$.
\end{lemma}
\begin{proof}
  We find a proper interval representation of $J$; this takes linear time~\cite{DengHH1996}.
  From now on, we will not distinguish between the vertices and their corresponding intervals.
  Number the intervals in order of their left endpoints: $v_1,\ldots,v_n$, where $n = |V(J)|$.
  We make two crucial claims.
  \begin{cclaim}
  \label{c:interval:1}
    $\dsg{J\setminus(Q \cup R)}{N[R]} = \ds{J\setminus(Q \cup N[R])}$ for any disjoint $Q, R \subseteq Z$.
  \end{cclaim}
  \begin{cproof}
    Let $D \subseteq V(J)\setminus(Q \cup R)$ be a smallest set dominating $V(J) \setminus (Q \cup N[R])$ that (among all such smallest dominating sets) has the least number of vertices in $N(R)$.
    Suppose that $D \cap N(R) \not= \emptyset$, and let $v_{i} \in N(R) \cap D$.
    By the choice of $D$, $v_{i}$ must have a neighbor $v_{j}$ that is not in $N[R]$, nor in the dominating set, nor dominated by $D\setminus N[R]$; otherwise, either $D$ is not smallest or $D \cap N(R)$ is not smallest, a contradiction.
    We will show that we can replace $v_i$ by $v_j$.
    Let $v_k \not= v_j$ be any neighbor of $v_i$ that is not in $N[R]$, nor in the dominating set, nor dominated by $D\setminus N[R]$.
    Note that any such~$v_k$ is uniquely dominated by $v_i$, and might not be dominated by $v_j$.
    Since $v_{j},v_{k}$ are not in $N(R)$ and $J$ is a proper interval graph, $j,k > i$ or $j,k < i$.
    Assume without loss of generality that it is the former.
    Then $i < j < k$ or $i < k < j$, and thus the fact that $J$ is a proper interval graph implies that~$v_{j}$ and~$v_{k}$ are adjacent.
    Hence, $(D \setminus \{v_{i}\})\cup \{v_{j}\}$ is also a smallest subset of $V(J) \setminus (Q \cup R)$ dominating $V(J) \setminus (Q \cup R)$, but with a strictly smaller number of vertices in $N(R)$, contradicting the choice of $D$. Therefore, $D \cap N(R) = \emptyset$, and the claim follows.
  \end{cproof}

  \begin{cclaim}
  \label{c:interval:2}
    The following greedy algorithm finds a smallest dominating set of a proper interval graph: consider the leftmost still undominated vertex $v$ and add the rightmost neighbor of $v$ to the dominating set.
    Moreover, the algorithm runs in linear time.
  \end{cclaim}
  \begin{cproof}
  Straightforward by induction using that the graph is a proper interval graph.
  \end{cproof}

  We now show the following.
  Consider arbitrary disjoint $Q, R \subseteq Z$.
  Let $D$ denote a smallest dominating set for $J \setminus (Q \cup N[R])$ as found by the strategy of Claim~\ref{c:interval:2}.
  Let $D^{u}$ denote the set of leftmost undominated vertices as considered by the strategy.
  For any $A \subseteq V(J)$, consider the graph $(J[A \cup D \cup D^{u}])\setminus(Q \cup N[R])$.
  Observe that the vertices in $A$ are irrelevant to the algorithm of Claim~\ref{c:interval:2}, because the algorithm only uses the vertices in $D^{u}$ to compute $D$.
  Therefore, $\ds{(J[A \cup D \cup D^{u}])\setminus(Q \cup N[R])} = \ds{J \setminus (Q \cup N[R])}$. 

  We now construct $J'$ as follows. For each disjoint $Q, R \subseteq Z$, it follows from Claim~\ref{c:interval:1} that $\dsg{J\setminus(Q \cup R)}{N[R]} = \ds{J\setminus(Q \cup N[R])}$.
  Using the algorithm of Claim~\ref{c:interval:2}, find a smallest dominating set $D$ for $J\setminus(Q \cup N[R])$.
  Let $D^{u}$ be as in the above paragraph.
  Add $D$ and $D^{u}$ to $J'$.
  Also ensure that the vertices of $Z$ are added to $J'$ and set $Z' = Z$.
  The observation of the previous paragraph together with Claim~\ref{c:interval:1} now implies that $\dsg{J\setminus(Q \cup R)}{N[R]} = \dsg{J'\setminus(Q \cup R)}{N[R]}$ for any disjoint $Q,R \subseteq Z$.
  Moreover, the neighborhood of each $z \in Z$ is nonempty in $J'$ by construction.
  Hence, $J'$ is a claw-free stripe. Finally, since there are nine choices for $Q,R$, it follows that $|V(J')| \leq 18\,\ds{J} + 2$.
\end{proof}

\begin{lemma}
\label{lem:kernel:z2}
  Let $(J,Z)$ be a reduced stripe such that $(J,Z)$ is a thickening $\mc{W}$ of a member $(J'',Z'')$ of $\mc{Z}_{2}$, $J$ is a graph, and $(J,Z)$ is semi-thickened.
  Then we can compute in linear time a claw-free stripe $(J',Z')$ such that $Z = Z'$, $|V(J')| \leq 26$, and $\dsg{J\setminus(Q \cup R)}{N[R]} = \dsg{J'\setminus(Q \cup R)}{N[R]}$ for any disjoint $Q,R \subseteq Z$.
\end{lemma}
\begin{proof}
  Let $n$, $a_0$, $b_0$, $A = \{a_{1},\ldots,a_{n}\}$, $B=\{b_{1},\ldots,b_{n}\}$, $C=\{c_{1},\ldots,c_{n}\}$, and $X$ be as in the definition of $\mc{Z}_{2}$. As in Lemma~\ref{lem:recog:z2}, we may assume without loss of generality that $a_{1},b_{1}$ are semiadjacent (and thus $c_{1} \in X$), that $a_{2},c_{2}$ are semiadjacent (and thus $b_{2} \in X$), and that $b_{3},c_{3}$ are semiadjacent (and thus $a_{3} \in X$). Moreover, the analysis of Lemma~\ref{lem:recog:z2} implies that without loss of generality $|X \cap \{a_i,b_i,c_i\}| \in \{1,3\}$ for $i=1,2,3$. 

  Throughout the remainder of the proof, let $i,i' \in \{1,\ldots,n\}$.
  Suppose that $a_i, a_{i'} \not\in X$ and $b_i,b_{i'},c_i,c_{i'} \in X$.
  Then $i,i' \geq 4$ and $W_{a_i} \cup W_{a_{i'}}$ is a homogenous clique, which contradicts that $(J,Z)$ is semi-thickened, because $a_i$ and $a_{i'}$ are not incident on a semi-edge in $J'$.
  Hence, there is at most one $i \geq 4$ such that $a_i \not\in X$, $b_i,c_i \in X$, and $W_{a_i} \not= \emptyset$, and for such a value of $i$, $|W_{a_i}| \leq 1$ because $a_i$ is not incident on a semi-edge and $(J,Z)$ is reduced.
  Similar observations apply if $b_i \not\in X$ and $a_i,c_i \in X$, or $c_i \not\in X$ and $a_i,b_i \in X$.

  Suppose that $a_{i},a_{i'},b_{i},b_{i'} \not\in X$ and $c_i,c_{i'} \in X$.
  Observe that any vertex in $W_{a_i}$ combined with any vertex in $W_{b_{i'}}$ dominates $J$, but no single vertex of $W_{a_{i}}$, $W_{a_{i'}}$, $W_{b_{i}}$, or $W_{b_{i'}}$ does.
  Hence, we can remove all vertices of $W_{a_{i''}}$ and $W_{b_{i''}}$ for any $i'' \geq 4$ with $i'' \not= i,i'$, $a_{i''},b_{i''} \not\in X$, and $c_{i''} \in X$.
  Without loss of generality, $i' \geq 4$.
  Since neither $a_{i'}$ nor $b_{i'}$ are incident on a semi-edge, $|W_{a_{i'}}|, |W_{b_{i'}}| = 1$, because $(J,Z)$ is reduced.
  If $i \geq 4$, then we can similarly observe that $|W_{a_{i}}|, |W_{b_{i}}| = 1$.
  Otherwise, $i=1$.
  Since $a_1$ and $b_1$ are semiadjacent, $|W_{a_1} \cup W_{b_1}| \leq 3$, because $(J,Z)$ is reduced.
  Similar reductions apply if $a_{i},a_{i'},c_{i},c_{i'} \not\in X$ and $b_i,b_{i'} \in X$, or $b_{i},b_{i'},c_{i},c_{i'} \not\in X$ and $a_i,a_{i'} \in X$.
  (Note that in these cases, for example, $W_{a_i}$ combined with $W_{c_{i'}}$ dominates $J$ except one vertex from $Z$; however, this does not influence the correctness of the reduction rule that is applied here.)

  Suppose that $a_{i},a_{i'},b_{i},b_{i'},c_{i},c_{i'} \not\in X$.
  Then $i,i' \geq 4$.
  Again, any vertex in $W_{a_i}$ combined with any vertex in $W_{b_{i'}}$ dominates $J$, but no single vertex of $W_{a_{i}}$, $W_{a_{i'}}$, $W_{b_{i}}$, or $W_{b_{i'}}$ does.
  Hence, we can remove all vertices of $W_{a_{i''}}$ and $W_{b_{i''}}$ for any $i'' \geq 4$ with $i'' \not= i,i'$ and $a_{i''},b_{i''},c_{i''} \not\in X$.
  Since $a_{i},a_{i'},b_{i},b_{i'},c_{i},c_{i'}$ are all not incident on a semi-edge, $|W_{a_{i}}|, |W_{a_{i'}}|, |W_{b_{i}}|, |W_{b_{i'}}|,|W_{c_{i}}|, |W_{c_{i'}}| = 1$, because $(J,Z)$ is reduced.

  Let $J'$ denote stripe that remains after applying the above reduction operations, and let $Z'=Z$. From the above reductions, $|V(J')| \leq 15 + 9 + 2 = 26$.
  Moreover, it is straightforward to see that $(J',Z')$ is still a thickening of a member of $\mc{Z}_{2}$. Finally, $\dsg{J\setminus(Q \cup R)}{N[R]} = \dsg{J'\setminus(Q \cup R)}{N[R]}$ for any disjoint $Q,R \subseteq Z$ by construction.
\end{proof}

\begin{lemma}
\label{lem:kernel:z3}
  Let $(J,Z)$ be a reduced stripe such that $(J,Z)$ is a thickening $\mc{W}$ of a member $(J'',Z'')$ of $\mc{Z}_{3}$, $J$ is a graph, and $(J,Z)$ is semi-thickened.
  Then we can compute in linear time a claw-free stripe $(J',Z')$ such that $Z = Z'$, $|V(J')| \leq 23$, and $\dsg{J\setminus(Q \cup R)}{N[R]} = \dsg{J'\setminus(Q \cup R)}{N[R]}$ for any disjoint $Q,R \subseteq Z$.
\end{lemma}
\begin{proof}
  Let $H$, $h_1,\ldots,h_5$ be as in the definition of $\mc{Z}_{3}$, and let $P = \{h_1,\ldots,h_5\}$.
  Recall that $J''$ is a line trigraph of $H$, where the vertex corresponding to the edge $h_2h_3$ and the vertex corresponding to the edge $h_3h_4$ are made strongly antiadjacent or semiadjacent.
  By the assumptions of the lemma, we know $H$ and $P$.
  \begin{figure}
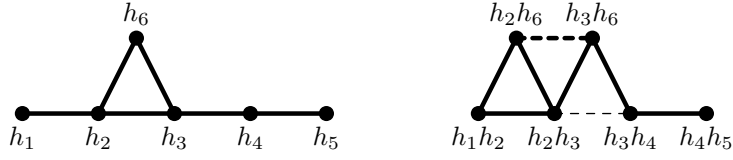

    \begin{center}
      \ig{z3-semi-ex1.mps}
    \end{center}
    \caption{On the left: A graph $H$ that follows the constraints of the definition of $\mc{Z}_{3}$ with one vertex adjacent to~$h_2$ and $h_{3}$. On the right: A line trigraph of $H$. In particular, this is a member of $\mc{Z}_{3}$. The thick dotted edge is an edge or a semi-edge. The thin dotted edge is a semi-edge or a non-edge.}
  \label{fig:z3:example}
  \end{figure}
  \begin{cclaim}
  \label{c:z3:degree}
    Without loss of generality, no two vertices of $P$ are adjacent to more than one vertex of $V(H) \setminus P$ of degree~$2$.
  \end{cclaim}
  \begin{cproof}
    Let $T$ denote the set of vertices in $V(H) \setminus P$ that have degree two and that are adjacent to~$h_2$ and $h_3$ (see Fig.~\ref{fig:z3:example}).
    Suppose that $|T| \geq 2$.
    Let $E$ denote the set of edges in $H$ incident on $h_2$ and a vertex of $T$, and let $F$ denote the set of edges in $H$ incident on $h_3$ and a vertex of $T$.
    Note that $|E|=|F|=|T| \geq 2$.
    Moreover, since $h_2$ and $h_3$ both have degree at least~$3$, no semi-edge of $(J'',Z'')$ is incident on exactly one vertex corresponding to an edge of $E$ or $F$.

    We provide a modified member $(J''',Z''')$ of $\mc{Z}_{3}$ and a modified thickening $\mc{W}'''$ to $(J,Z)$.
    Initially, $(J''',Z''')$ is equal to $(J'',Z'')$ and $\mc{W}''' = \mc{W}$, and in particular, $H'''$ and $H$ are the same.
    Now add a new vertex $v$ to $H'''$ and make it adjacent to $h_2$ and $h_3$. Let $e$ and $f$ denote the corresponding edges. 
    Make $e$ and $f$ semiadjacent in $J'''$.
    Now set $W'''_e = \bigcup_{a \in E} W_a$ and $W'''_f = \bigcup_{a \in F} W_{a}$.
    Remove $E$, $F$, and $T$ from $H$ and blank $W'''_a$ for each $a \in E \cup F$.
    Since each $a \in E$ is adjacent in $J''$ to exactly one $b \in E$ and vice versa, it follows from the assumption that $|E|=|F| \geq 2$ that $W'''_e$ is neither complete nor anticomplete to $W'''_f$.
    By construction, $(J''',Z''')$ is a member of $\mc{Z}_{3}$ based on $H'''$, and $(J,Z)$ is the thickening $\mc{W}'''$ of $(J''',Z''')$.

    Similar arguments can be made with respect to $h_2$ and $h_4$, and $h_3$ and $h_4$.
  \end{cproof}
  Note that the above modification could possibly destroy the property that $(J,Z)$ is reduced.
  However, then we simply apply Lemma~\ref{lem:kernel:zis2-base}.
  Therefore, we may assume that $(J,Z)$ is reduced and that no two vertices of $P$ are adjacent to more than one vertex of $V(H) \setminus P$ of degree~$2$.

  The fact that $(J,Z)$ is reduced now implies that vertices of $V(H) \setminus P$ of degree~$2$ are responsible for at most $3 \cdot 3 = 9$ vertices of $V(J)$.
  Also observe that if $v,v' \in V(H) \setminus P$ are of degree~$1$, then any two vertices from $W_{e} \cup W_{e'}$ would imply the existence of twins in $J$ (where $e,e'$ are the edges of $H$ incident on $v,v'$ respectively), which contradicts that $(J,Z)$ is semi-thickened.
  Hence, these vertices of $V(H) \setminus P$ are responsible for at most $3$ vertices of $V(J)$.
  Also observe that $h_2h_3$ and $h_3h_4$ are antiadjacent in $J''$, and therefore, the fact that $(J,Z)$ is reduced implies that $|W_{h_2h_3}| + |W_{h_3h_4}| \leq 3$.

  Now observe that any vertex from $W_{h_2h_3}$ and any vertex from $W_{h_3h_4}$ together dominate $J$.
  Hence, if $T$ is the set of vertices of $V(H) \setminus P$ of degree~$3$ and there exists a $T' \subseteq T$ of size~$2$, then $W_{e}$ can be removed for any $e$ incident on a vertex of $T \setminus T'$.
  Indeed, if for some disjoint $Q,R \subseteq Z$ a dominating set $D$ that attains $\dsg{J\setminus(Q \cup R)}{N[R]}$ contains a vertex $v$ from $W_{e}$ where $e$ is incident on a vertex of $T \setminus T'$, then there is a vertex $v'$ from $W_{e'}$ that is not dominated by $v$ where $e'$ is an edge incident on a vertex of $T'$.
  Hence, $|D| \geq 2$, and any vertex from $W_{h_2h_3}$ and any vertex from $W_{h_3h_4}$ combined would form an equally small dominating set.
  After applying this reduction, vertices of $V(H) \setminus P$ of degree~$3$ are responsible for at most $2 \cdot 3 = 6$ vertices of $V(J)$.

  Let $J'$ denote stripe that remains after applying the above reduction operations, and let $Z'=Z$.
  From the above reductions, $|V(J')| \leq 23$.
  Moreover, it is straightforward to see that $(J',Z')$ is still a thickening of a member of $\mc{Z}_{2}$.
  Finally, $\dsg{J\setminus(Q \cup R)}{N[R]} = \dsg{J'\setminus(Q \cup R)}{N[R]}$ for any disjoint $Q,R \subseteq Z$ by construction.
\end{proof}

\begin{lemma}
\label{lem:kernel:z4}
  Let $(J,Z)$ be a reduced stripe such that $(J,Z)$ is a thickening $\mc{W}$ of a member $(J'',Z'')$ of $\mc{Z}_{4}$ and $J$ is a graph.
  Then we can compute in constant time a claw-free stripe $(J',Z')$ such that $Z = Z'$, $|V(J')| \leq 11$, and $\dsg{J\setminus(Q \cup R)}{N[R]} = \dsg{J'\setminus(Q \cup R)}{N[R]}$ for any disjoint $Q,R \subseteq Z$.
\end{lemma}
\begin{proof}
  It suffices to observe that $(J',Z')$ has $9$ vertices and two semi-edges.
  Since $(J,Z)$ is reduced, it follows that $|V(J)| \leq 11$.
  Hence, we simply output $J' = J$ and $Z'=Z$.
\end{proof}

\begin{lemma}
\label{lem:kernel:z5}
  Let $(J,Z)$ be a reduced stripe such that $(J,Z)$ is a thickening $\mc{W}$ of a member $(J'',Z'')$ of $\mc{Z}_{5}$ and $J$ is a graph.
  Then we can compute in constant time a claw-free stripe $(J',Z')$ such that $Z = Z'$, $|V(J')| \leq 14$, and $\dsg{J\setminus(Q \cup R)}{N[R]} = \dsg{J'\setminus(Q \cup R)}{N[R]}$ for any disjoint $Q,R \subseteq Z$.
\end{lemma}
\begin{proof}
  It suffices to observe that $(J',Z')$ has at most $13$ vertices and one semi-edge. Since $(J,Z)$ is reduced, it follows that $|V(J)| \leq 14$.
  Hence, we simply output $J' = J$ and $Z'=Z$.
\end{proof}

\subsection{Reducing the Number of Strips}
\label{sec:kernel:reduce}
Throughout this subsection, let $G$ be a connected claw-free graph that has a purified strip-structure $(H,\eta)$ of nullity zero such that $1 \leq |\overline{F}| \leq 2$ for each $F \in E(H)$. Let $k$ be an integer.

Observe that the strip-graph $H$ can be seen as a multigraph (with self-loops) with vertex set $V(H)$ and edge set $E(H)$.
Indeed, each $F \in E(H)$ with $|\overline{F}| = 2$ can be seen as a normal edge, and each $F \in E(H)$ with $|\overline{F}| = 1$ can be seen as a self-loop.

\begin{definition}
  Let $\HT$ denote the multigraph obtained from $H$ (seen as a multigraph) by adding a self-loop to each $h \in \mc{P}(H)$.
  We call these self-loops \emph{force edges}.
\end{definition}
Recall Definition~\ref{def:ds:striped} for the definition of $\mc{P}(H)$.

Now recall that $\eds{\HT}$ is the size of a smallest edge dominating set of $\HT$.
The following lemma will be very helpful to reduce the number of strips of $(H,\eta)$.

\begin{lemma}
\label{lem:kernel:eds}
  If $\ds{G} \leq k$, then $\eds{\HT} \leq k$.
\end{lemma}
\begin{proof}
  Let $D$ be a dominating set of $G$ of size at most $k$.
  Let $L$ denote the set of all $F \in E(H) \subseteq E(\HT)$ for which $D \cap \eta(F) \not= \emptyset$.
  By construction, $|L| \leq |D|$.
  We claim that $L$ is an edge dominating set of $\HT$.
  For sake of contradiction, let $F \in E(\HT)$ not share an endpoint with an edge of $L$.
  We consider two cases, and obtain a contradiction in both, proving the claim and thus the lemma.

  Suppose that $F$ is a force edge, incident on $h \in V(H)$.
  By definition, there exists an $F' \in E(H)$ such that $h \in \overline{F'}$ and $F'$ corresponds to a stripe.
  By assumption, $\eta(h) \cap D = \emptyset$ and $\eta(F') \cap D = \emptyset$.
  However, since $F'$ corresponds to a stripe, it follows by the definition of stripes and strip-structures that every vertex of $G$ that is adjacent to a vertex of $\eta(F',h)$ is in $\eta(h) \cup \eta(F')$.
  Hence, every vertex of $\eta(F',h)$ is not dominated by $D$, a contradiction.

  Suppose that $F \in E(H)$.
  By assumption, $\eta(F) \cap D = \emptyset$ and $\eta(h) \cap D = \emptyset$ for each $h \in \overline{F}$.
  However, it follows by the definition of strip-structures that every vertex of $G$ that is adjacent to a vertex of $\eta(F)$ is in $\eta(F) \cup \bigcup_{h \in \overline{F}} \eta(h)$.
  Hence, the vertices of $\eta(F)$ are not dominated by $D$, a contradiction.
\end{proof}
The preceding lemma enables the transfer of ideas from polynomial kernels for \textsc{Edge Dominating Set}~\cite{Fernau2006,XiaoKP2013}.
Note that the presence of stripes makes a direct application of these kernels infeasible.
Although it is possible to use the basic ideas from any one of these kernels\footnote{In the conference version of this paper, we relied on the kernel of Fernau~\cite{Fernau2006}. Using the ideas from the kernel by Xiao~\etal\cite{XiaoKP2013}, we end up with a slightly smaller bound on the size of the kernel, as well as a slightly easier description.}, we will in fact use a simplification of the ideas in the kernel by Xiao~\etal\cite{XiaoKP2013}.
For this, we need the following straightforward corollary of Lemma~\ref{lem:kernel:eds}.
Recall that $\mm{\HT}$ denotes the size of a largest matching of~$\HT$.

\begin{corollary}
\label{cor:kernel:mm}
  If $\ds{G} \leq k$, then $\mm{\HT} \leq 2k$.
\end{corollary}
\begin{proof}
  Let $L$ denote an edge dominating set of $\HT$ of size at most $k$; $L$ exists by Lemma~\ref{lem:kernel:eds}.
  Let~$M$ be a largest matching of $\HT$. If $F \in M$, then at least one of the endpoints must be shared with an $F' \in L$. 
  Hence, $|M| \leq 2|L| \leq 2k$.
\end{proof}
Given a set $M \subseteq E(\HT)$, let $\overline{M}$ denote the set of endpoints of the edges of $M$. Note that $\overline{M} \subseteq V(H)$.

\begin{definition} \label{def:kernel:ve}
  Let $M = M(\HT)$ be any largest matching of $\HT$.
  Let $\mc{V}(H)$ be the union of $\overline{M}$ and any inclusion-wise minimal set $B$ of vertices of $\HT$ such that each $v \in \overline{M}$ has at least $\min\{2k+1,|N_{\HT}(v)|\}$ neighbors in $\HT$ in the set $\overline{M} \cup B$.
  Let $\mc{E}(H)$ denote the set of $F \in E(\HT)$ such that $\overline{F} \subseteq \mc{V}(\HT)$.
  The strip-graph induced by $\mc{V}(H)$ and $\mc{E}(H)$ is denoted $\mc{H}$. The restriction of $\eta$ to $\mc{E}(H)$ is denoted $\zeta$.
  Let $\mc{G}$ be the graph obtained from $G$ by removing each vertex of $\eta(F)$ for each $F \in E(H) \setminus \mc{E}(H)$.
\end{definition}
We make the following observation.

\begin{proposition}
\label{prp:kernel:ht-obs}
  Let $M = M(\HT)$. For any $F \in E(H)$, if $|\overline{F}|=2$, then $\overline{F} \cap \overline{M} \not= \emptyset$.
  If $F$ corresponds to a stripe, then $F \in \mc{E}(H)$.
  Finally, $(\mc{H},\zeta)$ is a purified strip-structure of nullity zero for $\mc{G}$ with $\eta(F) = \zeta(F)$ for any $F \in \mc{E}(H)$.
\end{proposition}
\begin{proof}
  For the first part, since $M$ is a largest matching of $\HT$, there is no $F \in E(H)$ with $\overline{F} \cap \overline{M} = \emptyset$ and $|\overline{F}| = 2$.

  For the second part, note that since $M$ is a largest matching of $\HT$, the force edges ensure that $\mc{P}(H) \subseteq \overline{M}$.
  Hence, if $F \in E(H)$ corresponds to a stripe, then $\overline{F} \subseteq \mc{P}(H) \subseteq \overline{M}$, and thus $F \in \mc{E}(H)$.

  For the third part, it suffices to observe that for any $h \in \mc{V}(H)$, there are distinct $F,F' \in \mc{E}(H)$ with $h \in \overline{F} \cap \overline{F'}$.
  This is indeed true, because $F'' \in E(H) \setminus \mc{E}(H)$ if and only if there is a $h \in \overline{F''}$ such that $h \not\in \mc{V}(H)$.
  Then we note that for any $h \in V(H)$, there are distinct $F,F' \in E(H)$ with $h \in \overline{F} \cap \overline{F'}$ by the definition of a strip-structure.
\end{proof}

We now prove two interesting consequences of Definition~\ref{def:kernel:ve}.

\begin{lemma}
  $\ds{G} \leq k$ if and only if $\ds{\mc{G}} \leq k$.
\end{lemma}
\begin{proof}
  Let $M = M(\HT)$.
  Suppose that $\ds{G} \leq k$.
  Let $D$ denote a dominating set of $G$ of size at most $k$.
  We construct a dominating set $D'$ of $\mc{G}$ of size at most $|D|$.
  Consider each $v \in D$ in turn.
  Suppose that $v \in \eta(F)$ for some $F \in \mc{E}(H)$.
  Then add $v$ to $D'$.
  Suppose that $v \in \eta(F)$ for some $F \in E(H) \setminus \mc{E}(H)$. 
  If $\overline{F} \cap \mc{V}(H) = \emptyset$, then $|\overline{F}| = 1$; otherwise, by Proposition~\ref{prp:kernel:ht-obs}, $\overline{F} \cap \overline{M} \not= \emptyset$, which using the fact that $\overline{M} \subseteq \mc{V}(H)$ would contradict the assumption that $\overline{F} \cap \mc{V}(H) = \emptyset$.
  However, if $|\overline{F}|=1$, then $F$ corresponds to a stripe, and by Proposition~\ref{prp:kernel:ht-obs}, $F \in \mc{E}(H)$, a contradiction. 
  It follows that $\overline{F} \cap \mc{V}(H) \not= \emptyset$.
  Since $F$ would be in $\mc{E}(H)$ if $\overline{F} \subseteq \mc{V}(H)$, $|\overline{F}| = 2$.
  Hence, by Proposition~\ref{prp:kernel:ht-obs}, $\overline{F} \cap \overline{M} \not= \emptyset$.
  Let $\overline{F} = \{h,h'\}$, and assume without loss of generality that $h \in \overline{M}$.
  Add $v'$ to $D'$, where $v'$ is any vertex of $\eta(F',h)$ and $F'$ satisfies $F' \in \mc{E}(H)$ and $h \in \overline{F'}$.
  Indeed, $v$ only dominates vertices in $\eta(h) \cup \eta(h') \cup \eta(F)$.
  However, $h' \not\in \mc{V}(H)$, and thus, by construction, $F'' \not\in \mc{E}(H)$ for each $F'' \in E(H)$ for which $h' \in \overline{F''}$.
  Hence, $(\eta(h') \cup \eta(F)) \cap V(\mc{G}) = \emptyset$.
  Since $\eta(h)$ is a clique, $v'$ dominates $\eta(h)$, and thus the resulting set $D'$ is indeed a dominating set of $\mc{G}$. Therefore, $\ds{\mc{G}} \leq k$.

  Suppose that $\ds{\mc{G}} \leq k$.
  Let $D$ denote a dominating set of $\mc{G}$ of size at most $k$.
  We claim that $D$ is also a dominating set of $G$. To this end, it suffices to prove that any $w \in \eta(F)$ for any $F \in E(H) \setminus \mc{E}(H)$ has a neighbor in $D$.
  First, observe that the force edges ensure that $\mc{P}(H) \subseteq \overline{M}$.
  Hence, $F$ does not correspond to a stripe; otherwise, since $\overline{F} \subseteq \mc{P}(H) \subseteq \overline{M} \subseteq \mc{V}(H)$, $F \in \mc{E}(H)$, a contradiction.
  It follows that $F$ corresponds to a spot. Second, observe that by Proposition~\ref{prp:kernel:ht-obs}, $\overline{M} \cap \overline{F} \not= \emptyset$. Let $h \in \overline{M} \cap \overline{F}$. 
  By assumption, there is an $h' \in \overline{F} \setminus \mc{V}(H)$.
  Since $h' \not \in \mc{V}(H)$, it follows from the construction of $\mc{V}(H)$ that at least $2k+1$ neighbors of $h$ in $\HT$ are in $\mc{V}(H)$.
  Let $N$ denote this set of neighbors of $h$, and let $\mc{F}$ denote the set of $F \in \mc{E}(H)$ with $\overline{F} \subseteq N \cup \{h\}$.
  Call $h' \in N$ marked if $D \cap \eta(h') \not= \emptyset$ or if $\eta(F') \cap D \not=\emptyset$ for some $F' \in \mc{F}$ with $h' \in \overline{F'}$.
  Since $|\overline{F'}| \leq 2$ for each $F' \in E(H)$, at most $2|D| \leq 2k$ neighbors in $N$ can be marked, and in particular, there is a $h' \in N$ that is unmarked.
  Consider any $F' \in \mc{F}$ with $h' \in \overline{F'}$; note that $F'$ is indeed well-defined.
  Since $h'$ is unmarked, $D \cap \eta(h') \not= \emptyset$ and $\eta(F') \cap D \not=\emptyset$.
  Hence, $h \in \overline{F'}$ and the vertices of $\eta(F')$ are dominated by a vertex of $\eta(h)$; that is, $\eta(h) \cap D \not= \emptyset$.
  Since $F$ corresponds to a spot, $\eta(F) = \eta(F,h) \subseteq \eta(h)$. As $\eta(h)$ is a clique, $\eta(F)$ is dominated by $D$.
  Therefore, $D$ is a dominating set of $G$.
\end{proof}

\begin{lemma}
\label{lem:kernel:bounds}
  Let $M = M(\HT)$.
  If $\ds{G} \leq k$, then $|\mc{V}(H)| \leq |\mc{V}(H)| \leq (2k+2) |\overline{M}| \leq 8k(k+1)$, $|\{\overline{F} \mid F \in \mc{E}(H), h \in \overline{F}\}| \leq |\mc{V}(H)| \leq 8k(k+1)$ for any $h \in \overline{M}$, and $|\{\overline{F} \mid F \in \mc{E}(H)\}| \leq (|\overline{M}|+1)\, |\mc{V}(H)| \leq 8k(4k+1)(k+1)$.
\end{lemma}
\begin{proof}
  Since $\ds{G} \leq k$, it follows from Corollary~\ref{cor:kernel:mm} that $|M| \leq 2k$, and thus $|\overline{M}| \leq 4k$.
  For each $v \in \overline{M}$, mark $\min\{2k+1,|N_{\HT}(v)|\}$ neighbors of $v$ in $\HT$ that are in $\mc{V}(H)$.
  Let $B$ denote the set used in the construction of $\mc{V}(H)$.
  Then at most $(2k+1)\, |\overline{M}|$ vertices of $B$ are marked, and since $B$ is inclusion-wise minimal, $|B| \leq (2k+1)\, |\overline{M}|$.
  Hence, $|\mc{V}(H)| \leq (2k+2) |\overline{M}| \leq 8k(k+1)$.

  The bound on $|\{\overline{F} \mid F \in \mc{E}(H), h \in \overline{F}\}|$ for each $h \in \overline{M}$ is immediate from the definition of $\mc{E}(H)$.

  By Proposition~\ref{prp:kernel:ht-obs}, $\overline{F} \cap \overline{M} \not= \emptyset$ for every $F \in \mc{E}(H)$ with $|\overline{F}| = 2$.
  The bound on $|\{\overline{F} \mid F \in \mc{E}(H)\}|$ is immediate from the bound on $|\{\overline{F} \mid F \in \mc{E}(H), h \in \overline{F}\}|$ for each $h \in \overline{M}$.
\end{proof}
Note that this lemma does not yet bound $|\mc{E}(H)|$; this requires a more careful argument that we make in the next subsection when we present the final kernel.

\subsection{The Kernel}
\label{sec:kernel:kernel}
We are now ready to present the kernel.

\begin{theorem}
  Let $G$ be a connected claw-free graph with $n$ vertices and $m$ edges, and let $k$ be an integer.
  Then there is an algorithm that runs in $O(n^5)$ time and that returns a graph $G'$ with $O(k^{3})$ vertices and an integer $k'$ such that $\ds{G} \leq k$ if and only if $\ds{G'} \leq k'$.
\end{theorem}
\begin{proof}
  We first preprocess the graph as in Theorem~\ref{thm:ds:fpt}.
  In $O(n^{5})$ time, this returns a graph that has no twins and no proper W-joins, which we call $G$ as well by abusing notation.

  We test whether $\ds{G} \leq 3$ using the algorithm of Lemma~\ref{lem:ds:constant} in $O(n^{4})$ time.
  If so, then the algorithm actually determines $\ds{G}$; hence, we return a graph~$G'$ that consists of a single vertex and an integer $k'$ that is $1$ if $\ds{G} \leq k$ and that is $0$ otherwise.
  This is clearly correct.
  Now we may assume that $\ds{G} > 3$. If $\alpha(G) \leq 3$, then $\ds{G} \leq 3$ by Proposition~\ref{prp:ds:is-size}, a contradiction. Hence, $\alpha(G) > 3$.

  We then apply the algorithm of Theorem~\ref{thm:main-base-kernel} in $\stripTime$ time.
  If it outputs that $G$ is a proper circular-arc graph, then we can compute $\ds{G}$ in linear time by Theorem~\ref{thm:ds:circular}.
  We return a graph~$G'$ that consists of a single vertex and an integer $k'$ that is $1$ if $\ds{G} \leq k$ and that is $0$ otherwise.
  If it outputs that $G$ is a thickening of an XX-trigraph, then we can compute $\ds{G}$ in $O(n^{4})$ time by Lemma~\ref{lem:ds:s2}.
  Again, we return a graph $G'$ that consists of a single vertex and an integer $k'$ that is $1$ if $\ds{G} \leq k$ and that is $0$ otherwise.
  It remains that it outputs that $G$ admits a strip-structure $(H,\eta)$ such that for each strip $(J,Z)$ satisfies the condition that:
  \begin{itemize}
    \item $(J,Z)$ is a trivial line graph strip, or
    \item $(J,Z)$ is a stripe for which $J$ is connected and
      \begin{itemize}
        \item $|Z|=1$, $\alpha(J) \leq 3$, and $V(J) \setminus N[Z] \not= \emptyset$,
        \item $|Z|=1$, $J$ is a proper circular-arc graph, and either $J$ is a strong clique or $\alpha(J) > 3$,
        \item $|Z|=2$ and $J$ is a proper interval graph,
        \item $|Z|=1$, $\alpha(J) = 4$, and $V(J) \setminus N[Z] \not= \emptyset$, or
        \item $|Z|=2$ and $(J,Z)$ is a thickening $\mc{W}$ of a member $(J',Z')$ of $\mc{Z}_{2} \cup \mc{Z}_{3} \cup \mc{Z}_{4} \cup \mathcal{Z}_{5}$. Moreover, we know $\mc{W}$, $(J',Z')$, and the class that $(J',Z')$ belongs to.
      \end{itemize}
  \end{itemize}
  We now determine the set $\mc{P}(H)$ and the multigraph (with self-loops) $\HT$, which takes linear time.
  Then we can determine $M(\HT)$, $\mc{V}(H)$, $\mc{E}(H)$, $\mc{G}$, $\mc{H}$, and $\zeta$.
  By Proposition~\ref{prp:kernel:ht-obs}, $(\mc{H},\zeta)$ is a purified strip-structure of nullity zero for $\mc{G}$.
  Moreover, since $\mc{V}(H) \subseteq V(H)$, $\mc{E}(H) \subseteq E(H)$, and $\zeta$ is the restriction of $\eta$ to $\mc{E}(H)$, it is immediate from Proposition~\ref{prp:kernel:ht-obs} that each strip of $(\mc{H},\zeta)$ still satisfies the above condition.

  We provide a slight modification of $\mc{H}$.
  For any $F,F' \in \mc{E}(H)$ with $\overline{F} = \overline{F'}$ and $|\overline{F}| = |\overline{F'}| = 2$ such that $\zeta(F)$ and $\zeta(F')$ are a disjoint union of two cliques, remove $F'$ from $\mc{E}(H)$ and add $\zeta(F')$ to $\zeta(F)$.
  This replacement is indeed correct, because both $F$ and $F'$ correspond to a thickening of a stripe in $\mc{Z}_{3}$; this follows immediately from the definition of $\mc{Z}_{3}$.
  In fact, both stripes are a thickening of a four-vertex path where the middle edge is a semi-edge (see also the proof of Lemma~\ref{lem:recog:z3}).
  Now let $\overline{F} = \{h,h'\}$.
  Observe that $\zeta(F,h) \subseteq \zeta(h)$ and $\zeta(F',h) \subseteq \zeta(h)$, and thus $\zeta(F,h)$ is complete to $\zeta(F',h)$, because $\zeta(h)$ is a clique.
  The same observation holds mutatis mutandis with respect to $h'$.
  Hence, $\zeta(F) \cup \zeta(F')$ is a disjoint union of two cliques, and therefore, the stripe resulting from the previous replacement is a thickening of a member of $\mc{Z}_{3}$.
  Moreover, we know this thickening and the corresponding member of $\mc{Z}_{3}$.
  We perform the above replacement operation iteratively.
  Since it takes linear time to determine which stripes are a disjoint union of two cliques, the replacement operation can be performed in linear time in total.
  By abuse of notation, we still use the notation $\mc{E}(H)$ and $(\mc{H},\zeta)$ for the resulting set and strip-structure; note that the graph $\mc{G}$ has remained unchanged.

  To obtain the kernel, we distinguish between stripes $(J,Z)$ of $(\mc{H}, \zeta)$ with $|Z| = 1$ and those with $|Z|=2$.
  Let $k'' = 0$.
  If $|Z| = 1$, then we only consider those stripes for which $V(J) \setminus N[Z] \not= \emptyset$.
  Then either $\alpha(J) \leq 4$, or $J$ is a proper circular-arc graph and either $J$ is a strong clique or $\alpha(J) > 3$. Hence, using a straightforward extension of Lemma~\ref{lem:ds:stripe-is} or Lemma~\ref{lem:ds:circular-stripe} respectively, we can compute $\dsg{J\setminus(Q \cup R)}{N[R]}$ in $O(n^{5})$ time for any disjoint $Q,R \subseteq Z$.
  Then we apply the reduction of Lemma~\ref{lem:kernel:zis1} to $J$, which takes linear time and reduces the size of the stripe to have at most four vertices.
  Add the value $k'$ computed by the algorithm of Lemma~\ref{lem:kernel:zis1} to $k''$.
  The fact that this reduction is safe follows from Lemma~\ref{lem:kernel:equiv}.

  If $|Z| = 2$, then either $J$ is a proper interval graph or $(J,Z)$ is a thickening $\mc{W}$ of a member $(J',Z')$ of $\mc{Z}_{2} \cup \mc{Z}_{3} \cup \mc{Z}_{4} \cup \mathcal{Z}_{5}$.
  Suppose that $J$ is a proper interval graph.
  Then we compute $\ds{J}$ in linear time using Theorem~\ref{thm:ds:circular}.
  If $\ds{J} > k$, then we know that $\ds{G} > k$; hence, we return a graph $G'$ that consists of a single vertex and an integer $k'=0$.
  Otherwise, we apply the reduction of Lemma~\ref{lem:kernel:interval} to $J$, which takes linear time and reduces the size of the stripe to have at most $18k+2$ vertices.
  The fact that this reduction is safe follows from Lemma~\ref{lem:kernel:equiv}.

  Suppose that $(J,Z)$ is a thickening $\mc{W}$ of a member $(J',Z')$ of $\mc{Z}_{2} \cup \mc{Z}_{3} \cup \mc{Z}_{4} \cup \mathcal{Z}_{5}$.
  By assumption, we know $\mc{W}$, $(J',Z')$, and the class that $(J',Z')$ belongs to.
  Lemma~\ref{lem:kernel:semi} and the fact that $G$ does not admit twins imply that $(J,Z)$ is semi-thickened.
  We then apply Lemma~\ref{lem:kernel:zis2-base} to ensure that $(J,Z)$ is reduced.
  Note that the algorithm of Lemma~\ref{lem:kernel:zis2-base} runs in linear time and preserves that $(J,Z)$ is semi-thickened. 
  Depending on the class that $(J',Z')$ belongs to, we then apply Lemma~\ref{lem:kernel:z2}, \ref{lem:kernel:z3}, \ref{lem:kernel:z4}, or~\ref{lem:kernel:z5}.
  This takes linear time and reduces the size of the stripe to have at most $26$ vertices.
  The fact that these reductions are safe follows from Lemma~\ref{lem:kernel:equiv}.

  Let $G'$ be the resulting graph.
  By Lemma~\ref{lem:kernel:equiv}, we know that $\ds{G} \leq k$ if and only if $\ds{G'} \leq k - k''$, so let $k' = k-k''$.
  It thus suffices to bound the number of vertices in $G'$. To this end, we distinguish several types of strips $(J,Z)$ of $(\mc{H},\zeta)$:

  \ccase{1} $(J,Z)$ is a stripe with $|Z|=1$ and $V(J) \setminus N[Z] = \emptyset$.\\
  Suppose that for some $h \in \mc{V}(H)$ there are $F,F' \in \mc{E}(H)$ with $\overline{F} = \overline{F'} = \{h\}$ such that $F,F'$ correspond to a stripe of this type.
  Then any pair of vertices in $\zeta(F) \cup \zeta(F')$ forms twins in $G$, a contradiction.
  Hence, each such stripe has size at most~$2$ and there at most $|\mc{V}(H)| = O(k^{2})$ such stripes in total by Lemma~\ref{lem:kernel:bounds}.
  Hence, $O(k^{2})$ vertices exist in such stripes.

   \ccase{2} $(J,Z)$ is a stripe with $|Z| = 2$ and $V(J) \setminus N[Z] = \emptyset$.\\
  Observe that any such stripe is a disjoint union of two cliques.
  By the above modification of $(\mc{H},\zeta)$, there is at most one $F \in \mc{E}(H)$ that corresponds to such a stripe with $\overline{F} = \{h,h'\}$ for each $h,h' \in V(\mc{H})$.
  By Lemma~\ref{lem:kernel:bounds}, there are $O(k^{3})$ such stripes in total.
  As observed previously, any such stripe is a thickening of a member of $\mc{Z}_{3}$, and therefore, these stripes are reduced to have size at most $23$ by Lemma~\ref{lem:kernel:z3}.
  Hence, $O(k^{3})$ vertices remain in such stripes.

  \ccase{3} $(J,Z)$ is a stripe with $V(J) \setminus N[Z] \not= \emptyset$.\\
  Let $D$ be a dominating set of $G$.
  Then $D$ contains at least one vertex of $V(J) \setminus Z$ by the definition of a strip-structure.
  Hence, if $\ds{G} \leq k$, then there can be at most $k$ such stripes.
  Moreover, after the above reduction, each such stripe has only $O(k)$ vertices remaining.
  Hence, $O(k^{2})$ vertices remain in such stripes.

  \ccase{4} $(J,Z)$ is a spot.\\
  Suppose that for some $h,h' \in \mc{V}(H)$ there are $F,F' \in \mc{E}(H)$ with $\overline{F} = \overline{F'} = \{h,h'\}$ such that $F,F'$ correspond to a spot.
  Then the pair of vertices in $\zeta(F) \cup \zeta(F')$ forms twins in $G$, a contradiction. Hence, there $O(k^{3})$ spots by Lemma~\ref{lem:kernel:bounds}.
  Hence, $O(k^{3})$ vertices exist in spots.

  \medskip\noindent
  These four cases are exhaustive, and imply that $G'$ has $O(k^{3})$ vertices.
\end{proof}

\begin{corollary}
  \textsc{Dominating Set} on claw-free graphs has a kernel with $O(k^{3})$ vertices.
\end{corollary}

\section*{Part III -- Hardness Results and Lower Bounds}

\section{Hardness Results and Lower Bounds} \label{sec:hard}
In this section, we show that various generalizations and improvements of the parameterized and kernelization algorithms in this paper are unlikely to exist.
In particular, we show that:
\begin{itemize}
  \item \problemDS{} and \problemCDS{} are \Wh[1]-hard on $K_{1,4}$-free graphs;
  \item \problemWDS{} and \problemWCDS{} are \Wh[2]-hard on co-bipartite graphs, and thus also on claw-free graphs;
  \item \problemCDS{} on line graphs (and thus also on claw-free graphs) has no polynomial kernel, unless the polynomial hierarchy collapses to the third level, contrasting the polynomial kernel we saw for \problemDS{};
  \item \problemDS{} and \problemCDS{} on line graphs (and thus also on claw-free graphs) have no subexponential-time algorithm, unless the Exponential Time Hypothesis fails.
\end{itemize}
We prove each of these results in turn.

Recall that to prove the hardness of a parameterized problem $P$, one can give a \emph{parameterized reduction} from a $\mathsf{W}[\cdot]$-hard problem~$P'$ to $P$ so that every instance $I'$ of $P'$ with parameter $k'$ is mapped to an instance $I$ of $P$ with parameter $k \leq g(k')$ for some computable function $g$.
We demand that $I$ can be computed in $f(k') \cdot |I'|^{O(1)}$ time for some computable function $f$, and that $I$ is a ``yes''-instance if and only if $I'$ is a ``yes''-instance.
If~$f$ and $g$ are polynomials, then the reduction is called a \emph{polynomial parameter transformation}.

\subsection{Hardness on \texorpdfstring{$K_{1,4}$}{K1,4}-free Graphs} \label{sec:a:hardness}
We show that \problemDS{} is \Wh[1]-hard when restricted to $K_{1,4}$-free graphs.
Our hardness result is obtained by the so-called \emph{\problemMCC{} reduction technique}, introduced and explained in~\cite{FellowsHRV2009}.
The problem from which we reduce in this technique is defined as follows:
\defparproblem{\problemMCC}{$k$}
{A graph $H$ and a vertex coloring $c: V(H) \to \{1,2,\dots,k\}$ of $H$.}
{Does $(H,c)$ have a \emph{multicolored clique} (a set $K \subseteq V(H)$ such that $\{u,v\} \in E(H)$ and $c(u) \neq c(v)$ for all distinct $u,v \in K$) of size~$k$?}

\begin{theorem}[\cite{FellowsHRV2009}]
  \problemMCC{} is \Wh[1]-complete.
\end{theorem}

The general idea in the \problemMCC{} reduction technique is to organize gadgets into three categories: vertex selection, edge selection, and validation.
The role of the first two is to encode the selection of $k$ vertices and $\binom{k}{2}$ edges that together form the $k$-multicolored clique in the instance of the \problemMCC{} problem.
The task of the validation gadget is, as its name suggests, to validate the selection of vertices and edges.
In other words, it makes sure that the edges selected are in fact incident to the selected vertices.

We now use this technique to prove that \problemDS{} on $K_{1,4}$-free graphs is \Wh[1]-hard.

\begin{theorem} \label{thm:hardness:ds4}
  \problemDS{} on $K_{1,4}$-free graphs is \Wh[1]-hard.
\end{theorem}
\begin{proof}
  Consider an instance $(H,c,k)$ of \problemMCC{}.
  We construct an instance $(G,k')$ of \problemDS{} as follows: The graph $G$ consists of three big cliques, each one corresponding to a different gadget (\ie vertex selection, edge selection, and validation), and some additional dummy vertices.
  The vertex selection clique is formed by the vertices $A:=\{a_v : v \in V(H)\}$, the edge selection clique is formed by $B:=\{b_{\{u,v\}} : \{u,v\} \in E(H)\}$, and the validation clique by $C:=\{c_{(u,v)},c_{(v,u)} : \{u,v\} \in E(H)\}$.
  Note that there are two ``directed'' validation vertices, $c_{(u,v)}$ and~$c_{(v,u)}$, for each edge $\{u,v\} \in E(H)$.
  The dummy vertices will be denoted by $X:=\{x_i : 1 \leq i \leq k\}$ and $Y:=\{y_{\{i,j\}} : 1 \leq i < j \leq k\}$.
  These are all the vertices of $G$, \ie $V(G):=A \cup B \cup C \cup X \cup Y$.

  We next describe the edges connecting the vertices of $G$. The first three sets of edges connect each pair of vertices in the same clique:
  \begin{itemize}\setlength{\parskip}{-0.1cm}
    \item $E_1:= \big\{\{a_u,a_v\} : a_u,a_v \in A \textrm{ and } u \neq v\big\}$.
    \item $E_2:= \big\{\{b_{\{u,v\}},b_{\{u',v'\}}\} : b_{\{u,v\}},b_{\{u',v'\}} \in B \textrm{ and } \{u,v\} \neq \{u',v'\}\big\}$.
    \item $E_3:= \big\{\{c_{(u,v)},c_{(u',v')}\} : c_{(u,v)},c_{(u',v')} \in C \textrm{ and } (u,v) \neq (u',v')\big\}$.
  \end{itemize}
  The next two sets of edges connect the dummy vertices to vertices in the selection gadgets.
  These edges will ensure that exactly $k$ vertices will be chosen from the vertex selection gadget, one for each color, and exactly $\binom{k}{2}$ vertices will be chosen from the edge selection gadget, one for each pair of colors.

  \begin{itemize}
    \item $E_4:= \big\{\{x_i,a_v\} : x_i \in X, a_v \in A, \textrm{ and }c(v) = i\big\}$.
    \item $E_5:= \big\{\{y_{\{i,j\}},b_{\{u,v\}}\} : y_{\{i,j\}} \in Y,b_{\{u,v\}} \in B, \textrm { and }
    \{c(u),c(v)\} = \{i,j\}\big\}$.
  \end{itemize}
  Finally, we add the two sets of edges which connect vertices in the selection gadgets to vertices in the validation gadget:

  \begin{itemize}
    \item $E_6:= \big\{\{a_v,c_{(v,u)}\} : a_v \in A \textrm{ and } c_{(v,u)} \in C\big\}$.
    \item $E_7:= \big\{\{b_{\{u,v\}},c_{(u',v')}\} : b_{\{u,v\}} \in B,c_{(u',v')} \in C, \{c(u),c(v)\} =
\{c(u'),c(v')\}, \mbox{ and } \{u,v\} \neq \{u',v'\}\big\}$.
  \end{itemize}

  Setting $E(G):= \bigcup_{1 \leq i \leq 7} E_i$, and $k':=k+\binom{k}{2}$, completes the description of our construction. See Fig.~\ref{Figure: Construction} below for a illustrative example.
  Since the size of $G$ is polynomial in the size of $H$ and $k$ and~$k'$ is polynomial in $k$, this reduction is a parameterized reduction.
  Furthermore, observe that~$G$ is indeed $K_{1,4}$-free since the neighborhood of each vertex can be partitioned into at most three cliques.
  Thus, to complete our argument, we show that $(H,c)$ has a multicolored clique of size $k$ if and only if $G$ has dominating of size $k'$.
  \begin{figure}
    \begin{center}
      \ig[scale=1]{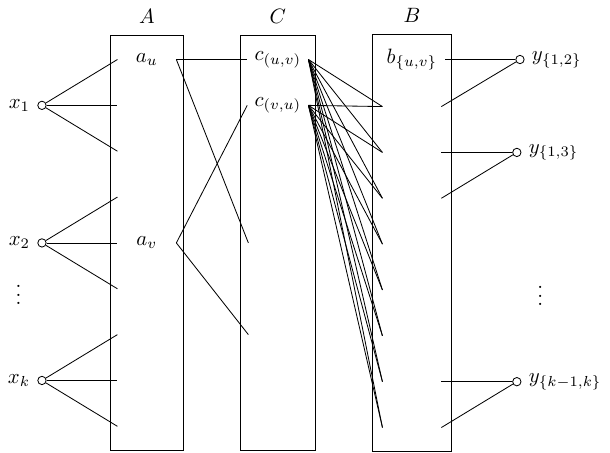}
    \end{center}
    \caption{A graphical example of the construction used in the reduction. Here vertices $u$ and $v$ are adjacent vertices in $H$ from different color classes, \ie $c(u)\neq c(v)$.
    Note that $c_{(u,v)}$ and $c_{(v,u)}$ are not adjacent to $b_{\{u,v\}}$, but are adjacent to all other vertices of $B$.}
  \label{Figure: Construction}
  \end{figure}

  Suppose that $(H,c)$ has a multicolored clique $K$ of size $k$.
  Then for any $i \in \{1,\ldots,k\}$, there is some $v \in K$ with $c(v) =i$, and for any distinct $u,v \in K$, we have $\{u,v\} \in E(H)$.
  We argue that $D:=\{a_v : v \in K\} \cup \{b_{\{u,v\}} : u,v \in K, u \neq v\}$ is a dominating set in $G$.
  Indeed, all vertices of $A$ and $B$ are dominated by some $a_v \in D$ and some $b_{\{u,v\}} \in D$.
  For any $i \in \{1,\ldots,k\}$, the dummy vertex $x_i$ is dominated by $a_v \in D$, where $v \in K$ and $c(v) = i$, and for any distinct $i,j \in \{1,\ldots,k\}$, the vertex $y_{\{i,j\}}$ is dominated by $a_{\{u,v\}} \in D$, where $u,v \in K$ and $\{c(u),c(v)\} = \{i,j\}$.
  Finally, take any vertex $c_{(u,v)} \in C$. If $u \in K$, then $a_u \in D$ dominates $c_{(u,v)}$, and otherwise, the vertex $b_{\{u',v'\}} \in D$ with $\{c(u'),c(v')\}=\{c(u),c(v)\}$  dominates $c_{(u,v)}$.
  Hence, $G$ has a dominating set of size~$k'$.

  Conversely, suppose that $G$ has a dominating set $D$ of size $k'$.
  Then it is not difficult to see that due to the dummy vertices, for each $i \in \{1,\ldots,k\}$, there is exactly one vertex $a_v \in D$ with $c(v) = i$, and for each distinct $i,j \in \{1,\ldots,k\}$, exactly one vertex $b_{\{u,v\}} \in D$ with $\{c(u),c(v)\} = \{i,j\}$.
  Furthermore, there are no other vertices in $D$. Let $K:=\{u : a_u \in D\}$.
  Observe that by the above arguments, $K$ has exactly one vertex for each color $i \in \{1,\ldots,k\}$, and so it remains to argue that~$K$ forms a clique in $H$.
  For this, it is enough to show that for any pair of distinct vertices $u,v \in K$, we must have $b_{\{u,v\}} \in D$.
  Suppose that this is not the case, and let $u$ and $v$ denote two vertices in $K$ with $b_{\{u,v\}} \notin D$.
  Let $b_{\{u',v'\}}$ be the vertex in $D$ such that $\{c(u'),c(v')\} = \{c(u),c(v)\}$.
  Then $\{u',v'\} \neq \{u,v\}$, so, \wloge $u' \neq u$ and $u' \neq v$.
  But then, by construction, the validation vertex $c_{(u',v')}$ is not dominated by $D$.
  Thus, we get that $b_{\{u,v\}} \in D$ for any pair of distinct vertices $u,v \in K$, which implies that $\{u,v\} \in E(H)$ for any pair of distinct vertices $u,v \in K$.
  Hence, $(H,c)$ has a multicolored clique of size $k$.
\end{proof}

\begin{theorem}
  \problemCDS{} on $K_{1,4}$-free graphs is \Wh[1]-hard.
\end{theorem}
\begin{proof}
  We slightly modify the construction of Theorem~\ref{thm:hardness:ds4}.
  Consider again an instance $(H,c,k)$ of \problemMCC{}, and construct the instance $(G,k')$ of \problemDS{} as in Theorem~\ref{thm:hardness:ds4}. Recall that $G$ was $K_{1,4}$-free, as the neighborhood of each vertex in $G$ is a union of at most three cliques.
  Let $A,B$ be as in that construction. Now add two new vertices $v$ and $w$, make $v$ complete to $A$ and $B$, and make $w$ adjacent to $v$.
  Call the resulting graph $G'$.
  Observe that the neighborhood of each vertex in $G'$ is still a union of at most three cliques, and thus $G'$ is $K_{1,4}$-free. 
  Moreover, any (connected) dominating set of $G'$ must contain $v$.
  Then the same arguments as in Theorem~\ref{thm:hardness:ds4} can be used to show that $G'$ has a connected dominating set of size $k'+1$ if and only if $H$ has a multicolored clique of size $k$.
\end{proof}

\subsection{Hardness of the Weighted Case on Claw-Free Graphs}
We define the weighted version of \problemDS{} as follows:

\defparproblem{\problemWDS}{$k'+K'$}{A graph $G'$, a weight function $w' : V(G) \rightarrow \mathbb{N}$, an integer $k'$, and an integer $K'$.}{Does $G'$ have a dominating set $D' \subseteq V(G')$ of size at most $k'$ such that $w(D') = \sum_{d' \in D'} w(d') \leq K'$~?}

We can similarly define \problemWCDS{}, by insisting that $G'[D']$ is connected. We show that both problems are unlikely to be fixed-parameter tractable on co-bipartite graphs, and thus also on claw-free graphs.

\begin{theorem}
  \problemWDS{} and \problemWCDS{} on co-bipartite graphs are \Wh[2]-hard.
\end{theorem}
\begin{proof}
  We give a polynomial parameter transformation from \problemDS{} on general graphs, which is known to be \Wh[2]-hard~\cite{DowneyF1992}.
  Let $(G,k)$ be an instance of \problemDS{}.
  Consider two copies $V_{1}$ and $V_{2}$ of $V(G)$, and define the graph $G'$ with vertex set $V_{1} \cup V_{2}$ and edges such that $V_{1}$ and $V_{2}$ each form a clique, and $v \in V_{1}$ is adjacent to $u \in V_{2}$ if and only if $u \in N_{G}[v]$. 
  Observe that $G'$ is indeed co-bipartite.
  Define a weight function $w'$ such that $w'(v) = 1$ for each $v \in V_{1}$ and $w'(v) = k+1$ for each $v \in V_{2}$.
  Finally, define $K' = k$ and $k' = k$.
  This completes the instance $(G',w',k',K')$ of \problemWDS{}. It suffices to observe that $(G,k)$ is a ``yes''-instance if and only if $(G',w',k',K')$ is.

  Observe that the same reduction also yields the transformation for \problemWCDS{}.
\end{proof}

\subsection{Polynomial Kernel for \problemCDS{} on Claw-Free Graphs Unlikely}
In contrast to \problemDS{}, which has a polynomial kernel on claw-free graphs, \problemCDS{} is unlikely to have a polynomial kernel on claw-free graphs. In fact, we show an even stronger result.

We need to define the following problems:
\defparproblem{\problemCVC}{$k$}{A graph $G$ and an integer $k$.}{Does $G$ have a \emph{connected vertex cover} (a set $C \subseteq V(G)$ such that $G[C]$ is connected, and $u\in C$ or $v \in C$ for each edge $\{u,v\} \in E(G)$) of size at most $k$?}

\defparproblem{\problemCEDS}{$k$}{A graph $G$ and an integer $k$.}{Does $G$ have a \emph{connected edge dominating set} (a set $D \subseteq E(G)$ such that $G[D]$ $=(\{v \mid \{u,v\} \in D\},D)$ is connected and $D$ contains an edge incident on $u$ or $v$ for any edge $\{u,v\} \in E(G)$) of size at most $k$?}

\begin{theorem}
\label{hardness:cds:nopoly}
  \problemCDS{} on line graphs has no polynomial kernel, unless the polynomial hierarchy collapses to the third level.
\end{theorem}
\begin{proof}
  We provide a polynomial parameter transformation from \problemCVC{}, which is known to not have a polynomial kernel on general graphs unless the polynomial hierarchy collapses to the third level~\cite{DomLS2009}.

  As an intermediate problem, consider \problemCEDS{}.
  To see a polynomial parameter transformation from \problemCVC{} to \problemCEDS{}, we claim that $G$ has a connected vertex cover of size at most~$k$ if and only if $G$ has a connected edge dominating set of size at most $k-1$.
  Indeed, if $C$ is a connected vertex cover of $G$, then the edges of any spanning tree of $G[C]$ form a connected edge dominating set of $G$ of size $|C|-1$.
  Conversely, if~$D$ is a connected edge dominating set of $G$, then $V(G[D])$ is a connected vertex cover of $G$ of size at most $|D| +1$.

  The polynomial parameter transformation from \problemCEDS{} to \problemCDS{} on line graphs is straightforward, as $G$ has a connected edge dominating set of size at most $k$ if and only if $L(G)$ has a connected dominating set of size at most $k$.
\end{proof}
Note that the same proof also implies that \problemCDS{} is \NPh-hard on line graphs, even on line graphs of planar graphs of maximum degree four, as \problemCVC{} is \NPh-hard already on planar graphs of maximum degree four~\cite{GareyJ1977}.
We remark here that \NPh-hardness of \problemCDS{} on line graphs was independently proven by Munaro~\cite[Lemma~45]{Munaro2017} and that \problemCVC{} is also \NPh-hard on line graphs~\cite[Lemma~36]{Munaro2017}.

\subsection{Subexponential-Time Algorithms on Claw-Free Graphs Unlikely}
We now show that it will be hard to substantially improve our parameterized algorithms on claw-free graphs under the so-called \emph{Exponential Time Hypothesis}~\cite{ImpagliazzoPZ2001,FominK2010}.
We require it in the following formulation: no algorithm can decide instances of \problemSAT{} in $2^{o(m)}$ time, where $m$ is the number of clauses of the instance.

\begin{theorem}
Unless the Exponential Time Hypothesis fails, no algorithm can decide \problemEDS{} instances $(G,k)$ in $2^{o(k)}\poly(|G|)$ time, even if $G$ is bipartite.
\end{theorem}
\begin{proof}
It suffices to observe that the \NPh-hardness reduction for \problemEDS{} on bipartite graphs by Yannakakis and Gavril~\cite[Theorem~2]{YannakakisG1980} implies a transformation from an instance of \problemSAT{} with $m$ clauses to an instance of \problemEDS{} where $G$ is bipartite and $k = O(m)$.
\end{proof}

\begin{corollary}
Unless the Exponential Time Hypothesis fails, no algorithm can decide \problemDS{} instances $(G,k)$ in $2^{o(k)}\poly(|G|)$ time, even if $G$ is the line graph of a bipartite graph.
\end{corollary}

\begin{theorem}
Unless the Exponential Time Hypothesis fails, no algorithm can decide \problemCEDS{} instances $(G,k)$ in $2^{o(k)}\poly(|G|)$ time.
\end{theorem}
\begin{proof}
  We observe that the \NPh-hardness reduction for \problemVC{} by Garey and Johnson~\cite[Theorem~3.3]{GareyJ1979} implies a transformation from an instance of \problemSAT{} with $m$ clauses to an instance $(G,k)$ of \problemVC{} where $k = O(m)$.
  Let $G'$ be the graph obtained from $G$ by adding a new vertex $v$ that is adjacent to all vertices of $G$, and then adding a new vertex that is only adjacent to $v$.
  Then $(G,k)$ is a ``yes''-instance of \problemVC{} if and only if $(G',k+1)$ is a ``yes''-instance of \problemCVC{}.
  It remains to observe that the polynomial parameter transformation from \problemCVC{} to \problemCEDS{} of Theorem~\ref{hardness:cds:nopoly} is in fact linear.
  Hence, we obtain a transformation from an instance of \problemSAT{} with $m$ clauses to an instance of \problemCEDS{} with $k = O(m)$.
\end{proof}

\begin{corollary}
Unless the Exponential Time Hypothesis fails, no algorithm can decide \problemCDS{} instances $(G,k)$ in $2^{o(k)}\poly(|G|)$ time, even if $G$ is a line graph.
\end{corollary}

\end{document}